\documentclass[final,conference,10pt,twocolumn,cm]{IEEEtran}

\usepackage{etex}
\newif\iffull%
\fulltrue%
\newif\ifwithappendix\withappendixtrue

\usepackage[english]{babel}
\let\oldsfdefault\sfdefault
\usepackage{times}
\renewcommand{\sfdefault}{\oldsfdefault}

\usepackage[usenames,dvipsnames]{color}

\usepackage{amsmath}
\usepackage{amsfonts,mathrsfs,amssymb}
\usepackage{amsthm}
\usepackage{mathtools}
\usepackage{stmaryrd}
\usepackage{DotArrow}

\usepackage{booktabs}

\newcommand{\subparagraph}[1]{{\em #1}\hspace{0.5em plus 0.25em minus 0.25em}}

\usepackage{multirow}

\usepackage[inline,shortlabels]{enumitem}
\setlist[itemize]{label=-}
\setlist[enumerate]{label=\small(\roman*), labelwidth=3pt}
\newlist{asparaenum}{enumerate}{3}
\setlist[asparaenum]{label=(\roman*), align=left, leftmargin=0pt, labelindent=!, listparindent=\parindent, labelwidth=0pt, itemindent=!}
\newlist{inparaenum}{enumerate*}{3}
\setlist[inparaenum]{label=(\roman*)}
\newlist{asparaitem}{enumerate}{3}
\setlist[asparaitem]{label=-, align=left, leftmargin=0pt, labelindent=!, listparindent=\parindent, labelwidth=0pt, itemindent=!}
\newlist{inparaitem}{itemize*}{3}

\usepackage{anyfontsize}

\usepackage[stable]{footmisc}

\usepackage{bibspacing}
\usepackage{booktabs}
\usepackage{keyval}

\usepackage[final]{graphicx}
\usepackage[final]{listings}
\usepackage{algorithm2e}
\usepackage{pdfsync}

\theoremstyle{definition}
\newtheorem{definition}{Definition}
\newtheorem*{definition*}{Definition}
\newtheorem{proposition}{Proposition}
\newtheorem{corollary}{Corollary}
\newtheorem{lemma}{Lemma}
\newtheorem*{lemma*}{Lemma}

\newtheorem*{decisionproblem*}{Decision Problem}

\theoremstyle{plain}
\newtheorem{theorem}{Theorem}

\newcommand{\customthmno}{}
\newtheorem*{customthm}{Theorem \customthmno}
\newenvironment{customtheorem}[1][]{\renewcommand{\customthmno}{#1}\begin{customthm}}{\end{customthm}}

\newcommand{\customlemno}{}
\newtheorem*{customlem}{Lemma\customlemno}
\newenvironment{customlemma}[1][]{\renewcommand{\customlemno}{ #1}\begin{customlem}}{\end{customlem}}

\newcommand{\customcorno}{}
\newtheorem*{customcor}{Corollary \customcorno}
\newenvironment{customcorollary}[1][]{\renewcommand{\customcorno}{ #1}\begin{customcor}}{\end{customcor}}

\newcommand{\custompropno}{}
\newtheorem*{customprop}{Proposition \custompropno}
\newenvironment{customproposition}[1][]{\renewcommand{\custompropno}{ #1}\begin{customprop}}{\end{customprop}}

\newtheorem*{claim*}{Claim}

\theoremstyle{remark}

\newtheorem*{remark*}{Remark}

\newtheorem{example}{Example}

\newenvironment{proof*}
    {\begin{proof}[Proof \textup(Sketch\textup)]}
    {\end{proof}}

\usepackage[final,
             bookmarks,
             bookmarksopen,
             colorlinks,
             final,
             linkcolor=blue,
             citecolor=brown,
             pdfstartview=FitH ]%
{hyperref}

\newcounter{custlabelcount}

\newcommand{\customref}[2][\empty]{%
\ifx\empty#1%
\hyperref[#1]{#1}%
\else%
\hyperref[#2]{#1}%
\fi%
}

\usepackage[nompar,nosign]{commenting}

\colorlet{defaultcol}{black!90!yellow}
\declareauthor{lo}{Luke}{red}
\authorcommand{lo}{comment}
\declareauthor{jk}{Jonathan}{blue!90!black}
\authorcommand{jk}{comment} 
\newif\ifhighlight%
\highlighttrue%

\input{macros_global}                       

\newcommand{\leqMinner}{\ensuremath{\leq_{\scriptscriptstyle \Minner}}}
\newcommand{\greaterMinner}{\ensuremath{>_{\scriptscriptstyle \Minner}}}

\newcommand{\leqconfig}[1][]{\ensuremath{\leq_{\scriptscriptstyle \mathit{Config}#1}}}
\newcommand{\lessconfig}[1][]{\ensuremath{<_{\scriptscriptstyle \mathit{Config}#1}}}

\newcommand{\leqACPS}[1][]{\ensuremath{\leq_{\scriptscriptstyle \mathit{ACPS}#1}}}
\newcommand{\leqAPCPS}[1][]{\ensuremath{\leq_{\scriptscriptstyle \mathit{APCPS}#1}}}

\newcommand{\ConfigSimple}{\ensuremath{\Configuration_{\S}}}
\newcommand{\BTS}[3]{\ensuremath{\B_{#3}[#1]}}
\newcommand{\nsimple}{\ensuremath{n_{\S}}}
\newcommand{\ninner}{\ensuremath{n_{\I}}}
\newcommand{\ncomplex}{\ensuremath{n_{\C}}}
\newcommand{\enumelem}[2]{\ensuremath{#1_{\paren{#2}}}}

\newcommand{\dom}[0]{\ensuremath{\mathit{dom}}}

\newcommand{\Expspace}[0]{{\scshape Expspace}}
\newcommand{\Exptime}[0]{{\scshape Exptime}}
\newcommand{\Elementary}[0]{{\scshape El\-e\-men\-ta\-ry}}
\newcommand{\AltExpspace}[0]{{\scshape AltExpspace}}
\newcommand{\Tower}[0]{{\scshape Tower}}

\newcommand{\Ack}[0]{{\scshape Ack}}
\newcommand{\HAck}[0]{{\scshape HAck}}
\newcommand{\FastGrowing}[1]{\ensuremath{\mathbf{F}_{#1}}}
\newcommand{\FOmegaUpThree}{\ensuremath{\FastGrowing{\omega^{\omega^{\omega}}}}}

\newcommand{\IofG}[1][\GofP]{I\!\paren{#1}}
\newcommand{\GofP}[1][\calP]{\calG\!\paren{#1}}
\newcommand{\SigmaofP}[1][\calP]{\Sigma\paren{#1}}
\newcommand{\NonTofP}[1][\calP]{\NonT\!\paren{#1}}
\newcommand{\RofP}[1][\calP]{\calR\paren{#1}}

\newcommand{\NonT}{\mathcal{N}}
\newcommand{\PPL}{\mathcal{L}}
\newcommand{\Rules}{\ensuremath{\mathcal{R}}}
\newcommand{\eqvI}[1][I]{\ensuremath{\simeq_{#1}}}

\newcommand{\ComSigma}[0]{{\ensuremath{\Sigma^{\text{com}}}}}
\newcommand{\NComSigma}[0]{{\ensuremath{\Sigma^{\neg\text{com}}}}}
\newcommand{\ComN}[0]{{\ensuremath{\NonT^{\text{com}}}}}
\newcommand{\NComN}[0]{{\ensuremath{\NonT^{\neg\text{com}}}}}

\newcommand\calG{\mathcal{G}}
\newcommand\calV{\mathcal{V}}
\newcommand\calM{\mathcal{M}}

\newcommand{\toSF}[0]{\ensuremath{\to_{\text{seq}}}} 
\newcommand{\toCF}[0]{\ensuremath{\to_{\text{con}}}} 
\newcommand{\toCM}[0]{\ensuremath{\to_{\text{con}'}}} 

\newcommand{\snd}[2]{{{#1} \mathbin{\text{!}} {#2}}}
\newcommand{\rec}[2]{{{#1} \mathbin{\text{?}} {#2}}}
\newcommand{\spn}[1]{\ensuremath{\nu #1}}

\newcommand{\cproc}[1]{{#1}}
\newcommand{\cchan}[2]{\mcchan{\mset{#1}}{#2}} 
\newcommand{\mcchan}[2]{{#1}^{#2}} 

\newcommand{\APCPSConfig}{\ensuremath{%
    \mathit{Config}^{\mathit{APCPS}}%
    }}
\newcommand{\Config}{\ensuremath{\mathit{Config}}}

\newcommand{\Chan}{\ensuremath{\mathit{Chan}}}
\newcommand{\Channels}{\ensuremath{\mathit{Chans}}}
\newcommand{\TermCache}[0]{\ensuremath{\mathit{TermCache}}}

\newcommand{\MixedCache}[0]{\ensuremath{\mathit{MixedCache}}}

\newcommand{\MMsg}{\ensuremath{\mathit{Msg}}}
\newcommand{\ChanPar}{\ensuremath{\mathop{\,\lhd \,}}}
\newcommand{\Parallel}{\ensuremath{\mathop{\,\big|\big| \,}}}

\newcommand{\Pred}[0]{\ensuremath{\mathit{Pred}}}

\newcommand{\calP}[0]{\ensuremath{\mathcal{P}}}
\newcommand{\calF}[0]{\ensuremath{\mathcal{F}}}
\newcommand{\length}[3][]{\ensuremath{\rho^{#1}_{#2}\!(#3)}}
\newcommand{\dist}[0]{\ensuremath{d}}

\newcommand{\B}[0]{\ensuremath{\mathbb{B}}}
\newcommand{\budget}[0]{\text{budget}}
\newcommand{\activ}[0]{\text{active}}
\newcommand{\mode}[0]{\text{mode}}
\newcommand{\spawns}[0]{\text{spawns}}
\newcommand{\action}[0]{}
\newcommand{\eval}[0]{\text{ev}}
\newcommand{\shaperegion}[0]{\text{csh}}
\newcommand{\processctrl}[0]{\text{cpr}}
\newcommand{\channel}[0]{\text{chn}}
\newcommand{\ncom}[0]{\neg\text{com}}
\makeatletter

\newcommand{\notation@place}[3][]{p^{#2\raisebox{2pt}{\ensuremath{\scriptstyle #1}}}_{#3}}
\newcommand{\budgetplace}[1][]{\notation@place[#1]{}{\budget}}
\newcommand{\evalplace}[2][]{\notation@place[#1]{\eval}{#2}}
\newcommand{\multisetplace}[3][]{\notation@place[#1]{\action}{#2,#3}}
\newcommand{\activeplace}[3][]{\notation@place[#1]{[#2]}{#3}}
\newcommand{\activeshapeplace}[2][]{\notation@place[#1]{[#2]}{\shaperegion}}

\newcommand{\activencomplace}[3][]{\notation@place[#1]{#2}{\ncom(#3)}}

\newcommand{\passiveplace}[3][]{\notation@place[#1]{[#2]}{#3}}
\newcommand{\passiveshapeplace}[2][]{\notation@place[#1]{[#2]}{\shaperegion}}
\newcommand{\passivencomplace}[3][]{\notation@place[#1]{[#2]}{\ncom(#3)}}

\newcommand{\channelplace}[2][]{\notation@place[#1]{\channel}{#2}}
\newcommand{\regionctrlplace}[2][]{\notation@place[#1]{\shaperegion}{#2}}
\newcommand{\processctrlplace}[2][]{\notation@place[#1]{\processctrl}{#2}}
\newcommand{\ncomplace}[3][]{\notation@place[#1]{[#2]}{\ncom(#3)}}
\newcommand{\activemodeplace}[1][]{\notation@place[#1]{\mode}{\activ}}
\newcommand{\evalmodeplace}[1][]{\notation@place[#1]{\mode}{\eval}}
\newcommand{\spawnsvarplace}[1][]{\notation@place[#1]{\spawns}{}}
\makeatother

\makeatletter
\newcommand{\Lang}{%
	\texorpdfstring%
        {${{\ifx\f@series\my@test@bf\boldsymbol{\lambda}\else\lambda\fi}{\textsc{\scriptsize Actor}}}$}%
		{Lambda-Actor}%
}
\makeatother

\newcommand{\ecsrule}[4][\pid]{#1\colon #2\xrightarrow{#3}#4}

\def\eTau#1[#2-->#3]{\ecsrule[\ifx!#1!\pid\else #1\fi]{#2}{\tau}{#3}}
\def\eSend#1[#2--#3!#4-->#5]{\ecsrule[\ifx!#1!\pid\else #1\fi]{#2}{#3\mathbin{!}#4}{#5}}
\def\eRec#1[#2--?#3-->#4]{\ecsrule[\ifx!#1!\pid\else #1\fi]{#2}{?#3}{#4}}
\def\eSpawn#1[#2--v#3.#4-->#5]{\ecsrule[\ifx!#1!\pid\else #1\fi]{#2}{\nu#3.#4}{#5}}

\newcommand{\pid}{\iota}

\newcommand{\State}{\dom{State}}

\newcommand{\Procs}[1][\dom]{#1{\mathit{Procs}}}

\newcommand{\Msg}[1][\dom]{\ensuremath{#1{\mathit{Msg}}}}

\newcommand{\cfaTo}{\leadsto}

\gdef\abs#1{\ifx\to#1\cfaTo\else\widehat#1\fi}

\newlength{\colwd}

\makeatletter

\newcounter{RULE}

\begingroup
\catcode`\|=\active%
\gdef\@redefbar{%
    \catcode`\|=\active%
    \def|{\hspace{3pt}\mid\hspace{3pt}}
}
\endgroup

\newenvironment{syntax}{%
    \begin{equation*}
    \def\prod{\mathrel{::=}}
    
    \@redefbar
    \begin{array}{rl}%
}
{%
    \end{array}
    \end{equation*}
}

\makeatother

\input{macros_project_specific_names}       

\usepackage[final]{listings}

\usepackage{xcolor}
\colorlet{keyword}{blue!50!black}
\colorlet{atom}{red!50!black}
\colorlet{module}{green!30!black}
\colorlet{comment}{black!70}
\colorlet{coderules}{black!50}
\colorlet{lineno}{black!50}

\makeatletter
\newcommand{\srcsize}{\@setfontsize{\srcsize}{7.625pt}{7.625pt}}
\newcommand{\srcinsize}{\@setfontsize{\srcsize}{8.5pt}{8.5pt}}
\makeatother

\lstdefinelanguage{CoreErlang}
  {morekeywords={fun,and,case,letrec,let,in,catch,div,end,exit,export,halt,%
      if,import,link,make_ref,module,monitor,of,or,receive,self,send,spawn,throw,to,%
      unlink%
      },%
   morekeywords={[2]error,false,nil,ok,true,undefined,pid},%
   otherkeywords={!},%
   morecomment=[l]\%,%
   morestring=[b]"%
  }[keywords,comments,strings]%

\lstdefinestyle{sans}{
    xleftmargin=20pt,
    tabsize=4,
    showstringspaces=false,
    columns=[l]fullflexible,
    breaklines,
    fontadjust,
    numbers=left,
    numberstyle={\tiny\color{lineno}},
    literate={
        {->}{{$\hspace{-1pt}\rightarrow\hspace{1ex}$}}2
        {...}{{$\ldots$}}3
    },
    basicstyle={\sffamily\srcsize}, 
    keywordstyle={\color{keyword}\bf\srcsize},
    keywordstyle={[2]\color{atom}\srcsize},
    commentstyle={\scriptsize\color{comment}},
    emphstyle={\color{atom}\srcsize},
    emphstyle={[2]\color{module}},
    moredelim=[is][emphstyle]{\#}{\#},
    moredelim=[is][emphstyle2]{@}{@},
    escapechar=§,mathescape
}
\lstdefinestyle{boxed}{frame=single,frameround=tttt,backgroundcolor=\color{keyword!5}}
\lstdefinestyle{head}{style=boxed, frame=tlr, frameround=tfft}
\lstdefinestyle{middle}{style=boxed, frame=lr, firstnumber=last}
\lstdefinestyle{tail}{style=boxed, frame=lrb, frameround=fttf, firstnumber=last}

\lstdefinestyle{inl}{columns=[c]flexible, basicstyle={\sffamily\srcinsize}, emph={stop,ready,task_bag,despatcher_done,system}}

\newcommand{\erl}[2][]{\ifmmode\expandafter\text\fi{\lstinline[style=inl,#1]{#2}}}

\lstnewenvironment{erlang}[1][]{
 \lstset{language=CoreErlang, style=sans, #1}
}{}
\lstnewenvironment{erlang*}[1][]{
	\lstset{language=CoreErlang, style=sans, numbers=none, #1}
}{}
\lstnewenvironment{erlango}[1][]{
  \lstset{language=CoreErlang, style=sans, style=boxed, frame=tlr,%
  frameround=ttff, #1} }{}

\newcommand{\erlpds}[1]{\textsf{\scalebox{0.85}{#1}}}
\newcommand{\erlpdsatom}[1]{\textsf{\scalebox{0.85}{\textcolor{atom}{#1}}}}

\usepackage{setspace}
\allowdisplaybreaks%

\usepackage{tikz}
\usetikzlibrary{calc}
\usetikzlibrary{backgrounds}
\usetikzlibrary{shapes,arrows,positioning,fit,matrix,automata}

\usepackage{pdfsync}
\usepackage[normalem]{ulem}
\usepackage{varwidth,bussproofs}

\usepackage[font=small,labelfont=bf]{caption}

\usepackage[numbers]{natbib}

\usepackage[capitalise]{cleveref}
\let\oldciteauthor\citeauthor
\renewcommand{\citeauthor}[1]{%
\hypersetup{citecolor=black}%
\oldciteauthor{#1}%
\hypersetup{ citecolor=brown }%
}

\newcommand{\executeiffilenewer}[3]{%
\ifnum\pdfstrcmp{\pdffilemoddate{#1}}%
{\pdffilemoddate{#2}}>0%
{\immediate\write18{#3}}\fi%
}
\newif\ifsvg%
 \svgtrue

\newcommand*{\addFileDependency}[1]{
  \typeout{(#1)}
  \IfFileExists{#1}{}{\typeout{No file #1.}}
}

\begin{document}

\title{{\fontsize{21.56}{26.7}\selectfont 
Decidable Models of Recursive Asynchronous Concurrency}}
\author{Jonathan Kochems \qquad \qquad C.-H.~Luke Ong}
\maketitle 

\begin{abstract}
Asynchronously communicating pushdown systems (ACPS) that satisfy the empty-stack constraint 
(a pushdown process may receive only when its stack is empty) 
are a popular decidable model for recursive programs with asynchronous atomic procedure calls. 
We study a relaxation of the empty-stack constraint for ACPS that permits concurrency and communication actions at any stack height, called 
the \emph{shaped stack} constraint, \changed[lo]{thus enabling a larger class of concurrent programs to be modelled.} 
\changed[lo]{We establish a close connection between ACPS with shaped stacks and a novel extension of Petri nets:} 
\emph{Nets with Nested Coloured Tokens} (NNCTs).
Tokens in NNCTs are of two types: \emph{simple} and \emph{complex}. 
Complex tokens carry an arbitrary number of coloured tokens. The rules of NNCT can synchronise 
complex and simple tokens, inject coloured tokens into a complex token, and eject all tokens of a 
specified set of colours to predefined places.
We show that the coverability problem for NNCTs is \Tower-complete. 
To our knowledge, NNCT is the first extension of Petri nets---
in the class of nets with an infinite set of token \emph{types}---that \changed[lo]{is proven to} have primitive recursive coverability.
This result implies \Tower-completeness of coverability for ACPS with shaped stacks.
\end{abstract} 






\comment[jk]{
Page Breakdown:
\begin{itemize}
	\item Abstract (4 sentences, cur.: $\sim$1/6 page, 162 words)
	\item Introduction (1 page, cur.: $\sim$2 1/4 pages)
	\item ACPS (3.5 pages, cur.: $\sim$4 1/4 pages)
	\begin{itemize}
		\item ACPS vs APCPS (2 page, cur. $\sim$1 1/2 page)
		\item APCPS standard semantics vs alternative Reduction to NNCT (1.5 pages, cur. $\sim$1 2/3 pages)
	\end{itemize}
	\item NNCT Definition, Coverability Problems and Connection to APCPS (2 pages, cur. $\sim$ 2 pages)
	\item Upper Bound (3 pages, cur. $\sim$ 3 pages)
	\item Lower Bound (1.5 pages, cur.$\sim$ 2 pages)
	\item Related work (1 pages, cur.$\sim$1 1/4 page)
	\item Conclusions, further work (1/6 pages, cur. $\sim$1/6 pages)
	\item References (2/3 pages, cur. $\sim$ 1 1/4 page)
\end{itemize}
Total:12, cur.$\sim$ 15 pages
}

\jk{
	Auto-criticism:
	\begin{itemize}
		\item \done{How relevant are the results in a PL context?}
		\item \done{We need to commit to studying only coverability or coverability, boundedness and termination}
		\item \done{What is the shaped stack constraint for ACPS}
		\item \done{What is the precise relationship of coverability for shaped ACPS, shaped APCPS and NNCT}
		\item \done{Colours/inner places (inner places feature little in upper bound, but in lower bound)}
		\item Flow of paper is improvable
		\item \done{Counter abstraction uses poor notation and is not explained as well as it could}
		\item Notation is not globally consistent enough to be a help to the reader
		\item Section titles should be improved: more concise and space saving
		\item Treatment of Complexity, \Tower
	\end{itemize}
}

\section{Introduction}\noindent
In recent years the study of decision procedures for concurrent pushdown systems has proved immensely fruitful.
Substantial advances have been made 
in the algorithmic verification of \emph{asynchronous programs}, i.e.~recursive programs with asynchronous atomic procedure calls \cite{Sen:2006,Jhala:2007,Ganty:2012}. 
Asynchronous programs can be 
modelled naturally by \emph{asynchronously communicating pushdown systems} (ACPS)---a dynamic network of concurrent pushdown systems that communicate via a fixed, finite set of unbounded and unordered channels---subject to the  ``empty-stack restriction'', which means that a pushdown process cannot receive messages unless its stack is empty. 

The empty-stack restriction prohibits arbitrary synchronisations between processes, 
thus ruling out classes of interesting programs for analysis by ACPS.
For example, the server program in \Cref{fig:task_server} gives rise to an ACPS. The program spawns two processes, one running \erl{server}, the other {\erl{despatcher}.}\linebreak[4] 
The \erl{server} process posts tasks to the channel \erl{task_bag} which are continually removed and executed by the \erl{despatcher} process (possibly posting further tasks to \erl{task_bag} or spawning new processes). The \erl{despatcher} process non-deterministically chooses to wait for a \erl{stop} message from the \erl{server}, or calls itself recursively. 
An interesting question for this program is whether the messages \erl{ready} and \erl{despatcher_done}~can erroneously reside in the channels \erl{task_bag} and \erl{system} resp. at the same time. Such a question may be formulated as a \emph{coverability} problem: is it possible to reach a configuration $s$ that \emph{covers} a configuration $s_{\text{cov}}$ i.e.~$s_{\text{cov}} \leq s$ where $\leq$ is a preorder on the configuration-space. Unfortunately coverability is undecidable for ACPS in general \cite{Ramalingam:2000}, however it is decidable for ACPS with the empty-stack restriction. 
Notice that the \erl{server} process increases its call-stack at every recursive call while executing receives and sends: it does not satisfy the empty-stack restriction.
\begin{figure}[bt]
\vspace{-0.75ex}
\input{figures/server_ex}
\vspace{-10pt}
\caption{Server in asynchronous programming style. The procedures \erl{do_server} and \erl{do_task} may be arbitrary terminating recursive procedures that are allowed to send messages and spawn new 
processes.}
\label{fig:task_server}
 \vspace{-2ex} 
\end{figure} %
%
\newcommand{\noncn}[1]{\textcolor{BrickRed}{#1}}
\newcommand{\cn}[1]{\textcolor{OliveGreen}{#1}}
Is it possible to relax the empty-stack restriction on ACPS while preserving the decidability of coverability? Fortunately, the answer is yes. In previous work \cite{KochemsO:2013}, we introduced \emph{asynchronous partially commutative pushdown systems} (APCPS), a model of recursive asynchronous concurrency in the form of a class of (partially commutative) \emph{context free grammar}. An APCPS process is an equivalence class of words (over a set of non-terminal and terminal symbols) that allow the commutation of certain commutative non-terminals; the terminals determine the concurrency and communication side-effects, and the transition relation is essentially the leftmost derivation of the grammar. 
Intuitively, a non-terminal (or procedure) is \emph{non-commutative} just if a ``blocking operation'', such as receive, may be invoked when running it. In the APCPS setting, the empty-stack restriction is replaced
by the \emph{$K$-shape constraint} where $K \geq 0$, which limits the number of \emph{non-commutative non-terminals} that may occur in any reachable process to no more than $K$. Processes may perform concurrent actions, message-send / receive or spawns at any time, provided all reachable processes fit the $K$-shape constraint. 
Consider a receive transition of the \erl{server} process:
\begin{align*}
(&\rec{\erlpdsatom{task\_bag}}{\erlpdsatom{ok}}, \cn{\erlpds{do\_server()}} 
\cdot \noncn{\erlpds{post\_task()}} \cdot \noncn{\text{L}_4} \cdot \cn{\text{L}_8} \cdots \cn{\text{L}_8}) \parallel \cdots \\
&\;\;\;\;\;\;\;\;\to[effect=\rec{\textit{task\_bag}}{\textit{ok}}] (\erlpds{do\_server()}, \noncn{\erlpds{post\_task()}} \cdot \noncn{\text{L}_4} \cdot \cn{\text{L}_8} \cdots \cn{\text{L}_8}) \parallel \cdots
\end{align*}
where non-commutative and commutative procedures are marked in {\color{BrickRed} red} and {\color{OliveGreen} green} resp. 
This transition is impossible under the empty-stack constraint; to satisfy the latter the \erl{server} process is forced to empty its call-stack before it can make a receive transition: remembering to execute {\erl{do_server()},}\linebreak[4]\erl{post_task()} and to return to line 4 (L$_4$) and eventually to perform all recursive executions of line 8 (L$_8$) is not possible.  
The state of the \erl{server} process in this transition is representative; they always fit a suffix of the shape 
$(q, \noncn{\blacksquare} \cn{\blacksquare}^* \noncn{\blacksquare} \cn{\blacksquare}^* \noncn{\blacksquare} \cn{\blacksquare}^*)$. This is in fact true for all processes in the program in \Cref{fig:task_server}; thus the program fits the shape constraint. 
In general, the $K$-shaped constraint allows a pushdown process to remember \emph{infinite state} information along a receive transition, whereas the empty-stack constraint limits it to be \emph{finite state}.
\begin{figure*}[tbh]
\newcommand{\ctModel}[1]{\textbf{\textit{#1}}}
\newcommand{\ctCoverability}{Cov}
\newcommand{\ctSimpleCoverability}{SCov}
\newcommand{\ctBoundedness}{Bound}
\newcommand{\ctTermination}{Term}
\newcommand{\ctStandardSemantics}{Std Sem}
\newcommand{\ctAlternativeSemantics}{Alt Sem}
\newcommand{\ctNormalform}{NF}
\newcommand{\ctPolytime}{PT}
\newcommand{\ctConsttime}{CT}
\newcommand{\ctExptime}{ExpT}
\newcommand{\ctElemtime}{ET}
\newcommand{\ctConjecture}{(?)}
\newcommand{\ctPrimitiveRecursive}{PR}
\centering

\newcommand{\contribtableheading}[1]{\textbf{#1}}
\newcommand{\contribtablesubheading}[1]{\emph{#1}}
\newcommand{\ttm}[1]{$\displaystyle #1$}
\newcommand{\rulespace}{\hspace{0.885ex}}
\gdef\contribtablesep{\hspace{1eM}}
\gdef\cts{\contribtablesep}
\newcommand{\place}[1]{\multicolumn{2}{c}{#1}}
\newcommand{\placesep}[1]{\multicolumn{3}{c}{#1}}
\newcommand{\interreduce}[2][]{$\xleftrightarrow[\text{#1}]{\text{#2}}$}
\newcommand{\reduce}[2][]{$\xrightarrow[\text{#1}]{\text{#2}}$}
\tikzstyle{line} = [draw, 
                    dashed, black!50]
  \begin{tikzpicture}
  [
   lbl/.style={font=\scriptsize},
   every node/.style={font=\small}
  ]
    \node at (0,0) (ACPSCoverability) {\ctCoverability};
    \node[right=1.25cm of ACPSCoverability]         (ACPSSimpleCoverability) {\ctSimpleCoverability};
    \node[right=1.25cm of ACPSSimpleCoverability,align=center]   (ACPSNFSimpleCoverability) {\ctSimpleCoverability\\ \ctNormalform};
    \node[right=2cm of ACPSNFSimpleCoverability,align=center] (APCPSSimpleCoverabilityStd) 
          {\ctSimpleCoverability\\ \ctStandardSemantics};
    \node[right=1.5cm of APCPSSimpleCoverabilityStd,align=center] (APCPSSimpleCoverabilityAlt) 
          {\ctSimpleCoverability\\ \ctAlternativeSemantics};
    \node[right = 3.5cm of APCPSSimpleCoverabilityAlt,align = center] (NNCTCoverability) 
          {Coverability\\\Tower\ Thm~\ref{thm:covering_radius:bound}};
    \node[below = 0.475cm of NNCTCoverability,align=center] (NNCTTTSimpleCoverability) 
          {\ctSimpleCoverability\ / \ctBoundedness\ / \ctTermination \\for Total Transfer\\\Tower-hard\ Thm~\ref{thm:nnct:coverability:tower_hard}};

    \node[below right=1cm and -1.25cm of ACPSCoverability,align=center]   (ACPSBoundTerm) {\ctBoundedness\ / \ctTermination};
    \node[below=0.645cm of ACPSNFSimpleCoverability,align=center]   (ACPSNFBoundTerm) {\ctBoundedness\ / \ctTermination\\
    \ctNormalform};
    \node[below=0.645cm of APCPSSimpleCoverabilityStd,align=center] (APCPSBoundTermStd) 
          {\ctBoundedness\ / \ctTermination\\ \ctStandardSemantics};
    \node[below=0.645cm of APCPSSimpleCoverabilityAlt,align=center] (APCPSBoundTermAlt) 
          {\ctBoundedness\ / \ctTermination\\ \ctAlternativeSemantics};

    \node[above= 0.51cm of ACPSSimpleCoverability] (ACPS)  {\ctModel{$K$-shaped ACPS}};
    \node[above right= 0.35cm and -0.5cm of APCPSSimpleCoverabilityStd] (APCPS) {\ctModel{$K$-shaped APCPS}};
    \node[above= 0.35cm of NNCTCoverability] (NNCT)  {\ctModel{NNCT}};

    \draw[thick, black] (-1   ,0.75) edge (17 ,0.75);
    \draw[thick, black] (-1   ,-2)   edge (17   ,-2);

    \draw[dashed, black!50] (5.6  ,0.75)    edge (5.6  ,0.4);
    \draw[dashed, black!50] (5.6  ,-0.4)    edge (5.6  ,-1);
    \draw[dashed, black!50] (5.6  ,-1.50)   edge (5.6  ,-2);
    \draw[dashed, black!50] (11.8 ,0.75)    edge (11.8 ,0.1);
    \draw[dashed, black!50] (11.8  ,-1.00)  edge (11.8  ,-1.35);
    \draw[dashed, black!50] (11.8  ,-1.50)  edge (11.8  ,-2);

    \draw[<->] (ACPSCoverability) edge 
                    node[below,lbl]{Lem~\ref{lem:acps:cov:eq:simplecov}} 
                    node[above,lbl]{\ctPolytime} (ACPSSimpleCoverability);
    \draw[<->] (ACPSSimpleCoverability) edge 
                  node[below,lbl]{Prop~\ref{prop:acps:normalform}} 
                  node[above,lbl]{\ctPolytime} 
               (ACPSNFSimpleCoverability);
    \draw[<->] (ACPSNFSimpleCoverability) edge 
                  node[below,lbl]{Prop~\ref{prop:acps:apcps:interreduction}}
                  node[above,lbl]{\ctPolytime} 
               (APCPSSimpleCoverabilityStd);
    \draw[<->] (APCPSSimpleCoverabilityStd) edge 
                  node[below,lbl]{\cite{KochemsO:2013}} 
                  node[above,lbl]{\ctConsttime} 
               (APCPSSimpleCoverabilityAlt);               
    \draw[->] (APCPSSimpleCoverabilityAlt) edge 
                  node[below,lbl]{\hspace{2eM}Thm~\ref{thm:APCPS-to-NNCT}}
                  node[above,lbl]{\hspace{2eM}\ctElemtime} 
               (NNCTCoverability);    
    \draw[->] (NNCTTTSimpleCoverability) edge 
                  node[below,sloped,lbl]{Thm~\ref{thm:NNCT-to-APCPS}}
                  node[above,sloped,lbl]{\ctElemtime} 
               (APCPSSimpleCoverabilityAlt);      
    \draw[<->] (ACPSBoundTerm) edge 
                  node[below,lbl]{Prop~\ref{prop:acps:normalform}} 
                  node[above,lbl]{\ctPolytime} 
               (ACPSNFBoundTerm);
    \draw[<->] (ACPSNFBoundTerm) edge 
                  node[above,lbl]{\ctPolytime \ctConjecture} 
               (APCPSBoundTermStd);
    \draw[<->] (APCPSBoundTermStd) edge 
                  node[above,lbl]{\ctPrimitiveRecursive \ctConjecture} 
               (APCPSBoundTermAlt);
    \draw[->] (NNCTTTSimpleCoverability) edge 
                  node[below,lbl]{\hspace{2eM}Thm~\ref{thm:NNCT-to-APCPS}}
                  node[above,lbl]{\hspace{2eM}\ctElemtime} 
               (APCPSBoundTermAlt);      
  \end{tikzpicture}
  \vspace{-10pt}
\caption{A ``map'' of this paper's results. ``$P_1 \stackrel{\theta}{\longrightarrow} P_2$'' means that (decision problem) $P_1$ $\theta$-reduces to $P_2$. 
If the $K$-shaped is dropped the ACPS and APCPS part of the graph remains valid.
\emph{Legend:
\ctPolytime = polynomial time, 
\ctConsttime = constant time, 
\ctElemtime = exponential time, 
\ctPrimitiveRecursive = primitive recursive time, 
\ctNormalform = normal form,
\ctStandardSemantics = standard semantics,
\ctAlternativeSemantics = alternative semantics,
\ctConjecture = conjecture,
\ctCoverability = coverability, 
\ctSimpleCoverability = simple coverability, 
\ctBoundedness = boundedness, and 
\ctTermination = termination.}}
\label{table:contributions}
\vspace{-1ex}
\end{figure*}

The main result of \emph{op.~cit.~}is that coverability is decidable for APCPS that satisfy the $K$-shape constraint. Though ``semantic'' in nature, the shape constaint follows from an easy-to-check syntactic condition \cite{KochemsO:2013} 
which seems natural and readily satisfied by recursive programs humans write. 

\paragraph*{Our contributions}
The APCPS model is a hybrid model. On the one hand, it has the form of a partially commutative context-free grammar (in the sense of \cite{Czerwinski:09}) equipped with an operational semantics that specifies the behaviour of the concurrency and communication side-effects, such as send, receive and spawn. On the other, an APCPS determines a transition system which is very similar to that of an ACPS (the main difference is that APCPS processes are defined modulo a commutation relation). Our first contribution is to clarify the connections between APCPS and the standard and much studied ACPS. We show that there is a corresponding \emph{$K$-shape constraint} for ACPS, which limits the number of non-commutative stack symbols that may occur in the ``reachable'' stacks. We prove that coverability for $K$-shaped ACPS is polynomial-time inter-reducible with (a simplified version of) coverability for $K$-shaped APCPS, which is decidable \cite{KochemsO:2013}; see Figure~\ref{table:contributions}. Notice that the $K$-shaped ACPS model strictly extends the ACPS model with empty-stack restriction; in fact, the latter satisfies the 1-shape constraint by definition.

What is the complexity of coverability for $K$-shaped ACPS? We know that ACPS satisfying the empty-stack restriction are closely related to Petri nets: for example, their respective coverability problems are inter-reducible \cite{Ganty:2012}. However the $K$-shape constraint captures a larger class of models than the empty-stack restriction. Is there an extension of Petri nets that corresponds to $K$-shaped ACPS? Our second, major, contribution, are answers to these questions. (See Figure~\ref{table:contributions} for an overview of the technical results.)
\begin{asparaenum}[(i)]
\item We introduce a non-trivial extension of Petri nets: \emph{nested nets with coloured tokens} (NNCT). As the name suggests, NNCT feature, 
\changed[lo]{in addition to ordinary tokens (called \emph{simple})},
\emph{complex tokens} that carry \emph{coloured tokens}. Transitions may inject coloured tokens into a complex token; or {transfer certain coloured tokens---those whose colour is from a specified set of \emph{active colours}---from a complex token to predefined places}. 


\item We show that coverability for NNCT is \Exptime\ inter-reducible with simple coverability for APCPS (via the alternative semantics), and hence also inter-reducible with coverability for $K$-shaped ACPS.

\item We prove that coverability for NNCT is \Tower-com\-plete, in the sense of \citeauthor{Schmitz:2013} \cite{Schmitz:2013}. To our knowledge, NNCT is the first extension of Petri nets
---in the class of nets with an infinite set of token types---that is proven to have primitive recursive coverability. To prove \Tower-membership of coverability for NNCT, we devise a geometrically inspired version of the \citeauthor{Rackoff:78} technique \cite{Rackoff:78}, which was originally used to prove the \Expspace-coverability for Petri nets. We obtain \Tower-hardness of coverability, boundedness and termination by adapting \citeauthor{Stockmeyer:1974}'s ruler construction \cite{Stockmeyer:1974} to NNCT. We also establish the decidability of \emph{boundedness} and \emph{termination} for NNCT.
Transfering our complexity analysis on NNCT implies, surprisingly, that the bound $K$ on the number of ``non-commutative procedure calls'' in the shaped stack constraint is not the expensive resource.
In fact, $K$ influences only the number of colours $\ninner$ and complex places $\ncomplex$ of the simulating NNCT $\calN$; coverability is then decidable in space bounded by an exponential tower of height $O(\nsimple + \slog(\nsimple \cdot \ninner \cdot \ncomplex))$ where $\nsimple$ is the number of simple places of $\calN$ which is independent of $K$. \end{asparaenum}
%
%
\paragraph*{Notation}
Let us write $\N^\infty = \N \union \set{\infty}$, $\range{n} = \set{1,...,n}$, and $\M[U]$ for the set of multisets over $U$. 
Explicit multisets are enclosed in $\mset{\cdot}$; e.g.~we write $\mset{u,u,v^2}$ for the multiset containing two occurrences each of $u$ and $v$. We write $\emptyset$ for both the empty set and the empty multiset. 
We say that $u$ is an element of the multiset $M \in \M[U]$, written $u \in M$, if $M(u) \geq 1$. 
If $M_1, M_2 \in \M[U]$, we write $M_1 \oplus M_2$ for the multiset union of $M_1$ and $M_2$. If $M = M_1 \oplus M_2$ then we define $M_1 = M \ominus M_2$. 
Let $M \in \M[U]$ and $U_0 \subseteq U$. We define $M \restriction U_0$ to be the multiset $M$ restricted to $U_0$ 
i.e.~$(M \restriction U_0) : u \mapsto M(u)$ if $u \in U_0$, and $0$ otherwise. 
We write $U^*$ for the set of finite sequences over $U$, i.e.~maps from $\N \to U$, 
let $\beta, \gamma, \ldots$ range over $U^{*}$ and we denote the \emph{length} of sequence $\beta$ by $|\beta|$. 
We define the \emph{Parikh image} of $\beta \in U^*$ by 
$\M(\beta) : u \mapsto |\{i \mid \beta(i) = u \}|$. 
%
%
Let $(U,\leq_U)$ be a preordered set; we extend $\leq_U$ to $\M[U]$ and $U^*$ as usual:
\begin{inparaenum}[(i)]
\item $M_1 \leq_{\M[U]} M_2$ just if $M_1\!=\!\mset{u_1,..., u_n}$, $M_2\!=\!\mset{u'_1,..., u'_m}$ and there is an injective map $h\!:\!\range{n}\!\to\!\range{m}$ such that for all $i\!\in\!\range{n}$ we have $u_i\!\leq_U\!u'_{h(i)}$; 
\item $\beta \higleq[U] \beta'$ if there is a strictly monotone function $h : \range{|\beta|} \to \range{|\beta'|}$ and $\beta(i) \leq_{U} \beta'(h(i))$ for all $i \in \range{|\beta|}$.
\end{inparaenum}
We write $U_1 \disjointunion U_2$ for the disjoint union of sets $U_1$ and $U_2$. 
Let $f_1 : U_1 \to V_1$ and $f_2 : U_2 \to V_2$ be maps, then we define the map $f_1 + f_2 : U_1 + U_2 \to V_1 + V_2$ by
\(
(f_1 + f_2) (\mathit{in_i}(u)) := \mathit{in}_i (f_i (u)),
\)
for $i \in \makeset{1, 2}$, where $\mathit{in}_i$ is the canonical injection of $U_i$ into $U_1 + U_2$. Henceforth, to save writing, we elide $\mathit{in}_i$. 
We extend the operators $\oplus$, $\ominus$ to functions, i.e.,~if $h_1, h_2 : V \to \M[U]$ then $(h_1 \oplus h_2)(v) := h_1(v) \oplus h_2(v)$ and $(h_1 \ominus h_2)(v) := h_1(v) \ominus h_2(v)$. For all sets $U$, $V$ we define the map $\vec{0} : U \to \M[V]$ by $\vec{0}(u) = \emptyset$ for all $u \in U$.
We write $\preupdate{f}{u_1 \mapsto u'_1, \ldots, u_n \mapsto u'_n}$ (we omit $f$ if $f = \vec{0}$) for the function $f'$ such that $f'(u) = f(u)$ for $u \neq u_i$, and $f'(u_i) = u'_i$ for all $i \in \range{n}$.

\section{Recursive Asynchronous Concurrency}\label{sec:acps}\noindent
ACPS are prevalent as concurrent systems and their algorithmic verification
is a central problem in verification. 
A variety of verification problems for asynchronous programs, and thus ACPS satisfying the empty-stack constraint, are polynomial-time inter-reducible to decision problems on Petri nets \cite{Ganty:2012}.
Due to this connection, we know that e.g.~verification of safety properties is \Expspace-complete.
It is our goal to exhibit a similar connection between ACPS satisfying the shape constraint and an extension of Petri nets. 

\begin{definition}
An \emph{asynchronously communicating pushdown system} (ACPS) $\calP$ is a quintuple $\calP = (\calQ,\calA,\Chan,\Msg[],\calR)$ composed of \emph{control states} $\calQ$, a \emph{stack alphabet} $\calA$, \emph{channel} names $\Chan$, \emph{messages} $\Msg[]$ and \emph{rules} $\calR$, all finite sets;
$\calR$ is a subset of $\calQ \times \calA^* \times \Lambda \times \calQ \times \calA^*$
where $\Lambda$ is the set of \emph{communication side-effects}: $\Lambda := \varset{\rec{c}{m},\snd{c}{m},\spn{(q,\beta)}\!:\!c\!\in \Chan, m\!\in\!\Msg[], q \in \calQ,\beta \in \calA^*} \union \set{\epsilon}$.
\end{definition}\noindent
We use the notation $(q,\beta) \to[effect=\lambda] (q',\beta')$ for rules.
An action $\lambda=\snd{c}{m}$ denotes the sending of the message $m$ to channel $c$, $\rec{c}{m}$ denotes the {retrieval} 
of $m$ from $c$, and $\spn(q,\beta)$ denotes the spawning of a new process in state $(q,\beta)$. 
An ACPS $\calP$ gives rise to an infinite state transition system over $\M[\calQ \times \calA^*] \times \varparen{\Chan \to \M[\Msg[]]}$. For simplicity, we write a configuration (say) $(\mset{(q,\beta), (q',\beta')}, \{ c_1 \mapsto \mset{m_a,m_b,m_b}, c_2 \mapsto \mset{}\})$ as ${\cproc{(q,\beta)} \parallel \cproc{(q',\beta')} \ChanPar \cchan{m_a,m_b,m_b}{c_1}, \cchan{}{c_2}}$.  We abbreviate a set of processes running in parallel as $\Pi$ and a set of channels by $\Gamma$ with names in $\Chan$. 
The transition relation $\to[TS=\calP]$ is defined as follows: suppose $(q,\beta) \to[effect=\lambda] (q',\beta') \in \calR$ and $\beta_0 \in \calA^*$ then 
\begingroup\makeatletter\def\f@size{9}\check@mathfonts
\begin{align*}
	&(q,\beta\beta_0)\!\parallel\!\Pi\!\ChanPar\!\Gamma \to[TS=\calP] (q',\beta'\beta_0)\!\parallel\!\Pi\!\ChanPar\!\Gamma 
	&\hspace{-2ex}(\lambda = \epsilon)
	\\
	&(q,\beta\beta_0)\!\parallel\!\Pi\!\ChanPar\!\Gamma \to[TS=\calP] (q',\beta'\beta_0)\!\parallel\!\Pi\!\ChanPar\!\Gamma\!\oplus\!\varupdate{}{c}{\mset{m}}
	&\hspace{-2ex}(\lambda = \snd{c}{m})
	\\
	&(q,\beta\beta_0)\!\parallel\!\Pi\!\ChanPar\!\Gamma\!\oplus\!\varupdate{}{c}{\mset{m}} \to[TS=\calP] (q',\beta'\beta_0)\!\parallel\!\Pi\!\ChanPar\!\Gamma
	&\hspace{-2ex}(\lambda = \rec{c}{m})
	\\ 
	&(q,\beta\beta_0)\!\parallel\!\Pi\!\ChanPar\!\Gamma \to[TS=\calP] (q',\beta'\beta_0)\!\parallel\!(q_0,\beta_1) \parallel\!\Pi\!\ChanPar\!\Gamma 
	&\hspace{-2ex}(\lambda = \spn{(q_0,\beta_1)})
\end{align*}%
\endgroup%
Many interesting verification problems can be expressed in terms of transition systems that are endowed with a preorder on the state space which we call \emph{pre-structured transition systems} (PSTS).
Formally, a PSTS is a triple $\calS = (S,\to[TS=\calS],\leq_{\calS})$ such that $\leq_{\calS}$ is a preorder on $S$, ${\to[TS=\calS]} \subseteq S \times S$ a transition relation and we denote its transitive closure as $\to[TS=\calS,*]$. 

\begin{definition}
Let $\calS\!=\!(S,\to[TS=\calS],\leq_{\calS})$ be a PSTS and ${s_0,s_{\text{cov}} \in S}$ be configurations.  
The triple $\mathpzc{Q} = \paren{\calS,s_0,s_{\text{cov}}}$ is a \emph{coverability query} which is said to be a \emph{yes-instance} for \emph{coverability} 
if there is a reachable configuration $s'$ that covers $s_{\text{cov}}$, i.e., $s_0 \to[TS=\calS,*] s'$ and $s_{\text{cov}} \leq_{\calS} s'$.
The \emph{boundedness problem} asks whether the set $\varset{s : s_0 \to[TS=\calS,*] s}$ is finite and the
\emph{termination problem} is to decide whether there exists an infinite path from $s_0$ in $\calS$.
\end{definition}\noindent
%
We augment ACPS with an order to yield a PSTS as follows: we order two processes $(q,\beta) \leq_{\calQ\times \calA^*} (q',\beta')$ just if $q = q'$ and either $\beta = \epsilon$ or $\beta = A\cdot\beta_0$, $\beta' = A\cdot\beta'_0$ and $\beta_0 \higleq \beta'_0$.
Configurations are then ordered using the usual multiset and function extension:
$\Pi \ChanPar \Gamma \leqACPS \Pi \ChanPar \Gamma$ just if $\Pi \leq_{\M[\calQ\times \calA^*]} \Pi'$ and
if $\Gamma(c) \leq_{\M[\Msg[]]} \Gamma'(c)$ for all $c \in \Chan$.

An ACPS process $(q, A\,\beta)$ may process its own stack, test whether $\beta \higgeq \beta_{cov}$, and record this in its local state $q$. We may thus assume that in a coverability query 
$(\calP,\Pi_0 \ChanPar \Gamma_0,\Pi \ChanPar \Gamma)$ all processes $(q,\beta)$ of $\Pi$ and $\Pi_0$ satisfy $\beta \in \set{\epsilon, A \in \calA}$ :
we call such a query \emph{simple} and the coverability problem restricted to simple queries 
\emph{simple coverability}.
\begin{lemma}\label{lem:acps:cov:eq:simplecov}
Coverability and simple coverability for ACPS poly\-no\-mi\-al-time inter-reduce.
\end{lemma}
%
%
%
%
\newcommand{\Fprenormalform}[1]{\ensuremath{\calF^{\text{pnf}}\!\paren{#1}}}
\newcommand{\FpnfTonormalform}[1]{\ensuremath{\calF^{\text{nf'}}\!\paren{#1}}}
\newcommand{\Fnormalform}[2][]{\ensuremath{\calF#1\!\paren{#2}}}
\noindent
Control-states may be encoded in an enlarged stack-alphabet and hence we may w.l.o.g.~consider ACPS in \emph{normal form}: for all $(q,\beta) \to[effect=\lambda,inl] (q',\beta') \in \calR$
\begin{inparaenum}[(i)]
\item(acps:nf:i) $q = q'$,
\item(acps:nf:ii) $\beta = A \in \calA$,
\linebreak[4]\item(acps:nf:v) if $\lambda = \spn{(q'',\beta'')}$ then $q'' = q$ and $\beta'' \in \calA$, and
\item(acps:nf:iii) $\beta' = \epsilon$ unless $\lambda=\epsilon$,
in which case: \item(acps:nf:iv) $\beta' \in \set{\epsilon, B\,C : B, C \in \calA}$.
\end{inparaenum}
\begin{proposition}\label{prop:acps:normalform}
Given an ACPS $\calP$, a simple coverability query $\mathpzc{Q}$ and a $\Pi^0 \ChanPar \Gamma^0$ there exist ACPS $\Fnormalform{\calP}$ in normal form, a simple coverability query $\Fnormalform{\mathpzc{Q}}$, and $\Fnormalform{\Pi^0 \ChanPar \Gamma^0}$ --- all poly\-no\-mi\-al-time computable --- such 
that:
$\mathpzc{Q}$ is a yes-instance if, and only if, $\Fnormalform{\mathpzc{Q}}$ is a yes-instance; and
$\calP$ is bounded (terminating) from $\Pi^0 \ChanPar \Gamma^0$ if, and only if, $\Fnormalform{\calP}$ is bounded (terminating respectively) from 
$\Fnormalform{\Pi^0 \ChanPar \Gamma^0}$.
\end{proposition} \noindent
Henceforth we shall elide the single state $q$ and write a rule simply as $\beta \to[effect=\lambda,inl]\beta'$.
There is a well-known connection between pushdown systems and \emph{context-free grammars} (CFG). It is trivial to see that transitions of a single ACPS process in normal form are (essentially) left-most derivations of a CFG.

\begin{definition}
Let $\Sigma$ be an alphabet of terminal symbols. A \emph{context-free grammar} (CFG) 
is a triple $\calG = (\Sigma, \NonT, \Rules)$ where $\Rules$ comprises rewrite rules of the form:
\begin{inparaenum}[(i)]
\item \label{item:rec} $X \rightarrow Y \, Z$ where $X,Y,Z \in \NonT$ and
\item \label{item:singleton} $X \rightarrow a$ where $a \in \Sigma \cup \makeset{\epsilon}$.
\end{inparaenum}
\end{definition} \noindent
Fix an ACPS $\calP = (\set{q}, \calA, \Chan, \Msg[], \calR)$ in normal form.
From $\calP$ we obtain a CFG $\GofP=(\SigmaofP, \NonTofP, \RofP)$ whose non-terminals $\NonTofP$ are the stack alphabet $\calA$, 
the set of terminals $\SigmaofP$ is trivially obtained from $\Lambda$ (we only replace
$\spn{(q,A)}$ by $\spn{A}$), and its rules $\RofP$ are derived by
\begin{align*}
	A \to[effect=\spn{A'}] \epsilon     &\;\mapsto\; A  \to \spn{A'}	&
	A \to[effect=\snd{c}{m}] \epsilon 	&\;\mapsto\; A  \to \snd{c}{m}	\\
	A \to[effect=\epsilon] \epsilon		&\;\mapsto\; A  \to \epsilon 	&
	A \to[effect=\rec{c}{m}] \epsilon 	&\;\mapsto\; A  \to \rec{c}{m}	\\
	A \to[effect=\epsilon] B\,C 		&\;\mapsto\; A  \to B\, C	
\end{align*} 
In moving from the operational view of pushdown systems to CFGs it is possible to build upon known results on how CFGs interact with reorderings in the words they produce.
In our setting, channels are unordered and a new process may start at its own leisure.
 Hence, the execution order of concurrency side-effects such as send and spawn is immaterial. However, the sequencing of other side-effects, notably message retrieval which requires synchronisation, is observable.

\subsubsection*{Exploiting Commutativity}

The use of an \emph{independence relation} is a common technique to formalise such sensitivity to order.
An independence relation $I$ over a set $U$ is a symmetric irreflexive relation over $U$. 
It induces a congruence relation $\eqvI$ on $U^*$ defined as the least equivalence relation $\mathpzc{R}$ such that: $(\beta, \beta') \in \mathpzc{R} \; \implies \; \forall \gamma, \gamma' \in U^{\ast} : (\gamma \, \beta \, \gamma', \gamma \, \beta' \, \gamma) \in \mathpzc{R}$ and $I \subseteq \mathpzc{R}$.
An element $u \in U$ is said to be \emph{non-commutative wrt to $I$} if $(u,u') \notin I$ for all $u' \in U$ and we denote the set of non-commutative elements by $U^{\neg\text{com}}$. Similarly we call a $u \in U$ \emph{commutative wrt to $I$} if $(u,u') \in I$ for all $u' \in U \setminus U^{\neg\text{com}}$, i.e.~commutative elements commute with all elements except non-commutative ones, and write $U^{\text{com}}$ for the set of commutative elements. 
An independence relation $I$ that partitions $U$ into commutative and non-commutative elements is said to be \emph{unambiguous}. Let $\Xi \subseteq U$, then we define the independence relation generated by $\Xi$ as
\[
\mathit{IndRel}_{U}(\Xi) := \makeset{(a, a'), (a', a) \mid a, a' \in \Xi, a \neq a'}.
\]

Since receive is the only blocking concurrency action, we specify the actions that we want to commute as $\ComSigma := \SigmaofP \setminus \varset{\rec{c}{m} \mid c \in \Chan, m \in \Msg[]}$. Then the independence relation $\mathit{IndRel}_{\SigmaofP}(\ComSigma)$ 
 allows us to commute all concurrency actions \emph{except} receive. Further, $\mathit{IndRel}_{\SigmaofP}(\ComSigma)$ 
is unambiguous and partitions $\SigmaofP$ into the commutative and non-commutative actions $\ComSigma$ and $\NComSigma$ respectively.

We expect an independence relation over non-terminals, that is consistent with $\mathit{IndRel}_{\SigmaofP}(\ComSigma)$, to classify a non-terminal as commutative if it is productive and rewrites only to commutative symbols. 
This intuition can be captured in terms of a suitable monotone function $F : \P[\NonTofP] \to \P[\NonTofP]$.
For a subset $U \subseteq \NonTofP$, a non-terminal $N$ is an element of the set $F(U)$ just if 
\begin{inparaenum} 
\item $\calL(N) \neq \emptyset$, and
\item for each $w$ such that $N \to w$, $\Sigma(w) \subseteq \ComSigma$ and $\NonT(w) \subseteq U$, writing $\Sigma(w)$ (respectively $\NonT(w)$) for the set of terminal (respectively non-terminal) symbols that occur in the word $w$.
\end{inparaenum}
We define $\ComN$ as the greatest fixpoint of $F$. Similarly to the case of concurrency actions,
this choice of commutative non-terminals gives rise to the independence relation $\IofG$ defined by 
$$\IofG \is \mathit{IndRel}_{\NonTofP \union \SigmaofP}(\ComSigma \union \ComN).$$ 
Again, the independence relation $\IofG$ yields a partition of $\NonTofP \union \SigmaofP$ into the commutative ($\ComSigma \union \ComN$) and non-commutative symbols $\NComSigma \union \NComN$. The set $\NComN$ denotes the set of 
 non-commutative non-terminals and it is easy to see that $\NComN = \NonTofP \setminus \ComN$.
%
Commutative non-terminals are productive and only produce commutative symbols. Rewriting non-commutative non-terminals may block.
For example, a non-terminal is an element of $\ComN$ if it produces only finite words of sends and spawns (possibly infinitely many), or infinite words of sends and spawns including all their prefixes. By contrast, a non-terminal that produces a receive terminal or no finite word is categorised as non-commutative.

A CFG $(\Sigma,\NonT,\calR)$ endowed with an independence relation on $\NonT \union \Sigma$ is called a \emph{partially commutative context-free grammar} (PCCFG) \cite{Czerwinski:09}.
In previous work, we introduced
a class of PCCFG, \emph{asynchronous partially commutative pushdown systems} (APCPS), defined over $\SigmaofP$ and endowed with $\IofG[\calG]$ \cite{KochemsO:2013}. Even though APCPS originate from a grammar they induce a
transition system semantics that is suitable for the analysis of asynchronous concurrent pushdown systems.

\begin{definition} \label{def:APCPS}
An \emph{asynchronous partially commutative pushdown system} (APCPS) is a PCCFG $\calG = (\SigmaofP, \IofG[\calG], \NonT, \Rules)$.
\end{definition}
\noindent In particular endowing $\GofP$ with $\IofG$ yields an APCPS.
Given an APCPS $\calG$ the \emph{standard semantics} gives rise to a transition system very similar to an ACPS. The main difference is that processes are equivalence classes induced by $\eqvI[{\IofG[\calG]}]$. 

\begin{table*}[tb]
\centering
\newcounter{stdrulecounter}
\creflabelformat{stdruletype}{#1#2#3}
\crefformat{stdruletype}{#2{(R-#1)}#3}
\Crefformat{stdruletype}{#2{(R-#1)}#3}
\crefalias{stdrulecounter}{stdruletype}
\newcommand{\stdrule}[1][]{\refstepcounter{stdrulecounter}(R-\thestdrulecounter)\label[stdruletype]{#1}}
\newcounter{altrulecounter}
\creflabelformat{altruletype}{#1#2#3}
\crefformat{altruletype}{#2{(R$'$-#1)}#3}
\Crefformat{altruletype}{#2{(R$'$-#1)}#3}
\crefalias{altrulecounter}{altruletype}
\newcommand{\altrule}[1][]{\refstepcounter{altrulecounter}(R$'$-\thealtrulecounter)\label[altruletype]{#1}}

\newcommand{\tableheading}[1]{\textbf{#1}}
\newcommand{\tablesubheading}[1]{\textbf{\emph{#1}}}
\newcommand{\ttm}[1]{$\displaystyle #1$}
\newcommand{\rulespace}{\hspace{0.885ex}}

\begin{tabular}{@{}l@{}p{0.40\textwidth}|@{}p{0.5455\textwidth}@{}}
\toprule
\multicolumn{3}{l}{
\begin{tabular}{@{}p{0.44\textwidth}@{}p{0.53\textwidth}@{}}
  \centering\tablesubheading{Standard semantics}     & \centering\tablesubheading{Alternative semantics}
\end{tabular}
}  \\
\midrule
\multicolumn{3}{c}{\fontsize{9.25}{11.46}\selectfont 
\begin{tabular}{@{}l@{\rulespace}r@{}l|l@{\rulespace}r@{}l@{}}
  \stdrule[stdrule:interleave:std]: &
  \ttm{A\, \gamma \,}     & \ttm{\parallel \Pi \ChanPar \Gamma \toCF \beta \, \gamma  \parallel \Pi \ChanPar \Gamma} &
  \altrule[altrule:interleave:std]: &
  \ttm{A\, M \, \gamma'\,} & \ttm{\parallel \Pi' \ChanPar \Gamma \toCM \beta \, M\, \gamma' \parallel \Pi' \ChanPar \Gamma} \\
  \stdrule[stdrule:interleave:com]: &
  \ttm{A\, \gamma \,}     & \ttm{\parallel \Pi \ChanPar \Gamma \toCF B\, C \, \gamma  \parallel \Pi \ChanPar \Gamma} &
  \altrule[altrule:interleave:com]: &
  \ttm{A\, M \, \gamma'\,} & \ttm{\parallel \Pi' \ChanPar \Gamma \toCM B \, (\M(w)\, \oplus\, M)\, \gamma' \parallel \Pi' \ChanPar \Gamma} \\
\stdrule[stdrule:receive]: &
  \ttm{\cproc{(\rec{c}{m})\,\gamma}\,} &
  \ttm{\parallel \Pi \ChanPar 
        \mcchan{(\mset{m} \oplus q)}{c}, \Gamma 
        \toCF
        \cproc{\gamma}                 \parallel \Pi \ChanPar \mcchan{q}{c}, \Gamma} &
\altrule[altrule:receive]: &
  \ttm{\cproc{(\rec{c}{m})\, \gamma'}\,} &
  \ttm{\parallel \Pi' \ChanPar  \mcchan{(\mset{m} \oplus q)}{c}, \Gamma\toCM 
       \cproc{\gamma'}                  \parallel \Pi' \ChanPar  \mcchan{q}{c}, \Gamma}\\
\stdrule[stdrule:send]: &
    \ttm{\cproc{(\snd{c}{m})\,\gamma}\,}&
    \ttm{\parallel \Pi \ChanPar \mcchan{q}{c}, \Gamma \toCF
    \cproc{\gamma}                \parallel \Pi \ChanPar \mcchan{(\mset{m} \oplus q)}{c},\Gamma}
    &
\altrule[altrule:send]: &
    \ttm{\cproc{(\snd{c}{m})\, \gamma'}\,}&
    \ttm{\parallel \Pi' \ChanPar  \mcchan{q}{c}, \Gamma\toCM 
    \cproc{\gamma'}                  \parallel \Pi' \ChanPar  \mcchan{(\mset{m} \oplus q)}{c}, \Gamma}\\
\stdrule[stdrule:spawn]: &
    \ttm{\cproc{(\nu X)\,\gamma\,}  }&
    \ttm{\parallel \Pi \ChanPar \Gamma \toCF
        \cproc{\gamma}         \parallel \cproc{X} \parallel \Pi \ChanPar \Gamma}
    &
\altrule[altrule:spawn]: &
    \ttm{\cproc{(\nu X)\, \gamma'}\, }&
    \ttm{\parallel \Pi'               \ChanPar \Gamma \toCM 
        \cproc{\gamma'}           \parallel \cproc{X} \parallel \Pi'  \ChanPar \Gamma }\\
 & & &
 \altrule[altrule:dispatch:term]: &
    \ttm{\cproc{M\, X \,\gamma'} \,}&
    \ttm{\parallel \Pi'  \ChanPar \Gamma \toCM 
        \cproc{X\, \gamma'}     \parallel \Pi' \parallel \Pi(M) \ChanPar \Gamma \oplus \Gamma(M)} \\
& & &
\altrule[altrule:dispatch:nonterm]: &
    \ttm{\cproc{M'\,  \gamma'}  \,}&
    \ttm{\parallel \Pi'  \ChanPar \Gamma \toCM 
        \cproc{M''\, \gamma'}  \parallel \Pi' \parallel \Pi(M') \ChanPar \Gamma \oplus \Gamma(M')}
\end{tabular}
}\\
\midrule
\multicolumn{3}{p{0.99\textwidth}}{\small 
\vspace{-9pt}
\changed[jk]{Let $(\dagger)$ be a condition on a $\beta \in \NonT^*$: $\beta = B\, C$ and $C$ is commutative.}
Rules \Cref{stdrule:interleave:com} and \Cref{altrule:interleave:com} have a side condition: \changed[jk]{$A \to B\,C \in \calG$}, $B\,C$ satisfies $(\dagger)$, and $C \to[*] w$.
 Rules \Cref{stdrule:interleave:std} and \Cref{altrule:interleave:std} have a side condition: $A \to \beta \in \calG$ and $\beta$ does \emph{not} satisfy $(\dagger)$.
In rule \Cref{altrule:dispatch:term} we require the multiset $M \in \M[\ComSigma]$ whereas in rule
\Cref{altrule:dispatch:nonterm} $M'$ may be an element of $\M[\ComSigma \union \NonT]$ and $M'' = M' \restriction \NonT$.
\changed[jk]{We use the abbreviations: $\Pi(M) \is \parallel_{A \in \NonT}\parallel_1^{M(\spn{A})} A$, and
$\Gamma(M) \is \Oplus_{c \in \Chan}\Oplus_{m \in \Msg[]} \varupdate{}{c}{\varmset{m^{M(\snd{c}{m})}}}$.}
} \\
\bottomrule
\end{tabular}

\caption{Transition semantics for APCPS}
\label{fig:standard_vs_alternative:semantics}
\vspace{-\baselineskip}
\end{table*}

\paragraph*{Standard Semantics} 
In the standard semantics a process, denoted $\gamma$, is a word $\delta\, \beta\, X_1\, \beta_1 \cdot\cdot\cdot X_n \beta_n$ ({modulo} $\eqvI[{\IofG[\calG]}]$) where $X_i \in \NComN$ and $\beta, \beta_i \in \ComN$ and $\delta \in \Sigma \cup \NonT \cup \set{\epsilon}$. We call the set of such processes $\Procs[]$. The standard semantics of an APCPS is a transition system over 
$\M[\Procs[]] \times \Channels$ where $\Channels = (\Chan \rightarrow \M[\MMsg])$. We order processes 
{$\delta \pi_0 \leq_{\Procs[]} \delta \pi_1$ if $\delta \pi'_0 \higleq[(\NonT\union\Sigma)] \delta \pi'_1$} for some $\pi_0'$ and $\pi_1'$ such that both $\delta \, \pi_i \eqvI[{\IofG[\calG]}] \delta \, \pi'_i$. We lift $\leq_{\Procs[]}$ to a {preorder} $\leqAPCPS$ on configurations, using the multiset and function extension, and obtain a PSTS $(\M[\Procs[]] \times \Channels,\toCF,\leqAPCPS)$ {where the transition relation $\toCF$ is defined in the left column of} \Cref{fig:standard_vs_alternative:semantics}. 
Processes change state by performing a \changed[lo]{leftmost} CFG derivation of $\calG$ until an action appears in the leftmost position (Rules \Cref{stdrule:interleave:std} and \Cref{stdrule:interleave:com}); the type of action then determines a concurrency side-effect that interacts with the rest of the configuration (Rules \Cref{stdrule:receive}--\Cref{stdrule:spawn}) causing the action to be consumed and enabling further leftmost reductions of $\calG$.
We say an APCPS coverability query $(\calG,\Pi_0 \ChanPar \Gamma_0,\Pi \ChanPar \Gamma)$ is simple if $\pi \eqvI[{\IofG[\calG]}] A \in \NonT$ for all $\pi$ in $\Pi$ and $\Pi_0$.

\subparagraph{The Shaped Constraint.} 
If $\beta_i\!\in\!\ComN^*$, $X_i\!\in\!\NComN$, $\delta\!\in\!\NonT$, $\bar{\delta}\!\in\!\NonT \union \Sigma \union \varset{\epsilon}$, and $n \leq K$,
we say an ACPS process $\pi = \delta  \beta_0 X_1 \beta_1 \cdot\cdot\cdot X_n \beta_n$ (an APCPS process $\bar{\pi} \eqvI[{\IofG[\calG]}]$\linebreak[4] $\bar{\delta}\beta X_1 \beta_1 \cdot\cdot\cdot X_n \beta_n$) is \emph{$K$-shaped}. An ACPS (APCPS) configuration $\Pi \ChanPar \Gamma$ is $K$-shaped if all processes in $\Pi$ are $K$-shaped and we say an ACPS $\calP$ (an APCPS $\calG$) has \emph{$K$-shaped stacks from $\Pi_0 \ChanPar \Gamma_0$} just if all reachable configurations from $\Pi_0 \ChanPar \Gamma_0$ are $K$-shaped.
Intuitively, the shaped constraint requires that, at all times, at most an \emph{a priori} fixed number $K$ of non-commutative non-terminals $K$ may reside in the stack. Because the restriction does not apply to commutative non-terminals, stacks can grow to arbitrary heights.

Given a simple coverability query $\mathpzc{Q}\!=\!(\calP, \Pi_0\!\ChanPar\!\Gamma_0, \Pi\!\ChanPar\!\Gamma)$ for an ACPS in normal form $\calP$ we may analyse $\mathpzc{Q}'= \commabr{(\GofP, \Pi_0\ChanPar\Gamma_0, \Pi\ChanPar\Gamma)}$ a simple coverability query for the APCPS $\GofP$ instead; the reduction preserves the shape constraint:
\begin{proposition}\label{prop:acps:apcps:interreduction}
$\mathpzc{Q}$ is a yes-instance, if and only if, $\mathpzc{Q}'$ is a yes-instance.
Hence simple coverability for ACPS and APCPS poly\-no\-mi\-al-time inter-reduce. 
Further $\calP$ is $K$-shaped from $\Pi_0\ChanPar\Gamma_0$ if, and only if, $\GofP$ is $K$-shaped from $\Pi_0\ChanPar\Gamma_0$.
\end{proposition}
\noindent Deciding coverability, boundedness or termination on the standard semantics looks daunting. 
Even a $K$-shaped process is infinite state, follows a stack discipline and may synchronise with an unbounded number of other processes. Since APCPS subsume ordinary concurrent pushdown systems, simple coverability is undecidable in general for the standard semantics.
However, the independence relation $\eqvI[{\IofG[\calG]}]$ enables a simplification which is formalised in the \emph{alternative semantics}.
The key idea is to summarise the effects of commutative non-terminals.
In the alternative semantics, rather than keeping track of the contents of the stack, 
we precompute the actions produced by commutative non-terminals
and store them in summaries on the stack. The non-com\-mu\-ta\-tive procedure calls, which are left on the stack, then act as separators for these summaries. 

\paragraph*{Alternative Semantics} 
In the alternative semantics a process, which we denote by $\gamma'$, has the shape 
$\delta M X_1 M_1\!\cdots\!X_n M_n$, with $X_i\!\in\!\NComN$, $M, M_i\!\in\!\M[\NonT\!\union\!\ComSigma]$ and $\delta \in \Sigma \cup \NonT \cup \set{\epsilon}$, and is said to be $K$-shaped if $n \leq K$. 
We denote the set of such processes by $\Procs[]'$.
Then the alternative concurrent semantics of an APCPS is a transition system over elements of $\M[\Procs[]'] \times \Channels$. We abbreviate a set of alternative processes running in parallel as $\Pi'$.

The right column of 
Table~\ref{fig:standard_vs_alternative:semantics} shows the alternative semantics along-side the standard semantics. Most transition rules, barring the different shape of processes, are essentially the same as the standard semantics.
The rules that are different implement and manage the summary of commutative non-terminals.
Rule \Cref{altrule:interleave:com} executes a rule $A \to B\, C$ by precomputing the actions $w$ of the  commutative non-terminal $C$ and inserting $w$'s \emph{Parikh image} $\M(w)$ into the summary $M$. This is the counterpart of a push in the alternative semantics.
The rules \Cref{altrule:dispatch:term} and \Cref{altrule:dispatch:nonterm} are the pop counterparts; they ensure that the precomputed actions are rendered effective at the appropriate moment. 
Rule \Cref{altrule:dispatch:term} is applicable when the summary $M$ contains exclusively commutative actions; such a summary denotes a sequence of commutative non-terminals whose computation terminates and generates concurrency actions. Rule \Cref{altrule:dispatch:nonterm} handles the case where the summary $M'$ contains non-terminals. Such a summary represents a partial computation of a sequence of commutative non-terminals. In this\linebreak case rule \Cref{altrule:dispatch:nonterm} dispatches all commutative actions and then blocks. It is necessary to consider this case since not all non-terminals have terminating computations. Thus rule \Cref{altrule:interleave:com} may non-deterministically abandon the pre-computation of actions. 

Again we can turn $\calG$ with $\toCM$ into a PSTS $(\M[\Procs[]']\times\Channels,\toCM,\leqAPCPS['])$ by endowing it with a preorder $\leqAPCPS[']$.
We order elements of $\M[\Sigma \union \NonT]$ with the usual multiset ordering $\leq_{\M[\Sigma \union \NonT]}$ and elements of $\Procs[]'$ are ordered by $\leq_{\Procs[]'} \is \higleq[{(\Sigma \union \NonT \union \M[\Sigma \union \NonT])}]$ which is lifted to configurations as before.
We say a coverability query $(\calG,\Pi_0 \ChanPar \Gamma_0,\Pi' \ChanPar \Gamma)$ for the alternative semantics is \emph{simple} if $\pi' = A \in \NonT$ for all $\pi' \in \Pi'$ and $\Pi_0$.
The standard and alternative semantics give rise to the same yes-instances of simple coverability:
\begin{theorem}[\cite{KochemsO:2013}]\label{thm:reduction_coverability}
A query $\mathpzc{Q}$ is a yes-instance for simple coverabililty on the standard semantics if, and only if, $\mathpzc{Q}$ is a yes-instance on the alternative semantics.
\end{theorem} \noindent
The alternative semantics for shape-constrained APCPS gives rise to a \emph{well-structured transition system} (WSTS) \cite{Finkel:01}, which implies the decidability of simple coverability for both alternative and standard semantics \cite{KochemsO:2013}. However, this result yields only coarse complexity results.
To study the complexity of simple coverability, we introduce a non-trivial extension of Petri nets, \emph{nets with nested coloured tokens}.
\newcommand{\upwardclosure}[1]{\mathord{\uparrow} #1}
\section{Nets with Nested Coloured Tokens}\label{sec:NNCT}\noindent
The alternative operational semantics for APCPS requires the ability to model configurations
that contain multisets of multisets --- a capability that appears to be beyond Petri nets.
Rules \Cref{altrule:dispatch:term} and \Cref{altrule:dispatch:nonterm} seem to require
a feature that allows the transfer or ``ejection'' of elements from inside a nested multiset. Fortunately, we know from
the literature \cite{Esparza:1997,Ganty:2012,EsparzaGKL:2011} that computing the summary of a commutative non-terminal,
as performed by rule \Cref{altrule:interleave:com}, can be achieved by a Petri net.

Inspired by \emph{nested Petri nets} \cite{LomazovaS:99}, which give rise to configurations of nested structures of multisets, we introduce \emph{nets with nested coloured tokens} (NNCT) which feature multisets of multisets and vertical transfers. 
NNCT is designed with the necessary features to implement the alternative semantics for APCPS, while also allowing APCPS to simulate NNCT.

\begin{definition}
A \emph{net with nested coloured tokens} (NNCT) is a quintuple $\calN = (\Psimple,\Pcomplex,\Pinner,\Rules,\colmap)$ where 
$\Psimple, \Pcomplex, \Pinner$ and $\Rules$ are the finite sets of \emph{simple places}, 
\emph{complex places}, \emph{colours}, and\hspace{-0.1mm} \emph{rules};\hspace{-0.1mm} and\hspace{-0.1mm} the\hspace{-0.1mm} \emph{colour\hspace{-0.1mm} mapping}\hspace{-0.1mm} $\colmap$\hspace{-0.1mm} is\hspace{-0.1mm} a\hspace{-0.1mm} map\hspace{-0.1mm} from\hspace{-0.1mm} $\Pinner$\hspace{-0.1mm} to\hspace{-0.1mm} $\Psimple$.

\subparagraph{Markings.}
We define \emph{markings} for $\Psimple, \Pinner$ and $\Pcomplex$ as
\begin{align*}
\Msimple &\is \set{m \mid m : \Psimple \to \M[\set{\bullet}]}\\
\Minner &\is \set{m \mid m : \Pinner \to \M[\set{\bullet}]}\\
\Mcomplex &\is \set{m \mid m : \Pcomplex \to \M[\Minner]}.
\end{align*} 
We also call a marking for $\Pinner$ (i.e.~an element of $\Minner$) a \emph{complex token}, which we view as a multiset of \emph{coloured tokens}. Thus a marking for $\Pcomplex$ fills each complex place with a multiset of complex tokens. A \emph{configuration}
of a NNCT is the disjoint union of a marking for $\Psimple$ and a marking for $\Pcomplex$, i.e.
$$\Configuration \is \set{m + m' \mid m \in \Msimple, m' \in \Mcomplex}. $$

\subparagraph{Rules.}
We partition $\Rules$ into 
$\SimpleRules$, $\ComplexRules$ and $\TransferRules$.
A rule $r\!\in\!\SimpleRules$ is a pair $(I,O)$ 
 such that $I,O \in \Configuration$ where
$I(p)(m) = 0$ if both ${p \in \Pcomplex}$ and $m \neq \vec{0}$, 
where, as a reminder, $\vec{0}(x) = \emptyset$ for any $x$. 
A rule $r \in \ComplexRules$ is a pair $((p,I),(p',\calc,O))$ such that $I,O \in \ConfigSimple \is \varset{m + \vec{0} \mid m \in \Msimple}$; $p, p' \in \Pcomplex$ and ${\calc \in \Minner}$. 
A rule $r \in \TransferRules$ is a pair {$((p,I),(p',P,O))$} such that $p, p' \in \Pcomplex$; $I,O \in \ConfigSimple$ and
$P \subseteq \Pinner$ such that $\colmap \restriction P : P \to \colmap(P)$ is bijective, i.e., $\colmap^{-1}$ is well-defined on $\colmap(P)$.
We sometimes need to refer to the mappings $I, O$ of a rule $r$, for which we write $I_r$ and $O_r$.

\subparagraph{Operational Semantics.}
Let $s \in \Configuration$.
\begin{asparaenum}[(R1),topsep=0pt,partopsep=0pt,parsep=0pt]
\item  Suppose $r = (I,O) \in \SimpleRules$. If 
rule $r$ is \emph{enabled} at $s$ 
i.e. $s = s_0 \oplus I$ for some $s_0 \in \Config$, then {$s \to[rule=r,TS=\calN] s'$}
where \changed[jk]{$s' := s_0 \oplus O$.} 

\item  Suppose $r = ((p,I),(p',\calc,O)) \in \ComplexRules$. If rule $r$ is \emph{enabled} at $s$ 
i.e. $s = s_0 \oplus I \oplus \update{}{p}{\mset{m}}$ for some $s_0 \in \Config$,
then {$s \to[rule=r,TS=\calN] s'$} where 
$s' = s_0 \oplus O \oplus \update{}{p'}{\mset{m \oplus \calc}}$.

\item Suppose $r = 
((p, I), ({p'}, P, O)) \in \TransferRules$. If rule $r$ is \emph{enabled} at $s$ 
i.e. $s = s_0 \oplus I \oplus \update{}{p}{\mset{m}}$ for some $s_0 \in \Config$,
then {$s \to[rule=r,TS=\calN] s'$} such that 
\[
s' = s_0 \oplus O \oplus \paren{m_{P} \compose \colmap^{-1}} \oplus \update{}{p'}{\mset{m_{\overline{P}}}}
\]
where $m_{P} = m \restriction P$ and $m_{\overline{P}} = m \restriction (\Pinner \setminus P)$.
\end{asparaenum}
\end{definition}\noindent
A simple rule may insert new complex tokens into complex places; 
it may also remove empty complex tokens. A complex rule removes a complex token $m$ from a complex place $p$ and inserts a new complex token $m \oplus \calc$ to a complex place $p'$. A set of \emph{active colours} $P \subseteq \Pinner$ is associated to each transfer rule, which removes a complex token $m$ from a complex place $p$, and inserts into $p'$ the complex token $m_{\overline{P}}$ which is obtained from $m$ \emph{less} all tokens with an active colour $c$; for each such $c \in P$, these tokens are transferred to the simple place $\colmap(c)$ which corresponds 
to a \emph{multiset-addition} of $m_{P} \compose \colmap^{-1}$.

\subparagraph{Subclasses of NNCT.}
Let $\calN = (\Psimple,\Pcomplex,\Pinner,\Rules,\colmap)$ be an NNCT. If a transfer rule $r=((p,I),(p',P,O))$
is such that $P = \Pinner$ then we say $r$ is a \emph{total} transfer rule. If all transfer rules of $\calN$
are total then we say $\calN$ is a \emph{total transfer} NNCT.

\newcommand\redPl{\mathit{red}}
\newcommand\greenPl{\mathit{green}}
\newcommand\bluePl{\mathit{blue}}
\newcommand\blackPl{\mathit{black}}
\begin{figure}[t]
\def\svgwidth{\columnwidth}
\ifsvg%
\executeiffilenewer{example_net.svg}{example_net.pdf}%
{/Applications/Inkscape.app/Contents/Resources/bin/inkscape -z -D --file=example_net.svg %
--export-pdf=example_net.pdf --export-latex}%
\addFileDependency{example_net.pdf_tex}%
\addFileDependency{example_net.aux}%
 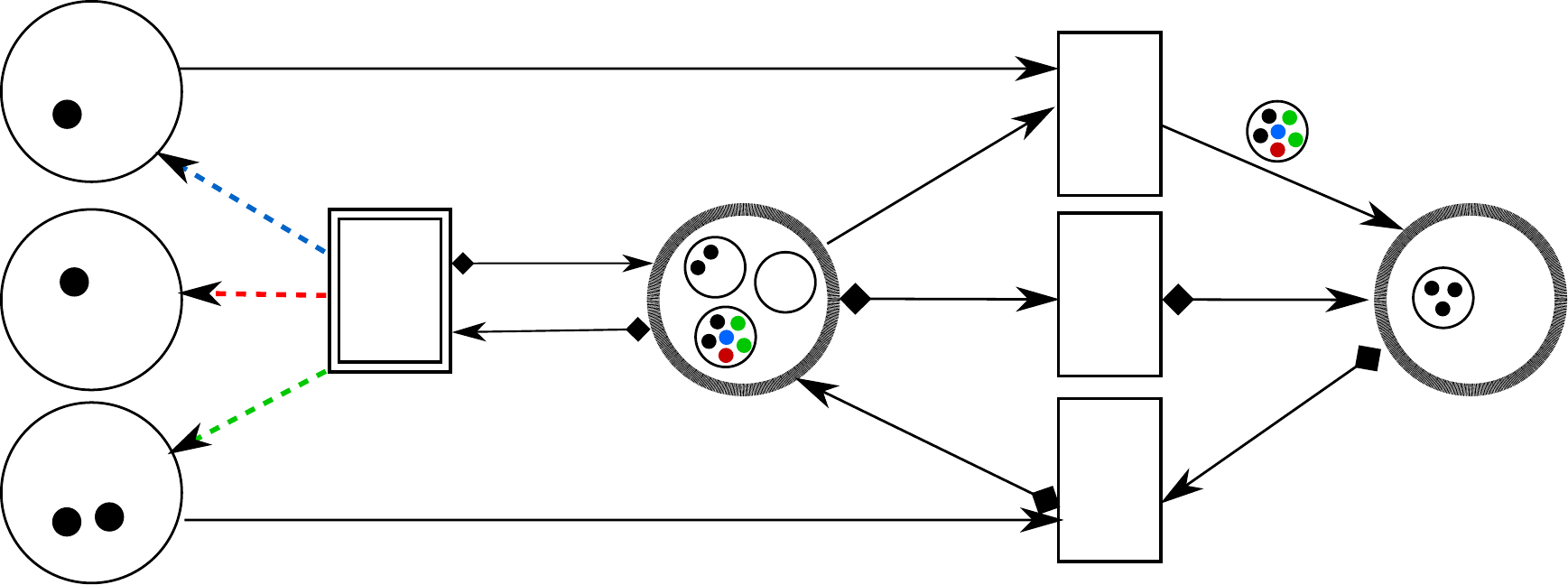
\fi

\vspace{-10pt}
\caption{The net of Example~\ref{eg:net} \label{fig:eg-net}}
\vspace{-5.29ex}%
\end{figure}
\begin{example}\label{eg:net}
We define a NNCT $\calN$ with places: $\Pcomplex=\makeset{p_1, p_2}, \discn \Psimple=\makeset{p_3, p_4, p_5}$, and 
$\Pinner = \makeset{\redPl, \greenPl, \bluePl, \blackPl}$, and colour mapping: $\colmap: \redPl \mapsto p_4, \greenPl \mapsto p_5, \bluePl \mapsto p_3$. 
The rule $r_1 = (\varset{p_1 \mapsto \mset{\vec{0}}, p_3 \mapsto \mset{\bullet}}, \varset{p_2 \mapsto \mset{m_1}})$ is simple where the complex token $m_1 = \varset{\blackPl \mapsto \varmset{\bullet^2}, \bluePl \mapsto \varmset{\bullet},\discn \greenPl \mapsto \varmset{\bullet^2},\redPl \mapsto \varmset{\bullet}}$.
The complex rules are:
\begin{align*}
r_2 &= ((p_1, \emptyset), (p_2, \varset{\blackPl \mapsto \mset{\bullet}}, \emptyset))\\
r_3 &= ((p_2, \varset{p_5 \mapsto \mset{\bullet}}), (p_1, \emptyset, \emptyset))
\end{align*}
The transfer rule is $r_4 = ((p_1,\emptyset), (p_1, \makeset{\redPl, \bluePl, \greenPl},\emptyset))$. 
Note that $\calN$ is \emph{not} a total transfer NNCT.
We graphically represent NNCT similarly to Petri nets.
In \Cref{fig:eg-net} we show a configuration of $\calN$. The complex place $p_1$ contains the complex token $m_1$, the empty complex token $\vec{0}$ (displayed as an empty circle) and the complex token $m_2 = \varset{\blackPl \mapsto \varmset{\bullet^2}}$.
The complex token $m_3 = \varset{\blackPl \mapsto \varmset{\bullet^3}}$ is located in complex place $p_2$. We  distinguish transfer rules by using double edged boxes and we display $\colmap$ and the set of active colours using dashed arrows. We indicate the origin and destination for complex tokens moved by complex and transfer rules with additional arcs that have a $\vardiamond$ end.
Boxes for complex rules are labelled with the colour marking $\calc$ to inject. Simple rules may have arcs
\emph{from} complex places labelled with $\vec{0}$ (indicating the removal of $\vec{0}$) and arcs \emph{to} complex places labelled with the complex token to be added. 
Rule $r_1$ removes an empty complex token from $p_1$ and a simple token from $p_3$, and adds the complex token $m_1$ to $p_2$. The complex rule $r_2$ non-deterministically selects a complex token e.g. $m_2$ and moves it to $p_2$ while
inserting one $\blackPl$-coloured token, i.e.~$m_2$ becomes $m_3$. 
%
Let us imagine for a moment that the simple places $p_3$, $p_4$, and $p_5$ are all empty. 
Rule $r_4$ non-deterministically selects a complex token in $p_1$, $m_1$ say, and distributes
its $\bluePl$, $\greenPl$ and $\redPl$-coloured tokens as indicated by the coloured dashed arrows to the simple
places $p_3$, $p_4$ and $p_5$ and turns them into simple (black) tokens. The result is that $p_3$ and $p_4$ contain a simple token each and $p_5$ two as displayed in \Cref{fig:eg-net}; the token $m_1$, less its $\bluePl$, $\greenPl$ and $\redPl$-coloured tokens, remains in place $p_1$ and becomes $m_2$.
\end{example} 

\subparagraph{NNCT are WSTS.}
We recall a few definitions:
let $(U,\leq)$ be a preordered set; we say $\leq$ is a \emph{well-quasi-order} (WQO) if for all infinite sequences $u_1,u_2,...$ there exist $i < j$ such that\linebreak[4] $u_i \leq u_j$. A WSTS is a PSTS $\calS=(S,\to[TS=\calS],\leq_{\calS})$ such that $\leq_{\calS}$ is a WQO and ${\to[TS=\calS]}$ is monotone with respect to $\leq_{\calS}$; 
a WSTS is said to be \emph{strict} if $\to[TS=\calS]$ is strictly monotone: 
we say $\to[TS=\calS]$ is \emph{(strictly) monotone} if $s \to[TS=\calS] s'$ and $s <_{\calS} t$ implies that there exists $t'$ such that $t \to[TS=\calS] t'$ and $s' \leq_{\calS} t'$ ($s' <_{\calS} t'$ respectively).

\noindent
We equip $\Configuration$ and $\Minner$ with preorders: let ${m, m'\!\in\!\Minner}$, we define $m \leqMinner m'$ if for all $p^{\I} \in \Pinner$ we have either $0 = |m(p^{\I})| = |m'(p^{\I})|$, or $0 < |m(p^{\I})| \leq |m'(p^{\I})|$. For $s,s'\!\in\!\Configuration$
we define $s\!\leqconfig\!s'$ if for all $p\!\in\!\Psimple$ we have $s(p)\!\leq_{\M[\set{\bullet}]}\!s(p')$
and $\forall p'\!\in\!\Pcomplex$ we have $s(p')\!\leq_{\M[\Minner]}\!s(p')$. 

\begin{lemma}\label{lem:config:wqo}
$(\commabr{\Minner,\leqMinner})$ and $(\commabr{\Configuration,\leqconfig})$ are WQO.
\end{lemma}\noindent
Thanks to the use of the refined order $\leqMinner$ 
the transition relation $\to[TS=\calN]$ for a NNCT $\calN$ is strictly monotone: 

\begin{proposition}\label{prop:nnct:wsts}
$(\Configuration,\to[TS={\calN}],\leqconfig)$ is a strict WSTS.
\end{proposition}\noindent
Let $\upwardclosure{S}\!=\!\varset{s : \exists s_0\!\in\!S, s_0\!\leq_{\calS}\!s}$ and $\Pred(S)\!=\!\varset{s : {s\!\to[TS=\calS]\!s'},\linebreak[4] s'\!\in\!S}$.
For a WSTS $\calS$, coverability, termination and bound-\linebreak[4]ed\-ness are decidable \cite{Finkel:01} provided that $\upwardclosure{\Pred(\upwardclosure{\{s\}})}$ is effectively computable for any $s \in S$, the WQO $\leq_\calS$ is decidable and (for bound\-ed\-ness:) $\calS$ is a strict WSTS.

\begin{theorem}\label{thm:nnct:cov:term:bound:decidable}
Coverability, termination and boundedness are decidable for NNCT.
\end{theorem}\noindent
By construction, NNCT can implement the alternative semantics of APCPS and \emph{vice versa}. We encode processes as complex tokens and allocate a colour for each pair in $\Sigma \times \varset{0,...,K}$ in order to encode summaries as coloured tokens.
Channels are encoded in simple places and rule \Cref{altrule:interleave:com} is implemented by a slightly modified CCFG widget \emph{\`a la} \citeauthor{Ganty:2012}~\cite{Ganty:2012}.
The resulting NNCT has $O(n\cdot|\Sigma||\NonT|)$ simple places, $O((|\NonT||\Sigma|)^{O(K)})$ complex places, 
$O(K\cdot|\Sigma|)$ colours and $O(n \cdot |\calR| \cdot (|\NonT||\Sigma|)^{O(K)})$ rules where $A_1 \parallel \cdots \parallel A_n \ChanPar \Gamma$ is the configuration to be covered.

\begin{theorem}\label{thm:APCPS-to-NNCT}
Simple coverability for $K$-shaped APCPS in the al\-ter\-na\-tive semantics \Exptime-time reduces to NNCT coverability.
\end{theorem}\noindent
Given a total-transfer NNCT $\calN$, we can construct a simulating APCPS $\calG$ that uses channels to encode simple places. 
A process with a single summary simulates a complex token and the coloured tokens it carries. 
$\calN$ runs
 a control process that executes rules by manipulating the state of processes and channels by communication through auxiliary channels. 
 An injection of a multiset of coloured tokens $\calc$ by a complex rule is simulated by forcing a process $\pi$ to add the summary of a non-terminal representing $\calc$ to $\pi$'s summary. 
 A transfer rule is simulated by\linebreak[4] forcing a process encoding a complex token to dispatch its summary. 
 Simple coverability may then be decided on $\calG$. A\linebreak[4] NNCT coverability query $\mathpzc{Q}\!=\!(\calN\!,s_0,s_{\text{cov}})$ is \emph{simple}, if $s_0$ and\linebreak[4] $s_{\text{cov}}$ contain \emph{no} complex tokens: $\forall p\!\in\!\Pcomplex\,\,\,s_0(p)\!=\!s_{\text{cov}}(p)\!=\!\emptyset$.

\begin{theorem}\label{thm:NNCT-to-APCPS}
Simple coverability, boundedness and termination for a total-transfer NNCT
\Exptime\ reduces to simple coverability, boundedness and termination respectively for a $\mathit{4}$-shaped APCPS in the alternative semantics.
\end{theorem}\noindent
We show in the next two sections that NNCT coverability is \Tower-complete.
\Tower\ was recently introduced by \citeauthor{Schmitz:2013} to provide a meaningful complexity class for decision problems of non-elementary complexity \cite{Schmitz:2013}. The notion of reduction used in \Tower\ are \Elementary-time reductions which give rise to \Tower-complete problems. The \Exptime\ reductions of \Cref{thm:APCPS-to-NNCT} and \Cref{thm:NNCT-to-APCPS} are thus sufficient to infer that coverability for APCPS is \Tower-complete.


\newcommand{\vsimrule}{\sim}
\section{Upper Bound: A Nested Rackoff Argument.}\label{sec:nnct:upperbound}\noindent
Let us fix a NNCT $\calN = (\Psimple,\Pcomplex,\Pinner,\Rules,\colmap)$ with a coverability query $(\calN,s_0,s_\text{cov})$ throughout this section. We enumerate $\calN$'s simple places $\Psimple = \varset{\enumelem{p}{1}, \ldots, \enumelem{p}{\nsimple}}$, complex places $\Pcomplex = \varset{\enumelem{p'}{1}, \ldots, \enumelem{p'}{\ncomplex}}$ and
 colours $\Pinner = \varset{\enumelem{p^{\I}}{1}, \ldots, \enumelem{p^{\I}}{\ninner}}$. 
 
Our proof of \Tower-membership for NNCT coverability is inspired by 
\citeauthor{Rackoff:78}'s method \cite{Rackoff:78}. 
%
The \citeauthor{Rackoff:78} method constructs a bound on the \emph{covering radius} of a given target state $s_\text{cov}$, which is the greatest \emph{covering distance} to $s_\text{cov}$ from an arbitrary state $s$, where the covering distance is the length of the shortest covering path. Typically the covering radius can then be used to establish a bound $\calB$ on the space required for a machine representation of any state along a covering path for $s_\text{cov}$. Depending on the flavour of VAS, it is then easy to see that a $\calB$-space bounded non-deterministic/alternating Turing machine can find a covering path, if there is one, and reject a path if its length exceeds the covering radius.

\begin{definition}
\label{def:covering-radius}
Suppose $\calS = (S,\to[TS=\calS],\leq_{\calS})$ is a PSTS.
We say $\vec{s} \in S^*$ is a \emph{path} if for all $1 \leq i < |\vec{s}|$, $\vec{s}(i) \to[TS=\calS] \vec{s}(i+1)$. A path $\vec{s}$ is \emph{covering} for $s'$ from $s$ if 
$s' \leq_{\calS} \vec{s}(|\vec{s}|)$ and $\vec{s}(1) = s$.

We define $\dist{s}{s'}{\calS}$, the \emph{covering distance} between $s$ and $s'$, 
as follows: 
if there exists a covering path for $s'$ from $s$ 
then 
$$\dist{s}{s'}{\calS} 
:= \min\set{|\vec{s}| \mid \vec{s} \text{
covering path for $s'$ from $s$}};$$
otherwise set $\dist{s}{s'}{\calS}\!:=\!0$.
We define the \emph{covering radius} for $s'$ as $\length{\calS}{s'} \is \sup\set{ \dist{s}{s'}{\calS} \mid s \in S_{\calS}}$, 
and
the \emph{covering diameter} of a subset $S' \subseteq S_{\calS} \times S_{\calS}$,
$\diameter[\calS]{S'}\!\is\!\sup\set{ \dist{s}{s'}{\calS} \mid (s,s') \in S'}$, i.e.~the maximum covering distance between any pair in $S'$.
\end{definition}\noindent
To determine a bound on the covering radius a more general, relativised problem is considered: 
suppose the contents of the last $n-i$ places are \emph{ignored} and the first $i$ places
are \emph{not ignored}, how can we bound the covering radius $\rho_i$ for $s_\text{cov}$? 

%
One way to ignore the contents of simple places in $\calN$ is to let their contents be ``unbounded'': define ${\Msimple}^\infty = \Psimple \to \M^{\infty}[\set{\bullet}]$
where we write $\M^\infty[U]$ for multisets over elements in $U$ with possibly infinite multiplicity, i.e.~the set of functions $U \to \N^\infty$.
Configurations with possibly infinite simple markings can be defined as $\Config^\infty = \varset{m + m' \mid m \in {\Msimple}^{\infty},m' \in \Mcomplex}$. We extend the transition relation $\to[TS=\calN]$ in the obvious way to $\Config^\infty$.
For each $i \leq \nsimple$ we then define a new transition relation $\to[TS=\calN_i]$ that formalises that we 
ignore the last $\nsimple-i$ simple places i.e.~$s \to[TS=\calN_i] s'$ just if $s \to[TS=\calN] s'$ and for all 
$i < j \leq \nsimple$ $s(\enumelem{p}{j})(\bullet) = s'(\enumelem{p}{j})(\bullet) = \infty$.
Writing $\calN_i = (\Config^\infty, \to[TS=\calN_i],\discn\leqconfig)$ the relativised problem is then to determine the 
covering radius $\length{\calN_i}{s_\text{cov}}$.

The core argument underlying the general \citeauthor{Rackoff:78} method establishes a recurrence relation that relates $\rho_{i+1}$ to $\rho_{i}$. 
This is achieved by a careful analysis of the size or \emph{norm} of con\-fig-u\-ra\-tions that occur along covering paths for $s_\text{cov}$. 
In general, we say $\norm{-}$ is a \emph{norm} for a set $S$ if $\norm{-}$ is a function from $S$ to $\N^{\infty}$ and we extend $\norm{-}$ to a norm $\norm{-}^*$ on sequences $S^*$ in the usual way $\norm{\vec{s}}^* = \max\set{\norm{\vec{s}} : 1 \leq i \leq |\vec{s}|} $. 
The general \citeauthor{Rackoff:78} method introduces a family of norms $\varparen{\norm[i]{-}}_{i=0}^{n}$ where each 
$\norm[i]{-}$ ignores the contents of the last $(n-i)$ places. 
Two facts are then established:
\begin{inparaenum}[(i)]
\item there exists an upper bound $B(\rho_i)$ on how many tokens can be removed in a path of length $\rho_i$; and
\item for any covering path $\vec{s}$ for $s_\text{cov}$ one of two cases apply: either
\begin{inparaenum}[(C$_1$)]
\item(path:case:1) $\norm[i+1]{\vec{s}}^* < B(\rho_i)$; or
\item(path:case:2)
$\vec{s}$ can be split at a \emph{pivot} $s_p$, i.e.~$\vec{s} = \vec{s}_1 \cdot s_p \cdot \vec{s}_2$, such that:
\begin{inparaenum}[(P$_1$)]
\item(pivot:property:1) $\norm[i+1]{\vec{s}_1}^* < B(\rho_i)$; and
\item(pivot:property:2) one \emph{not} ignored place of $s_p$, the $i+1$'s say, contains more than $B(\rho_i)$ tokens.
\end{inparaenum}
\end{inparaenum}
\end{inparaenum}
As a consequence of \ref{pivot:property:2} we may ignore the contents of the $i+1$'s place for any covering path
 for $s_{\text{cov}}$ from $s_p$ and thus we can replace $\vec{s}_2$ by a path $\vec{s}'_2$ with length at most $\rho_i$. 
 The recurrence relation is then established by noting that a path $\vec{s}$ to which case \ref{path:case:1} applies
 and a path $\vec{s}_1$ of case \ref{path:case:2} (which satisfies property \ref{pivot:property:1}) may be replaced
 by a path of length no more than the covering diameter $\mathpzc{d}_{i+1}$ of the set 
 $\set{(\vec{s}(1),\vec{s}(|\vec{s}|)) : \vec{s} \text{ path}, \norm[i+1]{\vec{s}}^{*} < B(\rho_i)}$. 
 This yields 
 \emph{the \citeauthor{Rackoff:78} recurrence relation} $\rho_{i+1}\!\leq\!\mathpzc{d}_{i+1}\!+\!\rho_i$.

It is challenging to apply \citeauthor{Rackoff:78}'s method to NNCT 
and we are forced to adjust the above argument in a few technical details. 
In the setting of NNCT we introduce two families of norms $\norm[\calN_i]{-}$ and $\norm[\calN_i;C]{-}$ on 
$\Config^{\infty}$. The norm $\norm[\calN_i;C]{-}$ ignores the simple places $\enumelem{p}{i+1},\ldots,\enumelem{p}{\nsimple}$ and the norm $\norm[\calN_i]{-}$ also ignores coloured tokens and complex places, i.e. for $s \in \Config^{\infty}$ we define
$\norm[\calN_i]{s} = \max\varset{|s(\enumelem{p}{j})| : j \leq i}$ and 
$\norm[\calN_i;C]{s} = \max\varparen{\varset{|s(\enumelem{p'}{j})|, \max_{m \in s(\enumelem{p'}{j})} \norm[\I]{m}
 : j \in \range{\ncomplex}} \union \varset{\norm[\calN_i]{s}}}$ where we further define the norm $\norm[\I]{-}$ on complex tokens by $\norm[\I]{m} = \max\varset{|m(\enumelem{p^{\I}}{j})| : j \in \range{\ninner}}$.

\subsection{Establishing the Rackoff Recurrence Relation}\noindent
First, we define bounds that are our key ingredients of $B(\rho_i)$:
\begin{align*}
R &:= \max\paren{\set{|I_r(p)|,|O_r(p)| \mid r \in \Rules, p \in \Psimple \union \Pcomplex}\union\set{1}},\\
R' &:= R + 1 + \changed[jk]{\max(\Xi \union \varset{\norm[\calN_{\nsimple};C]{s_{\text{cov}}}})}\text{ where}\\
&\Xi := \set{\norm[\I]{\calc} \mid \calc \in O_r(p'), p' \in \Pcomplex, r \in \SimpleRules}
\end{align*}\noindent
\begin{lemma}\label{lemma:pivot_strengthening:tools:i}
Suppose $\vec{s}$ is a covering path for $s_{\text{cov}}$ in $\calN_{i+1}$ and $|\vec{s}(1)(\enumelem{p}{i+1})| \geq R' \cdot \length{\calN_i}{s_{\text{cov}}}$. Then $\exists \vec{s}'$ a covering path
	   for $s_{\text{cov}}$ in $\calN_{i+1}$ such that $\vec{s}'(1) = \vec{s}(1)$ and $|\vec{s}'| \leq  \length{\calN_i}{s_{\text{cov}}}$.
\end{lemma}%
%
\noindent%
We then establish that covering paths fall into two categories.
%
\begin{lemma}\label{lem:rackoff_case_split}
For all covering paths $\vec{s}$ for $s_{\text{cov}}$ in $\calN_{i+1}$ either:
\begin{asparaenum}[(C$_1$)]
\item(lem:rackoff_case_split:c-1) $\norm[\calN_{i+1}]{\vec{s}}^* < R' \cdot \length{\calN_i}{s_{\text{cov}}}$; or
\item(lem:rackoff_case_split:c-2)  $\vec{s} = \vec{s}_1 \cdot s_p \cdot \vec{s}_2$ such that
$\norm[\calN_{i+1}]{\vec{s}_1}^* < R' \cdot \length{\calN_i}{s_{\text{cov}}}$ and
$\norm[\calN_{i+1}]{s_p} \geq R' \cdot \length{\calN_i}{s_{\text{cov}}}$. 
\end{asparaenum}
\end{lemma}\noindent
Unfortunately, we need stronger guarantees on the pivot con-\linebreak[4]fig\-u\-ra\-tion $s_p$. 
We require
that there exists a ``small'' predecessor $s_{p'}$ of $s_p$, i.e.~
$\norm[\calN_{i+1};C]{s_{p'}} < R' \cdot (\length{\calN_i}{s_{\text{cov}}} + 1)$, that 
is\linebreak[4] covered by $\vec{s}_1$. Of course, this does not hold for all pivot configurations, however, we can construct a pivot with this property.
We exploit property \ref{pivot:property:2} and two observations: 
let $\vec{s}_0$ be a path with $|\vec{s}| \leq L$ then 
\begin{inparaenum}[(O$_1$)]
\item along $\vec{s}_0$ only $L$ complex tokens can be moved/removed; and
\item for any complex token $m$ and occurring in $\vec{s}_0(1)$ at most $L \cdot R'$ carried coloured tokens of a given colour can be ejected and removed along $\vec{s}_0$.
\end{inparaenum}
\begin{lemma}\label{lemma:pivot_strengthening:tools:ii}
Suppose $\vec{s}$ is a covering path for $s_{\text{cov}}$ in $\calN_{i+1}$.
If $L \in \N$ such that $|\vec{s}| \leq L$ and $\norm[\calN_{i+1}]{\vec{s}(1)} < L \cdot R'$ then there exists a covering path $\vec{s}''$ for $s_{\text{cov}}$ in $\calN_{i+1}$ such that $\vec{s}''(1) \leqconfig \vec{s}(1)$, $\norm[\calN_{i+1};C]{\vec{s}''(1)} < L \cdot R'$, and $|\vec{s}''| \leq L$.
\end{lemma}\noindent
Superfluous complex and coloured tokens can thus be removed along $\vec{s}_1(|\vec{s}_1|) \cdot s_p \cdots$ to strengthen Lemma \ref{lem:rackoff_case_split}:
\begin{corollary}\label{cor:rackoff_case_split}
For all covering paths $\vec{s}$ for $s_{\text{cov}}$ in $\calN_{i+1}$ either
\begin{asparaenum}[(C$'_1$)]
\item(cor:rackoff_case_split:c-1) $\norm[\calN_{i+1}]{\vec{s}}^* < R' \cdot \length{\calN_i}{s_{\text{cov}}}$; or
\item(cor:rackoff_case_split:c-2) there exist paths $\vec{s}_1$ and $s_{p'} \cdot \vec{s}_2$ such that 
$\vec{s}_1$ is a covering path for $s_{p'}$, $\vec{s}_1(1) = \vec{s}(1)$ and
$\norm[\calN_{i+1}]{\vec{s}_1}^* < R' \cdot \length{\calN_i}{s_{\text{cov}}}$;
$s_{p'} \cdot \vec{s}_2$ is a covering path for $s_{\text{cov}}$ and
$|\vec{s}'_2| \leq \length{\calN_i}{s_{\text{cov}}}$; and
$\norm[\calN_{i+1};C]{s_{p'}} \leq R' \cdot \length{\calN_i}{s_{\text{cov}}}$. 
\end{asparaenum}
\end{corollary}\noindent
Let us write $B_i\! =\!R'\!\cdot\!\length{\calN_i}{s_{\text{cov}}}$, define $P_{(i,B)}$ to be the set of paths of norm bounded by $B$, i.e.~$P_{(i,B)} = \set{\vec{s} : \norm[\calN_{i}]{\vec{s}}^{*} < B}$, and $S_{(i,B)}$ to be the set
of pairs $(s,s')$ for which there exists a covering path in $P_{(i,B)}$ from $s$ for $s'$, i.e.~$S_{(i,B)} = \varset{(\vec{s}(1), s') : \vec{s} \in P_{(i,B)}, \norm[\calN_{i};C]{s'} \leq B, s' \leqconfig \vec{s}(|\vec{s}|)}$. 
Inspecting case \ref{cor:rackoff_case_split:c-2} we notice that $(\vec{s}_1(1), s_{p'}),  \in S_{(i,B_i)}$.
Hence, we can find a path $\vec{s}'_1$ of length at most $\prediameter[\calN_{i+1}]{\!(S_{(i+1,B_i)})}$ to replace $\vec{s}_1$. We can thus establish the \citeauthor{Rackoff:78} recurrence relation for NNCT:

\begin{proposition}\label{prop:rackoff_recurrence}
\begin{asparaenum}[(i)]
\item $\length{\calN_{0}}{s_{\text{cov}}} \leq \diameter[\calN_{0}]{S_{(0,R')}} $
\item $\length{\calN_{i+1}}{s_{\text{cov}}} \leq \diameter[\calN_{i+1}]{S_{(i+1,B_i)}} + 
\length{\calN_i}{s_{\text{cov}}}$.
\end{asparaenum}
\end{proposition}

\newcommand{\cabspetrinet}[1][i,B]{\ensuremath{\calV_{#1}}}
\newcommand{\cabspnetord}[1][i,B]{\ensuremath{\calV^{\leq}_{#1}}}
\newcommand{\cabspnetGen}[1][i,B]{\ensuremath{\calV_{#1}}}
\newcommand{\cabspnetdim}[1][i,B]{\ensuremath{\ind[#1]{d}}}
\newcommand{\cabspnetdimGen}[1][i,B]{\ensuremath{d_{i,B}}}
\newcommand{\cabspnetrules}[1][i,B]{\ensuremath{\ind[#1]{F}}}
\newcommand{\cabspnetrulesGen}[1][i,B]{\ensuremath{F_{i,B}}}

\newcommand{\cabspnetordrules}[1][i,B]{\ensuremath{\ind[#1]{F}^{\leq}}}
\newcommand{\cabspnetstart}{\ensuremath{\ind[0]{s}}}
\newcommand{\ind}[2][]{\ensuremath{\widehat{#2}_{#1}}}

\newcommand{\cabsfun}[1][i,B]{\ensuremath{\alpha_{#1}}}
\newcommand{\cconfunI}[1][B]{\ensuremath{\gamma^{\I}_{#1}}}

\subsection{The Covering Diameter of Bounded Paths}\label{sec:covdiameter}\noindent
The problem is thus reduced to bounding $\prediameter[\calN_{i}]{\!(S_{(i,B)})}$. For this purpose we expose a function $\cabsfun[i,B]$ and a Petri net $\cabspnetGen$ such that $\cabsfun[i,B]$ maps configurations of $\calN_i$ to configurations of 
$\cabspnetGen$. 
\begin{definition} A \emph{Petri net} is a tuple $\calV = \paren{d,F}$ where $d\!\in\!\N^{+}$ and $F\!\subseteq\!\Z^d\!\times\!\N^d$.
For $s,s'\!\in\!\N^d$ and $(\vec{A},\vec{B})\!\in\!F$ we have $s \to[rule={(\vec{A},\vec{B})},inl]_{\calV} s'$ if $s \geq \vec{B}$, $s' = s + \vec{A}$ and $s'\in \N^d$.
\end{definition}\noindent
We can think of $\cabspnetGen$ as a counter abstraction on $\calN_i$ which preserves covering paths (restricted to $P_{(i,B)}$) and their lengths.
A complex token $m$ that appears along $\vec{s}\in P_{(i,B)}$ cannot carry more than $B$ tokens of a particular colour $\mathit{col}$ if $m$'s $\mathit{col}$-coloured tokens are ejected along $\vec{s}$. This would lead to more than $B$ simple tokens in a configuration along $\vec{s}$ violating $\norm[\calN_{i}]{\vec{s}}^{*} < B$. 
If $m$'s $\mathit{col}$-coloured tokens are \emph{not} ejected along $\vec{s}$ then $m(\mathit{col})$ is immaterial and may be abstracted. 
Hence, for paths in $P_{(i,B)}$, we represent a complex token $m$ as a map $\hat{m}$ of type $\Pinner \to \set{0,1,\ldots,B-1,\infty}$.
Note that there are $(B+1)^{\ninner}$ maps of this type. 
In $\cabspnetGen$ we keep a counter for each $(p',\hat{m})$ 
representing the number of `abstract complex tokens' $\hat{m}$ in complex place $p'$ in addition to $i$ counters representing simple places. Petri net rules can simulate the simple and complex rules of 
$\calN_{i}$ along paths in $P_{(i,B)}$ and since the number of `abstract complex tokens' is finite we can also simulate transfer transitions as normal Petri net rules.
In the interest of readability we relegate the technical definitions of $\cabspnetGen$ and $\cabsfun[i,B]$ to the appendix and summarise their properties: 
\begin{theorem}\label{thm:existence_counter_abstraction_petri_net}
For all $i \leq \nsimple$, $B \in \N$ there exists a Petri net $\cabspnetGen = (\cabspnetdimGen, \cabspnetrulesGen)$ and a function $\cabsfun[i,B]$ such that
\begin{enumerate}[({A}1)]
\item(thm:existence_counter_abstraction_petri_net:i) $\cabspnetdimGen \leq i + (\ncomplex + 1) \times \paren{B+1}^{\ninner}$,
\item(thm:existence_counter_abstraction_petri_net:ii) $R \geq \max\varset{r(i) \mid r \in \cabspnetrulesGen}$, and 
\item(thm:existence_counter_abstraction_petri_net:iii) for all $s, s' \in S_{i,B}$:  $R' \geq \max_{j \in \range{\cabspnetdimGen}}(\cabsfun[i,B](s)(j))$,
\item(thm:existence_counter_abstraction_petri_net:iv) $\dist{s}{s'}{\calN_i} \leq \dist{\cabsfun[i,B](s)}{\cabsfun[i,B](s')}{\cabspnetGen}$.
\end{enumerate}
\end{theorem}
\begin{example}
Suppose we have a NNCT $\calN$ with $\Psimple=\set{p}$, $\Pcomplex = \set{p'}$, $\Pinner = \set{\mathit{g},\mathit{b}}$, $\colmap = \update[\mid c \in \Pinner]{}{c}{p}$, a complex rule $r$ and a transfer rule $r'$:
\begin{align*}
r  &= ((p',\emptyset), (p',\varupdate{}{\mathit{b}}{\mset{\bullet}},\emptyset)), & 
r' &= ((p',\emptyset), (p',\set{\mathit{b}},\emptyset)),
\end{align*}
i.e.~both rules $r$ and $r'$ take a complex token from $p'$ to $p'$; while doing so $r$ injects a $\mathit{b}$-token and $r'$ ejects all $\mathit{b}$-tokens to $p$. 
Consider the following configuration $s$ of $\calN$:
$$s = \preupdate{}{p \mapsto \mset{\bullet^5}, p' \mapsto \mset{m}} \text{ \!where } m = \preupdate{}{\mathit{g} \mapsto \mset{\bullet}, \mathit{b} \mapsto \mset{\bullet^{10}}}\hspace{-2.5pt}.$$ 
In $\cabspetrinet[1,2]$ we have a counter $j_{p',\mathit{g} = 1,\mathit{b} \geq 2}$ that represents complex to\-kens $m$ in $p'$ with $|m(\mathit{g})| = 1$ and $|m(\mathit{b})| \geq 2$. 
We represent $s$ as
the configuration $\cabsfun[1,2](s) = \preupdate{}{j_{p} \mapsto 5, j_{p',\mathit{g} = 1,\mathit{b} \geq 2} \mapsto 1}$ of $\cabspetrinet[1,2]$. 
The rules of $\cabspetrinet[1,2]$ maintain this representation. 
For example, for the complex rule $r$ there is a family of rules $\hat{r}_{m_0}$ in $\cabspetrinet[1,2]$ indexed by `abstract complex tokens'.
One of them, $\hat{r}_{m_0}$ say, is enabled at $\cabsfun[1,2](s)$ and removes a token from $j_{p',\mathit{g} = 1,\mathit{b} \geq 2}$ and adds a token to $j_{p',\mathit{g} = 1,\mathit{b} \geq 2}$.
Since $r$ injects a $\mathit{b}$-token and $j_{p',\mathit{g} = 1,\mathit{b} \geq 2}$ 
represents complex token with more than $2$ $b$-tokens, firing $\hat{r}_{m_0}$ does not change the configuration. To illustrate how transfer rules are encoded, consider a different complex token $m' = \preupdate{}{\mathit{g} \mapsto \mset{\bullet}, \mathit{b} \mapsto \mset{\bullet}}$. 
In $\cabspetrinet[1,2]$, $r'$ leads to a family of the rules $\hat{r}'_{m'}$ one of which removes a token from  
$j_{p',\mathit{g} = 1,\mathit{b} = 1}$, adds one token to $j_{p',\mathit{g} = 1,\mathit{b} = 0}$, and
adds one token to $j_{p}$. Note, the transfer rule $r'$ cannot apply along a path in $P_{(1,2)}$ to a complex token
$m'_0$ with $|m'_0(\mathit{b})| \geq 2$. Thus, for such \emph{`too large'} complex tokens no rule $\hat{r}'_{m'_0}$ exists in $\cabspetrinet[1,2]$. 
We overcome the mismatch between $\leq_{\N^{\cabspnetdimGen}}$ and $\leqconfig$ by adding rules
that allow transitions $\cabsfun[i,B](\preupdate{}{p' \mapsto \mset{m}}) \to[TS={\cabspetrinet[i,B]}] \cabsfun[i,B](\preupdate{}{p' \mapsto \mset{m'}})$ if $m \greaterMinner m'$.
\end{example}

Since $\cabspnetGen$ is a Petri net, we can appeal to a result on bounds on cover radii, due to \citeauthor{Bonnet:12}, that shows that \ref{thm:existence_counter_abstraction_petri_net:i}--\ref{thm:existence_counter_abstraction_petri_net:iii} are enough to control the covering radius $\rho_{\cabspnetGen}$.
\begin{lemma}[{\cite[Lemma 12 specialised to Petri nets]{Bonnet:12}}]\label{lemma:bound_lemma_bonnet}
Let $\calV\!=\!(d,F)$ be a Petri net, $M_\text{cov}$ be the marking to be covered, and 
$R_{\calV} = \max\varset{|\mathbf{A}(i)|, |\mathbf{B}(i)|\!:\!(\mathbf{A},\mathbf{B})\!\in\!F, i\!\in\!\range{d}}$.
If $R'_{\calV}=\max\limits_{i \in \range{d}}\varparen{M_\text{cov}(i), R_{\calV}}$ then
 $\length{\calV\,}{M_\text{cov}}\!\leq\!(6R_{\calV}R'_{\calV})^{(d+1)!}$.
\end{lemma} \noindent
Since $\cabsfun[i,B]$ is an expansive map~\ref{thm:existence_counter_abstraction_petri_net:iv}, \Cref{lemma:bound_lemma_bonnet} immediately gives us a concrete bound for $\diameter[\calN_{i}]{S_{(i,B)}}$:
\begin{corollary}\label{cor:length_bound_quasi_path}
Let $i \leq \nsimple$, $B \in \N$. Then
\begin{align*}
\diameter[\calN_{i}]{S_{(i,B)}} &\leq \max\set{\length{\cabspnetGen}{\cabsfun[i,B](s')} : \norm[\calN_i;C]{s'} \leq B}\\
&\leq (6\max\set{R,B,1}\max\set{R',B})^{\paren{\cabspnetdim+1}!}. 
\end{align*}
\end{corollary} \noindent
Thus we obtain a concrete bound on the covering radius for $s_{\text{cov}}$ by solving the \citeauthor{Rackoff:78} recurrence relation.

\begin{theorem}\label{thm:covering_radius:bound}
Let us write $\slog$, super-logarithm, for the inverse of $\tetration{(-)}$,  {tetration}, i.e.~$n = {\tetration{\slog(n)}}$. Then for all $i \leq \nsimple$:
\begin{asparaenum}[(i)]
\item $\length{\calN_{0}}{s_{\text{cov}}} \leq \tetration{2\slog(48\ncomplex \nsimple \ninner R')}$
\item $\length{\calN_{i+1}}{s_{\text{cov}}} \leq 2^{2^{\varparen{\varparen{\max\varset{\length{\calN_i}{s_{\text{cov}}},2}}^{48 \ninner \nsimple {\ncomplex}{R'}}}}}$ and
\item $\length{\calN\,}{s_{\text{cov}}} \leq \tetration{2\nsimple +2 \slog(48(\nsimple+1) \ninner \nsimple {\ncomplex}{R'})}$. 
\end{asparaenum}
\end{theorem}

\begin{corollary}\label{cor:cov_nnct_tower}
NNCT Coverability is decidable and in \Tower.
\end{corollary}       

\newcommand{\mathline}[3][]{#1\multicolumn{#2}{l}{$#3$}}
\newcommand{\tline}[4][]{#1\multicolumn{#2}{p{#4}}{{\em #3}}}
\section{A Lower Bound}\label{sec:lowerbound:maintext}\noindent
In this section, we show that simple coverability for total-transfer NNCT is \Tower-hard. 
We encode a bounded counter machine in NNCT which is constructed inductively. Given 
$n \geq 1$, we construct the \emph{yardstick} counters $c_1, \ldots, c_n$: each counter $c_i$ is bounded by $\tetration{i}$, and can be incremented, decremented and tested for zero; 
furthermore operations of the counter $c_{i+1}$ are implemented using operations of the counters $c_1, \ldots, c_i$.
Following \citeauthor{LazicNORW:07} \cite{LazicNORW:07} we present our proof using pseudo code rather than explicit NNCT rules, which we believe is clearer and more readable. 
We use a type of \citeauthor{Stockmeyer:1974}
 yardstick construction \cite{Stockmeyer:1974} that is reminiscent of \citeauthor{Lipton:76}'s \Expspace-hardness proof for coverability and reachability for VAS \cite{Lipton:76}. A $\tetration{i}$-bounded counter $c$ is represented by two places: $p_c$ and its \emph{complement place} $\overline{p}_c$. 
A valuation $v(c)$ of $c$ is re\-pre\-sented by $p_c$ containing $v(c)$ tokens and $\overline{p}_c$ containing $\tetration{i} - v(c)$ tokens. 
Places $p_c$ and $\overline{p}_c$ maintain the invariant that the number of tokens they carry sum up to $\tetration{i}$ at all times. An increment (decrement) of $c$ is then implemented by adding a token to $p_c$ ($\overline{p}_c$) and removing one from $\overline{p}_c$ ($p_c$). 
For a zero test $\iszero{c}$ an additional $\tetration{i}$-bounded counter $s_i$ is maintained. 
On the invocation of $\iszero{c}$ a non-deterministic number $k$ of tokens are removed from $\overline{p}_c$ and $k$ is added to $s_i$, which is assumed to be $0$ at the invocation of $\iszero{c}$.
The operation $\ismaxandreset{-}$ is then 
applied to $s_i$ which performs a decrement of precisely $\tetration{i}$ on $s_i$ and blocks if $s_i < \tetration{i}$. This of course means that $\ismaxandreset{-}$ can only succeed if $c=0$ to begin with and $k = \tetration{i}$.
Our construction 
is similar to \citeauthor{Lazic:2013}'s proof of \Tower-hardness for VAS with one stack \cite{Lazic:2013}. \citeauthor{Lazic:2013} uses the $\tetration{i}$-bounded counters to enumerate all possible stacks over a binary alphabet of height $\tetration{i}$ while decrementing the given counter one-by-one. In the case of NNCT stacks are not available, however, we can encode arrays into complex tokens in a yardstick fashion; and, instead of enumerating stacks, we enumerate all binary arrays of length $\tetration{i}$.


\pagebreak[4]
\begin{theorem}\label{thm:nnct:coverability:tower_hard}
Simple coverability, boundedness and termination for total-transfer NNCT is \Tower-hard.
\end{theorem}
\begin{proof}
We deduce \Tower-hardness by showing that a polynomial-time computable NNCT $\calN_{\calM}$ can weakly bisimulate a deterministic bounded two-counter machine
 $\calM$,
 of size $n$ and $\tetration{n}$-bounded counters. The machine $\calM$ supports the operations: $\inc{x}$, $\dec{x}$, $\reset{x}$, $\iszero{x}$, and $\ismax{x}$. 

Each simulation state of $\calN_{\calM}$ represents a valuation $v$ of $6n+2$ \emph{active} and \emph{inactive} counters, and $n$ arrays. In addition to the counters $x$, $y$ of $\calM$, the NNCT $\calN_{\calM}$ simulates the auxiliary counters $s_i$, $p_i$, $p'_i$,  $c_i$, $c'_i$, and an auxiliary array $a_i$ for each $i \leq n$. 
To simplify notation, we write $C_i$ for the set of counters $\set{s_i, p_i, p'_i, c_i, c'_i}$ when $i < n$,
and $C_n$ for the set of counters $\set{s_n, p_n, p'_n, c_n, c'_n, x, y}$.
Each active counter $d \in C_i$ is $\tetration{i}$-bounded, each inactive counter has an undefined value. 
Each array $a_i$ has length $(\tetration{i})+1$ and carries values $a_i(j) \in \set{0,1,2}$ for $j \leq \tetration{i}$. The NNCT $\calN_{\calM}$ has two simple places $\encp{d}$ and $\encp[\overline]{d}$ for each counter $d \in C_i$, three complex places $\encp{a_i}$, $\encp[\overline]{a_i}$, and $\encp{\mathit{aux}_i}$, two colours $\encp{j_i}$ and $\encp[\overline]{j_i}$ for each array $a_i$, a complex ``sink'' place $\encp{\mathit{disc}}$, and $\colmap$
maps $\encp{j_i}$ to $\encp{p'_i}$ and $\encp[\overline]{j_i}$ to $\encp[\overline]{p'_i}$.
Lastly, $\calN_{\calM}$ has a polynomial number of simple places encoding the control of $\calM$ and the internal control of $\calN_{\calM}$. 
Further, $\calN_{\calM}$'s transfer rules are all total, and hence $\calN_{\calM}$ is a total-transfer NNCT.

A valuation $v$ is represented by a configuration $s$ as follows:
\begin{itemize}[nosep]
\item For each $i$ and $d \in C_i$, if $d$ is active then there are exactly $v(d)$ $\bullet$-tokens in $\encp{d}$ and $\tetration{i} - v(d)$ $\bullet$-tokens in $\encp[\overline]{d}$.
\item For each inactive counter $d$, $\encp{d}$ and $\encp[\overline]{d}$ are empty.
\item Each array $a_i$ is represented as follows. For $k \leq \tetration{i}$,
let $m_{(i,k)}$ be a complex token such that $m_{(i,k)}$ contains exclusively tokens of colours $\encp{j_i}$ and $\encp[\overline]{j_i}$: $k$ tokens of colour $\encp{j_i}$ and $\tetration{i} - k$ tokens of colour  $\encp[\overline]{j_i}$.
Then, to represent $a_i$, there are exactly $v(a_i(k))$ tokens $m_{(i,k)}$ in $\encp{a_i}$ and
$2-v(a_i(k))$ tokens $m_{(i,k)}$ in $\encp[\overline]{a_i}$.
\item For each $i$ the place $\encp{\mathit{aux}_i}$ is empty.
\end{itemize}\noindent
The question whether $\calM$ reaches a halting 
control state from its initial state can then be answered by performing a simple coverability query on $\calN_{\calM}$'s simple places encoding $\calM$'s finite control. 
Assuming that only $\calM$'s halting control states have no successors, $\calM$'s halting problem also reduces to 
the termination problem for $\calN_{\calM}$. Augmenting $\calN_{\calM}$ with an additional simple place that is incremented with every transition shows that $\calM$'s halting problem reduces to boundedness of $\calN_{\calM}$. 

We implement further instructions to improve readability.
The NNCT $\calN_{\calM}$ simulates $\activate{d}$, $\deactivate{d}$ for ${d \in \Union_{i\in\range{n}} C_i}$,
and operations
$\reset{d}$, $\iszero{d}$, and $\ismax{d}$ for $d \in \Union_{i\in\range{n}} C_i \setminus \varset{s_i}$.
For each $i \in \range{n}$, we implement $\ismaxandreset{s_i}$ and the specialised operations:%
\vspace{2.5pt plus 1pt minus 1pt}

\noindent
\begin{tabular}{p{1eM}l}
 &$\isequal{p_i}{p'_i}$,  $\inc{a_i(p_i)}$, $\dec{a_i(p_i)}$, $\reset{a_i(p_i)}$, \\
 &$\iszero{a_i(p_i)}$, $\ismax{a_i(p_i)}$, $\activate{a_i}$.
\end{tabular}\vspace{2.5pt plus 1pt minus 1pt}

\noindent
These operations only succeed if the counters in question are active. 
The counters in $C_1$ are $2$-bounded so implementing operations on them is trivial.
For $i < n$, operations on $a_i$ are simulated using counters in $C_i$ and
operations on counters in $C_{i+1}$ are simulated using operations on 
counters in $C_i$ and $a_i$. 

\pagebreak
\begin{asparaenum}[(i),nosep]
\item \emph{The implementation of $\ismaxandreset{s_{i+1}}$}:\vspace{2.5pt plus 1pt minus 1pt}

\noindent 
\begin{tabular}{@{}p{0pt}@{$\,$}p{0.6eM}@{$\,$}p{0.6eM}l}
	\mathline[{&}]{3}{\textbf{for } p_i \is 0 \textbf{ to } \tetration{i} \textbf{ do } (\reset{a_i(p_i)};)}\\
	\mathline[{&}]{3}{\textbf{while }(\iszero{a_i(\tetration{i})}) \textbf{ do}}\\	
		\mathline[{&&}]{2}{\dec{s_{i+1}};\;\reset{p_i};\;\inc{a_i(p_i)};}\\
\end{tabular}

\noindent
\begin{tabular}{@{}p{0pt}@{$\,$}p{0.6eM}@{$\,$}p{0.6eM}l}
		\mathline[{&&}]{2}{\textbf{while } \ismax{a_i(p_i)} \textbf{ do } (\dec{a_i(p_i)};\;\inc{p_i};\;\inc{a_i(p_i)};)}\\
\end{tabular}\vspace{2.5pt plus 1pt minus 1pt}

\noindent
After the for-loop, we know that $a_i(x) = 0$ for all $x \leq \tetration{i}$.
The array $a_i$ is the binary representation of a number between $0$ and $\tetration{(i+1)}$. This number is initially $0$ and the outer while loop performs a long addition of $1$ for each iteration. If $v(a_i(v(p_i))) = 2$ then $v(p_i)$ is an index representing a carry bit.
For each number represented by $a_i$ we perform $\dec{s_{i+1}}$. Hence, if initially $v(s_{i+1}) = \tetration{(i+1)}$,
then after performing $\ismaxandreset{s_{i+1}}$ the resulting $v'$ sets $v'(s_{i+1})\!=\!0$. If $v(s_{i+1}) < \tetration{(i+1)}$, then after $v(s_{i+1})$ iterations the resulting valuation $v'$ sets $v'(s_{i+1}) = 0$ and $a_i$ represents the number $v(s_{i+1})$. Since $v(s_{i+1}) < \tetration{(i+1)}$ this implies that
$a_i(\tetration{i}) = 0$ and hence the body of the outer while loop is executed again leading to an invocation of
$\dec{s_{i+1}}$ which blocks. Hence $\ismaxandreset{s_{i+1}}$ blocks when executed in a configuration representing $v$ such that $v(s_{i+1}) < \tetration{(i+1)}$.

The implementation of $\ismaxandreset{s_{i+1}}$ assumes that operations on $a_i$ are correctly implemented. In the following we show how $\inc{a_i(p_i)}$ is simulated which is representative.
\item \emph{The implementation of $\inc{a_i(p_i)}$}: \vspace{2.5pt plus 1pt minus 1pt}

\noindent
\begin{tabular}{@{}p{0pt}@{$\,$}p{0.5eM}@{$\,$}p{0.5eM}l}
	\tline[{&}]{3}{Move a complex token from $\encp[\overline]{a_i}$ to $\encp{\mathit{aux}_i}$;}{7cm}\\
	\mathline[{&}]{3}{\deactivate{p'_i};}\\
	\tline[{&}]{3}{Eject the contents of a complex token in $\encp{\mathit{aux}_i}$ and place its remains into $\encp{\mathit{disc}}$;}{8.4cm}\\
	\mathline[{&}]{3}{\isequal{p_i}{p'_i};\; \reset{p'_i};}\\
	\mathline[{&}]{3}{\text{\em Create an empty complex token in $\encp{\mathit{aux}_i}$;}}\\
	\mathline[{&}]{3}{\textbf{while }\paren{p_i \neq p'_i} \textbf{ do } (\inc{p'_i};}\\
		\tline[{&&}]{2}{Inject\hspace{-0.02cm} a\hspace{-0.02cm} $\encp{j_i}$-coloured\hspace{-0.02cm} token\hspace{-0.02cm} into\hspace{-0.02cm} a\hspace{-0.02cm} complex\hspace{-0.02cm} token\hspace{-0.02cm} in\hspace{-0.02cm} $\encp{\mathit{aux}_i}$;)}{8.4cm}\\
	\mathline[{&}]{3}{\textbf{while }\paren{\neg(\ismax{p'_i})} \textbf{ do } (\inc{p'_i};}\\
		\tline[{&&}]{2}{Inject\hspace{-0.02cm} a\hspace{-0.02cm} $\encp[\overline]{j_i}$-coloured\hspace{-0.02cm} token\hspace{-0.02cm} into\hspace{-0.02cm} a\hspace{-0.02cm} complex\hspace{-0.02cm} token\hspace{-0.02cm} in\hspace{-0.02cm} $\encp{\mathit{aux}_i}$;)}{8.4cm}\\
	\mathline[{&}]{3}{\text{\em Move a complex token from $\encp{\mathit{aux}_i}$ to $\encp{a_i}$;\;$\reset{p'_i}$;}}\\
\end{tabular}\vspace{2.5pt plus 1pt minus 1pt}

\noindent
%
Suppose $\inc{a_i(p_i)}$ is executed in a configuration $s$ that represents valuation $v$ and $p_i, p'_i$ are active.
If $v(a_i(v(p_i))) < 2$ then there exists a complex token $m_{(i,v(p_i))}$ in $\encp[\overline]{a_i}$ and all complex tokens in $\encp[\overline]{a_i}$ are of the form $m_{(i,k)}$ for some $k \leq \tetration{i}$.
The move of one $m_{(i,k)}$ complex token from $\encp[\overline]{a_i}$ to $\encp{\mathit{aux}_i}$ results in
$\encp{\mathit{aux}_i}$ containing just $m_{(i,k)}$, since by assumption $\encp{\mathit{aux}_i}$ is empty before.
After deactivating $p'_i$, both $\encp{p'_i}$  and $\encp[\overline]{p'_i}$ are empty. Ejecting
the contents of $m_{(i,k)}$ removes $m_{(i,k)}$ from $\encp{\mathit{aux}_i}$, inserts $k$ $\bullet$-tokens into $\encp{p'_i}$ and $\tetration{i} - k$ $\bullet$-tokens into $\encp[\overline]{p'_i}$, and 
places the remaining empty complex token into $\encp{\mathit{disc}}$. Disregarding
$a_i$, the configuration we have reached represents a partial valuation $v'$ that sets $v'(p'_i) = k$ and $v'(p_i) = v(p_i)$. After executing $\isequal{p_i}{p'_i}$, the simulation only succeeds if $k = v(p_i)$.
Hence $\encp[\overline]{a_i}$ now contains $2-(v(a_i(v(p_i)))+1)$ complex tokens $m_{(i,v(p_i))}$ and the same number of other complex tokens $m_{(i,k)}$.
The two while loops carefully inject $\encp{j_i}$ and
$\encp[\overline]{j_i}$-coloured tokens into the newly created
token at $\encp{\mathit{aux}_i}$ to yield a new $m_{(i,v(p_i))}$ located in $\encp{\mathit{aux}_i}$
which we move to $\encp{a_i}$. 
Thus, $\encp{a_i}$ now contains $v(a_i(v(p_i)))+1$ complex tokens $m_{(i,v(p_i))}$ and the same number of other complex tokens $m_{(i,k)}$ as before. The configuration $s'$ we have reached thus represents a valuation $v''$ such that
$v''(a_i(j)) = v(a_i(j))$ for all $j \leq \tetration{i}$ and $j \neq v(p_i)$ and $v''(a_i(v(p_i))) = v(a_i(v(p_i)))+1$.
Otherwise, if $v(a_i(v(p_i))) = 2$, the simulation either blocks on the attempt to move some
$m_{(i,k)}$ from $\encp[\overline]{a_i}$ to $\encp{\mathit{aux}_i}$ or on the execution of $\isequal{p_i}{p'_i}$
since it is impossible to obtain $k = v(p_i)$.

\item \emph{The implementation of $\iszero{d}$}:%
\vspace{2.5pt plus 1pt minus 1pt}

\noindent
\begin{tabular}{@{}p{0pt}@{$\,$}p{0.5eM}@{$\,$}p{0.5eM}l}
	\mathline[{&}]{3}{\textbf{while }(*) \textbf{ do } (\inc{d}; \inc{s_{i+1}};);$\;$ \ismaxandreset{s_{i+1}};}\\	
	\mathline[{&}]{3}{\textbf{while }(*) \textbf{ do } (\dec{d}; \inc{s_{i+1}};);$\;$ \ismaxandreset{s_{i+1}}}\\	
\end{tabular}\vspace{2.5pt plus 1pt minus 1pt}

\noindent
Note that $s_{i+1}=0$ is a precondition for $\iszero{d}$ and we only modify $s_{i+1}$ in $\iszero{d}$ and $\ismaxandreset{s_{i+1}}$. Thus, except when executing $\iszero{d}$, we maintain $s_{i+1} = 0$.

If $\ismaxandreset{s_{i+1}}$ completes after the first while loop, then the loop must have incremented $s_{i+1}$ and $d$ to $\tetration{(i+1)}$. Since $d$ is 
$\tetration{(i+1)}$-bounded, $d$ had to be $0$ to begin with and is valued $\tetration{(i+1)}$ after $\ismaxandreset{s_{i+1}}$'s first invocation. The rest of the implementation then guarantees that $d$ is decremented to $0$ again.
Hence $\iszero{d}$ succeeds only if started in a configuration that values $d$ as $0$.

\item We omit the analogous implementations of the other operations%
\iffull%
\ifwithappendix%
 here. We refer the reader to \Cref{sec:lowerbound:full}.
\else
, which can be found in the long version \cite{KochemsO:2014a}.
\fi
\else
the long version of the paper \cite{KochemsO:2014a}.
\fi
\end{asparaenum}
The initial operation of $\calN_{\calM}$ executes $\activate{-}$
for all counters and arrays in turn, in order of the bound size and length. From this point on, $\calN_{\calM}$ weakly bisimulates $\calM$. Thus, we can reduce the halting problem for $\calM$ to simple coverability, termination, or boundedness of $\calN_{\calM}$.
\end{proof}       
\section{Related Work}
\paragraph*{Asynchronously Communicating Pushdown Systems}
In 2006, \citeauthor{Sen:2006} showed that safety verification is decidable for recursive asynchronous programs \cite{Sen:2006}, which give rise to ACPS that satisfy the empty-stack constraint. 
Liveness properties \cite{Ganty:2009} and practical analyses \cite{Jhala:2007} for asynchronous programs have since been studied.
Recently, \citeauthor{Ganty:2012} showed that a variety of verification problems for asynchronous programs are polynomial-time inter-reducible to decision problems on Petri nets, thus e.g. safety verification is \Expspace-complete.
%
Extensions of \citeauthor{Sen:2006}'s model \cite{Chadha:2007,CaiO:2013,Majumdar:2014:Draft} have been proposed that allow e.g.~the modelling of higher-order stack automata, the dynamic creation of task buffers, or WQO stack alphabets; however, the empty-stack restriction remains a key restriction. 
The empty-stack constraint fits nicely with the atomicity requirement on asynchronous procedures. 
An atomic procedure only needs to make synchronisations before and after its execution which may thus be performed with an empty call stack. The shaped-constraint enables a relaxation of the atomicity requirement:
it allows non-trivial synchronisations inside the execution of procedure call. 
%
This increase in expressive power, together with the ramping up of the computational complexity, 
confirms our intuition that shaped-stack ACPS are a much more general model than ACPS satisfying the empty-stack constraint.

\paragraph*{Concurrent Pushdown Systems (CPS)}
Numerous classes with decidable verification problems have been discovered: 
parallel flow graph systems \cite{Esparza:2000}, 
visibly pushdown automata with FIFO-channels \cite{Babic:2011},
CPS communicating\linebreak[4] over locks \cite{Kahlon:2009}, 
recursive programs with hierarchical com\-mu-\linebreak[4]ni\-ca\-tion \cite{Bouajjani:2005,Bouajjani:Popl:2012},
and CPS with FIFO-channels for which the empty-stack restriction applies to sends~\cite{Heussner:2010}.
Further, over-approximation \cite{Flanagan:2003,Henzinger:2003} and
under-ap\-prox\-i\-ma\-tion techniques \cite{Esparza:2011,QadeerR:2005,Bouajjani:Tacas:2012,Torre:2011} have been studied.
\citeauthor{Czerwinski:09} introduced PCCFG to study bisimulation for BPC~\cite{Czerwinski:09}, a process algebra
extending BPA and BPP \cite{Esparza:1997}. However, synchronisation between processes (which transforms PCCFG to APCPS) is not a feature considered by \citeauthor{Czerwinski:09}

\paragraph*{Extensions of Petri nets}
Coverability is a central decision problem
in the vast literature of Petri net extensions. However, any non-trivial extension, 
 such as \emph{reset arcs} or \emph{transfer arcs} \cite{Dufourd:98,Schnoebelen:10}, typically renders coverability non-primitive recursive.
\emph{Nested Petri nets} may appear closely related to NNCT, 
they are however much more expressive and coverability is Ackermann-hard \cite{LomazovaS:99}. Nested Petri nets allow arbitrary nesting of tokens, and nested layers of a token can synchronise. By contrast, internal synchronisation is not possible in NNCT: coloured tokens can only be ejected to simple places and then inspected. Our proof exploits this fact and it seems to 
explain the \Tower-membership of NNCT coverability. 

\emph{Data nets} \cite{LazicNORW:07} allow tokens to be drawn from an arbitrary linearly ordered set.
Recently, coverability and termination for data nets and a subclass, Petri data nets (PDN), were shown to be $\FOmegaUpThree$-complete \cite{HaddadSS:2012,Schmitz:2013}.
%
A more restricted subclass of PDN studied by \citeauthor{LazicNORW:07} gives rise to a \Tower-hard coverability problem, namely, \emph{unordered} PDN (UPDN) which features an equality check on tokens. 
\citeauthor{LazicNORW:07} show that coverability, termination and boundedness are all \Tower-hard, but no upper-bound is available. Unfortunately, the equality check on tokens 
makes it unclear whether coverability of UPDN reduces to NNCT coverability which is why we opted for a \Tower-hardness proof from first principles. 
Adding the ability to create fresh tokens to UPDNs yields $\nu$-Petri nets ($\nu$-PN)
for which 
coverability is \Ack-hard 
\cite{Rosa-VelardoF11}.

\emph{Rackoff technique.}
Originally introduced to show the \Expspace-membership of coverability and boundedness for VAS~\cite{Rackoff:78}, the \emph{Rackoff technique} has recently become popular. It has been used to establish \Expspace\ upper bounds for coverability and boundedness of \emph{strongly increasing affine nets} (SIAN) \cite{Bonnet:12}, selective unboundedness of VAS with states (VASS) \cite{Demri:2013}, model-checking Petri nets \cite{BlockeletS:2011},
an \AltExpspace\ upper bound for coverability and boundedness for branching VAS \cite{DemriJLL:2009}, and
a \Tower-upper bound for a coverability problem for alternating BVAS with states (ABVASS) \cite{LazicS:2014}.
Even though NNCT coverability and ABVASS coverability are \Tower-complete, it is not 
obvious how to inter-reduce the coverability problems between NNCT and ABVASS. 
It is hard to see how one can simulate in an ABVASS 
the ability of NNCT to model 
complex tokens, carrying an unbounded number of coloured tokens, that can interact via ejection with other complex tokens. 
In the other direction there is no clear counterpart to the tree-like runs of ABVASS in NNCT. 

\emph{Vector addition systems with one stack (SVAS).}
Recent work 
has shown that boundedness and termination are decidable for SVAS and that the problem lies in \HAck\ \cite{Leroux:2014}. The decidability of coverability and reachability is still an open question but are known to be \Tower-hard \cite{Lazic:2013}. 
 
\paragraph*{Future Directions}
%
We intend to identify 
a practical coverability algorithm for shaped ACPS, 
as they arise 
from the abstract interpretation of realistic Erlang programs \cite{Soter:2012,DOsualdoKO:2013}.

\paragraph*{Acknowledgments}
Financial support by EPSRC (grant EP/F036\-361/1 and
DTG doctoral studentship for the first author) is gratefully acknowledged. 
We thank Matthew Hague, Stefan Kiefer, James Worrell, Javier Esparza, Sylvain Schmitz, Ranko Lazic, and Alain Finkel for insightful comments.


\bibliographystyle{abbrvnat}
{\small
\bibliography{bibliography/pccfg,bibliography/ref,bibliography/vass,bibliography/asyncpc,bibliography/concpds,bibliography/complexity,bibliography/datanets}
}

\ifwithappendix
\clearpage
\appendix

\subsection{Proofs for Section \ref{sec:acps}}
\begin{customlemma}[\ref{lem:acps:cov:eq:simplecov}]
Coverability and simple coverability for ACPS poly\-no\-mi\-al-time inter-reduce.
\end{customlemma}
\begin{proof}
The reduction from simple coverability to coverability of ACPS is trivial since the former is a subproblem
of the latter. 
For the other direction suppose we have a coverability query $\mathpzc{Q} = \varparen{\calP,\Pi_0 \ChanPar \Gamma_0,\Pi_{\text{cov}} \ChanPar \Gamma_{\text{cov}}}$.
Let us assume that $\calP = (\calQ,\calA,\Chan,\Msg[],\calR)$,
$\Pi_{\text{cov}} = (q_1,\beta_1) \parallel \cdots \parallel (q_n,\beta_n)$ and
$\Pi_{0} = (q^{0}_1,\beta^{0}_1) \parallel \cdots \parallel (q^{0}_m,\beta^{0}_m)$
where $q_i,q^0_{j} \in \calQ$ and $\beta_i,\beta^{0}_j \in \calA^*$ for $i \in \range{n}$, $j \in \range{m}$.

\noindent Let us define the following ACPS $\calP' = (\calQ',\calA,\Chan,\Msg[],\calR')$
where
$\calQ' = \calQ \union \varset{q'_{(i,A,\beta')},q'_{(i,\epsilon)},\discn q^{0}_{(j,\beta'')} : \beta_i = A\cdot\beta\cdot\beta', 
\beta^{0}_j = \beta'' \cdot \beta''', \beta_k = \epsilon, i,k \in \range{n}, j \in \range{m}}$ and all control states 
$q'_{(i,A,\beta)}$, $q'_{(k,\epsilon)}$
$q^{0}_{(j,\beta)}$ are fresh
and 

\begin{align*}
\calR'& = \calR \\
&\union 
\set{\left.
	\begin{aligned}
	(q'_{(i,A,B\,\beta')},B)    			&\to[effect=\epsilon] (q'_{(i,A,\beta')},\epsilon),\\
	(q'_{(i,A,\beta'')},B')  		    &\to[effect=\epsilon] (q'_{(i,\beta'')},\epsilon),\\
	(q_{i},A)      		   	 				&\to[effect=\epsilon] (q'_{(i,A,\beta'_i)},\epsilon)
	\end{aligned}\, \right| 
	\begin{aligned}
	&\beta_i = A\, \beta'_i, \\
	&\beta'_i = \beta_0\cdot B \cdot\beta', \\
	&\beta'_i = \beta_1 \cdot \beta'', \\
	&i \in \range{n}
	\end{aligned}
}\\
&\union
\set{\left.
	\begin{aligned}
	(q_{k},\epsilon)        				&\to[effect=\epsilon] (q'_{(k,\epsilon)},\epsilon),\\
	(q^0_{(j,\beta''' \cdot C)},D_0)		&\to[effect=\epsilon] (q^0_{(j,\beta''')},D_0\,C),\\
	(q^0_{(j,\epsilon)},D_0)				&\to[effect=\epsilon] (q^0_j,\epsilon)
	\end{aligned}\, \right| 
	\begin{aligned}
	&\beta_k = \epsilon,\\
	&\beta^0_j = \beta'''\!\cdot\! C\!\cdot\! \beta''''\\
	&k \in \range{n}, j \in \range{m}
	\end{aligned}
}
\end{align*}
The ACPS $\calP'$ essentially implements the query $\mathpzc{Q}'$. The rules involving 
$q^0_{(j,\beta)}$ set up the start configuration with arbitrary stacks from a length one stack.
Rules involving $q'_{(k,\epsilon)}$ and $q'_{(i,A,\epsilon)}$ essentially check that the coverability query is satisfied.
In order to account for this we change the coverability query to $\mathpzc{Q}' = \varparen{\calP,\Pi'_0 \ChanPar \Gamma_0, \Pi'_{\text{cov}} \ChanPar \Gamma_{\text{cov}}}$
where 
$\Pi'_0 = (q^0_{(1,\beta^0_1)},D_0) \parallel \cdots \parallel (q^0_{(n,\beta^0_n)},D_0)$,
$\Pi'_{\text{cov}} = (q'_{(1,\tilde{\beta}_1)},\epsilon) \parallel \cdots \parallel (q'_{(n,\tilde{\beta}_n)},\epsilon)$ and if $\beta_i = \epsilon$ then $\tilde{\beta}_1 = \epsilon$; otherwise if $\beta_i = A \cdot \beta'$ then $\tilde{\beta}_i = A,\epsilon$.
By construction $\mathpzc{Q}'$ is a simple query.

Since any set of suffixes/prefixes of a sequence $\beta$ satisfies 
$$|\varset{\beta_0 : \beta_0\cdot\beta_1 = \beta}| = |\varset{\beta_1 : \beta_0\cdot\beta_1 = \beta}| \leq |\beta|$$
we can clearly set that $|\calQ'| \leq |\calQ| + (n+m) \times \max\varset{|\beta_i|,|\beta^0_j| : i \in \range{n}, j \in \range{m}}$
and $|\calR'| \leq |\calR| + 3 (n+m) \times \max\varset{|\beta_i|,|\beta^0_j| : i \in \range{n}, j \in \range{m}}$ and hence $\calP'$ and $\mathpzc{Q}'$ are clearly polynomial-time computable from $\calP$ and $\mathpzc{Q}$.

By construction $\calP'$ has the following property: $(q'_{i,A,\beta_1},A\,\beta) \parallel \Pi_0 \ChanPar \Gamma \discn
\to[*,TS=\calP'] (q'_{i,A,\epsilon},\epsilon) \parallel \Pi_0 \ChanPar \Gamma$ if and only if $\beta_i = A\,\beta_0\beta_1$, $\beta_1 \higleq \beta$, i.e.~checking coverability at a process level is correctly implemented by each process.
Further $\Pi_0 \ChanPar \Gamma_0 \to[TS=\calP] \Pi \ChanPar \Gamma$ if, and only if,
$\Pi'_0 \ChanPar \Gamma_0 \to[TS=\calP',*] \Pi \ChanPar \Gamma$, i.e.~$\Pi'_0$ correctly sets up $\Pi_0$.

Suppose now that $\mathpzc{Q}$ is a yes-instance. This means that
 $\Pi_0 \ChanPar \Gamma_0 \discn \to[*,TS=\calP] \Pi \ChanPar \Gamma$ such that 
 $\Pi_{\text{cov}} \ChanPar \Gamma_{\text{cov}} \leqACPS \Pi \ChanPar \Gamma$. Clearly it is then the case that
$\Pi'_0 \ChanPar \Gamma_0 \to[*,TS=\calP'] \Pi \ChanPar \Gamma$. Since $\Pi_{\text{cov}} \ChanPar \Gamma_{\text{cov}} \leqACPS \Pi \ChanPar \Gamma$ we know that
$\Pi = (q_1,\beta'_1) \parallel \cdots \parallel (q_n,\beta'_n) \parallel \Pi'$ and 
for all $i \in \range{n}$ either $\beta_i = \epsilon$ or $\beta_i = A \cdot \beta''_i$, 
$\beta'_i = A \cdot \beta'''_i$ and $\beta''_i \higleq \beta'''_i$.

And hence $\Pi'_0 \ChanPar \Gamma_0 \to[*,TS=\calP'] \Pi \ChanPar \Gamma \to[*,TS=\calP'] (q'_{(1,\tilde{\beta}_1)},\epsilon) \parallel \cdots \parallel (q'_{(n,\tilde{\beta}_n)},\epsilon) \parallel \Pi' \ChanPar \Gamma$ and hence $\mathpzc{Q}'$ is a yes-instance for coverability.

For the other direction, suppose $\mathpzc{Q}'$ is a yes-instance for coverability.
We then know that
$\Pi'_0 \ChanPar \Gamma_0 \to[*,TS=\calP'] (q'_{(1,\tilde{\beta}_1)},\beta'_1) \parallel \cdots \parallel (q'_{(n,\tilde{\beta}_n)},\beta'_n) \parallel \Pi' \ChanPar \Gamma' =: s'$ such that
$$(q'_{(1,\tilde{\beta}_1)},\epsilon) \parallel \cdots \parallel (q'_{(n,\tilde{\beta}_1)},\epsilon) \ChanPar \Gamma \leqACPS s'.$$
Then our observation above tells us that it must be the case that (possibly reordering locally-independent transitions)
$\Pi'_0 \ChanPar \Gamma_0 \to[*,TS=\calP'] (q_1,\beta''_1\beta'_1) \parallel \cdots \parallel (q_n,\beta''_n\beta'_n) \parallel \Pi' \ChanPar \Gamma' \to[*,TS=\calP'] (q'_{(1,\beta_1)},\beta''_1\beta'_1) \parallel \cdots \parallel (q'_{(n,\beta_n)},\beta''_n\beta'_n) \parallel \Pi' \ChanPar \Gamma'$ with
either $\beta_i = \epsilon$ or $\beta_i = A \beta^{\dagger}_i$, $\beta''_1 = A \beta'''_i$ and
$\beta^{\dagger}_i \higleq \beta'''_i$ and thus $\beta^{\dagger}_i \higleq \beta'''_i\beta'_i$ for all $i \in \range{n}$.
Further 
$\Pi_0 \ChanPar \Gamma_0 \to[*,TS=\calP] (q_1,\beta''_1\beta'_1) \parallel \cdots \parallel (q_n,\beta''_n\beta'_n) \parallel \Pi' \ChanPar \Gamma'$
Hence $\mathpzc{Q}$ is a yes-instance for coverability.

We can thus conclude that simple coverability and coverability are polynomial-time inter-reducible.
\end{proof}

\subsubsection{Proof of Proposition~\ref{prop:acps:normalform}}
In this section we will give a proof of the following Proposition:
\begin{customproposition}[\ref{prop:acps:normalform}]
Given an ACPS $\calP$, a simple coverability query $\mathpzc{Q}$ and a $\Pi^0 \ChanPar \Gamma^0$ there exists ACPS $\Fnormalform{\calP}$ in normal form, a simple coverability query $\Fnormalform{\mathpzc{Q}}$, and $\Fnormalform{\Pi^0 \ChanPar \Gamma^0}$ --- all poly\-no\-mi\-al-time computable --- such 
that:
$\mathpzc{Q}$ is a yes-instance if, and only if, $\Fnormalform{\mathpzc{Q}}$ is a yes-instance; and
$\calP$ is bounded (terminating) from $\Pi^0 \ChanPar \Gamma^0$ if, and only if, $\Fnormalform{\calP}$ is bounded (terminating respectively) from 
$\Fnormalform{\Pi^0 \ChanPar \Gamma^0}$.
\end{customproposition}

We will give a proof in two steps:
\begin{inparaenum}[(i)]
\item we first transform a general ACPS $\calP$ into an ACPS that satisfies a \emph{pre-normal form} as defined below; 
\item secondly, we show how to transform an ACPS in pre-normal form with the desired property.
\end{inparaenum}
We lay out our argument in the two Lemmas below, but first we define \emph{pre-normal form}:
We say an ACPS $\calP  = (\calQ,\calA,\Chan,\Msg[],\calR)$ is in pre-normal if
for all $(q,\beta) \to[effect=\lambda] (q',\beta') \in \calR$
\begin{inparaenum}[(i')]
\item(acps:nf:i') if $\lambda \neq \epsilon$ then $\beta = \beta' = \epsilon$,
otherwise either 
\item(acps:nf:ii') $\beta = A \in \calA$ and $\beta' = \epsilon$ or
\item(acps:nf:iii') $\beta = \epsilon$ and $\beta' = A' \in \calA$; or
\item(acps:nf:iv') $\beta = \beta' = \epsilon$; and
\item(acps:nf:v') if $\lambda = \spn{(q'',\beta'')}$ then $\beta'' = \epsilon$.
\end{inparaenum}

\begin{lemma}\label{app:lemma:pnf}
Given an ACPS $\calP$, a simple coverability query $\mathpzc{Q}$ and a start configuration $\Pi^0 \ChanPar \Gamma^0$  there exists ACPS $\Fprenormalform{\calP}$ in pre-normal form, simple coverability query $\Fprenormalform{\mathpzc{Q}}$, and start configuration $\Fprenormalform{\Pi^0 \ChanPar \Gamma^0}$ --- all polynomial-time computable --- such 
that:
\begin{inparaenum}[(A)]
\item $\mathpzc{Q}$ is a yes-instance if, and only if, $\Fprenormalform{\mathpzc{Q}}$ is a yes-instance;
\item $\calP$ is bounded from $\Pi^0 \ChanPar \Gamma^0$ if, and only if, $\Fprenormalform{\calP}$ is bounded from $\Fprenormalform{\Pi^0 \ChanPar \Gamma^0}$; and
\item $\calP$ is terminating from $\Pi^0 \ChanPar \Gamma^0$ if, and only if, $\Fprenormalform{\calP}$ is terminating from $\Fprenormalform{\Pi^0 \ChanPar \Gamma^0}$.
\end{inparaenum}
\end{lemma}
\begin{proof}

Let us fix an ACPS $\calP = (\calQ,\calA,\Chan,\Msg[],\calR)$ and 
let us define the ACPS $\calP_0 = (\calQ',\calA,\Chan,\Msg[],\calR')$ and
\begin{align*}
\calQ' &= \calQ \union \set{{q}^{\text{pop}}_{\beta_0}, {q'}^{\text{push}}_{\beta_1} \left| 
		\begin{aligned}
		&(q,\beta) \to[effect=\lambda] (q',\beta') \in \calR', 
		 \beta  = \beta'_0   \cdot\beta_0, \\
		&\beta' = \beta_1\cdot\beta'_1
		\end{aligned}\right.}\\
		&\;\;\;\;\union 
		\set{{q''}^{\text{push}}_{\beta''_0} \left| 
		\begin{aligned}
		&(q,\beta) \to[effect=\spn{(q'',\beta'')}] (q',\beta') \in \calR'\\
		& \beta'' = \beta''_0 \cdot \beta''_1
		\end{aligned}\right.}\\
		&\;\;\;\;\union
		\set{q^{\text{cov}} \mid q \in \calQ}\\
\calR' &= \set{({q}^{\text{pop}}_{\beta},\epsilon) \to[effect=\lambda'] ({q'}^{\text{push}}_{\beta'},\epsilon) 
	 \left|
	 \begin{aligned}
	 &(q,\beta) \to[effect=\lambda] (q',\beta') \in \calR \\
	 &\text{if } \lambda = \spn{(q'',\beta'')} \\
		&\;\text{ then } \lambda' = \spn{(q''^{\text{push}}_{\beta''},\epsilon)}\\
		&\text{otherwise } \lambda' = \lambda
	 \end{aligned}
	 \right.
	 }\\
	 &\;\;\;\;\union
	 \set{\left.
	 \begin{aligned}
	  &(q,\epsilon) 						 	\to[effect=\epsilon] ({q}^{\text{pop}}_{\epsilon},\epsilon),\\
	  &({q'}^{\text{push}}_{\epsilon},\epsilon) \to[effect=\epsilon] (q',\epsilon) 
	 \end{aligned}
	 \right| (q,\beta) \to[effect=\lambda] (q',\beta') \in \calR'
	 }\\
	 &\;\;\;\;\union
	\set{\left.
	\begin{aligned}
		&({q}^{\text{pop}}_{\beta},A) 	 \to[effect=\epsilon] ({q}^{\text{pop}}_{\beta\cdot A},\epsilon), \\
		&({q'}^{\text{push}}_{\beta'\cdot A'},\epsilon) \to[effect=\epsilon] ({q'}^{\text{push}}_{\beta'},A'), \\
	\end{aligned} \right|
	  \begin{aligned}
	  	&{q}^{\text{pop}}_{\beta},{q}^{\text{pop}}_{\beta\cdot A} \in \calQ'\\
	  	&{q'}^{\text{push}}_{\beta'\cdot A'},{q'}^{\text{push}}_{\beta'} \in \calQ'
	  \end{aligned}
	}
\end{align*}
The rules of $\calP_0$ simply implement a $(q,\beta) \to[effect=\lambda] (q',\beta')$ by popping
$\beta$ one symbol at the time and then pushing $\beta'$ one symbol at the time.
It is easy to see that $\calP_0$ is in pre-normal form. Further
let $\Xi = \varset{|\beta|,|\beta'| : (q,\beta) \to[effect=\lambda] (q',\beta') \in \calR'} \union \varset{|\beta''| : (q,\beta) \to[effect=\spn{(q'',\beta'')}] (q',\beta') \in \calR'}$ then
$|\calQ'| \leq 2|\calQ| + 3 \times |\calR| \times \max\varparen{\Xi}$ and
$|\calR'| \leq 2 \times |\calQ'| + 3 \times |\calR|$
and so $\calP_0$ is clearly polynomial-time computable from $\calP$.

We will show that there is a \emph{weak reflexive bisimulation} between $\calP$ and $\calP_0$.
\begin{definition*}[Weak reflexive bisimulation]
Suppose $(S,\to[label=u,LTS=S])$ and $(S',\discn\to[label=u,LTS={S'}])$ are labelled transition systems we say a relation $\calB \subseteq S \times S'$ is a 
\emph{weak reflexive simulation} if for all $(s,s') \in \calB$, if for some $t \in S$ we have $s \to[label=u,LTS=S] t$  then either $(t,s') \in \calB$ and $u = \epsilon$ or there exists 
$t' \in S'$ such that $s' \to[label=u,LTS=S',*] t'$ and $(t,t') \in \calB$. 
We say $\calB$ is a \emph{weak reflexive bisimulation} relation just if both $\calB$ and $\calB^{-1}$ are weak reflexive simulation relations.
\end{definition*}

We temporarily label the transition systems $(\calP,\to[TS=\calP])$ and $(\calP_0,\to[TS=\calP_0])$ with rules of $\calR$. Let us label the transition $\Pi \ChanPar \Gamma \to[label=r,LTS=\calP] \Pi' \ChanPar \Gamma'$ if the rule $r \in \calR$ is used to justify the transition. We label $\calP_0$'s transition as follows:  
If $\Pi \ChanPar \Gamma \to[TS=\calP_0] \Pi' \ChanPar \Gamma'$ using a rule
$({q}^{\text{pop}}_{\beta},\epsilon) \to[effect=\lambda'] ({q'}^{\text{push}}_{\beta'},\epsilon) \in \calR'$
introduced because of a rule $r = (q,\beta) \to[effect=\lambda] (q',\beta') \in \calR$
we label the transition
$\Pi \ChanPar \Gamma \to[label=r,LTS=\calP_0] \Pi' \ChanPar \Gamma'$; otherwise
we label the transition with $\epsilon$, i.e.~
$\Pi \ChanPar \Gamma \to[label=\epsilon,LTS=\calP_0] \Pi' \ChanPar \Gamma'$.

Let first define a representation function for the state of a pushdown process:
\begin{align*}
F(q,\beta) =\,& \set{(q,\beta)} \\
			&\union 
			\set{(q^{\text{push}}_{\beta_1},\beta_2\cdot\beta''') 
				\left|
				\begin{aligned}
				&\exists (q',\beta') \to[effect=\lambda] (q,\beta'') \in \calR\\	
				&\beta = \beta'' \cdot \beta'''\\
				&\beta'' = \beta_1 \cdot \beta_2
				\end{aligned}
				\right.
			}\\
			&\union 
			\set{(q^{\text{push}}_{\beta_1},\beta_2\cdot\beta''') 
				\left|
				\begin{aligned}
				&\exists (q_0,\beta_0) \to[effect=\lambda] (q_1,\beta_1) \in \calR\\	
				&\lambda = \spn{(q,\beta)}\\
				&\beta = \beta_1 \cdot \beta_2
				\end{aligned}
				\right.
			}\\
			&\union 
			\set{(q^{\text{pop}}_{\beta_1},\beta_2\cdot\beta''') 
				\left|
				\begin{aligned}
				&\exists (q,\beta') \to[effect=\lambda] (q',\beta'') \in \calR\\	
				&\beta = \beta' \cdot \beta'''\\
				&\beta' = \beta_1 \cdot \beta_2
				\end{aligned}
				\right.
			}
\end{align*}
with which we can now relate configurations of $\calP$ and $\calP_0$ using the relation:
$\calB = \varset{(\Pi \ChanPar \Gamma, \Pi' \ChanPar \Gamma) : \Pi = \pi_1 \parallel \cdots \parallel \pi_n, \Pi' = \pi'_1 \parallel \cdots \parallel \pi'_n, \forall i \in \range{n}. \pi'_i \in F(\pi_i)}$.

Let us now prove that $\calB$ is a weak reflexive simulation. 
Suppose $\paren{\pi \parallel \Pi \ChanPar \Gamma, \pi_0 \parallel \Pi_0 \ChanPar \Gamma} \in \calB$, and clearly $\pi = (q,\beta)$, and
$(q,\beta) \parallel \Pi \ChanPar \Gamma \to[label=l,LTS=\calP] (q',\beta') \parallel \Pi' \ChanPar \Gamma'$.
Clearly this must happen using rule $l = (q,\beta_0) \to[effect=\lambda] (q',\beta'_0)$,
and, $\beta = \beta_0 \cdot \beta_1$ and $\beta' = \beta'_0 \cdot \beta_1$.
We will briefly show that we can assume $\pi_0 = (q^{\text{pop}}_{\beta_0},\beta_1)$.
We observe that we can then perform the following $\epsilon$-labelled transitions:
\begin{align*}
(q^{\text{push}}_{\beta''},\beta''') \parallel \Pi_0 \ChanPar \Gamma 
	&\to[label=\epsilon,LTS=\calP_0,*] 
		(q^{\text{push}}_{\epsilon},\beta'' \cdot  \beta''') \parallel \Pi_0 \ChanPar \Gamma\\
	&\to[label=\epsilon,LTS=\calP_0]
	(q_{\epsilon},\beta'' \cdot  \beta''') \parallel \Pi_0 \ChanPar \Gamma\\
	&= (q_{\epsilon},\beta) \parallel \Pi_0 \ChanPar \Gamma\\
	&\to[label=\epsilon,LTS=\calP_0]
	(q^{\text{pop}}_{\epsilon},\beta) \parallel \Pi_0 \ChanPar \Gamma\\
	&= (q^{\text{pop}}_{\epsilon},\beta_0 \cdot \beta_1) \parallel \Pi_0 \ChanPar \Gamma\\
	&\to[label=\epsilon,LTS=\calP_0,*]
	(q^{\text{pop}}_{\beta_0},\beta_1) \parallel \Pi_0 \ChanPar \Gamma
\end{align*}
it should be clear that for all $\pi_0 \in F(q,\beta)$ we can $\epsilon$-transition
$\pi_0 \parallel \Pi_0 \ChanPar \Gamma \to[label=\epsilon,LTS=\calP_0,*]
	(q^{\text{pop}}_{\beta_0},\beta_1) \parallel \Pi_0 \ChanPar \Gamma$. Hence we will assume in the following that
	$\pi_0 = (q^{\text{pop}}_{\beta_0},\beta_1)$.

First we note that there is a rule 
$({q}^{\text{pop}}_{\beta_0},\epsilon) \to[effect=\lambda'] ({q'}^{\text{push}}_{\beta'_0},\epsilon) \in \calR'$.
We can thus make a case analysis on $\lambda$.
\begin{itemize}[$\bullet$,leftmargin=*]
\item \emph{Case: } $\lambda = \epsilon$. \newline
Then clearly $\Pi = \Pi'$ and $\Gamma' = \Gamma$ and $\lambda' = \epsilon$
and:
$$(q^{\text{pop}}_{\beta_0},\beta_1) \parallel \Pi_0 \ChanPar \Gamma \to[label=l,LTS=\calP_0] ({q'}^{\text{push}}_{\beta'_0},\beta_1) \parallel \Pi_0 \ChanPar \Gamma.$$
Thus we can see that $({q'}^{\text{push}}_{\beta'_0},\beta_1) \in F(q',\beta')$  and hence
$((q',\beta') \parallel \Pi' \ChanPar \Gamma, ({q'}^{\text{push}}_{\beta'_0},\beta_1) \parallel \Pi_0 \ChanPar \Gamma) \in \calB$ 
which is what we wanted to prove.
\item \emph{Case: } $\lambda = \snd{c}{m}$. \newline
Then clearly $\Pi = \Pi'$ and $\Gamma' = \Gamma \oplus \varupdate{}{c}{\mset{m}}$ and $\lambda' = \snd{c}{m}$
and:
$$(q^{\text{pop}}_{\beta_0},\beta_1) \parallel \Pi_0 \ChanPar \Gamma \to[label=l,LTS=\calP_0] ({q'}^{\text{push}}_{\beta'_0},\beta_1) \parallel \Pi_0 \ChanPar \Gamma'.$$
Thus we can see that $({q'}^{\text{push}}_{\beta'_0},\beta_1) \in F(q',\beta')$  and hence
$((q',\beta') \parallel \Pi' \ChanPar \Gamma', ({q'}^{\text{push}}_{\beta'_0},\beta_1) \parallel \Pi_0 \ChanPar \Gamma') \in \calB$ 
which is what we wanted to prove.
\item \emph{Case: } $\lambda = \rec{c}{m}$. \newline
Then clearly $\Pi = \Pi'$ and $\Gamma = \Gamma' \oplus \varupdate{}{c}{\mset{m}}$ and $\lambda' = \rec{c}{m}$
and:
$$(q^{\text{pop}}_{\beta_0},\beta_1) \parallel \Pi_0 \ChanPar \Gamma \to[label=l,LTS=\calP_0] ({q'}^{\text{push}}_{\beta'_0},\beta_1) \parallel \Pi_0 \ChanPar \Gamma'.$$
Thus we can see that $({q'}^{\text{push}}_{\beta'_0},\beta_1) \in F(q',\beta')$  and hence
$((q',\beta') \parallel \Pi' \ChanPar \Gamma', ({q'}^{\text{push}}_{\beta'_0},\beta_1) \parallel \Pi_0 \ChanPar \Gamma') \in \calB$ 
which is what we wanted to prove.
\item \emph{Case: } $\lambda = \spn{(q'',\beta'')}$. \newline
Then clearly $\Pi' = (q'',\beta'') \parallel \Pi$ and $\Gamma' = \Gamma$ and $\lambda' = \spn{(q''^{\text{push}}_{\beta''},\epsilon)}$
and:
$$(q^{\text{pop}}_{\beta_0},\beta_1) \parallel \Pi_0 \ChanPar \Gamma \to[label=l,LTS=\calP_0] ({q'}^{\text{push}}_{\beta'_0},\beta_1) \parallel (q''^{\text{push}}_{\beta''},\epsilon) \parallel \Pi_0 \ChanPar \Gamma'.$$
Thus we can see that $({q'}^{\text{push}}_{\beta'_0},\beta_1) \in F(q',\beta')$,
$(q''^{\text{push}}_{\beta''},\epsilon) \in F(q'',\beta'')$
  and hence
$((q',\beta') \parallel \Pi' \ChanPar \Gamma', ({q'}^{\text{push}}_{\beta'_0},\beta_1) \parallel (q''^{\text{push}}_{\beta''},\epsilon) \parallel \Pi_0 \ChanPar \Gamma') \in \calB$ 
which is what we wanted to prove.
\end{itemize}
Hence we can conclude that $\calB$ is a weak reflexive simulation.

Let us turn now to $\calB^{-1}$.
Suppose $\paren{\pi \parallel \Pi \ChanPar \Gamma, \pi_0 \parallel \Pi_0 \ChanPar \Gamma} \in \calB$ and
$\pi_0 \parallel \Pi_0 \ChanPar \Gamma \to[label=l,LTS=\calP] \pi'_0 \parallel \Pi'_0 \ChanPar \Gamma'_0$
using rule $r \in \calR'$.

Let us perform a case analysis on $r$
\begin{itemize}
\item \emph{Case: } $r = (q,\epsilon) \to[effect=\epsilon] ({q}^{\text{pop}}_{\epsilon},\epsilon)$.  \newline
Then clearly $\pi_0 = (q,\beta)$, $\pi=(q,\beta)$ and 
$\pi'_0 = ({q}^{\text{pop}}_{\epsilon},\beta)$, $\Pi'_0 = \Pi_0$, and $\Gamma'_0 = \Gamma$.
Further $({q}^{\text{pop}}_{\epsilon},\beta) \in F(q,\beta)$ and thus
$\paren{\pi \parallel \Pi \ChanPar \Gamma, \pi'_0 \parallel \Pi'_0 \ChanPar \Gamma'_0} \in \calB$ which is what we wanted to prove.
\item \emph{Case: } $r = ({q}^{\text{push}}_{\epsilon},\epsilon) \to[effect=\epsilon] (q,\epsilon)$.  \newline
Clearly $\pi_0 = ({q}^{\text{push}}_{\epsilon},\beta)$, $\pi=(q,\beta)$ and 
$\pi'_0 = (q,\beta)$, $\Pi'_0 = \Pi_0$, and $\Gamma'_0 = \Gamma$.
Further $(q,\beta) \in F(q,\beta)$ and thus
$\paren{\pi \parallel \Pi \ChanPar \Gamma, \pi'_0 \parallel \Pi'_0 \ChanPar \Gamma'_0} \in \calB$ which is what we wanted to prove.
\item \emph{Case: } $r = ({q}^{\text{pop}}_{\beta},A) \to[effect=\epsilon] ({q}^{\text{pop}}_{\beta\cdot A},\epsilon)$.  \newline
Clearly $\pi_0 = ({q}^{\text{pop}}_{\beta},A \cdot \beta')$, $\pi=(q,\beta \cdot A \cdot \beta')$ and 
$\pi'_0 = ({q}^{\text{pop}}_{\beta\cdot A},\beta')$, $\Pi'_0 = \Pi_0$, and $\Gamma'_0 = \Gamma$.
Further $({q}^{\text{pop}}_{\beta\cdot A},\beta') \in F(q,\beta \cdot A \cdot \beta')$ and thus
$\paren{\pi \parallel \Pi \ChanPar \Gamma, \pi'_0 \parallel \Pi'_0 \ChanPar \Gamma'_0} \in \calB$ which is what we wanted to prove.
\item \emph{Case: } $r = ({q}^{\text{push}}_{\beta \cdot A},\epsilon) \to[effect=\epsilon] ({q}^{\text{push}}_{\beta},A)$.  \newline
Clearly $\pi_0 = ({q}^{\text{push}}_{\beta \cdot A}, \beta')$, $\pi=(q,\beta \cdot A \cdot \beta')$ and 
$\pi'_0 = ({q}^{\text{push}}_{\beta},A \cdot \beta')$, $\Pi'_0 = \Pi_0$, and $\Gamma'_0 = \Gamma$.
Further $({q}^{\text{pop}}_{\beta},\beta') \in F(q,\beta \cdot A \cdot \beta')$ and thus
$\paren{\pi \parallel \Pi \ChanPar \Gamma, \pi'_0 \parallel \Pi'_0 \ChanPar \Gamma'_0} \in \calB$ which is what we wanted to prove.
\item \emph{Case: } $r =({q}^{\text{pop}}_{\beta},\epsilon) \to[effect=\lambda] ({q'}^{\text{push}}_{\beta'},\epsilon)$. \newline
Clearly $\pi_0 = ({q}^{\text{pop}}_{\beta}, \beta'')$, $\pi=(q,\beta \cdot \beta'')$ and 
$\pi'_0 = ({q'}^{\text{push}}_{\beta'}, \beta'')$.
Further there is a rule $(q,\beta) \to[effect=\lambda'] (q',\beta') \in \calR$.
Let us do case analysis on $\lambda$
\begin{itemize}
\item \emph{Case: } $\lambda = \epsilon$. \newline
Then $\Pi'_0 = \Pi_0$, $\Gamma'_0 = \Gamma$ and $\lambda' = \epsilon$.
And
$(q,\beta\beta'') \parallel \Pi \ChanPar \Gamma \to[label=r,LTS=\calP] (q',\beta'\beta'') \parallel \Pi \ChanPar \Gamma$.
Clearly $({q'}^{\text{push}}_{\beta},\beta'') \in F(q',\beta'\beta'')$ and thus
$\varparen{(q',\beta'\beta'') \parallel \Pi \ChanPar \Gamma, \pi'_0 \parallel \Pi'_0 \ChanPar \Gamma'_0}\discn \in \calB$ which is what we wanted to prove.
\item \emph{Case: } $\lambda = \snd{c}{m}$ \newline
Then $\Pi'_0 = \Pi_0$, $\Gamma'_0 = \Gamma \oplus \varupdate{}{c}{\mset{m}}$ and $\lambda' = \snd{c}{m}$.
And
$(q,\beta\beta'') \parallel \Pi \ChanPar \Gamma \to[label=r,LTS=\calP] (q',\beta'\beta'') \parallel \Pi \ChanPar \Gamma'_0$.
Clearly $({q'}^{\text{push}}_{\beta},\beta'') \in F(q',\beta'\beta'')$ and thus
$\varparen{(q',\beta'\beta'') \parallel \Pi \ChanPar \Gamma'_0, \pi'_0 \parallel \Pi'_0 \ChanPar \Gamma'_0} \in \calB$ which is what we wanted to prove.
\item \emph{Case: } $\lambda = \rec{c}{m}$ \newline
Then $\Pi'_0 = \Pi_0$, $\Gamma = \Gamma'_0 \oplus \varupdate{}{c}{\mset{m}}$ and $\lambda' = \rec{c}{m}$.
And
$(q,\beta\beta'') \parallel \Pi \ChanPar \Gamma \to[label=r,LTS=\calP] (q',\beta'\beta'') \parallel \Pi \ChanPar \Gamma'_0$.
Clearly $({q'}^{\text{push}}_{\beta},\beta'') \in F(q',\beta'\beta'')$ and thus
$\varparen{(q',\beta'\beta'') \parallel \Pi \ChanPar \Gamma'_0, \pi'_0 \parallel \Pi'_0 \ChanPar \Gamma'_0} \in \calB$ which is what we wanted to prove.
\item \emph{Case: } $\lambda = \spn{({q''}^{\text{push}}_{\beta'''},\epsilon)}$ \newline
Then $\Pi'_0 = ({q''}^{\text{push}}_{\beta'''},\epsilon) \parallel \Pi_0$, $\Gamma = \Gamma'_0$ and
 $\lambda' = \spn{(q'',\beta''')}$.
And
$(q,\beta\beta'') \parallel \Pi \ChanPar \Gamma \to[label=r,LTS=\calP] (q',\beta'\beta'') \parallel (q'',\beta''') \parallel \Pi \ChanPar \Gamma$.
Firstly $({q'}^{\text{push}}_{\beta},\beta'') \in F(q',\beta'\beta'')$ and 
$({q''}^{\text{push}}_{\beta'''},\epsilon) \in F(q'',\beta''')$.
Thus
$\varparen{(q',\beta'\beta'') \parallel (q'',\beta''') \parallel \Pi \ChanPar \Gamma'_0, \pi'_0 \parallel \Pi'_0 \ChanPar \Gamma'_0} \in \calB$ which is what we wanted to prove.
\end{itemize}
\end{itemize}
Hence $\calB^{-1}$ is a weak reflexive simulation and hence $\calB$ is a weak reflexive bisimulation.

Let us now define $\Fprenormalform{\calP} = (\calQ',\calA,\Chan,\Msg[],\calR'\union\calR_{\text{cov}})$ where
$$
	\calR_{\text{cov}} = 
	\set{\left.
	\begin{aligned}
	    &({q}^{\text{pop}}_{\epsilon},\epsilon)      \to[effect=\epsilon] ({q}^{\text{cov}},\epsilon), \\
		&({q}^{\text{push}}_{\epsilon},\epsilon)     \to[effect=\epsilon] ({q}^{\text{cov}},\epsilon), \\
		&({q}^{\text{pop}}_{A \cdot \beta},\epsilon)      \to[effect=\epsilon] ({q}^{\text{cov}},A), \\
		&({q}^{\text{push}}_{A' \cdot \beta'},\epsilon)   \to[effect=\epsilon] ({q}^{\text{cov}},A'), \\
		&({q},\epsilon)   \to[effect=\epsilon] ({q}^{\text{cov}},\epsilon), \\
	\end{aligned} \right|
	  \begin{aligned}
	  	&{q}^{\text{pop}}_{\epsilon},{q}^{\text{push}}_{\epsilon'} \in \calQ',\\
	  	&{q}^{\text{pop}}_{A \cdot \beta},{q}^{\text{push}}_{A' \cdot \beta'} \in \calQ', \\
	  	&q \in \calQ
	  \end{aligned}
	}.
$$
Adding the rules $\calR_{\text{cov}}$ to $\calP_0$ (which are non-reversible) only changes which configurations are reachable/coverable by a one step transition. Obviously $\Fprenormalform{\calP}$ remains polynomial time computable from $\calP$

Suppose that $\mathpzc{Q} = (\calP, \Pi_0\ChanPar\Gamma_0, \Pi \ChanPar \Gamma)$ is a simple coverbility query
where $\Pi = (q_1,\beta_1) \parallel \cdots \parallel (q_k,\beta_k)$ and $\beta_i \in \calA \union \set{\epsilon}$ then let 
$\Fprenormalform{\mathpzc{Q}} = (\Fprenormalform{\calP}, \Pi_0\ChanPar\Gamma_0, \Pi' \ChanPar \Gamma)$ be a simple coverability query
such that $\Pi' = (q_1^{\text{cov}},\beta_1) \parallel \cdots \parallel (q_k^{\text{cov}},\beta_k)$.

Suppose $\mathpzc{Q}$ is a yes-instance
then $\Pi_0\ChanPar\Gamma_0 \to[TS=\calP,*] \Pi_1 \ChanPar \Gamma_1$ such that 
$\Pi \ChanPar \Gamma \leqACPS \Pi_1 \ChanPar \Gamma_1$. Since $\calB$ is a reflexive weak bisimulation we know
$\Pi_0\ChanPar\Gamma_0 \to[TS=\calP_0,*] \Pi'_1 \ChanPar \Gamma_1$ such that 
$(\Pi_1 \ChanPar \Gamma_1,\discn \Pi'_1 \ChanPar \Gamma_1) \in \calB$. 
Since $\Pi \ChanPar \Gamma \leqACPS \Pi_1 \ChanPar \Gamma_1$ we know that
$\Pi_1 = (q_1,\beta'_1) \parallel \cdots \parallel (q_k,\beta'_k) \parallel \Pi_2$
such that for all $i \in \range{k}$ either $\beta_i = \epsilon$ or $\beta_i = A_i \in \calA$ and
$\beta'_1 = A_i \, \beta''_1$.
Hence we can deduce that $\Pi'_1 = \pi'_1 \parallel \cdots \parallel \pi'_k \parallel \Pi'_2$ such that
$\pi'_i \in F(q_i,\beta'_i)$ for all $i \in \range{k}$. Thus it is easy to see that
$\Pi'_1 \ChanPar \Gamma_1 \to[TS=\Fprenormalform{\calP},*] (q_1^{\text{cov}},\beta'''_1) \parallel \cdots \parallel (q_k^{\text{cov}},\beta'''_k) \parallel \Pi'_2 \ChanPar \Gamma_1$ such that for all $i \in \range{k}$ it is the case that
$\beta'''_i = A \beta''''_i$ and $\beta'_i = A \beta^\dagger_i$
and hence $\Fprenormalform{\mathpzc{Q}}$ is a yes-instance.

Conversely, suppose $\Fprenormalform{\mathpzc{Q}}$ is a yes instance then
$\Pi_0\ChanPar\Gamma_0 \discn \to[TS=\Fprenormalform{\calP},*] (q_1^{\text{cov}},\beta'_1) \parallel \cdots \parallel (q_k^{\text{cov}},\beta'_k) \parallel \Pi'_1 \ChanPar \Gamma_1$
such that for all $i \in \range{k}$ either $\beta_i = \epsilon$ or $\beta_i = A_i$ and $\beta'_i = A_i\, \beta'_i$. 
Hence clearly (by reversing transitions from $\calR_{\text{cov}}$)
$\Pi_0\ChanPar\Gamma_0 \to[TS=\calP_0,*] \pi_1 \parallel \cdots \parallel \pi_k \parallel \Pi'_1 \ChanPar \Gamma_1$ where
$\pi'_i \in F(q_i,\beta''_i)$ for some $\beta''_1,\ldots,\beta''_k$ such that either $\beta_i = \epsilon$ or $\beta''_i = A_i \beta'''_i$ for all $i \in \range{k}$ .
Since $\calB$ is a reflexive weak bisimulation we know that
$\Pi_0\ChanPar\Gamma_0 \to[TS=\calP,*] (q_1,\beta''_1) \parallel \cdots \parallel (q_k,\beta''_k) \parallel \Pi_1 \ChanPar \Gamma_1$. Hence $\mathpzc{Q}$ is a yes-instance.
We can thus conclude that $\mathpzc{Q}$ is a yes-instance iff $\Fprenormalform{\mathpzc{Q}}$ is a yes-instance.

For boundedness, let $\Pi^0 \ChanPar \Gamma^0$ be a start configuration and let $\Fprenormalform{\Pi^0 \ChanPar \Gamma^0} = \Pi^0 \ChanPar \Gamma^0$.
Suppose that 
$\varset{\Pi \ChanPar \Gamma : \Pi_0 \ChanPar \Gamma_0 \discn \to[TS=\calP,*] \Pi \ChanPar \Gamma}$ is a finite set.
We notice that for all $\Pi \ChanPar \Gamma$ the set $\varset{\Pi' \ChanPar \Gamma : (\Pi \ChanPar \Gamma,\Pi' \ChanPar \Gamma) \in \calB }$ is finite.
Thus using that $\calB$ is a reflexive weak bisimulation we can infer that 
$\varset{\Pi \ChanPar \Gamma : \Pi_0 \ChanPar \Gamma_0 \discn \to[TS=\calP_0,*] \Pi \ChanPar \Gamma}$ is a finite set.
Since the rules in $\calR_{\text{cov}}$ adds only a finite number of reachable configurations we can conclude
$\set{\Pi \ChanPar \Gamma : \Pi_0 \ChanPar \Gamma_0 \to[TS=\Fprenormalform{\calP},*] \Pi \ChanPar \Gamma}$ is a finite set and thus $\Fprenormalform{\calP}$ is bounded from $\Fprenormalform{\Pi^0 \ChanPar \Gamma^0}$.

Conversely, suppose $\set{\Pi \ChanPar \Gamma : \Pi_0 \ChanPar \Gamma_0 \to[TS=\Fprenormalform{\calP},*] \Pi \ChanPar \Gamma}$ is a finite set. 
Then since the rules in $\calR_{\text{cov}}$ adds only a finite number of reachable configurations we can infer
$\varset{\Pi \ChanPar \Gamma : \Pi_0 \ChanPar \Gamma_0 \to[TS=\calP_0,*] \Pi \ChanPar \Gamma}$
is a finite set. We further note:
$\varset{\Pi \ChanPar \Gamma : \Pi_0 \ChanPar \Gamma_0 \to[TS=\calP,*] \Pi \ChanPar \Gamma} \subseteq \set{\Pi \ChanPar \Gamma : \Pi_0 \ChanPar \Gamma_0 \to[TS=\calP_0,*] \Pi \ChanPar \Gamma}$ and thus $\calP$ is bounded from $\Pi_0 \ChanPar \Gamma_0$. 
Thus $\calP$ is bounded from $\Pi_0 \ChanPar \Gamma_0$ if and only if $\Fprenormalform{\calP}$ is bounded from $\Fprenormalform{\Pi_0 \ChanPar \Gamma_0}$.

For termination, let $\Pi_0 \ChanPar \Gamma_0$ be a start configuration and define again $\Fprenormalform{\Pi^0 \ChanPar \Gamma^0} = \Pi^0 \ChanPar \Gamma^0$.
Suppose there exists an infinite path $\vec{s}$ in $\calP$ starting from $\Pi_0 \ChanPar \Gamma_0$. 
Since $\calB$ is a weak reflexive bisimulation and $\vec{s}$ uses an infinite sequence of labels it is clear that there is a path $\vec{s}'$ in $\calP_0$ and $\vec{s}'$ is also an infinite path. The path $\vec{s}'$ is clearly also a
 path of $\Fprenormalform{\calP}$ hence $\Fprenormalform{\calP}$ is non-terminating from 
 $\Fprenormalform{\Pi^0 \ChanPar \Gamma^0}$.
Conversely, suppose that $\vec{s}'$ is an infinite path in $\Fprenormalform{\calP}$ starting from $\Pi_0 \ChanPar \Gamma_0$. 
Since a path can only be finitely extended by rules in $\calR_{\text{cov}}$ we can deduce that there is also 
an infinite path $\vec{s}''$ from $\Pi_0 \ChanPar \Gamma_0$ in $\calP_0$.
If $\vec{s}''$ gives rise to an infinite sequence of labels
then, since $\calB$ is a weak reflexive bisimulation, we clearly obtain a path $\vec{s}$ in $\calP$ that is also infinite (since all transitions in $\calP$ carry a label).
Suppose for a contradiction that $\vec{s}''$ gives rise only for a finite sequence of labels.
This implies there exists an infinite path $\vec{s}_0$ in $\calP_0$ such that for all $i$ the transition
$\vec{s}(i) \to[label=\epsilon,LTS=\calP] \vec{s}(i+1)$ is an $\epsilon$-transition.
Inspecting the definition $\calP_0$ we see this is implies all rules used 
must of the form $(q,\beta) \to[effect=\epsilon] (q',\beta')$.
We can conclude that this is impossible since
there are no cycles in $\calP_0$'s rules with no side effects.
Hence we can conclude that $\calP$ is non-terminating from $\Pi^0 \ChanPar \Gamma^0$ if, and only if,
$\Fprenormalform{\calP}$ is non-terminating from $\Fprenormalform{\Pi_0 \ChanPar \Gamma_0}$.
\end{proof}

We can now use a summarisation-inspired idea to encode control-states into the stack alphabet:

\begin{lemma}\label{app:lemma:pnfTonf}
Given an ACPS $\calP$ in pre-normal form, a simple coverability query $\mathpzc{Q}$ and a start configuration $\Pi^0 \ChanPar \Gamma^0$ there exists ACPS $\FpnfTonormalform{\calP}$ satisfying:
for all rules $(q,\beta) \to[effect=\lambda] (q',\beta')$ of $\FpnfTonormalform{\calP}$
\begin{inparaenum}
\item(acps:nf:i) $q = q'$,
\item(acps:nf:ii) $\beta = A \in \calA$,
\item(acps:nf:iii) if $\lambda \neq \epsilon$ then $\beta' = A' \in \calA$,
otherwise \item(acps:nf:iv) $\beta' \in \set{\epsilon, A', B\,C : A', B, C \in \calA}$, and
\item(acps:nf:v) $\lambda = \spn{(q'',\beta)}$ then $q'' = q$, and $\beta \in \calA$;
\end{inparaenum}
simple coverability query $\FpnfTonormalform{\mathpzc{Q}}$; and start configuration $\FpnfTonormalform{\Pi^0 \ChanPar \Gamma^0}$ --- all poly\-no\-mi\-al-time computable --- such 
that:
\begin{inparaenum}[(A)]
\item $\mathpzc{Q}$ is a yes-instance if, and only if, $\FpnfTonormalform{\mathpzc{Q}}$ is a yes-instance; and
\item $\calP$ is bounded (terminating) from $\Pi^0 \ChanPar \Gamma^0$ if, and only if, $\FpnfTonormalform{\calP}$ is bounded (terminating respectively) from 
$\FpnfTonormalform{\Pi^0 \ChanPar \Gamma^0}$.
\end{inparaenum}
\end{lemma}
\begin{proof}
Suppose $\calP = (\calQ,\calA,\Chan,\Msg[],\calR)$ is an ACPS in pre-normal form. We then define
$\calP' = (\calQ',\calA',\Chan,\Msg[],\calR')$ where $\calQ' = \set{q_0}$ and $q_0$ is a fresh control state,
$\calA' = \varset{A^{(q,q')} \mid q,q' \in \calQ, A \in \calQ \union \varset{\Theta}}$, where $\Theta$ is a fresh symbol, and $\calR'$ is obtained from $\calR$ as follows
\begin{align*}
\calR' =\,& \set{A^{(q,q'')} \to[effect=\lambda'] A^{(q',q'')} \left| 
	\begin{aligned}
		&(q,\epsilon) \to[effect=\lambda] (q',\epsilon) \in \calR, \\
		&A \in \calA \union \varset{\Theta}, q'',q_0' \in \calQ, \\
		&\text{if } \lambda = \spn{(q_0,\epsilon)} \text{ then } \lambda' = \spn{\Theta^{(q_0,q_0')}}\\
		&\text{otherwise } \lambda = \lambda'
	\end{aligned}\right.
	}\\
&\union \set{A^{(q,q'')} \to[effect=\epsilon] {B}^{(q',q''')} A^{(q''',q'')} \left| 
	\begin{aligned}
		&(q,\epsilon) \to[effect=\epsilon] (q',B) \in \calR, \\
		&A \in \calA \union \varset{\Theta},\\
		&q'',q''' \in \calQ 
	\end{aligned}\right.
}\\
&\union \set{A^{(q,q')}   \to[effect=\epsilon] \epsilon \mid (q,A) \to[effect=\epsilon] (q',\epsilon) \in \calR}
\end{align*}
where rather than writing $(q_0,A)$ we just write $A$.
It is easy to see that $\calP'$ is polynomial time computable from $\calP$.

We represent the state of a pushdown process by the following function:
\begin{align*}
F(q,&A_1 A_2 A_3\cdots A_n) = \\
	&\set{A_1^{(q,q_1)} A_2^{(q_1,q_2)} A_3^{(q_2,q_3)}\cdots A_n^{(q_{n-1},q_n)}\Theta^{q_{n},q_{n+1}} \mid \Xi}
\end{align*}
where $\Xi = q_1,\ldots q_{n+1} \in \calQ$.
We use the former to represent a $\calP$-configuration as a set:
\begin{align*}
G(\Pi \ChanPar \Gamma) = 
	\set{\Pi' \ChanPar \Gamma \left| 
	\begin{aligned}
		&\Pi  = \pi_1  \parallel \cdots \parallel \pi_n, \\
		&\Pi' = \pi'_1 \parallel \cdots \parallel \pi'_n,\\ 
		&\forall i \in \range{n}. \pi'_i \in F(\pi_i) 
	\end{aligned}
	\right.}.
\end{align*}
Further we define a relation of configurations and sets of configurations:
$$\calB = \set{(\Pi\ChanPar\Gamma, G(\Pi\ChanPar\Gamma))}.$$

Let us define a \emph{co-universal powerset lifting} of the transition system induced by $\calP'$ and $\to[TS=\calP_0]$ as follows $\P[\calP'] \is (\P[\M[\calQ' \times \calA^*]],\to[TS={\P[\calP']}])$ where $S \to[TS={\P[\calP']}] S'$ just if
for all $s' \in S'$ there exists $s \in S$ such that $s \to[TS=\calP_0] s'$.
Further we temporarily label $\calP$ and $\P[\calP']$ by \calP-configurations as follows:
if $s \to[TS=\calP] s'$ we label by $s'$ the transition $s \to[label=s',LTS=\calP] s'$.
If $S \to[TS={\P[\calP']}] S'$ such that $S'=G(s')$ for some $\calP$-configuration $s'$ we label by $s'$ the transition
$S \to[label=s',LTS={\P[\calP']}] S'$.

We will show that $\calB$ is a bisimulation relation for $\calP$ and $\P[\calP']$. As a first step
let us give a proof that $\calB$ is a simulation relation:
Suppose $(\pi \parallel \Pi \ChanPar \Gamma, G(\pi \parallel \Pi \ChanPar \Gamma)) \in \calB$ and
$\pi \parallel \Pi \ChanPar \Gamma \to[label=l,LTS=\calP] \pi' \parallel \Pi' \ChanPar \Gamma'$
using rule $r \in \calR$. We know that $l = \pi' \parallel \Pi' \ChanPar \Gamma'$.
So let $\pi'_1 \parallel \Pi'_1 \ChanPar \Gamma' \in G(\pi' \parallel \Pi' \ChanPar \Gamma')$ 
such that $\Pi'_1 \ChanPar \Gamma' \in G(\Pi' \ChanPar \Gamma')$ and $\pi'_1 \in F(\pi')$.
Let us perform a case analysis on $r$:
\begin{itemize}
\item \emph{Case: } $r=(q,\epsilon) \to[effect=\epsilon] (q',B)$. \newline
In this case $\pi = (q,A_1\cdots A_n)$, $\pi' = (q',B\,A_1\cdots A_n)$, $\Gamma = \Gamma'$ and $\Pi = \Pi'$.
From this we can deduce: $\pi'_1 = B^{(q',q_1)}\,A_1^{(q_1,q_2)}\cdots A_n^{(q_{n},q_{n+1})}\Theta^{(q_{n+1},q_{n+2})}$ for some $\commabr{q_1,\ldots,q_{n+2}}$.
Let $\pi_1 = A_1^{(q,q_2)}\cdots A_n^{(q_{n},q_{n+1})}\Theta^{(q_{n+1},q_{n+2})}$ then clearly
$\pi_1 \parallel \Pi'_1 \ChanPar \Gamma \to[TS=\calP_0] \pi'_1 \parallel \Pi'_1 \ChanPar \Gamma$
using rule $A_1^{(q,q_2)} \to[effect=\epsilon] B^{(q',q_1)}\,A_1^{(q_1,q_2)}$.
Further $\pi_1 \in F(\pi)$ and hence $\pi_1 \parallel \Pi'_1 \ChanPar \Gamma \discn \in G(\pi \parallel \Pi' \ChanPar \Gamma) = G(\pi \parallel \Pi \ChanPar \Gamma)$. And so since $\pi'_1 \parallel \Pi'_1 \ChanPar \Gamma'$ is an arbitrary element of $G(\pi' \parallel \Pi' \ChanPar \Gamma')$ we can deduce
$G(\pi \parallel \Pi \ChanPar \Gamma) \to[label=l,LTS={\P[\calP']}] G(\pi' \parallel \Pi' \ChanPar \Gamma')$.
and $(\pi' \parallel \Pi' \ChanPar \Gamma',G(\pi' \parallel \Pi' \ChanPar \Gamma')) \in \calB$.
\item \emph{Case: } $r=(q,A) \to[effect=\epsilon] (q',\epsilon)$. \newline
In this case $\pi = (q,A\,A_1\cdots A_n)$, $\pi' = (q',A_1\cdots A_n)$, $\Gamma = \Gamma'$ and $\Pi = \Pi'$.
From this we can deduce: $\pi'_1 = A_1^{(q',q_1)}\cdots A_n^{(q_{n-1},q_{n})}\Theta^{(q_{n},q_{n+1})}$ for some $q_1,\ldots,q_{n+1}$.
Let $\pi_1 = A^{(q,q')}\,A_1^{(q',q_1)}\cdots A_n^{(q_{n-1},q_{n})}\Theta^{(q_{n},q_{n+1})}$ then clearly
$\pi_1 \parallel \Pi'_1 \ChanPar \Gamma \to[TS=\calP_0] \pi'_1 \parallel \Pi'_1 \ChanPar \Gamma$
using rule $A^{(q,q')} \to[effect=\epsilon] \epsilon$.
Further $\pi_1 \in F(\pi)$ and hence $\pi_1 \parallel \Pi'_1 \ChanPar \Gamma \in G(\pi \parallel \Pi' \ChanPar \Gamma) = G(\pi \parallel \Pi \ChanPar \Gamma)$. And so since $\pi'_1 \parallel \Pi'_1 \ChanPar \Gamma'$ is an arbitrary element of $G(\pi' \parallel \Pi' \ChanPar \Gamma')$ we can deduce
$G(\pi \parallel \Pi \ChanPar \Gamma) \to[label=l,LTS={\P[\calP']}] G(\pi' \parallel \Pi' \ChanPar \Gamma')$.
and $(\pi' \parallel \Pi' \ChanPar \Gamma',G(\pi' \parallel \Pi' \ChanPar \Gamma')) \in \calB$.
\item \emph{Case: } $r=(q,\epsilon) \to[effect=\lambda] (q',\epsilon)$. \newline
In this case $\pi = (q,A_1\cdots A_n)$, $\pi' = (q',A_1\cdots A_n)$.
From this we can deduce: $\pi'_1 = A_1^{(q',q_1)}\cdots A_n^{(q_{n-1},q_{n})}\Theta^{(q_{n},q_{n+1})}$ for some $q_1,\ldots,q_{n+1}$.
Let $\pi_1 = A_1^{(q,q_1)}\cdots A_n^{(q_{n-1},q_{n})}\discn\Theta^{(q_{n},q_{n+1})}$.
Let us perform a case analysis on $\lambda$:
\begin{itemize}
\item \emph{Case: } $\lambda = \epsilon$. \newline
We have $\Pi = \Pi'$, $\Gamma = \Gamma'$,  and
$\pi_1 \parallel \Pi'_1 \ChanPar \Gamma \to[TS=\calP_0] \pi'_1 \parallel \Pi'_1 \ChanPar \Gamma$
using rule $A_1^{(q,q_1)} \to[effect=\epsilon] A_1^{(q',q_1)}$.
Further $\pi_1 \in F(\pi)$ and hence $\pi_1 \parallel \Pi'_1 \ChanPar \Gamma \in G(\pi \parallel \Pi' \ChanPar \Gamma) = G(\pi \parallel \Pi \ChanPar \Gamma)$.
Since $\pi'_1 \parallel \Pi'_1 \ChanPar \Gamma'$ is an arbitrary element of $G(\pi' \parallel \Pi' \ChanPar \Gamma')$ we can deduce
$G(\pi \parallel \Pi \ChanPar \Gamma) \to[label=l,LTS={\P[\calP']}] G(\pi' \parallel \Pi' \ChanPar \Gamma')$.
and $(\pi' \parallel \Pi' \ChanPar \Gamma',G(\pi' \parallel \Pi' \ChanPar \Gamma')) \in \calB$.
\item \emph{Case: } $\lambda = \snd{c}{m}$. \newline
We have $\Pi = \Pi'$, $\Gamma' = \Gamma \oplus \varupdate{}{c}{\mset{m}}$, and
$\pi_1 \parallel \Pi'_1 \ChanPar \Gamma \to[TS=\calP_0] \pi'_1 \parallel \Pi'_1 \ChanPar \Gamma'$
using rule $A_1^{(q,q_1)} \to[effect=\snd{c}{m}] A_1^{(q',q_1)}$.
Further $\pi_1 \in F(\pi)$ and hence $\pi_1 \parallel \Pi'_1 \ChanPar \Gamma \in G(\pi \parallel \Pi' \ChanPar \Gamma) = G(\pi \parallel \Pi \ChanPar \Gamma)$.
Since $\pi'_1 \parallel \Pi'_1 \ChanPar \Gamma'$ is an arbitrary element of $G(\pi' \parallel \Pi' \ChanPar \Gamma')$ we can deduce
$G(\pi \parallel \Pi \ChanPar \Gamma) \to[label=l,LTS={\P[\calP']}] G(\pi' \parallel \Pi' \ChanPar \Gamma')$.
and $(\pi' \parallel \Pi' \ChanPar \Gamma',G(\pi' \parallel \Pi' \ChanPar \Gamma')) \in \calB$.
\item \emph{Case: } $\lambda = \rec{c}{m}$. \newline
We have $\Pi = \Pi'$, $\Gamma = \Gamma' \oplus \varupdate{}{c}{\mset{m}}$, and
$\pi_1 \parallel \Pi'_1 \ChanPar \Gamma \to[TS=\calP_0] \pi'_1 \parallel \Pi'_1 \ChanPar \Gamma'$
using rule $A_1^{(q,q_1)} \to[effect=\rec{c}{m}] A_1^{(q',q_1)}$.
Further $\pi_1 \in F(\pi)$ and hence $\pi_1 \parallel \Pi'_1 \ChanPar \Gamma \in G(\pi \parallel \Pi' \ChanPar \Gamma) = G(\pi \parallel \Pi \ChanPar \Gamma)$.
Since $\pi'_1 \parallel \Pi'_1 \ChanPar \Gamma'$ is an arbitrary element of $G(\pi' \parallel \Pi' \ChanPar \Gamma')$ we can deduce
$G(\pi \parallel \Pi \ChanPar \Gamma) \to[label=l,LTS={\P[\calP']}] G(\pi' \parallel \Pi' \ChanPar \Gamma')$.
and $(\pi' \parallel \Pi' \ChanPar \Gamma',G(\pi' \parallel \Pi' \ChanPar \Gamma')) \in \calB$.
\item \emph{Case: } $\lambda = \spn{q'',\epsilon}$. \newline
We have $\Pi' = (q'', \epsilon) \parallel \Pi$, $\Gamma = \Gamma'$.
Hence $\Pi'_1 = \pi_0 \parallel \Pi_1$ such that 
$\Pi_1 \ChanPar \Gamma \in G(\Pi \ChanPar \Gamma)$ and $\pi_0 \in F(q'', \epsilon)$ which implies
$\pi_0 = \Theta^{q'',q'''}$ for some $q'''$.
Thus
$\pi_1 \parallel \Pi_1 \ChanPar \Gamma \to[TS=\calP_0] \pi'_1 \parallel \Theta^{q'',q'''} \parallel \Pi_1 \ChanPar \Gamma$
using rule $A_1^{(q,q_1)} \to[effect=\spn{\Theta^{q'',q'''}}] A_1^{(q',q_1)}$.
Further $\pi_1 \in F(\pi)$ and hence $\pi_1 \parallel \Pi'_1 \ChanPar \Gamma \in G(\pi \parallel \Pi \ChanPar \Gamma)$.
Since $\pi'_1 \parallel \Pi'_1 \ChanPar \Gamma'$ is an arbitrary element of $G(\pi' \parallel \Pi' \ChanPar \Gamma')$ we can deduce
$G(\pi \parallel \Pi \ChanPar \Gamma) \to[label=l,LTS={\P[\calP']}] G(\pi' \parallel \Pi' \ChanPar \Gamma')$.
and $(\pi' \parallel \Pi' \ChanPar \Gamma',G(\pi' \parallel \Pi' \ChanPar \Gamma')) \in \calB$.
\end{itemize}
\end{itemize}
We can thus conclude that $\calB$ is a simulation relation.

For a proof that $\calB^{-1}$ is a simulation relation we will first prove a little lemma:
\begin{lemma*}
$\varset{(s_0,s) : s_0 \in G(s)}$ is a simulation relation.
\end{lemma*}
\begin{proof}
Suppose $\pi_0 \parallel \Pi_0 \ChanPar \Gamma \in G(\pi \parallel \Pi \ChanPar \Gamma)$ and 
$\pi_0 \parallel \Pi_0 \ChanPar \Gamma \to[TS=\calP'] \pi'_0 \parallel \Pi'_0 \ChanPar \Gamma'$ using rule
$r \in \calR'$. Let us make a case analysis on $r$:
\begin{itemize}
\item \emph{Case: } $r = A^{(q,q')}  \to[effect=\epsilon] \epsilon$. \newline
Then $\Pi'_0 = \Pi_0$, $\Gamma = \Gamma'$, $\pi_0 = A^{(q,q')}A_1^{(q',q_1)}\cdots A_n^{(q_{n-1},q_{n})}\discn\Theta^{(q_{n},q_{n+1})}$ and $\pi'_0 = A_1^{(q',q_1)}\cdots A_n^{(q_{n-1},q_{n})}\Theta^{(q_{n},q_{n+1})}$ for some $q_1,\ldots,q_{n+1}$.
We can deduce that $\pi = (q,A\,A_1\cdots A_n)$ and we can let
$\pi' = (q',A_1\cdots A_n)$ so that $\pi'_0 \in F(\pi')$.
Further we know by construction that $(q,A) \to[effect=\epsilon] (q',\epsilon) \in \calR$
and thus $\pi \parallel \Pi \ChanPar \Gamma \to[TS=\calP] \pi' \parallel \Pi \ChanPar \Gamma$
and $\pi'_0 \parallel \Pi'_0 \ChanPar \Gamma' \in G(\pi' \parallel \Pi' \ChanPar \Gamma') = G(\pi' \parallel \Pi \ChanPar \Gamma)$ which is what we wanted to prove.
\item \emph{Case: } $r = A^{(q,q'')} \to[effect=\epsilon] {B}^{(q',q''')} A^{(q''',q'')}$. \newline
Then $\Pi'_0 = \Pi_0$, $\Gamma = \Gamma'$, $\pi_0 = A^{(q,q'')}A_1^{(q'',q_1)}\cdots A_n^{(q_{n-1},q_{n})}\discn\Theta^{(q_{n},q_{n+1})}$ and $\pi'_0 = {B}^{(q',q''')} A^{(q''',q'')}A_1^{(q'',q_1)}\cdots A_n^{(q_{n-1},q_{n})}\discn\Theta^{(q_{n},q_{n+1})}$ for some $q_1,\ldots,q_{n+1}$.
We can deduce that $\pi = (q,A\,A_1\cdots A_n)$ and we can let
$\pi' = (q',B\,A\,A_1\cdots A_n)$ so that $\pi'_0 \in F(\pi')$.
Further we know by construction that $(q,\epsilon) \to[effect=\epsilon] (q',B) \in \calR$
and thus $\pi \parallel \Pi \ChanPar \Gamma \to[TS=\calP] \pi' \parallel \Pi \ChanPar \Gamma$
and $\pi'_0 \parallel \Pi'_0 \ChanPar \Gamma' \in G(\pi' \parallel \Pi' \ChanPar \Gamma') = G(\pi' \parallel \Pi \ChanPar \Gamma)$ which is what we wanted to prove.
\item \emph{Case: } $r = A^{(q,q'')} \to[effect=\lambda'] A^{(q',q'')}$. \newline
Then $\Pi'_0 = \Pi_0$, $\Gamma = \Gamma'$, $\pi_0 = A^{(q,q'')}A_1^{(q'',q_1)}\cdots A_n^{(q_{n-1},q_{n})}\discn\Theta^{(q_{n},q_{n+1})}$ and $\pi'_0 = A^{(q',q'')}A_1^{(q'',q_1)}\cdots A_n^{(q_{n-1},q_{n})}\discn\Theta^{(q_{n},q_{n+1})}$ for some $q_1,\ldots,q_{n+1}$.
We can deduce that $\pi = (q,A\,A_1\cdots A_n)$ and we can let
$\pi' = (q',A\,A_1\cdots A_n)$ so that $\pi'_0 \in F(\pi')$.
Further we know by construction that $(q,\epsilon) \to[effect=\lambda] (q',\epsilon) \in \calR$.
Let us perform a case analysis on $\lambda'$.
\begin{itemize}
\item \emph{Case: } $\lambda'=\epsilon$. \newline
Then $\lambda = \epsilon$, $\Pi = \Pi'$, $\Gamma = \Gamma'$
and thus $\pi \parallel \Pi \ChanPar \Gamma \to[TS=\calP] \pi' \parallel \Pi \ChanPar \Gamma$
and $\pi'_0 \parallel \Pi'_0 \ChanPar \Gamma' \in G(\pi' \parallel \Pi' \ChanPar \Gamma') = G(\pi' \parallel \Pi \ChanPar \Gamma)$ which is what we wanted to prove.
\item \emph{Case: } $\lambda'=\snd{c}{m}$. \newline
Then $\lambda = \snd{c}{m}$, $\Pi = \Pi'$, $\Gamma' = \Gamma \oplus \varupdate{}{c}{\mset{m}}$
and thus $\pi \parallel \Pi \ChanPar \Gamma \to[TS=\calP] \pi' \parallel \Pi \ChanPar \Gamma'$
and $\pi'_0 \parallel \Pi'_0 \ChanPar \Gamma' \in G(\pi' \parallel \Pi' \ChanPar \Gamma') = G(\pi' \parallel \Pi \ChanPar \Gamma')$ which is what we wanted to prove.
\item \emph{Case: } $\lambda'=\rec{c}{m}$. \newline
Then $\lambda = \rec{c}{m}$, $\Pi = \Pi'$, $\Gamma = \Gamma' \oplus \varupdate{}{c}{\mset{m}}$
and thus $\pi \parallel \Pi \ChanPar \Gamma \to[TS=\calP] \pi' \parallel \Pi \ChanPar \Gamma'$
and $\pi'_0 \parallel \Pi'_0 \ChanPar \Gamma' \in G(\pi' \parallel \Pi' \ChanPar \Gamma') = G(\pi' \parallel \Pi \ChanPar \Gamma')$ which is what we wanted to prove.
\item \emph{Case: } $\lambda'=\spn{\Theta^{(q''',q'''')}}$. \newline
Then $\lambda = \spn{(q''',\epsilon)}$, $\Pi' = (q''',\epsilon) \parallel \Pi$, $\Gamma' = \Gamma$
and thus $\pi \parallel \Pi \ChanPar \Gamma \to[TS=\calP] \pi' (q''',\epsilon) \parallel \parallel \Pi \ChanPar \Gamma'$
and $\pi'_0 \parallel \Pi'_0 \ChanPar \Gamma' \in G(\pi' \parallel \Pi' \ChanPar \Gamma') = G(\pi' \parallel \Pi \ChanPar \Gamma')$ which is what we wanted to prove.
\end{itemize}
\end{itemize}
which concludes the proof.
\end{proof}

Now we can use the above Lemma to show that $\calB^{-1}$ is a simulation relation.
Hence suppose $(\pi \parallel \Pi \ChanPar \Gamma, G(\pi \parallel \Pi \ChanPar \Gamma)) \in \calB$ and
$G(\pi \parallel \Pi \ChanPar \Gamma) \to[label=l,LTS={\P[\calP']}] G(\pi' \parallel \Pi' \ChanPar \Gamma')$
then clearly $G(\pi' \parallel \Pi' \ChanPar \Gamma') \neq \emptyset$ and hence we have
$s \in G(\pi \parallel \Pi \ChanPar \Gamma)$ and $s' \in G(\pi' \parallel \Pi' \ChanPar \Gamma')$
such that $s \to[TS=\calP'] s'$ from which the above Lemma let's us conclude that 
$\pi \parallel \Pi \ChanPar \Gamma \to[TS=\calP] \pi' \parallel \Pi' \ChanPar \Gamma'$.
Hence we can conclude that $\calB$ is a bisimulation.

Let us define $\FpnfTonormalform{\calP} = (\calQ',\calA' \union \calA_{\text{cov}},\Chan,\Msg[],\calR' \union \calR_{\text{cov}})$ where
$\calA_{\text{cov}} = \varset{\Theta_{\text{cov}}^{(q,A)}, \Theta_{\text{cov}}^{(q,\epsilon)} : q \in \calQ, A \in \calA}$ and

\begin{align*}
\calR_{\text{cov}} =\,& \set{
A^{q,q'} \to[effect=\epsilon] \Theta_{\text{cov}}^{(q,A)},
A^{q,q'} \to[effect=\epsilon] \Theta_{\text{cov}}^{(q,\epsilon)} : q,q' \in \calQ, A \in \calA
}\\
&\union
\set{
\Theta^{q,q'} \to[effect=\epsilon] \Theta_{\text{cov}}^{(q,\epsilon)} : q,q' \in \calQ}.
\end{align*}
It is easy to see that $\FpnfTonormalform{\calP}$ is polynomial-time computable from $\calP$.

Now suppose $\mathpzc{Q} = (\calP,\Pi_0 \ChanPar \Gamma_0,\Pi_{\text{cov}} \ChanPar \Gamma_{\text{cov}})$ 
is a simple coverability query. Since $\mathpzc{Q}$ is simple we may  assume that
$\Pi_0 = (q^0_1,\epsilon) \parallel \cdots \parallel (q^0_k,\epsilon)$ and
$\Pi_{\text{cov}} = (q_1,\beta_1) \parallel \cdots (q_n,\beta_n)$ such that $\beta_i \in \calA \union \set{\epsilon}$ since a trivial polynomial-time transform $\calP$ may set up a general $\Pi_0 \ChanPar \Gamma_0$ from a simple query $\mathpzc{Q}$.

Fix a $q^{\dagger} \in \calQ$. We may define the query $\FpnfTonormalform{\mathpzc{Q}} = \commabr{(\FpnfTonormalform{\calP},\FpnfTonormalform{\Pi_0 \ChanPar \Gamma_0},\FpnfTonormalform{\Pi_{\text{cov}} \ChanPar \Gamma_{\text{cov}}})}$
where
\begin{gather*}
 \FpnfTonormalform{\Pi_0 \ChanPar \Gamma_0} = \Theta^{(q^0_1,q^{\dagger})} \parallel \cdots \parallel 
 \Theta^{q^0_k,q^{\dagger}} \ChanPar \Gamma_0 \text{ and  }\\
 \FpnfTonormalform{\Pi_{\text{cov}} \ChanPar \Gamma_{\text{cov}}} = \Theta_{\text{cov}}^{(q_1,\beta_1)} \parallel \cdots \parallel 
 \Theta_{\text{cov}}^{(q_n,\beta_n)} \ChanPar \Gamma.
 \end{gather*} 
 $\FpnfTonormalform{\Pi_0 \ChanPar \Gamma_0} \in G(\Pi_0 \ChanPar \Gamma_0)$ and
 $\FpnfTonormalform{\Pi_{\text{cov}} \ChanPar \Gamma_{\text{cov}}}\discn \in G(\Pi_{\text{cov}} \discn\ChanPar \Gamma_{\text{cov}})$ and both the former are polynomial-time computable.

Now suppose $\mathpzc{Q}$ is a yes-instance, then
$\Pi_0 \ChanPar \Gamma_0 \to[TS=\calP,*] \Pi \ChanPar \Gamma$ such that $\Pi_{\text{cov}} \ChanPar \Gamma_{\text{cov}} 
\discn \leqACPS \Pi \ChanPar \Gamma$ and since $\calB$ is a bisimulation we also know
that
$G(\Pi_0 \ChanPar \Gamma_0) \discn\to[TS={\P[\calP']},*] G(\Pi \ChanPar \Gamma)$.
Further we must have $\Pi = (q_1,\beta'_1) \parallel \cdots (q_n,\beta'_n) \parallel \Pi_1$
where either $\beta_i = \epsilon$ or $\beta_i = A_i$ and $\beta'_i = A_i\cdot \beta''_i$.
By definition for all $\Pi' \ChanPar \Gamma \in G(\Pi \ChanPar \Gamma)$ there exists 
a $\Pi'_0 \ChanPar \Gamma_0 \in G(\Pi_0 \ChanPar \Gamma_0)$ such that
 $\Pi'_0 \ChanPar \Gamma_0 \to[TS=\calP',*] \Pi' \ChanPar \Gamma$.
We can deduce that $\Pi' = A_1^{q_1,q'_1}\bar{\beta}_1 \parallel \cdots A_n^{q_n,q'_n}\bar{\beta}_n \parallel \Pi'_1$
where either $A_i = \beta_i$ or $\beta_i = \epsilon$ and $A_i = \Theta$.
We can then see that $\Pi' \ChanPar \Gamma \to[TS={\FpnfTonormalform{\calP}},*] \Theta_{\text{cov}}^{(q_1,\beta_1)}\bar{\beta}_1 \parallel \cdots \Theta_{\text{cov}}^{(q_n,\beta_n)}\bar{\beta}_n \parallel \Pi'_1$.
Since $\Pi'_0 = \Theta^{(q^0_1,q^1_1)} \parallel \cdots \parallel \Theta^{q^0_k,q^1_k}$ may differ from
$\FpnfTonormalform{\Pi_0 \ChanPar \Gamma_0}$ only by the choice of the $q^1_j$ which cannot play a r\^ole in
any reduction; this allows us to deduce:
$\FpnfTonormalform{\Pi_0 \ChanPar \Gamma_0} \to[TS={\FpnfTonormalform{\calP}},*] \Theta_{\text{cov}}^{(q_1,\beta_1)}\bar{\beta}_1 \parallel \cdots \Theta_{\text{cov}}^{(q_n,\beta_n)}\bar{\beta}_n \parallel \Pi''_1$.
Thus $\FpnfTonormalform{\mathpzc{Q}}$ is a yes-instance. 

Conversely, suppose $\FpnfTonormalform{\mathpzc{Q}}$ is a yes-instance
then $\FpnfTonormalform{\Pi_0 \ChanPar \Gamma_0} \discn \to[TS=\FpnfTonormalform{\calP},*] \Theta_{\text{cov}}^{(q_1,\beta_1)}\bar{\beta}_1 \parallel \cdots \Theta_{\text{cov}}^{(q_n,\beta_n)}\bar{\beta}_n \parallel \Pi'_1$.
Reversing transitions using rules from $\calR_{\text{cov}}$ we can see that
$\FpnfTonormalform{\Pi_0 \ChanPar \Gamma_0} \to[TS=\calP',*] 
A_1^{q_1,q'_1}\bar{\beta}_1 \parallel \cdots \parallel A_n^{q_n,q'_n}\bar{\beta}_n \parallel \Pi''_1$
where either $A_i = \beta_i$ or $\beta_i = \epsilon$ and $A_i = \Theta$.
Since $\FpnfTonormalform{\Pi_0 \ChanPar \Gamma_0} \in G(\Pi_0 \ChanPar \Gamma_0)$ the 
Lemma above implies that 
$A_1^{q_1,q'_1}\bar{\beta}_1 \parallel \cdots \parallel A_n^{q_n,q'_n}\bar{\beta}_n \parallel \Pi''_1 \in
G((q_1,\beta_1\,\beta'_1) \parallel \cdots \parallel (q_n,\beta_n\,\beta'_n) \parallel \Pi_1)$ for some
$\beta'_i$ and $\Pi_1$ and
$\Pi_0 \ChanPar \Gamma_0 \to[TS=\calP,*] (q_1,\beta_1\,\beta'_1) \parallel \cdots \parallel (q_n,\beta_n\,\beta'_n) \parallel \Pi_1$.
Thus $(\calP,\Pi_0 \ChanPar \Gamma_0,\Pi_{\text{cov}} \ChanPar \Gamma_{\text{cov}})$ is a yes-instance.

For boundedness, suppose given $\Pi_0 \ChanPar \Gamma_0$ suppose $\{ \Pi \ChanPar \Gamma : \Pi_0 \ChanPar \Gamma_0 \to[TS=\calP,*] \Pi \ChanPar \Gamma\}$ is a finite set.
Then since $\calB$ is a bisimulation this implies $\{ G(\Pi \ChanPar \Gamma) : G(\Pi_0 \ChanPar \Gamma_0) \to[TS={\P[\calP']},*] G(\Pi \ChanPar \Gamma)\}$
is a finite set which implies that $\{ \Pi \ChanPar \Gamma : \FpnfTonormalform{\Pi_0 \ChanPar \Gamma_0} \to[TS=\calP',*] \Pi \ChanPar \Gamma\}$ is a finite set since $G(\Pi \ChanPar \Gamma)$ is a finite set for all $\Pi \ChanPar \Gamma$. 
Since rules from $\calR_{\text{cov}}$ only add a finite of finite number of configurations we can deduce that
$\{ \Pi \ChanPar \Gamma : \FpnfTonormalform{\Pi_0 \ChanPar \Gamma_0} \to[TS=\FpnfTonormalform{\calP},*] \Pi \ChanPar \Gamma\}$ is a finite set.
Conversely, suppose $\{ \Pi \ChanPar \Gamma : \FpnfTonormalform{\Pi_0 \ChanPar \Gamma_0}\discn \to[TS=\FpnfTonormalform{\calP},*] \Pi \ChanPar \Gamma\}$ is a finite set then since rules from $\calR_{\text{cov}}$ only add a finite of finite number of configurations we can infer that
$\{ \Pi \ChanPar \Gamma : \FpnfTonormalform{\Pi_0 \ChanPar \Gamma_0} \to[TS=\calP',*] \Pi \ChanPar \Gamma\}$ is finite set.
Clearly this implies that $\{ G(\Pi \ChanPar \Gamma) : G(\Pi_0 \ChanPar \Gamma_0) \to[TS={\P[\calP']},*] G(\Pi \ChanPar \Gamma)\}$ is a finite set and thus $\{ \Pi \ChanPar \Gamma : \Pi_0 \ChanPar \Gamma_0 \to[TS=\calP,*] \Pi \ChanPar \Gamma\}$ is a finite set since $\calB$ is a bisimulation.

For termination, since $\calB$ is a bisimulation clearly there is an infinite path from $\Pi_0 \ChanPar \Gamma_0$ in $\calP$ iff
there is an infinite path from $G(\Pi_0 \ChanPar \Gamma_0)$ in $\P[\calP']$. And the latter is clearly possible if, and only if,
there is an infinite path from a $\FpnfTonormalform{\Pi_0 \ChanPar \Gamma_0} \in G(\Pi_0 \ChanPar \Gamma_0)$ which is implied by the Lemma above and by $\calB$ being a bisimulation.
This concludes the proof.
\end{proof}

We can now give a proof of Proposition~\ref{prop:acps:normalform}:
\begin{proof}[Proof of Proposition~\ref{prop:acps:normalform}]
Let $\Fnormalform[']{\calP} = \FpnfTonormalform{\Fprenormalform{\calP}}$,
$\Fnormalform[']{\mathpzc{Q}} = \FpnfTonormalform{\Fprenormalform{\mathpzc{Q}}}$, and
$\Fnormalform[']{\Pi_0 \ChanPar \Gamma_0} = \FpnfTonormalform{\Fprenormalform{\Pi_0 \ChanPar \Gamma_0}}$.
Lemma~\ref{app:lemma:pnf} and Lemma~\ref{app:lemma:pnfTonf} then implies that
$\Fnormalform[']{\mathpzc{Q}}$ is a yes-instance if, and only if, $\mathpzc{Q}$ is a yes-instance;
$\Fnormalform[']{\calP}$ is bounded from $\Fnormalform[']{\Pi_0 \ChanPar \Gamma_0}$ if, and only if,
$\calP$ is bounded from $\Pi_0 \ChanPar \Gamma_0$; and
$\Fnormalform[']{\calP}$ is terminating from $\Fnormalform[']{\Pi_0 \ChanPar \Gamma_0}$ if, and only if,
$\calP$ is terminating from $\Pi_0 \ChanPar \Gamma_0$.
$\Fnormalform[']{\calP}$ is not quite in normal form yet, since it contains rules of the form
$A \to[effect=\lambda] B$ where normal form requires the RHS to be $\epsilon$, clearly we can by introducing a
polynomial number of non-terminals remedy this; replacing such rules by pairs::
$$A \to[effect=\lambda] B \mapsto \set{A \to[effect=\epsilon] C^{\lambda}\,B, C^{\lambda} \to[effect=\lambda] \epsilon}.$$
We call the the resulting ACPS $\Fnormalform{\calP}$, and take the query $\Fnormalform{\mathpzc{Q}} = \Fnormalform[']{\mathpzc{Q}}$, and $\Fnormalform{\Pi_0 \ChanPar \Gamma_0} = \Fnormalform[']{\Pi_0 \ChanPar \Gamma_0}$ and it is trivial to see that the result holds.
\end{proof}

\subsubsection{Proof of Proposition~\ref{prop:acps:apcps:interreduction}}

We fix a $\calP$ in normal form for this section and hence have also a fixed $\GofP$.
The proof of Proposition~\ref{prop:acps:apcps:interreduction} exploits the fact, that we may accelerate
any transitions with commutative side-effects (send, spawn, $\epsilon$) and that we may delay blocking/non-commutative
transitions until the last possible point. This means we can synchronise reductions of $\calP$ and $\GofP$ at configurations where all processes have non-commutative non-terminals as a head symbol. We will relabel the transition 
relations several times in order to chose different synchronisation points.

Let us first extend $\eqvI[{\IofG}]$ to parallel compositions. We write
$\pi_1 \parallel \cdots \parallel \pi_n \eqvI[{\IofG}] \bar{\pi}_1 \parallel \cdots \parallel \bar{\pi}_n$
just if $\pi_i \eqvI[{\IofG}] \bar{\pi}_i$ for all $i \in \range{n}$.

For this section let us write $\to[TS=\GofP]$ for $\toCF$ to make clear that the transition relation is induced by 
$\GofP$. We temporarily label $\to[TS=\calP]$ and $\to[TS=\GofP]$ by $\calP$-configurations as follows:
If $\Pi \ChanPar \Gamma \to[TS=\calP] \Pi' \ChanPar \Gamma'$ then we label the transition
$\Pi \ChanPar \Gamma \to[label=\Pi'_0 \ChanPar \Gamma',LTS=\calP] \Pi' \ChanPar \Gamma'$ for $\Pi'_0 \eqvI[\IofG] \Pi'$ .
If $\Pi \ChanPar \Gamma \to[TS=\GofP] \Pi' \ChanPar \Gamma'$ and $\Pi' \eqvI[\IofG] \bar{\Pi'}$, $\bar{\Pi'} \in \M[\NonT^*]$ then we label
the transition $\Pi \ChanPar \Gamma \to[label=\bar{\Pi'} \ChanPar \Gamma',LTS=\GofP] \Pi' \ChanPar \Gamma'$; 
otherwise we label it by $\epsilon$, i.e.~$\Pi \ChanPar \Gamma \to[label=\epsilon,LTS=\GofP] \Pi' \ChanPar \Gamma'$.

\begin{lemma}\label{app:apcps:simulates:acps}
The relation $\mathpzc{R} = \varset{(\Pi \ChanPar \Gamma,\bar{\Pi} \ChanPar \Gamma) : \bar{\Pi} \eqvI[{\IofG}] \Pi,
\bar{\Pi} \in \M[\NonT^*]}$
is a weak simulation relation for $\calP$ and $\GofP$.
\end{lemma}
\begin{proof}
Suppose $(\pi \parallel \Pi \ChanPar \Gamma,\pi_0 \parallel \Pi_0 \ChanPar \Gamma) \in \mathpzc{R}$
and $\pi \parallel \Pi \ChanPar \Gamma \to[TS=\calP] \pi' \parallel \Pi' \ChanPar \Gamma'$.
We may assume that $\pi = A\,\beta$ and a rule $A \to[effect=\lambda] \beta'$ is used in the transition.
Since $\pi_0 \parallel \Pi_0 \eqvI[{\IofG}] A\,\beta \parallel \Pi$ we can w.l.o.g. assume that
$\pi_0 = A\, \bar{\beta}$ and $\bar{\beta} \eqvI[{\IofG}] \beta$ (by using transitivity of $\eqvI[{\IofG}]$ otherwise).
Let us perform a case analysis on $\lambda$:
\begin{itemize}
\item \emph{Case: } $\lambda=\epsilon$ \newline
Since $\calP$ is in normal form we know that
$\beta' \in \varset{\epsilon, B\,C : B,C \in \calA}$,
$\pi' = \beta'\,\beta$,
$\Pi' = \Pi$, $\Gamma' = \Gamma$ and $A \to \beta' \in \GofP$.
Then clearly $A\, \bar{\beta} \parallel \Pi_0 \ChanPar \Gamma \to[label=l,LTS=\GofP] \beta' \bar{\beta} \parallel \Pi_0 \ChanPar \Gamma$
and $\beta'\, \bar{\beta} \parallel \Pi_0 \eqvI[{\IofG}] \beta'\,\beta \parallel \Pi = \pi' \parallel \Pi'$.
Further clearly $\beta'\, \bar{\beta} \parallel \Pi_0 \in \M[\NonT^*]$ and thus we may take $l = \pi' \parallel \Pi'$.
\item \emph{Case: } $\lambda=\snd{c}{m}$ \newline
Since $\calP$ is in normal form we know that 
$\beta'=\epsilon$,
$\pi' = \beta$,
$\Pi' = \Pi$, $\Gamma' = \Gamma \oplus \varupdate{}{c}{\mset{m}}$ and $A \to \snd{c}{m} \in \GofP$.
Then clearly 
$A\, \bar{\beta} \parallel \Pi_0 \ChanPar \Gamma \to[label=\epsilon,LTS=\GofP] \snd{c}{m} \bar{\beta} \parallel \Pi_0 \ChanPar \Gamma
\to[label=l,LTS=\GofP] \bar{\beta} \parallel \Pi_0 \ChanPar \Gamma'$
and $\bar{\beta} \parallel \Pi_0 \eqvI[{\IofG}] \beta \parallel \Pi = \pi' \parallel \Pi'$.
Further clearly $\bar{\beta} \parallel \Pi_0 \in \M[\NonT^*]$ and thus we may take $l = \pi' \parallel \Pi'$.
\item \emph{Case: } $\lambda=\rec{c}{m}$ \newline
Since $\calP$ is in normal form we know that 
$\beta'=\epsilon$,
$\pi' = \beta$,
$\Pi' = \Pi$, $\Gamma = \Gamma' \oplus \varupdate{}{c}{\mset{m}}$ and $A \to \rec{c}{m} \in \GofP$.
Then clearly 
$A\, \bar{\beta} \parallel \Pi_0 \ChanPar \Gamma \to[label=\epsilon,LTS=\GofP] \rec{c}{m} \bar{\beta} \parallel \Pi_0 \ChanPar \Gamma
\to[label=l,LTS=\GofP] \bar{\beta} \parallel \Pi_0 \ChanPar \Gamma'$
and $\bar{\beta} \parallel \Pi_0 \eqvI[{\IofG}] \beta \parallel \Pi = \pi' \parallel \Pi'$.
Further clearly $\bar{\beta} \parallel \Pi_0 \in \M[\NonT^*]$ and thus we may take $l = \pi' \parallel \Pi'$.
\item \emph{Case: } $\lambda=\spn{A'}$ \newline
Since $\calP$ is in normal form we know that 
$\beta'=\epsilon$,
$\pi' = \beta$,
$\Pi' = A' \parallel \Pi$, $\Gamma = \Gamma'$ and $A \to \spn{A'} \in \GofP$.
Then clearly 
$A\, \bar{\beta} \parallel \Pi_0 \ChanPar \Gamma \to[label=\epsilon,LTS=\GofP] \spn{A'} \bar{\beta} \parallel \Pi_0 \ChanPar \Gamma
\to[label=l,LTS=\GofP] \bar{\beta} \parallel A' \parallel \Pi_0 \ChanPar \Gamma$.
Further $\bar{\beta} \parallel A' \parallel \Pi_0 \eqvI[{\IofG}] \beta \parallel A' \parallel \Pi = \pi' \parallel \Pi'$, and
Further clearly $\bar{\beta} \parallel A' \parallel \Pi_0 \in \M[\NonT^*]$ and thus we may take $l = \pi' \parallel \Pi'$.
\end{itemize}
Hence $\mathpzc{R}$ is a weak simulation which is what we wanted to prove.
\end{proof}

We temporarily relabel $\to[TS=\calP]$ and $\to[TS=\GofP]$ by side effects in $\Lambda$ as follows:
If $\Pi \ChanPar \Gamma \to[TS=\calP] \Pi' \ChanPar \Gamma'$ using rule $\beta \to[effect=\lambda] \beta'$ then we label the transition
$\Pi \ChanPar \Gamma \to[label=\lambda,LTS=\calP] \Pi' \ChanPar \Gamma'$.
If $\lambda\, \beta \parallel \Pi \ChanPar \Gamma \to[TS=\GofP] \beta \parallel \Pi' \ChanPar \Gamma'$ using rules \Cref{stdrule:send}, \Cref{stdrule:spawn} and $\lambda \in \SigmaofP$ then we label
the transition $\lambda\, \beta \parallel \Pi \ChanPar \Gamma \to[label=\lambda,TS=\GofP] \beta \parallel \Pi' \ChanPar \Gamma'$; 
otherwise we label it by $\epsilon$, i.e.~$\Pi \ChanPar \Gamma \to[label=\epsilon,LTS=\GofP] \Pi' \ChanPar \Gamma'$. Further we label the transitive closures as usual:
If $\Pi_0 \ChanPar \Gamma_0 \to[label=l_1,LTS=\GofP] \Pi_1 \ChanPar \Gamma_1 \to[label=l_2,LTS=\GofP] \cdots \to[label=l_n,LTS=\GofP] \Pi_n \ChanPar \Gamma_n$ then we label $\Pi_0 \ChanPar \Gamma_0 \to[label=l_1\cdots l_n,LTS=\GofP,*] \Pi_n \ChanPar \Gamma_n$ and similarly for $\to[label=\vec{l},LTS=\calP,*]$.

We capture the language of side effects of a $\beta \in \ComN$ in $\calP$ and $\GofP$ as follows:
$\calL_{\GofP}(\beta) = \varset{\vec{\lambda} \mid \beta \ChanPar \emptyset \to[label=\vec{\lambda},LTS=\GofP,*] \epsilon \parallel \Pi(\vec{\lambda}) \ChanPar \Gamma(\vec{\lambda})}$ and
$\calL_{\calP}(\beta) = \varset{\vec{\lambda} \mid \beta \ChanPar \emptyset \to[label=\vec{\lambda},LTS=\calP,*] \epsilon \parallel \Pi(\vec{\lambda}) \ChanPar \Gamma(\vec{\lambda})}$ where 
a $\vec{\lambda}$ determines a ``delta'' to the configuration in terms of sent messages and spawned processes:
\begin{align*}
\Pi(\vec{\lambda})    &\is \set{\Parallel_{A \in \NonT} \paren{\parallel_1^{\M(\vec{\lambda})(\spn{A})} A}},\\
\Gamma(\vec{\lambda}) &\is \set{\Oplus_{c \in \Chan}\Oplus_{m \in \Msg[]} \update{}{c}{\mset{m^{\M(\vec{\lambda})(\snd{c}{m})}}}}.
\end{align*}
Further we note that $\calL_{\GofP}(\beta)$ is effectively the \emph{Parikh} language defined by the CFG $\GofP$ of $\beta$ in the standard derivation sense.
We can now make precise how we can accelerate the execution of commutative non-terminals:
\begin{lemma}\label{app:apcps:acps:com:termination}
Suppose $\beta,\bar{\beta} \in \ComN^*$ then:
\begin{asparaenum}[(i)]
\item $\calL_{\calP}(\beta) \neq \emptyset$ and for all $\vec{\lambda} \in \calL_{\calP}(\beta)$ we have the transitions:
$\beta\,\beta' \parallel \Pi \ChanPar \Gamma \to[TS=\calP,*] \beta' \parallel \Pi \parallel \Pi(\vec{\lambda})  \ChanPar \Gamma \oplus \Gamma(\vec{\lambda})$; and
\item $\calL_{\GofP}(\bar{\beta}) \neq \emptyset$ and for all $\vec{\lambda} \in \calL_{\GofP}(\bar{\beta})$ we have the transitions:
$\bar{\beta}\,\bar{\beta}' \parallel \bar{\Pi} \ChanPar \bar{\Gamma} \to[TS=\GofP,*] \bar{\beta}' \parallel \bar{\Pi} \parallel \Pi(\vec{\lambda}) \ChanPar \Gamma(\vec{\lambda})$.
\end{asparaenum}
And if $\beta \eqvI[{\IofG}] \bar{\beta}$ then we have the set equality:
$$\set{(\Pi(\vec{\lambda}),\Gamma(\vec{\lambda})) \mid \vec{\lambda} \in \calL_{\calP}(\beta)} = 
\set{(\Pi(\vec{\lambda}),\Gamma(\vec{\lambda})) \mid \vec{\lambda} \in \calL_{\GofP}(\bar{\beta})}.$$
\end{lemma}
\begin{proof}
Let us first prove that $\calL_{\calP}(\beta) \neq \emptyset$.
First since $\beta \in \ComN^*$ we may see them as non-terminals of $\GofP$ and rewrite them.
Suppose $\beta = A_1\cdots A_n$ then we know that $A_i$ is commutative and thus $\calL(A_i) \neq \emptyset$ for all 
$i \in \range{n}$. So let us pick $\vec{\lambda}_i \in \calL(A_i)$ let us show in a brief induction that if
$\vec{\lambda} \in \calL(A)$ then $\vec{\lambda} \in \calL_{\calP}(A)$. Let us suppose $A \to^n \vec{\lambda}$.
If $n = 1$ then we use a rule $A \to \vec{\lambda}$ in $\GofP$; hence there exists a rule
$A \to[effect=\vec{\lambda}] \epsilon$ in $\calP$ and thus
$A \ChanPar \emptyset \to[label=\vec{\lambda},LTS=\GofP,*] \epsilon \parallel \Pi(\vec{\lambda}) \ChanPar \Gamma(\vec{\lambda})$.
If $n=k+1$ and the claim holds for all $A'$ and $k' \leq k$ then we can see that
$A \to \beta \to^k \vec{\lambda}$. Now since $\GofP$ is in Chomsky normal form the first step must be a rule of the form
$A \to B\,C$ and $\beta = B\,C$. This implies $\vec{\lambda} = \vec{\lambda}' \vec{\lambda}''$,
$B \to^{k'} \vec{\lambda}'$, $C \to^{k''} \vec{\lambda}''$, and $k = k' + k''$, which yields using the IH that
$\vec{\lambda}' \in \calL_{\calP}(B)$ and $\vec{\lambda}'' \in \calL_{\calP}(C)$. Further we know there is a rule
$A \to[effect=\epsilon] B\, C$ in $\calP$  and thus
$A \ChanPar \emptyset \to[label=\epsilon,LTS=\GofP,*] B\, C \ChanPar \emptyset
\to[label=\vec{\lambda}',LTS=\GofP,*] C \parallel \Pi(\vec{\lambda}') \ChanPar \Gamma(\vec{\lambda}')
\to[label=\vec{\lambda}'',LTS=\GofP,*] \epsilon \parallel \Pi(\vec{\lambda}'\vec{\lambda}') \ChanPar \Gamma(\vec{\lambda}'')$
which concludes the induction.
Hence we can infer that in fact $\vec{\lambda}_i \in \calL(A_i)$
and so we can see that
$A_1\cdots A_n \ChanPar \emptyset \to[label=\vec{\lambda}_1,LTS=\GofP,*] A_2\cdots A_n \parallel \Pi(\vec{\lambda}_1) \ChanPar \Gamma(\vec{\lambda}_1)
\to[label=\vec{\lambda}_2,LTS=\GofP,*] \cdots
\to[label=\vec{\lambda}_n,LTS=\GofP,*] \cdots
\epsilon \parallel \Pi(\vec{\lambda}_1\cdots\vec{\lambda}_n) \ChanPar \Gamma(\vec{\lambda}_1\cdots\vec{\lambda}_n)$
and thus $\vec{\lambda}_1\cdots\vec{\lambda}_n \in \calL_{\calP}(\beta)$
which is what we wanted to prove.
For the second claim of (i) suppose $\vec{\lambda} \in \calL_{\calP}(\beta)$ then by definition
$\beta \ChanPar \emptyset \to[label=\vec{\lambda},LTS=\GofP,*] \epsilon \parallel \Pi(\vec{\lambda}) \ChanPar \Gamma(\vec{\lambda})$. It is then trivial to see that
$\beta\,\beta' \parallel \Pi \ChanPar \Gamma \to[TS=\calP,*] \beta' \parallel \Pi \parallel \Pi(\vec{\lambda})  \ChanPar \Gamma \oplus \Gamma(\vec{\lambda})$.

For (ii), let us first prove that $\calL_{\GofP}(\bar{\beta}) \neq \emptyset$. First we note again that
$\bar{\beta} \in \ComN^*$ and thus $\calL(\bar{\beta}) \neq \emptyset$. 
Hence take $\vec{\lambda} \in \calL(\bar{\beta})$.
By continuously commuting non-terminals to the leftmost-position we can see that
$\bar{\beta} \ChanPar \emptyset \to[label=\epsilon,LTS=\GofP,*] \vec{\lambda} \ChanPar \emptyset$
and then clearly
$\vec{\lambda} \ChanPar \emptyset \to[label=\vec{\lambda},LTS=\GofP,*] 
\epsilon \parallel \Pi(\vec{\lambda}) \ChanPar \Gamma(\vec{\lambda})$ 
thus $\vec{\lambda} \in \calL_{\GofP}(\bar{\beta})$.
For the second claim of (ii) let $\vec{\lambda} \in \calL_{\GofP}(\bar{\beta})$ then by definition
$\bar{\beta} \ChanPar \emptyset \to[label=\vec{\lambda},LTS=\GofP,*] 
\epsilon \parallel \Pi(\vec{\lambda}) \ChanPar \Gamma(\vec{\lambda})$;
it is hence trivial to see that 
$\bar{\beta}\,\bar{\beta}' \parallel \bar{\Pi} \ChanPar \bar{\Gamma} \to[TS=\calP,*] \bar{\beta}' \parallel \bar{\Pi} \parallel \Pi(\vec{\lambda}) \ChanPar \bar{\Gamma} \oplus \Gamma(\vec{\lambda})$.

Since $\GofP$, restricted to non-terminals of $\ComN$, is effectively a completely commutative context-free grammar derived from $\calP$ we can see that
$\varset{\M(\lambda) \mid \lambda \in \calL_{\calP}(\beta)} = \varset{\M(\bar{\lambda}) \mid \bar{\lambda} \eqvI[\IofG] \lambda \in \calL_{\calP}(\beta)} = \varset{\M(\bar{\lambda}) \mid \bar{\lambda} \in \calL_{\GofP}(\bar{\beta})}$
which implies immediately that
$$\set{(\Pi(\vec{\lambda}),\Gamma(\vec{\lambda})) \mid \vec{\lambda} \in \calL_{\calP}(\beta)} = 
\set{(\Pi(\vec{\lambda}),\Gamma(\vec{\lambda})) \mid \vec{\lambda} \in \calL_{\GofP}(\bar{\beta})}.$$
\end{proof}

As a final step in our argument let us relabel the transition relation again.
We temporarily relabel $\to[TS=\calP]$ and $\to[TS=\GofP]$ by $\calP$-configurations as follows:
If $\Pi \ChanPar \Gamma \to[TS=\calP] \Pi' \ChanPar \Gamma'$ then we label the transition
$\Pi \ChanPar \Gamma \to[label=\Pi'_0 \ChanPar \Gamma',LTS=\calP] \Pi' \ChanPar \Gamma'$, if
$\Pi' \in \M[\NComN\cdot\calA^*]$ and $\Pi'_0 \eqvI[\IofG] \Pi'$; otherwise we label it with $\epsilon$:
$\Pi \ChanPar \Gamma \to[label=\epsilon,LTS=\calP] \Pi' \ChanPar \Gamma'$.
If $\Pi \ChanPar \Gamma \to[TS=\GofP] \Pi' \ChanPar \Gamma'$ 
and $\Pi' \eqvI[\IofG] \bar{\Pi'}$, $\bar{\Pi'} \in \M[\NComN \cdot \NonT^*]$ then we label
the transition $\Pi \ChanPar \Gamma \to[label=\bar{\Pi'} \ChanPar \Gamma',LTS=\GofP] \Pi' \ChanPar \Gamma'$; 
otherwise we label it by $\epsilon$, i.e.~$\Pi \ChanPar \Gamma \to[label=\Pi' \ChanPar \Gamma',LTS=\GofP] \Pi' \ChanPar \Gamma'$.

\begin{lemma}\label{app:lemma:acps:apcps:wbisimulation}
The relation $\mathpzc{R}' = \varset{(\Pi \ChanPar \Gamma,\bar{\Pi} \ChanPar \Gamma) : \bar{\Pi} \eqvI[{\IofG}] \Pi,
\bar{\Pi} \in \M[\NComN\cdot\NonT^*]}$
is a weak bisimulation relation for $\calP$ and $\GofP$.
\end{lemma}
\begin{proof}
Let us first prove that $\mathpzc{R}'$ is a weak simulation.
Suppose $(\Pi \ChanPar \Gamma, \Pi_0 \ChanPar \Gamma) \in \mathpzc{R}'$
and $\Pi \ChanPar \Gamma \to[label=\Pi' \ChanPar \Gamma',LTS=\calP,*] \Pi' \ChanPar \Gamma'$ such that $\Pi' \in \M[\NComN\cdot\calA^*]$.
Lemma~\ref{app:apcps:simulates:acps} then tells us that
$\Pi \ChanPar \Gamma \to[label=\Pi' \ChanPar \Gamma',LTS=\calP,*] \Pi'_0 \ChanPar \Gamma'$
and $\Pi'_0 \eqvI[\IofG] \Pi'$. The latter implies that $\Pi'_0 \in \M[\NComN\cdot\NonT^*]$ since $\Pi' \in \M[\NComN\cdot\calA^*]$, i.e. each process is headed by a non-commutative non-terminal which obviously are invariant under commutation.
Hence we can deduce that $\mathpzc{R}'$ is a weak simulation.

Let us first prove that $\mathpzc{R}'^{-1}$ is a weak simulation.
Suppose $(\Pi \ChanPar \Gamma, \Pi_0 \ChanPar \Gamma) \in \mathpzc{R}'$
and $\Pi_0 \ChanPar \Gamma \to[label=\Pi' \ChanPar \Gamma',LTS=\GofP,*] \Pi'_0 \ChanPar \Gamma'$ where
$\Pi'_0 \eqvI[\IofG] \Pi'$ for some $\Pi' \in \M[\NComN\cdot\calA^*]$.
W.l.o.g. we may assume this path may not be split to expose two labels $l,l'$ (otherwise we may consider a maximal splitting of the path to obtain several labels and apply the below to each in turn one-by-one).
We can infer that 
$\Pi_0 = A\,\beta_0 \parallel \Pi_1$ and the transition may be split
$A\,\beta_0 \parallel \Pi_1 \ChanPar \Gamma \to[label=l,LTS=\GofP] \beta'\,\beta_0 \parallel \Pi'_1 \ChanPar \Gamma_1 \to[label=l',LTS=\GofP,*] \Pi'_0 \ChanPar \Gamma'$, $l\cdot l' = \Pi' \ChanPar \Gamma'$ and using a rule
$A \to \beta' \in \GofP$.
Further we may infer that 
$\Pi = A\,\beta \parallel \Pi_2$.
Let us make a case analysis on $\beta'$.
\begin{itemize}
\item \emph{Case: } $\beta' = B\, C$, $B$ non-commutative. \newline
Firstly, we notice $\Pi'_1 = \Pi_1, \Gamma = \Gamma_1$ and
$B\,C\,\beta_0 \parallel \Pi_1 \in \M[\NComN\cdot\NonT^*]$ which implies
that $l = \Pi' \ChanPar \Gamma'$ and also $\Pi'_0 \ChanPar \Gamma' = \beta'\,\beta_0 \parallel \Pi_1$.
Further there is a rule $A \to[effect=\epsilon] B\,C$ in $\calP$ and thus
$A\,\beta \parallel \Pi_2 \ChanPar \Gamma \to[label=l,LTS=\calP] B\, C\,\beta \parallel \Pi_2 \ChanPar \Gamma$
and $B\, C\,\beta \parallel \Pi_2 \eqvI[\IofG] B\,C\,\beta_0 \parallel \Pi_1$
which is what we wanted to prove.
\item \emph{Case: } $\beta' = B\, C$, $B$ commutative. \newline
Firstly, we notice $\Pi'_1 = \Pi_1$ and $\Gamma = \Gamma_1$.
Since $A$ was non-commutative it must be the case that $C$ is non-commutative.
Further since $\Pi_0 \ChanPar \Gamma \to[label=\Pi' \ChanPar \Gamma',LTS=\GofP,*] \Pi'_0 \ChanPar \Gamma'$ may not be further split we may assume that for some $\vec{\lambda} \in \calL_{\GofP}(B)$ we may rewrite the transition as:
$A\,\beta_0 \parallel \Pi_1 \ChanPar \Gamma \to[label=\epsilon,LTS=\GofP] B\,C\,\beta_0 \parallel \Pi_1 \ChanPar \Gamma \to[label=l',LTS=\GofP,*] C\,\beta_0 \parallel \Pi_1 \parallel \Pi(\vec{\lambda}) \ChanPar \Gamma \oplus \Gamma(\vec{\lambda})$, i.e.~$\Pi'_0 \ChanPar \Gamma' = C\,\beta_0 \parallel \Pi_1 \parallel \Pi(\vec{\lambda}) \ChanPar \Gamma \oplus \Gamma(\vec{\lambda})$.
Now we know that 
there is a rule $A \to[effect=\epsilon] B\,C$ in $\calP$ and thus
$A\,\beta \parallel \Pi_2 \ChanPar \Gamma \to[label=\epsilon,LTS=\calP] B\, C\,\beta \parallel \Pi_2 \ChanPar \Gamma$
and Lemma~\ref{app:apcps:acps:com:termination} then gives us that
$B\,C\,\beta \parallel \Pi_2 \ChanPar \Gamma \to[label=l',LTS=\calP,*] C\,\beta \parallel \Pi_2 \parallel 
\Pi(\bar{\vec{\lambda}}) \ChanPar \Gamma \oplus \Gamma(\bar{\vec{\lambda}})$
for some $\bar{\vec{\lambda}} \in \calL_{\calP}(B)$ such that $\Pi(\bar{\vec{\lambda}}) = \Pi(\vec{\lambda})$,
$\Gamma(\vec{\lambda}) = \Gamma(\bar{\vec{\lambda}})$.
Further $C\,\beta \parallel \Pi_2 \parallel \Pi(\vec{\lambda}) \eqvI[\IofG] C\,\beta_0 \parallel \Pi_1 \parallel \Pi(\vec{\lambda})$ which is what we wanted to prove.
\item \emph{Case: } $\beta' = \epsilon$. \newline
Firstly, we notice $\Pi'_1 = \Pi_1$ and $\Gamma = \Gamma_1$ and we may write 
$\beta_0 = \beta'_0\beta''_0$ such that $\beta'_0 \in \ComN^*$ and $\beta''_0 \in (\NComN\ComN^*)^*$.
Further since $\Pi_0 \ChanPar \Gamma \to[label=\Pi' \ChanPar \Gamma',LTS=\GofP,*] \Pi'_0 \ChanPar \Gamma'$ may not be further split we may assume that for some $\vec{\lambda} \in \calL_{\GofP}(\beta'_0)$ we may rewrite the transition as:
$A\,\beta_0 \parallel \Pi_1 \ChanPar \Gamma \to[label=\epsilon,LTS=\GofP] \beta_0 \parallel \Pi_1 \ChanPar \Gamma \to[label=l',LTS=\GofP,*] \beta''_0 \parallel \Pi_1 \parallel \Pi(\vec{\lambda}) \ChanPar \Gamma \oplus \Gamma(\vec{\lambda})$, i.e.~$\Pi'_0 \ChanPar \Gamma' = \beta''_0 \parallel \Pi_1 \parallel \Pi(\vec{\lambda}) \ChanPar \Gamma \oplus \Gamma(\vec{\lambda})$.
Now we know that 
there is a rule $A \to[effect=\epsilon] \epsilon$ in $\calP$ and thus
$A\,\beta \parallel \Pi_2 \ChanPar \Gamma \to[label=\epsilon,LTS=\calP] \beta \parallel \Pi_2 \ChanPar \Gamma$
and since $\beta \eqvI[\IofG] \beta_0$ we may rewrite $\beta = \beta'\beta''$ such that
$\beta' \eqvI[\IofG] \beta'_0$ and $\beta'' \eqvI[\IofG] \beta''_0$, hence clearly $\beta' \in \ComN^*$ and $\beta'' \in (\NComN\ComN^*)^*$,
and Lemma~\ref{app:apcps:acps:com:termination} then gives us that
$\beta \parallel \Pi_2 \ChanPar \Gamma \to[label=l',LTS=\calP,*] \beta'' \parallel \Pi_2 \parallel 
\Pi(\bar{\vec{\lambda}}) \ChanPar \Gamma \oplus \Gamma(\bar{\vec{\lambda}})$
for some $\bar{\vec{\lambda}} \in \calL_{\GofP}(\beta')$ such that $\Pi(\bar{\vec{\lambda}}) = \Pi(\vec{\lambda})$,
$\Gamma(\vec{\lambda}) = \Gamma(\bar{\vec{\lambda}})$.
Further $\beta'' \parallel \Pi_2 \parallel \Pi(\vec{\lambda}) \eqvI[\IofG] \beta''_0 \parallel \Pi_1 \parallel \Pi(\vec{\lambda})$ which is what we wanted to prove.
\item \emph{Case: } $\beta' = \snd{c}{m}$. \newline
Firstly, we notice $\Pi'_1 = \Pi_1$ and $\Gamma = \Gamma_1$ and we may write 
$\beta_0 = \beta'_0\beta''_0$ such that $\beta'_0 \in \ComN^*$ and $\beta''_0 \in (\NComN\ComN^*)^*$.
Further since $\Pi_0 \ChanPar \Gamma \to[label=\Pi' \ChanPar \Gamma',LTS=\GofP,*] \Pi'_0 \ChanPar \Gamma'$ may not be further split we may assume that for some $\vec{\lambda} \in \calL_{\GofP}(\beta'_0)$ we may rewrite the transition as:
$A\,\beta_0 \parallel \Pi_1 \ChanPar \Gamma \to[label=\epsilon,LTS=\GofP] \snd{c}{m}\beta_0 \parallel \Pi_1 \ChanPar \Gamma \to[label=l',LTS=\GofP,*] \beta''_0 \parallel \Pi_1 \parallel \Pi(\snd{c}{m}\cdot\vec{\lambda}) \ChanPar \Gamma \oplus \Gamma(\snd{c}{m}\cdot\vec{\lambda})$, i.e.~$\Pi'_0 \ChanPar \Gamma' = \beta''_0 \parallel \Pi_1 \parallel \Pi(\snd{c}{m}\cdot\vec{\lambda}) \ChanPar \Gamma \oplus \Gamma(\snd{c}{m}\cdot\vec{\lambda})$.
Now we know that 
there is a rule $A \to[effect=\snd{c}{m}] \epsilon$ in $\calP$ and thus
$A\,\beta \parallel \Pi_2 \ChanPar \Gamma 
	\to[label=\epsilon,LTS=\calP] \beta \parallel \Pi_2 \parallel \Pi(\snd{c}{m}) \ChanPar \Gamma\oplus \Gamma(\snd{c}{m})$
and since $\beta \eqvI[\IofG] \beta_0$ we may use analogous reasoning to the $\beta'=\epsilon$ case to obtain that
$\beta \parallel \Pi_2 \parallel \Pi(\snd{c}{m}) \ChanPar \Gamma\oplus \Gamma(\snd{c}{m}) \to[label=l',LTS=\calP,*] \beta'' \parallel \Pi_2 \parallel 
\Pi(\snd{c}{m}\cdot\vec{\lambda}) \ChanPar \Gamma \oplus \Gamma(\snd{c}{m}\cdot\vec{\lambda})$
such that $\beta'' \parallel \Pi_2 \parallel \Pi(\snd{c}{m}\cdot\vec{\lambda}) \eqvI[\IofG] \beta''_0 \parallel \Pi_1 \parallel \Pi(\snd{c}{m}\cdot\vec{\lambda})$ which is what we wanted to prove.
\item \emph{Case: } $\beta' = \spn{A'}$. \newline
Firstly, we notice $\Pi'_1 = \Pi_1$ and $\Gamma = \Gamma_1$ and we may write 
$\beta_0 = \beta'_0\beta''_0$ such that $\beta'_0 \in \ComN^*$ and $\beta''_0 \in (\NComN\ComN^*)^*$.
Further since $\Pi_0 \ChanPar \Gamma \to[label=\Pi' \ChanPar \Gamma',LTS=\GofP,*] \Pi'_0 \ChanPar \Gamma'$ may not be further split we may assume that for some $\vec{\lambda} \in \calL_{\GofP}(\beta'_0)$ we may rewrite the transition as:
$A\,\beta_0 \parallel \Pi_1 \ChanPar \Gamma \to[label=\epsilon,LTS=\GofP] \spn{A'}\beta_0 \parallel \Pi_1 \ChanPar \Gamma \to[label=l',LTS=\GofP,*] \beta''_0 \parallel \Pi_1 \parallel \Pi(\spn{A'}\cdot\vec{\lambda}) \ChanPar \Gamma \oplus \Gamma(\spn{A'}\cdot\vec{\lambda})$, i.e.~$\Pi'_0 \ChanPar \Gamma' = \beta''_0 \parallel \Pi_1 \parallel \Pi(\spn{A'}\cdot\vec{\lambda}) \ChanPar \Gamma \oplus \Gamma(\spn{A'}\cdot\vec{\lambda})$.
Now we know that 
there is a rule $A \to[effect=\spn{A'}] \epsilon$ in $\calP$ and thus
$A\,\beta \parallel \Pi_2 \ChanPar \Gamma 
	\to[label=\epsilon,LTS=\calP] \beta \parallel \Pi_2 \parallel \Pi(\spn{A'}) \ChanPar \Gamma \oplus \Gamma(\spn{A'})$
and since $\beta \eqvI[\IofG] \beta_0$ we may use analogous reasoning to the $\beta'=\epsilon$ case to obtain that
$\beta \parallel \Pi_2 \parallel \Pi(\spn{A'}) \ChanPar \Gamma \oplus \Gamma(\spn{A'}) \to[label=l',LTS=\calP,*] \beta'' \parallel \Pi_2 \parallel 
\Pi(\spn{A'}\cdot\vec{\lambda}) \ChanPar \Gamma \oplus \Gamma(\spn{A'}\cdot\vec{\lambda})$
such that $\beta'' \parallel \Pi_2 \parallel \Pi(\snd{c}{m}\cdot\vec{\lambda}) \eqvI[\IofG] \beta''_0 \parallel \Pi_1 \parallel \Pi(\snd{c}{m}\cdot\vec{\lambda})$ which is what we wanted to prove.
\item \emph{Case: } $\beta' = \rec{c}{m}$. \newline
Firstly, we notice $\Pi'_1 = \Pi_1$ and $\Gamma = \Gamma_1$ and we may write 
$\beta_0 = \beta'_0\beta''_0$ such that $\beta'_0 \in \ComN^*$ and $\beta''_0 \in (\NComN\ComN^*)^*$.
Further since $\Pi_0 \ChanPar \Gamma \to[label=\Pi' \ChanPar \Gamma',LTS=\GofP,*] \Pi'_0 \ChanPar \Gamma'$ may not be further split we may assume that for some $\vec{\lambda} \in \calL_{\GofP}(\beta'_0)$ we may rewrite the transition as:
$A\,\beta_0 \parallel \Pi_1 \ChanPar \Gamma 
	\to[label=\epsilon,LTS=\GofP] \rec{c}{m}\,\beta_0 \parallel \Pi_1 \ChanPar \Gamma 
	\to[label=\epsilon,LTS=\GofP] \beta_0 \parallel \Pi_1 \ChanPar \Gamma \ominus \varupdate{}{c}{m}
	\to[label=l',LTS=\GofP,*] \beta''_0 \parallel \Pi_1 \parallel \Pi(\vec{\lambda}) \ChanPar \Gamma \oplus \Gamma(\vec{\lambda})$, i.e.~$\Pi'_0 \ChanPar \Gamma' = \beta''_0 \parallel \Pi_1 \parallel \Pi(\vec{\lambda}) \ChanPar \Gamma \ominus \varupdate{}{c}{m} \oplus \Gamma(\vec{\lambda})$.
Now we know that 
there is a rule $A \to[effect=\rec{c}{m}] \epsilon$ in $\calP$ and thus
$A\,\beta \parallel \Pi_2 \ChanPar \Gamma 
	\to[label=\epsilon,LTS=\calP] \beta \parallel \Pi_2 \parallel \ChanPar \Gamma\ominus \varupdate{}{c}{m}$
and since $\beta \eqvI[\IofG] \beta_0$ we may use analogous reasoning to the $\beta'=\epsilon$ case to obtain that
$\beta \parallel \Pi_2 \ChanPar \Gamma \ominus \varupdate{}{c}{m}\to[label=l',LTS=\calP,*] \beta'' \parallel \Pi_2 \parallel 
\Pi(\vec{\lambda}) \ChanPar \Gamma \ominus \varupdate{}{c}{m} \oplus \Gamma(\vec{\lambda})$
such that $\beta'' \parallel \Pi_2 \parallel \Pi(\vec{\lambda}) \eqvI[\IofG] \beta''_0 \parallel \Pi_1 \parallel \Pi(\vec{\lambda})$ which is what we wanted to prove.
\end{itemize}
This concludes the proof.

\end{proof}


\begin{customproposition}[\ref{prop:acps:apcps:interreduction}]
Suppose $\calP$ is an ACPS in normal form and $\mathpzc{Q} = (\calP, \Pi_0\ChanPar\Gamma_0, \Pi\ChanPar\Gamma)$ a simple coverability query. Then, $\mathpzc{Q}$ is a yes-instance, if and only if, $\mathpzc{Q}'= (\GofP, \Pi_0\ChanPar\Gamma_0, \Pi\ChanPar\Gamma)$ is a yes-instance.
Hence simple coverability, boundedness and termination for ACPS and APCPS poly\-no\-mi\-al-time inter-reduce. A simple APCPS query $(\calG,\Pi_0 \ChanPar \Gamma_0,\Pi \ChanPar \Gamma)$ satisfies: $\pi \eqvI[{\IofG[\calG]}] A \in \NonT$ for all $\pi$ in $\Pi$ and $\Pi_0$.
Further $\calP$ is K-shaped from $\Pi_0\ChanPar\Gamma_0$ if, and only if, $\GofP$ is $K$-shaped from $\Pi_0\ChanPar\Gamma_0$.
\end{customproposition}
\begin{proof}
First suppose
$\mathpzc{Q} = (\calP,\Pi_0 \ChanPar \Gamma_0,\Pi \ChanPar \Gamma)$ is a simple coverability query for ACPS in normal form, $\Pi = A_1 \parallel \cdots \parallel A_n$.
W.l.o.g. we may assume that all $A_i$ are non-commutative wrt to $\IofG$ 
(otherwise we may introduce a fresh channel $c_{\text{cov}}$, message $m_{\text{cov}}$, non-terminals $A'_i$ 
and local states $0,\ldots,K,K+1,\infty$, counting the number of non-commutative non-terminals on stack --- widening at $\geq K+2$, 
rules $(k,A_i) \to[effect=\epsilon] A'_i$ for $k \leq K+1$, 
$A'_i \to[effect=\rec{c_{\text{cov}}}{m_{\text{cov}}}] \epsilon$ and change $\Pi$ to
$\Pi = A'_1 \parallel \cdots \parallel A'_n$; applying polynomial-time normal form reduction afterwards; 
this transformation may increase the shaped constraint from $K$-shaped to $K+1$-shaped, since $(\infty,A_i)$ will clearly be commutative, and is also polynomial).

Defining then $\mathpzc{Q}' = \commabr{(\GofP,\Pi_0 \ChanPar \Gamma_0,\Pi \ChanPar \Gamma)}$ is then clearly a simple coverability query and can be polynomial-time computed. 

Since $\GofP[-]$ is clearly a bijection and both $\GofP[-]$ and $\GofP[-]^{-1}$ are polynomial time computable we may simply show that
$\mathpzc{Q}$ is a yes-instance if, and only if, $\mathpzc{Q}'$ is a yes-instance.

Suppose $\mathpzc{Q}$' is a yes-instance then 
$\Pi_0 \ChanPar \Gamma_0 \to[TS=\GofP,*] \Pi' \ChanPar \Gamma'$ such that
$\Pi \ChanPar \Gamma \leqAPCPS \Pi' \ChanPar \Gamma'$.
Let us maximally split this path so that
$\Pi_0 \ChanPar \Gamma_0 
	\to[label=\Pi_1 \ChanPar \Gamma_1,LTS=\GofP,*] \Pi_1 \ChanPar \Gamma_1 
	\to[label=\Pi_2 \ChanPar \Gamma_2,TS=\GofP,*] \cdots 
	\to[label=\Pi_n \ChanPar \Gamma_n,LTS=\GofP,*] \Pi_n \ChanPar \Gamma_n 
	\to[label=\epsilon,LTS=\GofP,*] \Pi' \ChanPar \Gamma'$ using the third labelling of this section.
Since $\mathpzc{Q}'$ is simple we know that $\Pi = A_1 \parallel \cdots \parallel A_n$
and thus $\Pi' \eqvI[{\IofG}] A_1\beta_1 \parallel \cdots \parallel A_n\beta_n \parallel \Pi''$.
Further since the path above is a maximal split we may assume that 
$\Pi_n \eqvI[{\IofG}] A_1\beta_1 \parallel \cdots \parallel A_n\beta_n \parallel \Pi'''$
and $\Pi''' \ChanPar \Gamma_n \to[label=\epsilon,LTS=\GofP,*] \Pi'' \ChanPar \Gamma'$ where the latter may not be split into two paths exposing a label.
Hence we may conclude that
$\Pi''' = B_1\,\beta^{\dagger}_1\,\beta^{\dagger'}_1 \parallel \cdots \parallel B_k\,\beta^{\dagger}_k\beta^{\dagger'}_k \parallel \Pi''''$ 
where $B_1,\ldots,B_k \in \NComN$, $\beta^{\dagger}_1,\ldots,\beta^{\dagger}_k \in \ComN^*$ and
$\beta^{\dagger'}_1,\ldots,\beta^{\dagger'}_k \in (\NComN\ComN^*)^*$.
that after a one rule step for each process $1$ up to $k$
we can go to
$\Pi''' \ChanPar \Gamma_n \to[label=\epsilon,LTS=\GofP,*] \beta^{\dagger''}_1\,\beta^{\dagger'}_1 \parallel \cdots \parallel \beta^{\dagger''}_k\beta^{\dagger'}_k \parallel \Pi'''' \ChanPar \Gamma^{\dagger} \to[label=\epsilon,LTS=\GofP,*] \Pi'' \ChanPar \Gamma'$
where $\beta^{\dagger''}_1,\ldots,\beta^{\dagger''}_k \in \ComN^*$, 
$\Pi'' = \Pi'''' \parallel \Pi(\vec{\lambda}^{\dagger}_1\cdots\vec{\lambda}^{\dagger}_k)$,
$\Gamma' = \Gamma^{\dagger} \oplus \Gamma(\vec{\lambda}^{\dagger}_1\cdots\vec{\lambda}^{\dagger}_k)$
for some $\vec{\lambda}^{\dagger}_i\vec{\lambda}^{\dagger'}_i \in \calL_{\GofP}(\beta^{\dagger''}_i)$ for each $i$.

We will now show that we can follow these transitions with $\calP$.
First
Lemma~\ref{app:lemma:acps:apcps:wbisimulation} tells us that we find the following path
$\Pi_0 \ChanPar \Gamma_0 
	\to[label=\Pi_1 \ChanPar \Gamma_1,LTS=\calP,*] \bar{\Pi}_1 \ChanPar \Gamma_1 
	\to[label=\Pi_2 \ChanPar \Gamma_2,TS=\calP,*] \cdots 
	\to[label=\Pi_n \ChanPar \Gamma_n,LTS=\calP,*] \bar{\Pi}_n \ChanPar \Gamma_n$ 
such that $\Pi_i \eqvI[{\IofG}] \bar{\Pi}_i$ for all $i \in \range{n}$.
Now $\bar{\Pi}_n = A_1\bar{\beta}_1 \parallel \cdots \parallel A_n\bar{\beta}_n \parallel \bar{\Pi}'''$
and $\Pi''' = B_1\,\bar{\beta}^{\dagger}_1\,\bar{\beta}^{\dagger'}_1 \parallel \cdots \parallel B_k\,\bar{\beta}^{\dagger}_k\bar{\beta}^{\dagger'}_k \parallel \bar{\Pi}''''$
since $A_1, \ldots, A_n,B_1,\ldots,B_k \in \NComN$. Now since $B_1,\ldots,B_k \in \NComN$ it is easy to see that we can follow the $\GofP$-path one rule step per process in:
$\bar{\Pi}''' \ChanPar \Gamma_n \to[label=\epsilon,LTS=\GofP,*] \bar{\beta}^{\dagger''}_1\,\bar{\beta}^{\dagger'}_1 \parallel \cdots \parallel \bar{\beta}^{\dagger''}_k\bar{\beta}^{\dagger'}_k \parallel \bar{\Pi'}''' \ChanPar \Gamma^{\dagger}$
where $\bar{\beta}^{\dagger''}_1, \ldots, \bar{\beta}^{\dagger''}_k \in \ComN^*$ and
$\bar{\beta}^{\dagger''}_i \eqvI[{\IofG}] \beta^{\dagger''}_i$ for all $i \in \range{k}$.
We can then pick $\bar{\vec{\lambda}}^{\dagger}_i \in \calL_{\calP}(\bar{\beta}^{\dagger''}_i)$ such that
$\bar{\vec{\lambda}}^{\dagger}_i \eqvI[{\IofG}] \vec{\lambda}^{\dagger}_i\vec{\lambda}^{\dagger'}_i$.
Lemma~\ref{app:apcps:acps:com:termination} then tells us we can reach 
$\bar{\Pi}''' \ChanPar \Gamma_n \to[label=\epsilon,LTS=\GofP,*] \bar{\Pi}^\dagger \parallel \bar{\Pi'}'''  \parallel \Pi(\bar{\vec{\lambda}}^{\dagger}_1\cdots\bar{\vec{\lambda}}^{\dagger}_k) \ChanPar \Gamma^{\dagger} \oplus \Gamma^{\dagger} \oplus \Gamma(\bar{\vec{\lambda}}^{\dagger}_1\cdots\bar{\vec{\lambda}}^{\dagger}_k)$
where we write $\bar{\Pi}^\dagger = \bar{\beta}^{\dagger'}_1 \parallel \cdots \parallel \bar{\beta}^{\dagger'}_k$.
We note that $\Gamma \leq \Gamma(\vec{\lambda}^{\dagger}_1\cdots\vec{\lambda}^{\dagger}_k) \leq \Gamma(\bar{\vec{\lambda}}^{\dagger}_1\cdots\bar{\vec{\lambda}}^{\dagger}_k)$.
We can now use these paths to extend our $\calP$ path to a covering configuration:
$\Pi_0 \ChanPar \Gamma_0 
	\to[LTS=\calP,*]  \bar{\Pi}_n \ChanPar \Gamma_n
	\to[LTS=\calP,*] A_1\bar{\beta}_1 \parallel \cdots \parallel A_n\bar{\beta}_n \parallel \bar{\Pi}^\dagger \parallel \bar{\Pi'}'''  \parallel \Pi(\bar{\vec{\lambda}}^{\dagger}_1\cdots\bar{\vec{\lambda}}^{\dagger}_k) \ChanPar \Gamma^{\dagger} \oplus \Gamma^{\dagger} \oplus \Gamma(\bar{\vec{\lambda}}^{\dagger}_1\cdots\bar{\vec{\lambda}}^{\dagger}_k)$
which clearly implies that 
$\Pi \ChanPar \Gamma$ is coverable in $\calP$ from  $\Pi_0 \ChanPar \Gamma_0$
Hence $\mathpzc{Q}$ is a yes-instance.

Conversely, suppose $\mathpzc{Q}$ is a yes-instance. Then
$\Pi_0 \ChanPar \Gamma_0 \to[TS=\calP,*] \Pi' \ChanPar \Gamma'$ such that
$\Pi \ChanPar \Gamma \leqACPS \Pi' \ChanPar \Gamma'$. Since $\mathpzc{Q}$ is a simple query
we know that $\Pi = A_1 \parallel \cdots \parallel A_n$ and hence we may conclude that
$\Pi' = A_1\,\beta_1 \parallel \cdots \parallel A_n\,\beta_n \parallel \Pi''$.
Lemma~\ref{app:apcps:simulates:acps} then allows us to conclude that
$\Pi_0 \ChanPar \Gamma_0 \to[TS=\GofP,*] \Pi'_0 \ChanPar \Gamma'$
where $\Pi'_0 \eqvI[{\IofG}] \Pi'$ i.e.~$\Pi'_0 \eqvI[{\IofG}] A_1\,\beta_1 \parallel \cdots \parallel A_n\,\beta_n \parallel \Pi''$ and thus clearly $\Pi \ChanPar \Gamma \leqAPCPS \Pi'_0 \ChanPar \Gamma'$ which implies 
$\mathpzc{Q}$' is a yes-instance.

For the second claim, 
we will prove that there is a reachable $K$-shaped configuration from $\Pi_0 \ChanPar \Gamma_0$ in $\calP$ if, and only if, there is a reachable $K$-shaped configuration from $\Pi_0 \ChanPar \Gamma_0$ in $\GofP$.

Suppose then that $\Pi_0 \ChanPar \Gamma_0 \to[TS=\calP,*] \Pi \ChanPar \Gamma$ and
that $\Pi \ChanPar \Gamma$ is $K$-shaped. An application of Lemma~\ref{app:apcps:simulates:acps} then yields that
$\Pi_0 \ChanPar \Gamma_0 \to[TS=\GofP,*] \Pi' \ChanPar \Gamma$
where $\Pi' \eqvI[{\IofG}] \Pi$ which clearly implies that $\Pi'$ is $K$-shaped as processes of $\Pi'$ are only (commutative) permutations of processes in $\Pi$.
Conversely, suppose $\Pi_0 \ChanPar \Gamma_0 \to[TS=\GofP,*] \Pi \ChanPar \Gamma$ and $\Pi \ChanPar \Gamma$ is $K$-shaped.
As before let us maximally split this path so that
$\Pi_0 \ChanPar \Gamma_0 
	\to[label=\Pi_1 \ChanPar \Gamma_1,LTS=\GofP,*] \Pi_1 \ChanPar \Gamma_1 
	\to[label=\Pi_2 \ChanPar \Gamma_2,TS=\GofP,*] \cdots 
	\to[label=\Pi_n \ChanPar \Gamma_n,LTS=\GofP,*] \Pi_n \ChanPar \Gamma_n 
	\to[label=\epsilon,LTS=\GofP,*] \Pi \ChanPar \Gamma$ using the third labelling of this section.
Now $\Pi = \beta_1\beta'_1 \parallel \cdots \parallel \beta_n\beta'_n$ with $\beta_i \in \ComN^*$ and
$\beta'_i \in (\NComN\ComN^*)^*$.
Using Lemma~\ref{app:apcps:acps:com:termination} we know that we can extend the path above to
$\Pi_0 \ChanPar \Gamma_0 \to[LTS=\GofP,*] \Pi \ChanPar \Gamma \to[LTS=\GofP,*] \beta'_1 \parallel \cdots \parallel \beta'_n \ChanPar \Gamma' =: \Pi' \ChanPar \Gamma'$.
Clearly $beta'_1 \parallel \cdots \parallel \beta'_n \in \M[\NComN\NonT^*]$ and since $\Pi \ChanPar \Gamma$ 
, and we have only ``executed off'' all commutative $\beta_i$'s, it is easy to see that $\Pi' \ChanPar \Gamma'$ is also $K$-shaped.
Since $beta'_1 \parallel \cdots \parallel \beta'_n \in \M[\NComN\NonT^*]$ we can invoke Lemma~\ref{app:lemma:acps:apcps:wbisimulation}
to see $\Pi_0 \ChanPar \Gamma_0 \to[LTS=\calP,*] \bar{\Pi}' \ChanPar \Gamma'$ 
such that $\bar{\Pi}' \eqvI[{\IofG}] \Pi'$. Hence as above we may conclude that $\bar{\Pi}' \ChanPar \Gamma'$ is $K$-shaped.

Hence we can conclude that there is a reachable $K$-shaped configuration from $\Pi_0 \ChanPar \Gamma_0$ in $\calP$ if, and only if, there is a reachable $K$-shaped configuration from $\Pi_0 \ChanPar \Gamma_0$ in $\GofP$.
From which we can deduce that $\calP$ is $K$-shaped if, and only if, $\GofP$ is $K$-shaped.
\end{proof}

\subsection{Proofs for Section \ref{sec:NNCT}}

\begin{customlemma}[\ref{lem:config:wqo}]
$(\commabr{\Minner,\leqMinner})$ and $(\commabr{\Configuration,\leqconfig})$ are WQO.
\end{customlemma}
\begin{proof}
Let for all $P \subseteq \Pinner$ define $M(P)$ to be set of complex tokens that exactly contain coloured tokens only of colours not in $P$, i.e.~$M(P) = \set{m \in \Minner : \forall p \in P. m(p) = \emptyset, \forall p' \notin p'. m(p') \neq \emptyset }$.
It should be clear that for each $m \in \Minner$ there is a unique $P \subseteq \Pinner$ such that $m \in M(P)$.
Hence we can see that $\Minner$ is isomorphic to the \emph{finite} disjoint union $\sum_{P \subseteq \Pinner} M(P)$.
Let $P \subseteq \Psimple$ we can see that $M(P)$ is isomorphic to $\Pi_{p \in P} \set{\emptyset} \times \Pi_{p \notin P}{\M{\set{\bullet}}}$ and so clearly $\leq_{M(P)} \is \Pi_{p \in P} \mathord{=} \times \Pi_{p \notin P} \subseteq_{\M{\set{\bullet}}}$ is a WQO (eliding the isomorphism) for $M(P)$.
Hence $\sum_{P \subseteq \Pinner} \leq_{M(P)}$ is a WQO for $\sum_{P \subseteq \Pinner} M(P)$ and thus for $\Minner$
(eliding the isomorphism).

Suppose $m, m' \in \Minner$ such that $m \leqMinner m'$. Let $P = \set{p : m(p) = \emptyset}$
since $m \leqMinner m'$ it is easy to see that $P = \set{p : m'(p) = \emptyset}$. Hence $m, m' \in M(P)$
and we also know that for all $p \notin P$ we have $0 < |m(p)| \leq |m(p')|$, i.e.~$m(p) \subseteq_{\M[\set{\bullet}]} m'(p)$, hence $m \leq_{M(P)} m'$ and thus $m \leq_{\sum_{P \subseteq \Pinner} M(P)} m'$.
In the opposite direction suppose $m, m'$ are such $m \leq_{\sum_{P \subseteq \Pinner} M(P)} m'$ then
both $m,m' \in M(P)$ for some $P \subseteq \Pinner$ and $m \leq_{M(P)} m'$ which means
that for all $p \in P$ it is the case that $0 = |m(p)| = |m'(p)| = |\emptyset|$ and for all
$p \notin P$ we have both $m(p), m'(p) \neq \emptyset$ and $m(p) \subseteq_{\M[\set{\bullet}]} m'(p)$,
i.e.~$0 < |m(p)| \leq |m'(p)|$ and hence $m \leqMinner m'$.
$(\commabr{\Minner,\leqMinner})$ is then isomorphic to a WQO and hence is a WQO.

For $(\commabr{\Configuration,\leqconfig})$.
It is clear that $\paren{\M[\bullet],\subseteq_{\M[\bullet]}}$ is a WQO; since $\Psimple$ is finite, Dickson's lemma tells us that $\paren{\Msimple,\leq_{\Msimple}}$ is a WQO where $\leq_{\Msimple}$ is the extension from $\subseteq_{\M[\bullet]}$ to 
$\subseteq_{\Psimple \to \M[\bullet]}$. 
Further since multi sets over a WQO are again a WQO, $\commabr{\varparen{\M[{\Minner}], \subseteq_{\M[\Minner]}}}$ is a WQO and so, we can infer that $(\commabr{\Configuration,\leqconfig})$ is a WQO which concludes the proof.
\end{proof}

\begin{lemma}\label{lem:oplus:monotone}
$\oplus$ is monotone in both arguments with respect to $\leqconfig$ and $\leqMinner$.
\end{lemma}
\begin{proof}
Suppose we have $m,m',m_0,m'_0 \in \Minner$ such that $m \leqMinner m_0$ and $m' \leqMinner m'_0$.
Let $p \in \Pinner$
If $|(m \oplus m')(p)| = 0$ we can conclude that $|m(p)| = |m'(p)| = 0$ then $\leqMinner$ guarantees
that $|m(p)| = |m_0(p)| = 0$ and $|m'(p)| = |m'_0(p)| = 0$ and hence $|(m_0 \oplus m'_0)(p)| = 0$.

Otherwise either $|m(p)|$ or $|m'(p)| > 0$. W.l.o.g. assume $|m(p)| > 0$ then we know that 
$|m(p)| \leq_{\Minner} |m_0(p)|$
and $|m_0(p)| > \vec{0}$. Further we can easily see that $|m'(p)| \leq |m'_0(p)|$.
Hence we can see that $0 < |(m \oplus m')(p)| \leq |(m_0 \oplus m'_0)(p)|$. 
Hence we can deduce that $m \oplus m' \leqMinner m_0 \oplus m'_0$ which 
is what we want to prove.

Suppose we have $s,s',s_0,s'_0 \in \Configuration$ and $s \leqconfig s_0$ and $s' \leqconfig s'_0$.
Let $p \in \Psimple$ then clearly
$$|(s \oplus s')(p)| = |s(p)|+|s'(p)| \leq |s_0(p)|+|s'_0(p)| = |(s_0 \oplus s'_0)(p)|.$$
Further let $p' \in \Pcomplex$ then
\begin{align*}
(s \oplus s')(p') &= \mset{m_1,\ldots,m_n,m'_1,\ldots,m'_{n'}} \text{ and }\\
(s_0 \oplus s'_0)(p') &= \mset{M_1,\ldots,M_N,M'_1,\ldots,M'_{N'}}
\end{align*}
where we may assume:
\begin{align*}
s(p') &= \mset{m_1,\ldots,m_n}, &
s'(p') &= \mset{m'_1,\ldots,m'_{n'}},\\
s_0(p') &= \mset{M_1,\ldots,M_N} \text{ and } &
s'_0(p') &= \mset{M'_1,\ldots,M'_{N'}}.
\end{align*}
Further since $s \leqconfig s_0$ and $s' \leqconfig s'_0$ we know that
$n \leq N$ and $n' \leq N'$ and for all $i \in \range{n}$ we have $m_i \leqMinner M_i$
and $m'_j \leqMinner M'_j$ for all $j \in \range{n'}$.
This gives us an injection that pairs up $m_i$ with $M_i$ and $m'_j$ with $M'_j$ for 
$i \in \range{n}$ and $j \in \range{n'}$. We can thus conclude that
$(s \oplus s')(p') \leq_{\M[\Minner]} (s_0 \oplus s'_0)(p')$ which of course
implies
$s \oplus s' \leqconfig s_0 \oplus s'_0$.
\end{proof}

\begin{lemma}\label{lem:ominus:monotone:first}
Suppose we have $s,s_0,s' \in \Configuration$, $s \leqconfig s_0$ and $s'(p)(m) = 0$ if both $p \in \Pcomplex$ and $m \neq \vec{0}$ then $s \ominus s' \leq s_0 \ominus s'$.
\end{lemma}
\begin{proof}
Let $p \in \Psimple$ then clearly
$$|(s \ominus s')(p)| = |s(p)|-|s'(p)| \leq |s_0(p)|-|s'(p)| = |(s_0 \ominus s')(p)|.$$
Further let $p' \in \Pcomplex$ then
\begin{align*}
(s \ominus s')(p')(\vec{0}) &= s(p')(\vec{0}) - s'(p')(\vec{0}) \\
							&\leq s_0(p')(\vec{0}) - s'(p')(\vec{0}) \\
							&= (s_0 \ominus s')(p')(\vec{0})\\
(s \ominus s')(p')\restriction &\paren{\M[\Minner] \setminus \set{\vec{0}}} \\
						    &= s(p')\restriction \paren{\M[\Minner] \setminus \set{\vec{0}}} \\
							&\leq s_0(p')\restriction \paren{\M[\Minner] \setminus \set{\vec{0}}} \\
							&= (s_0 \ominus s')(p')\restriction \paren{\M[\Minner] \setminus \set{\vec{0}}} \\
\end{align*}
Hence $s \oplus s' \leqconfig s_0 \oplus s'_0$.
\end{proof}

\begin{lemma}\label{lem:ominus:monotone:second}
Let $p \in \Pcomplex$ and $s,s' \in \Config$.
Suppose $m,m' \in \Minner$ such that $m' = \min\set{m_0 \in s'(p) : m \leqMinner m_0}$.
If $s \oplus \update{}{p}{m} \leqconfig s' \oplus \update{}{p}{m'}$ 
then $s \leqconfig s'$.
\end{lemma}
\begin{proof}
Let $p \in \Psimple$ then clearly
$$|s(p)| = |(s \oplus \update{}{p}{m})(p)| \leq |(s' \oplus \update{}{p}{m'})(p)| = |s'(p)|.$$
Further let $p' \in \Pcomplex$ such that $p \neq p'$ then
\begin{align*}
s(p') &= (s \oplus \update{}{p}{m})(p') \leq_{\M[\Minner]} (s' \oplus \update{}{p}{m})(p')\\
	  &= s'(p').
\end{align*}
Focussing on $p$ we see:
\begin{align*}
s(p) &= \mset{m_1,\ldots,m_n}, &
s'(p) &= \mset{m'_1,\ldots,m'_{n'}},
\end{align*}
where we know that
\begin{align*}
(s \oplus \update{}{p}{m})(p) &= \mset{m_1,\ldots,m_n,m} \text{ and }\\
(s_0 \oplus \update{}{p}{m'})(p) &= \mset{m'_1,\ldots,m_{n'},m'}
\end{align*}
where $n \leq n'$ and there is an bijection $h$ from $\commabr{M \is \varmset{m_1,\ldots,m_n,m}}$ to $M' \is \mset{m'_1,\ldots,m_{n},m'}$.
such that $m^0 \leqMinner h(m)^0$ for all $m^0 \in M$. 
Suppose $h$ pairs up $m$ with $m'$ then clearly $h$ is an injection justifying
$s(p) \leq_{\M[\Minner]} s'(p)$. Otherwise say $h(m_i) = m'$ then since $m \leqMinner h(m)$
we know that $m' \leqMinner h(m)$ since $m'$ is a minimum.
Thus $\update{h}{m_i}{h(m)}$ is an injection justifying $s(p) \leq_{\M[\Minner]} s'(p)$.
We can thus conclude $s \leqconfig s'$.
\end{proof}

\begin{lemma}\label{lem:restriction:monotone}
$\restriction$ is monotone in the first argument with respect to $\leqMinner$.
\end{lemma}
\begin{proof}
Let $m, m' \in \Minner$ such that $m \leqMinner m'$ and $P \subseteq \Pinner$.
Suppose $p \in P$ and $0 < m(p)$ then
$0 < |m(p)| = |(m \restriction P)(p)| = |m(p)| \leq |m'(p)| = |(m' \restriction P)(p)|$.
If $p \in P$ and $|m(p)| = 0$ then
$0 = |m(p)| = |(m \restriction P)(p)| = |m(p)| = |m'(p)| = |(m' \restriction P)(p)|$.
Otherwise $p' \notin P$ then
$|(m \restriction P)(p')| = 0 = |(m' \restriction P)(p')|$.
Hence clearly $m \restriction P \leqMinner m' \restriction P$.
\end{proof}

\begin{customproposition}[\ref{prop:nnct:wsts}]
Given an NNCT $\calN$, $(\Configuration,\to[TS={\calN}],\leqconfig)$ is a strict WSTS.
\end{customproposition}
\begin{proof}
Let $\calN = (\Psimple,\Pcomplex,\Pinner,\Rules,\colmap)$ and
suppose $s, s' \in \Configuration$ such that $s \lessconfig s'$. Further suppose we can make the transition
$s \to[rule=r,TS=\calN] t$ using rule $r \in \Rules$.
Let us perform a case analysis on $r$
\begin{itemize}
\item \emph{Case: } $r = (I,O) \in \SimpleRules$. \newline
Since $r$ is enabled at $s$ we know that $s \ominus I \in \Configuration$.
Lemma~\ref{lem:ominus:monotone:first} then yields that $s \ominus I \leqconfig s' \ominus I$
and hence clearly $r$ is enabled at $s'$.
Thus $s \to[rule=r,TS=\calN] s' \ominus I \oplus O =: t'$.
Since $t = s \ominus I \oplus O$ Lemma~\ref{lem:oplus:monotone} gives us
$t \leqconfig t'$. Since $s \neq s'$ it is clear that $t \neq t'$ and
hence we obtain $t \lessconfig t'$ which is what we want to prove.
\item \emph{Case: } $r = ((p,I),(p',\calc,O)) \in \ComplexRules$. \newline
Since $r$ is enabled at $s$ we know that for some $m \in s(p)$
$s \ominus I \update{}{p}{m} \in \Configuration$ and 
$$t = s \ominus I \update{}{p}{m} \oplus O \oplus \update{}{p'}{m \oplus \calc}.$$
First Lemma~\ref{lem:ominus:monotone:first} then yields that $s \ominus I \leqconfig s' \ominus I$.
Since $s \lessconfig s'$ there exists $m' \in s'(p)$ such that $m \leqMinner m'$;
w.l.o.g. we can assume that $m' = \min\varset{m_0 \in s'(p) : m \leqMinner m_0}$.
Since $I \in \ConfigSimple$ it is also the case that $m' = \min\set{m_0 \in s'(p) \ominus I: m \leqMinner m_0}$.
Lemma~\ref{lem:ominus:monotone:second} then yields that
$s \ominus I \update{}{p}{m} \leqconfig s' \ominus I \update{}{p}{m'}$ hence it is easy to see that
$r$ is enabled at $s'$.
Further 
$$s' \to[rule=r,TS=\calN] s' \ominus I \update{}{p}{m'} \oplus O \oplus \update{}{p'}{m' \oplus \calc} =: t'.$$
Since $m \leqMinner m'$ Lemma~\ref{lem:oplus:monotone} yields 
$m \oplus \calc \leqMinner m' \oplus \calc$.
Hence it is easy to see $\update{}{p'}{m \oplus \calc} \leqconfig \update{}{p}{m' \oplus \calc}$.
Lemma~\ref{lem:oplus:monotone} then gives us that 
$t \leqconfig t'$. 
Since $s \neq s'$ either $s' \ominus I \update{}{p}{m'} \neq s \ominus I \update{}{p}{m}$ 
or $m \neq m'$. Noting this we can see that $t \neq t'$ and
hence $t \lessconfig t'$ which is what we want to prove.
\item \emph{Case: } $r = ((p,I),(p',P, O)) \in \TransferRules$. \newline
Since $r$ is enabled at $s$ we know that for some $m \in s(p)$
$s \ominus I \update{}{p}{m} \in \Configuration$ and 
$$t = s \ominus I \update{}{p}{m} \oplus O \oplus \update{}{p'}{m_{\overline{P}}} \oplus (m_P \compose \colmap^{-1})$$
where $m_P = m \restriction P$ and $m_{\overline{P}} = m \restriction (\Pinner \setminus P)$.
First Lemma~\ref{lem:ominus:monotone:first} then yields that $s \ominus I \leqconfig s' \ominus I$.
Since $s \lessconfig s'$ there exists $m' \in s'(p)$ such that $m \leqMinner m'$;
w.l.o.g. we can assume that $m' = \min\varset{m_0 \in s'(p) : m \leqMinner m_0}$.
Since $I \in \ConfigSimple$ it is also the case that $m' = \min\set{m_0 \in s'(p) \ominus I: m \leqMinner m_0}$.
Lemma~\ref{lem:ominus:monotone:second} then yields that
$s \ominus I \update{}{p}{m} \leqconfig s' \ominus I \update{}{p}{m'}$ hence it is easy to see that
$r$ is enabled at $s'$.
Further 
$$s' \to[rule=r,TS=\calN] s' \ominus I \update{}{p}{m'} \oplus O \oplus \update{}{p'}{m'_{\overline{P}}} \oplus (m'_P \compose \colmap^{-1})  =: t'.$$
where $m_P = m \restriction P$ and $m_{\overline{P}} = m \restriction (\Pinner \setminus P)$.
Since $m \leqMinner m'$ Lemma~\ref{lem:restriction:monotone} yields both
$m_P \leqMinner m'_P$ and $m_{\overline{P}} \leqMinner m'_{\overline{P}}$.
Hence it is easy to see $\update{}{p'}{m_{\overline{P}}} \leqconfig \update{}{p}{m'_{\overline{P}}}$.
Further clearly $|(m_P \compose \colmap^{-1})(p)| \leq |(m_P \compose \colmap^{-1})(p)|$ for all $p \in \Psimple$
and thus $(m_P \compose \colmap^{-1}) \leqconfig (m_P \compose \colmap^{-1})$.
Lemma~\ref{lem:oplus:monotone} then gives us that 
$t \leqconfig t'$. 
Since $s \neq s'$ either $s' \ominus I \update{}{p}{m'} \neq s \ominus I \update{}{p}{m}$ 
or $m \neq m'$. The later implies that either $m_P \neq m'_P$ or $m_{\overline{P}} \neq m'_{\overline{P}}$.
Noting this we can see that $t \neq t'$ and
hence $t \lessconfig t'$ which is what we want to prove.
\end{itemize}
\end{proof}           
\newpage
\newcommand\CacheStack[1]{{\ulcorner #1 \urcorner}}
\begin{customtheorem}[\ref{thm:APCPS-to-NNCT}]
Simple coverability for $K$-shaped APCPS in the alternative semantics \Exptime-time reduces to NNCT coverability.
\end{customtheorem}

\begin{proof}
\makeatletter
\newcommand{\tt@place}[2]{p_{#1}^{\scriptscriptstyle #2}}
\newcommand{\tt@rho}{\colmap}
\newcommand{\tt@msg}[1][]{\mathit{msg}_{#1}}
\newcommand{\tt@chan}[1][]{\mathit{c}_{#1}}
\newcommand{\tt@marking}[2][]{\mathit{s}^{#1}_{#2}}
\newcommand{\tt@bsfun}{\mathit{\calF}}
\newcommand{\tt@bsfun@simple}{\tt@bsfun^{\scriptscriptstyle{\Gamma}}}
\newcommand{\tt@bsfun@proc}{\tt@bsfun^{\scriptscriptstyle{\pi}}}
\newcommand{\tt@bsfun@Pi}{\tt@bsfun^{\scriptscriptstyle{\Pi}}}
\newcommand{\tt@ctoken}[1][]{m}
\newcommand{\tt@ctoken@fun}[1][]{\tt@bsfun^{\scriptscriptstyle{I}}}
\newcommand{\tt@letter}[1][]{e}
\newcommand{\tt@q}[1][]{\$}
\newcommand{\tt@bisim}{\mathit{R}}

\let\oldPPL\PPL
\renewcommand{\PPL}{\varset{A^{\text{cov}}_{i} : i \in \range{n}}}

Fix a $K$-shaped APCPS $\calG = (\Sigma, I, \NonT, \Rules, S)$ from $\Pi_0 \ChanPar \Gamma_0$ where $K \geq 1$ and 
a simple coverability query $(\calG,\Pi_0 \ChanPar \Gamma_0,A^{\text{cov}}_1 \parallel \cdots \parallel A^{\text{cov}}_n \ChanPar \Gamma^{\text{cov}})$. We first define a simulating NNCT 
$\calN = (\Psimple,\Pcomplex,\Pinner,\Rules,\tt@rho)$. 
\begin{itemize}[-]
\item For each $\tt@msg \in \MMsg$ and $\tt@chan \in \Chan$, we introduce a simple place $\tt@place{\tt@chan, \tt@msg}{S}$. 
\item For each $X \in \NonT$, we introduce a simple place $\tt@place{\nu X}{S}$. 
\item For each $0 \leq k \leq K$, $X_k \cdots X_1 \in (\NComN)^\ast$, $\ell_{k+1} \, \ell_k \cdots \ell_1 \in (\PPL \cup \makeset{+,-})^\ast$ and $\tt@q \in \Sigma \cup \set{\epsilon} \cup \NonT$, 
we introduce a complex place $\tt@place{\tt@q', \ell_{k+1} \, X_k \, {\ell_k} \cdots X_1 \, {\ell_1}}{C}$; 
we also introduce a complex place 
$\tt@place{\tt@q', (+)\,  N \, \ell_{k+1} \, X_k \, {\ell_k} \cdots X_1 \, {\ell_1}}{C}$
for each $\tt@q' \in \Sigma \union \set{\epsilon}$ and $N \in \NonT$.
\item We introduce an auxiliary simple place $\tt@place{\mathit{budget}}{\mathit{CCFG}}$ and
for each $X \in \NonT$, we introduce a simple place $\tt@place{X}{\mathit{CCFG}}$;
for each $0 \leq k \leq K$, $X_k \cdots X_1 \in (\NComN)^\ast$, $\ell_{k+1} \, \ell_k \cdots \ell_1 \in (\P[\PPL] \cup \makeset{+})^\ast$ and $\tt@q \in \Sigma \cup \set{\epsilon} \cup \NonT$, 
we introduce a complex place $\tt@place{\tt@q, \ell_{k+1} \, X_k \, {\ell_k} \cdots X_1 \, {\ell_1}}{\mathit{CCFG}}$,
which will be used by $\calN$ to implement a CCFG widget.
\item For each non-terminal $A^{\text{cov}}_i$ in the coverability query where $i \in \range{n}$ $\calN$ has a simple place $\tt@place{A_i}{\mathit{cov}}$.
\item The NNCT $\calN$ will further have four special simple places: $\tt@place{}{\mathit{sim}}$, $\tt@place{}{\mathit{CCFG}}$, $\tt@place{}{\mathit{CCFG}+}$ and $\tt@place{}{\mathit{Query}}$.
\item For each $1 \leq k \leq K+1$ and $\tt@letter \in \ComSigma$, we introduce a colour 
$\tt@place{k, \tt@letter}{I}$.
\item We define a map $\tt@rho : \Pinner \to \Psimple$ by $\tt@place{k,\snd{\tt@chan}{\tt@msg}}{I} \mapsto \tt@place{c,m}{S}$ and 
$\tt@place{k, \nu X}{I} \mapsto \tt@place{\nu X}{S}$.
\end{itemize}
Let us further define three special simple markings 
$\tt@marking{\mathit{sim}} = \varupdate{}{\tt@place{}{\mathit{sim}}}{\mset{\bullet}}$,
$\tt@marking{\mathit{CCFG}} = \update{}{\tt@place{}{\mathit{CCFG}}}{\mset{\bullet}}$,
$\tt@marking{\mathit{CCFG}+} = \update{}{\tt@place{}{\mathit{CCFG}+}}{\mset{\bullet}}$ and
$\tt@marking{\mathit{Query}} = \update{}{\tt@place{}{\mathit{Query}}}{\mset{\bullet}}$.

An APCPS configuration 
$\Pi \ChanPar \Gamma \in \APCPSConfig$ is represented as an NNCT configuration as follows: 
\begin{asparaitem}
\item For each $\tt@chan \in \Chan$ and $\tt@msg \in \MMsg$ place $\tt@place{\tt@chan, \tt@msg}{S}$ contains
precisely one $\bullet$-token for each occurrence of $\tt@msg$ in $\tt@chan$ ---
we can formalise this as a function 
$\tt@bsfun@simple(\Gamma) = \varupdate[\mid \tt@chan \in \Chan, \tt@msg \in \MMsg]%
			{}{\tt@place{\tt@chan, \tt@msg}{S}}%
			{\varmset{\bullet^{\Gamma (\tt@chan) (\tt@msg)}}}$.

\item Let $0 \leq k \leq K+1$. The representation of the state of a process $\pi = \tt@q \, M_{k+1} \, \gamma \in \Pi$ with $\gamma = X_k \, M_k \cdots X_1 \, M_1$ is defined by a case analysis in three cases -- a general case and two edge cases:
\begin{asparaitem}[$\circ$]
\item If either $0 < k \leq K$ or $k = 0$ and it is not the case that both $\tt@q \in \NonT$ and $M_1 = \emptyset$ then
we represent $\pi$ as follows:
\begin{inparaenum}[(i)]
\item for each $e \in M_i$ there is one $\tt@place{i, \tt@letter}{I}$-coloured token in $\tt@ctoken$ 
 where $1 \leq i \leq k+1$ and $\tt@letter \in \ComSigma$
--- or equivalently $\tt@ctoken = \tt@ctoken@fun(M_1,\ldots,M_{k+1})$ where we define
the function $\tt@ctoken@fun(M_1,\ldots,M_{k+1}) = \varupdate[\mid  1 \leq i \leq k+1, \tt@letter \in \discn\ComSigma]{}{\tt@place{i, \tt@letter}{I}}{\varmset{\bullet^{M_i(e)}}}$;
\item for each $1 \leq i \leq k+1$, if the cache $M_i \in \TermCache$ then let $\ell_i = +$; 
otherwise let $\ell_i \in \set{-} \union \varset{A^{\text{cov}}_j \mid j \in \range{n}, \discn M_i(A^{\text{cov}}_j) > 0}$. 
We refer to a possible value of $\ell_i$ as a \emph{character} of $M_i$.
The sequence of $\tt@q$, the non-commutative non-terminals $X_k \cdots X_1$ and a possible choice of characters $\ell_{k+1},\cdots,\ell_1$ of $M_{k+1},\cdots M_1$ is represented as the complex place in which $\tt@ctoken$ located, i.e.~
$\tt@ctoken$ is placed in $\tt@place{\tt@q, \ell_{k+1} \, X_k \, {\ell_k} \cdots X_1 \, {\ell_1}}{C}$.
We formalise the representation of one process as the set of markings
$\tt@bsfun@proc(\tt@q \, M_{k+1} \, X_k \, M_k \cdots X_1 \, M_1) = \discn
\varset{\varupdate{}{\tt@place{\tt@q, \ell_{k+1} \, X_k \, {\ell_k} \cdots X_1 \, {\ell_1}}{C}}{\discn\varmset{\tt@ctoken@fun(M_1,\ldots,M_{k+1})}} \mid \ell_i \text{ character of } \discn M_i, i \in \range{k+1}}$,
where the possible characters of $M_{k+1},\cdots M_1$ can be thought of non-deterministically chosen.
\end{inparaenum}
\item If $k = 0$, $\tt@q \in \NonT$, and $M_1 = \emptyset$ then we may represent $\pi$ in addition to the representation above also as a $\bullet$ token in $\tt@place{\spn{\tt@q}}{S}$ hence the representation of $\pi$
is defined as $\tt@bsfun@proc(\pi) = \varset{\varupdate{}{\tt@place{\spn{\tt@q}}{S}}{\mset{\bullet}}} \union \tt@bsfun@proc(\tt@q\, \emptyset)$ where $\tt@bsfun@proc(\tt@q\, \emptyset)$ is as defined in the general case above.

\item If $k = K+1$ and $\pi = \tt@q \, M_{K+2} \, \gamma \in \Pi$ with $\gamma = X_{K+1} \, M_{K+1} \,\discn X_{K} \, M_{K} \cdots X_1 \, M_1$ then since $\calG$ is a $K$-shaped APCPS it will be the case that $M_{K+2} = \emptyset$ and $X_{K+1} \in \NonT$ and $\tt@q \in \Sigma \union \epsilon$. We notice that any character for $M_{k+2}$ must be $+$ and so 
$\tt@q \, M_{K+2} \, \gamma$ is represented by a complex token $\tt@ctoken = \tt@ctoken@fun(M_1,\ldots,M_{K+2}) \is \tt@ctoken@fun(M_1,\ldots,M_{K+1})$ and with the set of markings
$\tt@bsfun@proc(\tt@q \, M_{K+2} \, \gamma) \discn = 
\varset{\varupdate{}{\tt@place{\tt@q, (+)\, X_{K+1}\, \ell_{K+1} \, X_k \, {\ell_k} \cdots X_1 \, {\ell_1}}{C}}{\mset{\tt@ctoken@fun(M_1,\ldots,M_{K+2})}} \mid \ell_i \text{ character of } M_i, i \in \range{k+1}}$. 
For uniformity we will not treat this special case explicitly in the following but we will note that this special representation applies when $k = K+1$.
\end{asparaitem}
\smallskip

Representing all processes in $\Pi = \pi_1 \parallel \cdots \parallel \pi_n$ can then be formalised as
$$\tt@bsfun@Pi(\Pi) = \set{\Oplus_{i=1}^n s_i \mid s_i \in \tt@bsfun@proc(\pi_i)}.$$
\end{asparaitem}

The representation of $\Pi \ChanPar \Gamma$ is a set of configurations where $\Gamma$ is represented in $\calN$'s simple places by $\tt@bsfun@simple(\Gamma)$ and $\Pi$ is represented in $\calN$'s complex and simple places by a configuration from
$\tt@bsfun@Pi(\Pi)$ together with a $\bullet$-token in $\tt@place{\mathit{sim}}{}$ or formally as:
$\tt@bsfun(\Pi \ChanPar \Gamma) = \set{s \oplus (\tt@bsfun@simple(\Gamma) \oplus \tt@marking{\mathit{sim}}) \mid s \in \tt@bsfun@Pi(\Pi)}$.
We can define the relation $\tt@bisim \subseteq \APCPSConfig \times \P[\Config]$ by
$$\tt@bisim = \set{(\Pi \ChanPar \Gamma,\tt@bsfun(\Pi \ChanPar \Gamma)) \mid \Pi \ChanPar \Gamma \in \APCPSConfig};$$
we will prove in the following that $\tt@bisim$ is a weak bisimulation between a labelled version of
$\paren{\APCPSConfig, \toCM}$ and a labelled version of a ``co-universal powerset lifting'' of  $(\Config,\to[TS=\calN])$.

Let us turn to the implementation of the alternative semantics as defined in the right column of \cref{fig:standard_vs_alternative:semantics} in $\calN$'s rules.
In the following we write $\Xi = X_k \, \ell_k \cdots X_1 \, \ell_1$ to save space. 

The alternative semantics defines transitions of the form $\gamma \parallel \Pi \ChanPar \Gamma \toCM \gamma' \parallel \Pi' \ChanPar \Gamma'$. We will describe $\calN$'s rules by a case analysis on the rule that justifies the transition and relate the forms of $\gamma$, $\gamma'$, $\Pi$, $\Pi'$, $\Gamma$ and $\Gamma'$ to guide the reader's intuition.
\begin{asparaitem}[$\bullet$]
\item \textbf{Rule} \Cref{altrule:interleave:std}: $\Pi = \Pi'$, $\Gamma = \Gamma'$ and $A \to \beta$ is a $\calG$ rule.\newline
We perform case analysis on $\beta$:

\begin{itemize}[-]
\item $\beta = B\,C$ and $C$ non-commutative.\newline
 For each sequence of non-commutative non-terminals and cache characters
 $A \, \ell_{k+1} \, X_k \, \ell_k \cdots X_1 \, \ell_1$ where $k < K$ 
 we introduce a complex rule
  $\commabr{((\tt@place{A, \ell_{k+1} \, \Xi}{C}, \tt@marking{\mathit{sim}}), (\tt@place{B, \mathord{(+)} \, C \, \ell_{k+1} \, \Xi}{C}, \emptyset, \tt@marking{\mathit{sim}}))}$
 that moves a complex token from a place encoding $A\, \ell_{k+1} \, \Xi$ to a place representing $B\, \mathord{(+)} \, C \, \ell_{k+1} \, \Xi$.

\item $\beta = a$ where $a \in \Sigma \union \{\epsilon\}$\newline
For each sequence of non-commutative non-terminals and cache characters
 $A \, \ell_{k+1} \, X_k \, \ell_k \cdots X_1 \, \ell_1$ where $k \leq K$ 
 we introduce a complex rule
$\commabr{((\tt@place{A, \ell_{k+1} \, \Xi}{C}, \tt@marking{\mathit{sim}}), (\tt@place{a, \ell_{k+1} \, \Xi}{C}, \emptyset, \tt@marking{\mathit{sim}}))}$
that moves a complex token from a place encoding
$A\, \ell_{k+1} \, \Xi$ to a place representing $a\, \ell_{k+1} \, \Xi$.
\end{itemize}

\item \textbf{Rule} \Cref{altrule:interleave:com}: $A \rightarrow B\,C$ is a $\calG$-rule and $C$ commutative.\newline
Note that a reduction $C \toSF^* w \in (\ComN \union \ComSigma)^*$ is that of a commutative context-free grammar (CCFG) for which a Petri-net encoding exists \cite{Ganty:2012} by \citeauthor{Ganty:2012} which builds on an earlier encoding \cite{Esparza:1997} by \citeauthor{Esparza:1997}. \citeauthor{Ganty:2012} leverage a recent result 
\cite{EsparzaGKL:2011}: every word of a CCFG has a \emph{bounded-index derivation}. The CCFG widget can thus be augmented with a budget counter that ensures that the Petri-net encoding respects boundedness of index. Termination of such a CCFG computation is signaled by a transition which is only enabled when the budget counter reaches the set budget. 

We make use of a trivially modified CCFG widget \emph{\`a la} \citeauthor{Ganty:2012} to implement
transitions justified by Rule \Cref{altrule:interleave:com}. We will first define 
the rules of the CCFG widget and how it is activated. 

Let us define a few abbreviations. Let $\tt@marking[\mathit{budget}]{b} = \varupdate{}{\tt@place{\mathit{budget}}{\mathit{CCFG}}}{\varmset{\bullet^b}}$
and for each $N \in \ComN$ let $\tt@marking{N} = \varupdate{}{\tt@place{N}{\mathit{CCFG}}}{\varmset{\bullet}}$.

For each sequence of non-commutative non-terminals and cache characters
 $X_k \, \ell_k \cdots X_1 \, \ell_1$ where $k \leq K$ 
 we introduce the following complex rules:
\begin{asparaitem}[$*$]
  \item $\commabr{((\tt@place{A, (+) \, \Xi}{C}, \tt@marking{\mathit{sim}}), (\tt@place{B, \mathord{(+)} \, \Xi}{\mathit{CCFG}}, \emptyset, \tt@marking{\mathit{CCFG}+} \oplus O))}$,
  \item $\commabr{((\tt@place{A, (+) \, \Xi}{C}, \tt@marking{\mathit{sim}}), (\tt@place{B, \mathord{(-)} \, \Xi}{\mathit{CCFG}}, \emptyset, \tt@marking{\mathit{CCFG}} \oplus O))}$,
  \item $\commabr{((\tt@place{A, (-) \, \Xi}{C}, \tt@marking{\mathit{sim}}), (\tt@place{B, \mathord{(-)} \, \Xi}{\mathit{CCFG}}, \emptyset, \tt@marking{\mathit{CCFG}} \oplus O))}$ and
  \item $\commabr{((\tt@place{A, \ell \, \Xi}{C}, \tt@marking{\mathit{sim}}), (\tt@place{B, \ell \, \Xi}{\mathit{CCFG}}, \emptyset, \tt@marking{\mathit{CCFG}} \oplus O))}$ for each $\ell \in \PPL$
\end{asparaitem}
where 
$O = \tt@marking[\mathit{budget}]{|\calN|} \oplus \tt@marking{C}$.
Each of rules activates the simulation of the CCFG widget. The topmost cache character is set according to whether the cache is a $\TermCache$ ($+$), a potential $\MixedCache$ (-) or exposing non-terminal $A^{\text{cov}}_i$ in a $\MixedCache$. 
Further the $\tt@marking{\mathit{CCFG}+}$-``mode'' will enforce that the CCFG widget computes a $\TermCache$
while the $\tt@marking{\mathit{CCFG}}$ -``mode'' allows a $\MixedCache$ and the exposal of a non-terminal.
Further a token is placed in place $\tt@place{C}{\mathit{CCFG}}$ and the budget place is initialised with 
$|\calN|$ $\bullet$-tokens. The CCFG widget maintains the invariant that the number of tokens in the set of places 
$\Union_{N \in \NonT} \tt@place{N}{\mathit{CCFG}}$ plus the contents of 
$\tt@place{}{\mathit{budget}}$ equals $|\calN|+1$.

To exit the simulation of the CCFG widget we add the rules:
\begin{asparaitem}[$*$]
  \item $\commabr{((\tt@place{B, (+) \, \Xi}{\mathit{CCFG}}, I), (\tt@place{B, \mathord{(+)} \, \Xi}{C}, \emptyset, \tt@marking{\mathit{sim}}))}$, 
  \item $\commabr{((\tt@place{B, (-) \, \Xi}{\mathit{CCFG}}, I), (\tt@place{B, \mathord{(+)} \, \Xi}{C}, \emptyset, \tt@marking{\mathit{sim}}))}$ and
  \item $\commabr{((\tt@place{B, \ell \, \Xi}{\mathit{CCFG}}, I), (\tt@place{B, \ell \, \Xi}{C}, \emptyset, \tt@marking{\mathit{sim}}))}$ for each $\ell \in \PPL$
\end{asparaitem}
where 
$I = \tt@marking{\mathit{CCFG}} \oplus \tt@marking[\mathit{budget}]{|\calN|+1}$. 
Since we require that $\tt@place{}{\mathit{budget}}$ contains $|\calN|+1$ $\bullet$-tokens
the invariant tells us that the set of places $\Union_{N \in \NonT} \tt@place{N}{\mathit{CCFG}}$
must be empty when one of the above rules is enabled.
  
We turn to how the CCFG widget implements $\calG$'s commutative rules. 
For each $\calG$-rule $r$ which involves only commutative non-terminals we do a case analysis on $r$:
\begin{enumerate}[(I)]
\item(CCFG:I) $r = X \to Y\, Z$ and $X$, $Y$, $Z$ commutative \newline
We add the complex rules 
\begin{asparaitem}
\item $\commabr{((\tt@place{B, (+) \, \Xi}{\mathit{CCFG}}, \tt@marking{\mathit{CCFG}+} \oplus I), (\tt@place{B, (+) \, \Xi}{\mathit{CCFG}},\emptyset,\tt@marking{\mathit{CCFG}+} \oplus O))}$,
\item $\commabr{((\tt@place{B, (-) \, \Xi}{\mathit{CCFG}}, \tt@marking{\mathit{CCFG}} \oplus I), (\tt@place{B, (-) \, \Xi}{\mathit{CCFG}},\emptyset,\tt@marking{\mathit{CCFG}} \oplus O))}$, and
\item $\commabr{((\tt@place{B, \ell \, \Xi}{\mathit{CCFG}}, \tt@marking{\mathit{CCFG}} \oplus I), (\tt@place{B, \ell \, \Xi}{\mathit{CCFG}},\emptyset,\tt@marking{\mathit{CCFG}} \oplus O))}$ for each $\ell \in \PPL$
where
$I = \tt@marking{X} \oplus \tt@marking[\mathit{budget}]{2}$ and $O = \tt@marking{Y} \oplus \tt@marking{Z} \oplus \tt@marking[\mathit{budget}]{1}$.
Further if $W \in \varset{Y,Z}$ and $W = A^\text{cov}_i$ then we add the complex rule
$\commabr{((\tt@place{B, (-) \, \Xi}{\mathit{CCFG}}, \tt@marking{\mathit{CCFG}} \oplus I), (\tt@place{B, W \, \Xi}{\mathit{CCFG}},\emptyset,\tt@marking{\mathit{CCFG}} \oplus O))}$.
We notice that a non-terminal can only be non-deterministically exposed in $\tt@marking{\mathit{CCFG}}$-``mode'' and not changed
after it has been set.
\end{asparaitem}

\item(CCFG:IV) $r = X \to \tt@letter$, $\tt@letter \in \ComSigma \union \set{\epsilon}$ \newline
We add the complex rules
\begin{asparaenum}
\item $\commabr{((\tt@place{B, (+) \, \Xi}{\mathit{CCFG}},\tt@marking{\mathit{CCFG}+} \oplus I), (\tt@place{B, (+) \, \Xi}{\mathit{CCFG}},\calc,\tt@marking{\mathit{CCFG}+} \oplus O))}$,
\item $\commabr{((\tt@place{B, (-) \, \Xi}{\mathit{CCFG}},\tt@marking{\mathit{CCFG}} \oplus I), (\tt@place{B, (-) \, \Xi}{\mathit{CCFG}},\calc,\tt@marking{\mathit{CCFG}} \oplus O))}$ and
\item $\commabr{((\tt@place{B, \ell \, \Xi}{\mathit{CCFG}},\tt@marking{\mathit{CCFG}} \oplus I), (\tt@place{B, \ell \, \Xi}{\mathit{CCFG}},\calc,\tt@marking{\mathit{CCFG}} \oplus O))}$ for each $\ell \in \PPL$
where
$I = \tt@marking{X}$,
$O = \tt@marking[\mathit{budget}]{1}$ and 
$\calc = \varupdate{}{\tt@place{k+1, \tt@letter}{I}}{\mset{\bullet}}$.
\end{asparaenum}
\item(CCFG:nonterminal:terminal) For each $X \in \ComN$ \newline
We add the simple rule:
$(\tt@marking{\mathit{CCFG}} \oplus I, \tt@marking{\mathit{CCFG}} \oplus O)$ and
where
$I = \tt@marking{X}$ and $O = \tt@marking[\mathit{budget}]{1}$.
\end{enumerate}

The CCFG widget defined above is essentially the same as \citeauthor{Ganty:2012}'s in \cite{Ganty:2012}
with one difference:
our CCFG widget injects for each terminal symbol $\tt@letter \in \ComSigma$
a token of colour $\tt@place{k, \tt@letter}{I}$ into the unique complex token 
located in some $\tt@place{B, \ell \, \Xi}{\mathit{CCFG}}$ instead of
placing a token into a designated place $p_{\tt@letter}$.

Further our CCFG widget can be thought of as implementing two CCFGs derived from $\calG$'s rules.
One, indicated by the ``mode'' $\tt@marking{\mathit{CCFG}+}$, implements the CCFG induced by $\calG$'s commutative rules; the other, indicated by the ``mode'' $\tt@marking{\mathit{CCFG}}$, allows: 
\begin{inparaenum}[(a)]
\item $\calG$ to produce partial words using rules introduced by \ref{CCFG:nonterminal:terminal},
which can be thought of allowing a non-terminal $X$ to rewrite to a terminal $\bar{X}$ that is ignored; and
\item a non-terminal $A^{\text{cov}}_i$ may change the location of the unique complex token $m$ 
from some place $\tt@place{B, (-) \, \Xi}{\mathit{CCFG}}$ to $\tt@place{B, A^{\text{cov}}_i \, \Xi}{\mathit{CCFG}}$
which reflects the representation of the topmost cache's character for the process represented by $m$.
\end{inparaenum}
This concludes the description of the implementation of Rule \cref{altrule:interleave:com}.


\item \textbf{Rule} \Cref{altrule:receive}: $\Pi = \Pi'$, $\Gamma = \mcchan{(\mset{\tt@msg} \oplus q)}{\tt@chan}, \Gamma'$ and
$\gamma = (\rec{\tt@chan}{\tt@msg})\, \gamma'$. \newline
Let $\tt@q = \rec{\tt@chan}{\tt@msg}$.
For each sequence of non-terminals and cache characters
 $\tt@q \, \ell_{k+1} \, X_k \, \ell_k \cdots X_1 \, \ell_1$ where $k \leq K+1$ 
 we introduce a complex rule $((\tt@place{\tt@q, \ell_{k+1} \, \Xi}{C}, I), (\tt@place{\epsilon, \ell_{k+1} \, \Xi}{C}, \emptyset, O))$ where $I = \update{}{\tt@place{\tt@chan,\tt@msg}{S}}{\mset{\bullet}} \oplus \tt@marking{\mathit{sim}}$ and $O = \tt@marking{\mathit{sim}}$
 that moves a complex token from a place encoding
$\tt@q\, \ell_{k+1} \, \Xi$ to a place representing $\epsilon\, \ell_{k+1} \, \Xi$ while
removing a $\bullet$-token from the simple place representing
messages $\tt@msg$ in channel $\tt@chan$.

\item \textbf{Rule} \Cref{altrule:send}: $\Pi = \Pi'$, $\Gamma, \mcchan{(\mset{\tt@msg} \oplus q)}{\tt@chan} = \Gamma'$ and
$\gamma = (\snd{\tt@chan}{\tt@msg})\, \gamma'$. \newline
 Let $\tt@q = \snd{\tt@chan}{\tt@msg}$.
For each sequence of non-terminals and cache characters
 $\tt@q \, \ell_{k+1} \, X_k \, \ell_k \cdots X_1 \, \ell_1$ where $k \leq K+1$ 
 we introduce a complex rule $((\tt@place{\tt@q, \ell_{k+1} \, \Xi}{C}, I), (\tt@place{\epsilon, \ell_{k+1} \, \Xi}{C}, \emptyset, O))$ where $I = \tt@marking{\mathit{sim}}$ and $O = \update{}{\tt@place{\tt@chan,\tt@msg}{S}}{\mset{\bullet}} \oplus \tt@marking{\mathit{sim}}$
 that moves a complex token from a place encoding
$\tt@q\, \ell_{k+1} \, \Xi$ to a place representing $\epsilon\, \ell_{k+1} \, \Xi$ while
adding a $\bullet$-token from the simple place representing
messages $\tt@msg$ in channel $\tt@chan$.

\item \textbf{Rule} \Cref{altrule:spawn}. $\Pi \parallel X = \Pi'$, $\Gamma = \Gamma'$ and
$\gamma = \spn{X}\, \gamma'$

Let $\tt@q = \spn{X}$. 
For each sequence of non-terminals and cache characters
 $\tt@q \, \ell_{k+1} \, X_k \, \ell_k \cdots X_1 \, \ell_1$ where $k \leq K+1$ 
 we introduce a complex rule $((\tt@place{\tt@q, \ell_{k+1} \, \Xi}{C}, I), (\tt@place{\epsilon, \ell_{k+1} \, \Xi}{C}, \emptyset, O))$
 where $I = \tt@marking{\mathit{sim}}$ and $O = \update{}{\tt@place{\spn{X}}{S}}{\mset{\bullet}} \oplus \tt@marking{\mathit{sim}}$
 that moves a complex token from a place encoding
$\tt@q\, \ell_{k+1} \, \Xi$ to a place representing $\epsilon\, \ell_{k+1} \, \Xi$ while
adding a $\bullet$-token to the simple ``spawning'' place $\tt@place{\spn{X}}{S}$.
Additionally, for every $N \in \NonT$ we introduce a simple rule $r_{N} =(\update{}{\tt@place{\spn{N}}{S}}{\mset{\bullet}} \oplus \tt@marking{\mathit{sim}}, \update{}{\tt@place{N, +}{C}}{\mset{\emptyset}} \oplus \tt@marking{\mathit{sim}})$ that removes a $\bullet$-token from $\tt@place{\spn{N}}{S}$ and
adds the empty complex token $\emptyset$ to a complex place representing $N \, +$.
For reference we will call the rule $r_{N}$ a weak spawn rule for the non-terminal $N$.\label{tt:rule:weak_spawn}


\item \textbf{Rule} \Cref{altrule:dispatch:term}: $\Pi' = \Pi \oplus \Pi(M)$, $\Gamma' = \Gamma \oplus \Gamma(M)$,
$\gamma = M\, \gamma'$ and $M \in \TermCache$. \newline
For each sequence of non-terminals and cache characters
 $\epsilon \, (\mathord{+}) \, \discn X_k \, \ell_k \cdots X_1 \, \ell_1$ where $k \leq K+1$ 
 we introduce a transfer rule 
$\commabr{((\tt@place{\epsilon, + \, X_{k} \, \ell_{k} \cdots X_1 \, \ell_1}{C}, \tt@marking{\mathit{sim}}), (\tt@place{X_{k}, \ell_{k} \cdots X_1 \, \ell_1}{C}, P, \tt@marking{\mathit{sim}}))}$ 
where $P = \set{\tt@place{k+1, e}{I} \mid e \in \ComSigma}$ if $k \leq K$ and $P = \emptyset$ otherwise. 
This rule moves a complex token from a place encoding $\epsilon, +, \Xi$ to a place representing $\Xi$ while it
simulates the immediate despatch of the commutative concurrency actions, send $\snd{\tt@chan}{\tt@msg}$ and spawn $\spn{X}$, that are present in the top-level cache $M$ which is in $\TermCache$ since its character is $+$.
\item \textbf{Rule} \Cref{altrule:dispatch:nonterm}: $\Pi' = \Pi \oplus \Pi(M)$, $\Gamma' = \Gamma \oplus \Gamma(M)$,
$\gamma = M\, \gamma_0$, $\gamma' = M'\, \gamma_0$, $M' = M \restriction {(\ComN \union \PPL)}$ and $M \in \MixedCache$ \newline
For each sequence of non-terminals and cache characters
 $\epsilon \, \ell_{k+1} \, \discn X_k \, \ell_k \cdots X_1 \, \ell_1$ where $k \leq K+1$, $\ell_{k+1} \in \set{-} \union \PPL$ where $k \leq K$
  we introduce a transfer rule 
$\commabr{((\tt@place{\epsilon, \ell_{k+1} \, \Xi}{C}, \tt@marking{\mathit{sim}}), (\tt@place{\epsilon, \ell_{k+1} \, \Xi}{C}, P, \tt@marking{\mathit{sim}})}$ 
where $P = \set{\tt@place{k+1, e}{I} \mid e \in \ComSigma}$. 
This rule moves a takes a complex token from a place encoding $\epsilon, +, \Xi$ and places it back while it
simulates the immediate despatch of the commutative concurrency actions, send $\snd{\tt@chan}{\tt@msg}$ and spawn $\spn{X}$, that are present in the top-level cache $M$ which is in $\MixedCache$ since its character is $\ell_{k+1} \in \set{-} \union \PPL$.
\end{asparaitem}

\medskip

The coverability query $(\calG,\Pi_0 \ChanPar \Gamma_0,A^{\text{cov}}_1 \parallel \cdots \parallel A^{\text{cov}}_n \ChanPar \Gamma^{\text{cov}})$ is implemented by a family of widgets $W_1, \cdots, W_n$ where for each $i$, $W_i$ is a disjunction of the non-emptiness test of complex places of the shapes $\tt@place{\epsilon, A^{\text{cov}}_i \, \Xi}{C}$ and $\tt@place{A^{\text{cov}}_i, \ell \, \Xi'}{C}$, as $\ell$ ranges over $\widehat{\mathpzc{L}} := \PPL \cup \set{+, -}$, and $\Xi$ and $\Xi'$ range over $(\NComN \cdot \widehat{\mathpzc{L}})^\ast$ of the appropriate lengths. A widget signals that the non-emptiness test is satisfied by placing a $\bullet$-token into the simple place $\tt@place{A_i}{\mathit{cov}}$.
The intention is that we can use a coverability query for $\calN$ asking for at least one token in each of $\tt@place{A_1}{\mathit{cov}}, \ldots, \tt@place{A_n}{\mathit{cov}}$ to implement the coverability query $(\calG,\Pi_0 \ChanPar \Gamma_0,A^{\text{cov}}_1 \parallel \cdots \parallel A^{\text{cov}}_n \ChanPar \Gamma^{\text{cov}})$.

Formally, there is a simple rule $(\tt@marking{\mathit{sim}},\tt@marking{\mathit{Query}})$ that non-de\-ter\-min\-is\-ti\-cal\-ly terminates the simulation of $\calG$ and activates the processing of the coverability query.
Then for each $A^{\text{cov}}_i$ where $i \in \range{n}$, we implement the widget $W_i$ by the following complex rules:
for all non-commutative non-terminals and cache characters
 $X_k \, \ell_k \cdots X_1 \, \ell_1$ where $k \leq K$, and $\tt@q \in \Sigma \cup \set{\epsilon} \cup \NonT$ we introduce complex rules
$((\tt@place{A^{\text{cov}}_i, \ell_{k+1} \, \Xi}{C}, \tt@marking{\mathit{Query}}), (\tt@place{\epsilon, (-) \, \Xi}{C}, \emptyset, \tt@marking{\mathit{Query}} \oplus O))$ and
$((\tt@place{\epsilon, A^{\text{cov}}_i \, \Xi}{C}, \tt@marking{\mathit{Query}}), (\tt@place{\epsilon, (-) \, \Xi}{C}, \emptyset, \tt@marking{\mathit{Query}} \oplus O))$
where $O = \varupdate{}{\tt@place{A_i}{\mathit{cov}}}{\mset{\bullet}}$.

Let us briefly inspect the size of $\calN$ and let us assume that all
$n, K, |\calR|, |\calN|, |\Sigma| > 0$. It is clear that there exists a constants $c$, $c'$ and $c''$
such that $\calN$ has no more than $c \cdot n \cdot |\calR| \cdot |\Pcomplex|^2$ complex rules,
$\calN$ has no more than $c' \cdot |\calR| \cdot |\calN|$ simple rules and
$\calN$ has no more than $c''\cdot |\Pcomplex|^2$ transfer rules.
Hence there exists a constant $c'''$ such that $\calN$ has no more than
$c''' \cdot n \cdot |\calR| \cdot |\calN| \cdot |\Pcomplex|^2$ rules.
Further there exists constants $d, d', d'', d'''$
such that $\calN$ has no more than $d \cdot n \cdot |\Sigma| \cdot |\NonT|$ simple places,
$\calN$ has no more than $d' \cdot K \cdot |\Sigma|$ colours, and 
$\calN$ has no more than $d'' \cdot |\NonT|^{d''' \cdot K} \cdot |\Sigma|^{d''' \cdot K}$ complex places.
It is thus easy to see that $\calN$ can be computed from $\calG$ in \Exptime.

We will now prove that it checking whether $(\calG,\Pi_0 \ChanPar \Gamma_0,A^{\text{cov}}_1 \parallel \cdots \parallel A^{\text{cov}}_n \ChanPar \Gamma^{\text{cov}})$ is a yes-instance of simple coverability reduces to a coverability check on $\calN$.

In order to clarify which model induces a transition we will write
$\to[TS=\calG]$ for $\toCM$ in the following.
We label the transition system $(\APCPSConfig,\to[TS=\calG])$ in the following way:
if $\Pi \ChanPar \Gamma$, $\Pi' \ChanPar \Gamma' \in \APCPSConfig$ such that
$\Pi \ChanPar \Gamma \to[TS=\calG]\Pi' \ChanPar \Gamma'$ then the labelled version
has the transition $\Pi \ChanPar \Gamma \to[label=\Pi' \ChanPar \Gamma',LTS=\calG]\Pi' \ChanPar \Gamma'$.
Next, let us define the transition system $(\P[\Config],\to[TS={\P[\calN]}])$ which is
 a \emph{weak} co-universal power lifting of the transition system $(\Config,\to[TS=\calN])$.
 Suppose $S,S' \subseteq \Config$ then $S \to[TS={\P[\calN]}] S'$ just if
 for all $s' \in S'$ there exists an $s \in S$ such that 
 there are number of weak spawn rules (c.f.\ref{tt:rule:weak_spawn}) $r_1, \ldots, r_m$ where $m \geq 0$
 such that
 $s \to[rule=r_1,TS=\calN] s_1 \to[rule=r_2,TS=\calN] \cdots \to[rule=r_m,TS=\calN] s_m$
 and $s_m \to[TS=\calN] s'$.
 We label $(\P[\Config],\to[TS={\P[\calN]}])$ in the following way: suppose
 $S,S' \subseteq \Config$ such that $S \to[TS={\P[\calN]}] S'$;
 if there exists $\Pi \ChanPar \Gamma \in \APCPSConfig$ such that $S' = \tt@bsfun(\Pi \ChanPar \Gamma)$
 then we label the transition by $S \to[label=\Pi \ChanPar \Gamma,LTS={\P[\calN]}] S'$;
 otherwise we label the transition by $S \to[label=\epsilon,LTS={\P[\calN]}] S'$.

We will now prove that $\tt@bisim$ is a weak bisimulation between the labelled versions of 
$(\APCPSConfig,\to[TS=\calG])$ and $\paren{\P[\Config],\to[TS={\P[\calN]}]}$.
\medskip%

Let us first prove $\tt@bisim$ is a weak simulation. Suppose $\commabr{\varparen{\gamma \parallel \Pi \ChanPar \Gamma,\tt@bsfun(\gamma \parallel \Pi \ChanPar \Gamma)} \in \tt@bisim}$ and 
$\gamma \parallel \Pi \ChanPar \Gamma \to[label={\Pi_0 \ChanPar \Gamma_0},LTS=\calG,*] \gamma' \parallel \Pi' \ChanPar \Gamma'$ then clearly
$\gamma \parallel \Pi \ChanPar \Gamma \to[label={\gamma' \parallel \Pi' \ChanPar \Gamma'},LTS=\calG] \gamma' \parallel \Pi' \ChanPar \Gamma'$.

We want to show that $\tt@bsfun(\gamma \parallel \Pi \ChanPar \Gamma) \to[label={\gamma' \parallel \Pi' \ChanPar \Gamma'},LTS={\P[\calN]},*] \tt@bsfun(\gamma' \parallel \Pi' \ChanPar \Gamma')$.
So let $s' \in \tt@bsfun(\gamma' \parallel \Pi' \ChanPar \Gamma')$ then by definition 
$s' = s'_0 \oplus (\tt@bsfun@simple(\Gamma') \oplus \tt@marking{\mathit{sim}})$ for some $s'_0 \in \tt@bsfun@Pi(\gamma' \parallel \Pi')$. We can further decompose $s'_0$ and obtain $s'_0 = s'_1 \oplus s'_2$ where $s'_1 \in \tt@bsfun@proc(\gamma')$ and $s'_2 \in \tt@bsfun@Pi(\Pi')$.

Let us perform a case analysis on the rule that justifies the transition
$\gamma \parallel \Pi \ChanPar \Gamma \to[label={\gamma' \parallel \Pi' \ChanPar \Gamma'},LTS=\calG] \gamma' \parallel \Pi' \ChanPar \Gamma'$:
\begin{asparaitem}
\item \textbf{Rule} \Cref{altrule:interleave:std}: $\Pi = \Pi'$ , $\Gamma = \Gamma'$, $\gamma = A \, M \, X_k\, M_{k} \cdots X_1\, M_{1}$, and $\gamma' = \delta \, M \, X_k\, M_{k} \cdots X_1\, M_{1}$ for some $\calG$ rule $A \to \delta$.\newline
We can see that 
$s'_1 = \varupdate{}{\tt@place{\overline{\delta}, \ell_{k+1} \, \Xi}{C}}{\mset{\tt@ctoken@fun(M_1,\ldots,M_{k+1})}}$
where $\ell_i$ is a character of $M_i$ for each $i \in \range{k+1}$,
and $\overline{\delta}$ depends on the form of $\delta$: 
$\overline{\delta} = B \, (+) \, C$ if $\delta = B\,C$, $k < K$ and $C$ non-commutative; and
$\overline{\delta} = a$ if $\delta = a$ and $a \in \Sigma \union \set{\epsilon}$.
Note that it is now obvious that it is impossible that $\gamma'$ is represented in a simple place.

Let $s = (s_1 \oplus s'_2) \oplus (\tt@bsfun@simple(\Gamma) \oplus \tt@marking{\mathit{sim}})$
we write $s_1 \is \varupdate{}{\tt@place{A, \ell_{k+1} \, \Xi}{C}}{\mset{\tt@ctoken@fun(M_1,\ldots,M_{k+1})}}$.
Clearly $s_1 \in \tt@bsfun@proc(A M\discn X_k M_k \cdots X_1 M_1) = \tt@bsfun@proc(\gamma)$ and thus
$s \in \tt@bsfun(\gamma \parallel \Pi \ChanPar \Gamma)$.

We know that $\calN$ has a complex rule 
$\commabr{((\tt@place{A, \ell_{k+1} \, \Xi}{C}, \tt@marking{\mathit{sim}}), (\tt@place{\overline{\delta} \, \ell_{k+1} \, \Xi}{C}, \emptyset, \tt@marking{\mathit{sim}}))}$
Thus we know we can make the transition
$s \to[TS=\calN] (s'_1 \oplus s_2) \oplus (\tt@bsfun@simple(\Gamma) \oplus \tt@marking{\mathit{sim}})$
and clearly $s' = (s'_1 \oplus s_2) \oplus (\tt@bsfun@simple(\Gamma) \oplus \tt@marking{\mathit{sim}})$.

Since $s \in \tt@bsfun(\gamma' \parallel \Pi' \ChanPar \Gamma')$ was arbitrary we can conclude that in fact
$\tt@bsfun(\gamma \parallel \Pi \ChanPar \Gamma) \to[TS={\P[\calN]}] \tt@bsfun(\gamma' \parallel \Pi' \ChanPar \Gamma')$ and thus clearly
$$\tt@bsfun(\gamma \parallel \Pi \ChanPar \Gamma) \to[label=\gamma' \parallel \Pi' \ChanPar \Gamma',LTS={\P[\calN]},*] \tt@bsfun(\gamma' \parallel \Pi' \ChanPar \Gamma')$$
which is what we wanted to prove.
\item \textbf{Rule} \Cref{altrule:interleave:com}: $\Pi = \Pi'$ , $\Gamma = \Gamma'$,
$\gamma = A \, M \, X_k\, M_{k} \cdots X_1\, M_{1}$, $\gamma' = \delta \, M' \, X_k\, M_{k} \cdots X_1\, M_{1}$ 
 for some $\calG$ rule $A \to B\,C$, $C$ commutative, $C \to[*] w$, and $M' = M \oplus \M(w)$.\newline
In this case we can assume that
$\commabr{s'_1 = \varupdate{}{\tt@place{B, \ell' \, \Xi}{C}}{\varmset{\tt@ctoken@fun(M_1,\ldots,M_k,M')}}}$
where $\ell'$ is a character for $M' = M \oplus \M(w)$.

Let us define cache characters $\ell, \overline{\ell}$ and $\overline{\ell}'$ by a case analysis on $\ell'$ and $\M(w)$:
\begin{enumerate}[(C1)]
\item(ell:case:1) \emph{Case: } $\ell'= +$:
Define $\ell = \overline{\ell} = \overline{\ell}' = +$.
\item(ell:case:2) \emph{Case: } $\ell'= -$ and $M \in \MixedCache$:
Define $\ell = \overline{\ell} = \overline{\ell}' = -$.
\item(ell:case:3) \emph{Case: } $\ell'= -$ and $M \in \TermCache$:
Define $\ell = +$ and $\overline{\ell} = \overline{\ell}' = -$.
\item(ell:case:4) \emph{Case: } $\ell' \in \PPL$, $M \in \MixedCache$ and $\M(w)(\ell') > 0$:
Define $\ell = \overline{\ell} = -$ and $\overline{\ell}' = \ell'$.
\item(ell:case:5) \emph{Case: } $\ell' \in \PPL$, $M \in \TermCache$ and $\M(w)(\ell') > 0$:
Define $\ell = +$ and $\overline{\ell} = -$ and $\overline{\ell}' = \ell'$.
\item(ell:case:6) \emph{Case: } $\ell' \in \PPL$ and $\M(w)(\ell') = 0$:
Define $\ell = \overline{\ell} = \overline{\ell}' = \ell'$.
\end{enumerate}
Let us show that the above choices ensure that $\ell$ is a character for $M$.
If $M \in \TermCache$ then either $M' = M \oplus \M(w) \in \TermCache$ and so $\ell' = +$  and hence
$\ell = +$ and thus $\ell$ is a character for $M$. 
Otherwise $\ell' \in \set{-} \union \PPL$ and so since $M \in \TermCache$ by definition we have
$\ell = +$ and hence $\ell$ is a character for $M$.
If $M \in \MixedCache$ then clearly $M \oplus \M(w) \in \MixedCache$ and so either $\ell' = -$ or 
$\ell' \in \PPL$. So if $\ell' = -$ or $\M(w)(\ell') > 0$ then $\ell = -$ which is a character $M$.
Otherwise, $\ell' \in \PPL$ and $\M(w)(\ell') = 0$ then $\ell = \ell'$ and since $\ell'$ is not a character for $\M(w)$
but for $M'$ it must be that case that $\ell = \ell'$ is character for $M$.
Hence we can conclude that in all cases $\ell$ is a character for $M$.

Define $s^{s'}$ by 
$s^{s'} = s^{s'}_1 \oplus s'_2 \oplus (\tt@bsfun@simple(\Gamma) \oplus \tt@marking{\mathit{sim}})$
where $s^{s'}_1 = \varupdate{}{\tt@place{A, \ell \, \Xi}{C}}{\mset{\tt@ctoken@fun(M_1,\ldots,M_k,M)}}$.
It is clear that $s^{s'}_1 \in \tt@bsfun@proc(\gamma)$ and hence $s^{s'} \in \tt@bsfun(\gamma \parallel \Pi \ChanPar \Gamma)$.



Our definitions of $\ell$ and $\overline{\ell}$ ensure that there is a complex rule of the form
$\vec{r}(1) = \commabr{((\tt@place{A, \ell \, \Xi}{C}, \tt@marking{\mathit{sim}}), (\tt@place{B, \overline{\ell} \, \Xi}{\mathit{CCFG}}, \emptyset, \tt@marking{\mathit{CCFG}+?} \oplus O))}$
which is enabled at $\vec{s}^{s'}(1)$ to active a CCFG widget computation
where $\mathit{CCFG}+? = \mathit{CCFG}+$ if $\overline{\ell} = +$; and
$\mathit{CCFG}+? = \mathit{CCFG}$ otherwise; and
$O = \tt@marking[\mathit{budget}]{|\calN|} \oplus \varupdate{}{\tt@place{C}{\mathit{CCFG}}}{\varmset{\bullet}}$.

Hence we can make a transition
$\vec{s}^{s'}(1) \to[rule=\vec{r}(1),TS=\calN] \vec{s}^{s'}(2)$ where 
$\vec{s}^{s'}(2) = ({\vec{s}_1}^{s'}(2) \oplus s'_2) \oplus (\tt@bsfun@simple(\Gamma)\oplus \tt@marking{\mathit{CCFG}+?} \oplus \tt@marking[\mathit{budget}]{|\calN|} \oplus \varupdate{}{\tt@place{C}{\mathit{CCFG}}}{\varmset{\bullet}})$ and 
${\vec{s}'}^{s'}(2) = \varupdate{}{\tt@place{B, \overline{\ell} \, \Xi}{CCFG}}{\mset{\tt@ctoken@fun(M_1,\ldots,M_{k+1})}}$.

We appeal to \cite{Ganty:2012} that the CCFG widget allows us to simulate the computation $C \to[*] w$
in $L$ steps, where $L$ only depends on the derivation $C \to[*] w$ and not on ${\vec{s}'}^{s'}(2)$,
so that we get
$\vec{s}^{s'}(1+1) \to[rule=\vec{r}(2),TS=\calN] \vec{s}^{s'}(1+2) \cdots \to[rule=\vec{r}(1+L),TS=\calN] \vec{s}^{s'}(1+L)$
where for all $2 \leq i \leq L+1$ the rule $\vec{r}(i)$ is introduced by cases \ref{CCFG:I}--\ref{CCFG:nonterminal:terminal} above; and
for each terminal $\tt@letter \in \ComSigma$ occurring in $w$ we inject a $\tt@place{k+1, \tt@letter}{I}$-coloured token into
$\tt@ctoken@fun(M_1,\ldots,M_{k+1})$, i.e.~if $M_0 \leq_{\M} \M(w)$ then injecting a $\tt@place{k+1, \tt@letter}{I}$-coloured token $\tt@ctoken@fun(M_1,\ldots,M_{k+1} \oplus M_0)$ yields
$\tt@ctoken@fun(M_1,\ldots,M_{k+1} \oplus (M_0 \oplus \mset{\tt@letter}))$,
 and for each $2 \leq i \leq L+1$ the configuration $\vec{s}^{s'}(i)$ has $\tt@marking{\mathit{CCFG+?}}$ as a submarking. 
Further, we can see that in the computation of $C \to[*] w$ it is possible to expose a non-terminal $\ell_0$ if $\M(w)(\ell_0) > 0$ and $\overline{\ell} = -$. We can thus see that
if cases \ref{ell:case:1}--\ref{ell:case:5} applied for the definition of $\overline{\ell}$ and $\overline{\ell}'$ then we can assume that the above computation ignores all non-terminal; if case \ref{ell:case:6} applied then
$\M(w)(\ell') > 0$ and $\overline{\ell} = -$, so we can expose $\ell' =\overline{\ell}'$ along the above computation.

Hence we can conclude that
$\vec{s}^{s'}(1+L) =  ({\vec{s}_1}^{s'}(1+L) \oplus s_2) \oplus (\tt@bsfun@simple(\Gamma)\oplus \tt@marking{\mathit{CCFG}+?} \oplus \tt@marking[\mathit{budget}]{|\calN|+1})$
where ${\vec{s}'}^s(1+L) = \varupdate{}{\tt@place{B, \overline{\ell}' \, \Xi}{\mathit{CCFG}}}{\mset{\tt@ctoken@fun(M_1,\ldots,M_{k+1} \oplus (\M(w) \restriction \ComSigma))}}$.

We note that $\overline{\ell}' = \ell'$ in all cases of \ref{ell:case:1}--\ref{ell:case:6}.

It is then the case that a complex rule of
the form $\commabr{((\tt@place{B, \ell' \, \Xi}{\mathit{CCFG}}, \tt@marking{\mathit{CCFG}+?} \oplus \tt@marking[\mathit{budget}]{|\calN|+1}), (\tt@place{B, \ell' \, \Xi}{C}, \emptyset, \tt@marking{\mathit{sim}}))}$
is enabled and we can make the transition
$\vec{s}^{s'}(L+1) \to[TS=\calN] \vec{s}^{s'}(L+2)$
where $\vec{s}^{s'}(L+2) = ({\vec{s}_1}^{s'}(L+2) \oplus s_2) \oplus (\tt@bsfun@simple(\Gamma)\oplus \tt@marking{\mathit{sim}})$
and $\vec{s}'^{s'}(L+2) = \varupdate{}{\tt@place{B, \ell' \, \Xi}{C}}{\mset{\tt@ctoken@fun(M_1,\ldots,M_{k+1} \oplus (\M(w) \restriction \ComSigma))}}$.

Now $\tt@ctoken@fun(M_1,\ldots,M_{k+1} \oplus (\M(w) \restriction \ComSigma)) = 
\commabr{\tt@ctoken@fun(M_1,\ldots,M \oplus \M(w))}$. 
Thus $\vec{s}'^{s'}(L+2) = s'_1$ and hence $\vec{s}^{s'}(L+2) = s'$.

We now need to lift these paths to $\P[\calN]$. The recipe above gives us for each 
$s' \in \tt@bsfun(\gamma' \parallel \Pi' \ChanPar \Gamma')$ a path
$\vec{s}^{s'}(1) \cdots \vec{s}^{s'}(L+2)$.
We note that $L$ is in fact independent of $s'$ since it is only dependent on the derivation $C \to[*] w$.
Further $\vec{s}^{s'}(1) \in \tt@bsfun(\gamma \parallel \Pi \ChanPar \Gamma)$, $\vec{s}^{s'}(L+2) = s'$
and each for all $2 \leq i < L+1$ the configuration $\vec{s}^{s'}(i)$ contains either submarking $\tt@marking{\mathit{CCFG+}}$ or $\tt@marking{\mathit{CCFG}}$ and hence $\nexists \Pi_0 \ChanPar \Gamma_0$ such that
$\vec{s}^{s'}(i) \in \tt@bsfun(\Pi_0 \ChanPar \Gamma_0)$.
Let us define the following subsets of $\Config$: for $1 \leq i \leq L+2$ let
$\vec{S}(i) = \set{\vec{s}^{s'}(i) \mid s' \in \tt@bsfun(\gamma' \parallel \Pi' \ChanPar \Gamma')}$.
By definition we have for all $\vec{s}^{s'}(i+1) \in \vec{S}(i+1)$ that $\vec{s}^{s'}(i) \to[TS=\calN] \vec{s}^{s'}(i+1) \in \vec{S}(i+1)$ if $1 \leq i < L +2$
and hence
$\vec{S}(1) \to[TS={\P[\calN]}] \vec{S}(2) \to[TS={\P[\calN]}] \cdots \to[TS={\P[\calN]}] \vec{S}(L+1) \to[TS={\P[\calN]}] \tt@bsfun(\gamma' \parallel \Pi' \ChanPar \Gamma')$
since clearly $\vec{S}(L+2) = \tt@bsfun(\gamma' \parallel \Pi' \ChanPar \Gamma')$.
Looking at the labelled version of $\P[\calN]$ our reasoning above implies that
For all $1 < i < L +2$ $\nexists \Pi_0 \ChanPar \Gamma_0$ such that $\vec{S}(i) = \tt@bsfun(\Pi_0 \ChanPar \Gamma_0)$
and hence
$\tt@bsfun(\gamma \parallel \Pi \ChanPar \Gamma) \supseteq \vec{S}(1) \to[label=\gamma' \parallel \Pi' \ChanPar \Gamma',LTS={\P[\calN]},*] \tt@bsfun(\gamma' \parallel \Pi' \ChanPar \Gamma')$
which is what we wanted to prove.

\item \textbf{Rule} \Cref{altrule:receive},\Cref{altrule:send}: 
$\gamma = \tt@q\, \gamma'$ and $\tt@q \in \Sigma$. \newline 
We will prove the case where the transition is justified by rule  \Cref{altrule:receive};
the case using rule \Cref{altrule:send} is analogous.

In this case $\Pi = \Pi'$, $\Gamma = \mcchan{(\mset{\tt@msg} \oplus q)}{\tt@chan}, \Gamma_0$,
$\Gamma' = \mcchan{q}{\tt@chan}, \Gamma_0$ and $\tt@q = \rec{\tt@chan}{\tt@msg}$.

If $\gamma' = M_{k+1}\, X_k\, M_{k} \cdots X_1\, M_{1}$ is easy to see that 
$s'_1 = \varupdate{}{\tt@place{\epsilon, \ell_{k+1} \, \Xi}{C}}{\mset{\tt@ctoken@fun(M_1,\ldots,M_{k+1})}}$
where $\ell_i$ is a character of $M_i$ for each $i \in \range{k+1}$; and
$\tt@bsfun@simple(\Gamma) = \tt@bsfun@simple(\Gamma') \oplus \varupdate{}{\tt@place{\tt@chan,\tt@msg}{S}}{\mset{\bullet}}$.

Hence let $s = (s_1 \oplus s'_2) + (\tt@bsfun@simple(\Gamma) \oplus \tt@marking{\mathit{sim}})$ where
we write $s_1 \is \varupdate{}{\tt@place{\rec{\tt@chan}{\tt@msg}, \ell_{k+1} \, \Xi}{C}}{\mset{\tt@ctoken@fun(M_1,\ldots,M_{k+1})}}$.

The complex rule 
$\commabr{((\tt@place{\rec{\tt@chan}{\tt@msg}, \ell_{k+1} \, \Xi}{C}, \update{}{\tt@place{\tt@chan,\tt@msg}{S}}{\mset{\bullet}} \oplus \tt@marking{\mathit{sim}}), (\tt@place{\epsilon, \ell_{k+1} \, \Xi}{C}, \emptyset, \tt@marking{\mathit{sim}}))}$ 
is then enabled at $s$ and we can make the transition
$s \to[TS=\calN] (s'_1 \oplus s'_2) + (\tt@bsfun@simple(\Gamma')\oplus \tt@marking{\mathit{sim}}) = s'$

Since $\gamma = (\rec{\tt@chan}{\tt@msg})\, \gamma'$ we know that $s_1 \in \tt@bsfun@proc((\rec{\tt@chan}{\tt@msg})\, \gamma')$ from which we can conclude that
$s_1 \oplus s'_2 \in \tt@bsfun@Pi(\gamma \parallel \Pi') = \tt@bsfun@Pi(\gamma \parallel \Pi)$ and hence
$(s_1 \oplus s'_2) \oplus (\tt@bsfun@simple(\Gamma)\oplus \tt@marking{\mathit{sim}}) \in \tt@bsfun(\gamma \parallel \Pi \ChanPar \Gamma)$.

Since $s' \in \tt@bsfun(\gamma' \parallel \Pi' \ChanPar \Gamma')$ was arbitrary we can conclude that in fact
$\tt@bsfun(\gamma \parallel \Pi \ChanPar \Gamma) \to[TS={\P[\calN]}] \tt@bsfun(\gamma' \parallel \Pi' \ChanPar \Gamma')$ and thus clearly
$$\tt@bsfun(\gamma \parallel \Pi \ChanPar \Gamma) \to[label=\gamma' \parallel \Pi' \ChanPar \Gamma',LTS={\P[\calN]},*] \tt@bsfun(\gamma' \parallel \Pi' \ChanPar \Gamma')$$
which concludes this case.
\item \textbf{Rule} \Cref{altrule:spawn}: 
 $\Pi' = \Pi \parallel X$, $\Gamma = \Gamma'$, and $\gamma = \spn{X}\, \gamma'$. \newline 
If $\gamma' = M_{k+1}\, X_k\, M_{k} \cdots X_1\, M_{1}$ is easy to see that 
$s'_1 = \varupdate{}{\tt@place{\epsilon, \ell_{k+1} \, \Xi}{C}}{\mset{\tt@ctoken@fun(M_1,\ldots,M_{k+1})}}$
where $\ell_i$ is a character of $M_i$ for each $i \in \range{k+1}$.
Further since $s'_2 \in \tt@bsfun@Pi(\Pi') = \tt@bsfun@Pi(\Pi \parallel X)$ we know that
$s'_2 = s''_2 \oplus s'_X$ where $s''_2 \in \tt@bsfun@Pi(\Pi)$ and $s'_X \in \tt@bsfun@proc(X)$.

Thus we can define $s = (s_1 \oplus s''_2) + (\tt@bsfun@simple(\Gamma) \oplus \tt@marking{\mathit{sim}})$ where
we write $s_1 \is \varupdate{}{\tt@place{\spn{X}, \ell_{k+1} \, \Xi}{C}}{\mset{\tt@ctoken@fun(M_1,\ldots,M_{k+1})}}$.

Hence the complex rule 
$$\commabr{((\tt@place{\spn{X}, \ell_{k+1} \, \Xi}{C}, \tt@marking{\mathit{sim}}), (\tt@place{\epsilon, \ell_{k+1} \, \Xi}{C}, \emptyset, \update{}{\tt@place{\spn{X}}{S}}{\mset{\bullet}} \oplus \tt@marking{sim}))}$$
is enabled at $s$ and we can make the transition
$s \to[TS=\calN] (s'_1 \oplus s''_2 \oplus s_X) \oplus (\tt@bsfun@simple(\Gamma')\oplus \tt@marking{\mathit{sim}})$
where we write $s_X = \update{}{\tt@place{\spn{X}}{S}}{\mset{\bullet}}$.

Since $s'_X \in \tt@bsfun@proc(X)$ either $s'_X = s_X$ or $s'_X = \update{}{\tt@place{X, (+)}{C}}{\mset{\emptyset}}$.
In the latter case we can use a weak spawn rule to make the transition
$$(s'_1 \oplus s''_2 \oplus s_X) \oplus (\tt@bsfun@simple(\Gamma')\oplus \tt@marking{\mathit{sim}}) \to[TS=\calN]
s'.$$

Since $\gamma = (\spn{X})\, \gamma'$ we know that $s_1 \in \tt@bsfun@proc(\gamma)$ from which we can conclude that
$s_1 \oplus s''_2 \in \tt@bsfun@Pi(\gamma \parallel \Pi)$.
Therefore
$(s_1 \oplus s''_2) \oplus (\tt@bsfun@simple(\Gamma) \oplus \tt@marking{\mathit{sim}}) \in \tt@bsfun(\gamma \parallel \Pi \ChanPar \Gamma)$.

Since $s' \in \tt@bsfun(\gamma' \parallel \Pi' \ChanPar \Gamma')$ was arbitrary we can conclude that
$\tt@bsfun(\gamma \parallel \Pi \ChanPar \Gamma) \to[TS={\P[\calN]}] \tt@bsfun(\gamma' \parallel \Pi' \ChanPar \Gamma')$ and thus clearly
$$\tt@bsfun(\gamma \parallel \Pi \ChanPar \Gamma) \to[label=\gamma' \parallel \Pi' \ChanPar \Gamma',LTS={\P[\calN]},*] \tt@bsfun(\gamma' \parallel \Pi' \ChanPar \Gamma')$$
which is what we wanted to prove.
\item \textbf{Rule} \Cref{altrule:dispatch:term}--\Cref{altrule:dispatch:nonterm}: $\Pi' = \gamma' \parallel \Pi \parallel \Pi(M)$, $\Gamma' = \Gamma \oplus \Gamma(M)$,
$\gamma = \epsilon\, M\, \gamma_0$ and $\gamma' = \epsilon \, M'\,\gamma_0$. \newline 
We will prove the case where the transition is justified by rule \Cref{altrule:dispatch:term};
the case using rule \Cref{altrule:dispatch:nonterm} is proved similarly.

In this case, in fact $M' = \emptyset$ and thus $\gamma' = \gamma_0$ which implies $\gamma = \epsilon\, M\, \gamma'$.
If $\gamma' = X_k\, M_{k} \cdots X_1\, M_{1}$ is easy to see that 
$s'_1 = \varupdate{}{\tt@place{X_k\, \ell_{k} \cdots X_1\, \ell_{1}}{C}}{\mset{\tt@ctoken@fun(M_1,\ldots,M_{k})}}$
where $\ell_i$ is a character of $M_i$ for each $i \in \range{k}$.
Further since $s'_2 \in \tt@bsfun@Pi(\Pi') = \tt@bsfun@Pi(\Pi \parallel \Pi(M))$
we decompose $s'_2$ into $s'_2 = s''_2 \oplus s'_{\Pi(M)}$ where $s''_2 \in \tt@bsfun@Pi(\Pi)$
and $s'_{\Pi(M)} \in \tt@bsfun@Pi(\Pi(M))$.

We also know that $M \in \TermCache$, so let $\ell = +$,
$s = (s_1 \oplus s''_2) \oplus (\tt@bsfun@simple(\Gamma) \oplus \tt@marking{\mathit{sim}})$
and $s_1 = \varupdate{}{\tt@place{\epsilon, + \, \Xi}{C}}{\mset{\tt@ctoken@fun(M_1,\ldots,M_{k},M)}}$.
Thus $s_1 \in \tt@bsfun@proc(M \, \gamma') = \tt@bsfun@proc(\gamma)$ and hence 
$s \in \tt@bsfun(\gamma \parallel \Pi \ChanPar \Gamma)$.

The transfer rule 
$r = \commabr{((\tt@place{\epsilon, + \, \Xi}{C}, \tt@marking{\mathit{sim}}), (\tt@place{\Xi}{C}, P, \tt@marking{\mathit{sim}}))}$ is then enabled at $s$
where $P = \set{\tt@place{k+1, e}{I} \mid e \in \ComSigma}$ 
if $k \leq K$ and $P = \emptyset$ otherwise (in which case $M = \emptyset$).
Let $\commabr{m = \tt@ctoken@fun(M_1,\ldots,M_{k},M)}$, $m_P = m \restriction P$ and $m_{\overline{P}} = m \restriction (\Pinner \setminus P)$. Then we know that using $r$ we can make the transition
$$s \to[TS=\calN] (s_{\overline{P}} \oplus s''_2) \oplus (\tt@bsfun@simple(\Gamma)\oplus \tt@marking{\mathit{sim}}) \oplus
\paren{m_{P} \compose \colmap^{-1}} =: s'_0$$
where $s_{\overline{P}} = \update{}{\tt@place{\Xi}{C}}{\mset{m_{\overline{P}}}}$.

It is not hard to see that 
$\commabr{m_{\overline{P}} = \tt@ctoken@fun(M_1,\ldots,M_{k},\emptyset) = \tt@ctoken@fun(M_1,\ldots,M_{k})}$ and hence $s_{\overline{P}} = s'_1$.

Let $\calG^\Gamma = \set{\tt@place{\tt@chan,\tt@msg}{S} \mid \tt@chan \in \Chan, \tt@msg \in \MMsg}$
and $\calG^\nu = \varset{\tt@place{\spn{X}}{S} \mid\discn X \in \NonT}$.
We notice that for all $\tt@chan \in \Chan, \tt@msg \in \MMsg$
\begin{align*}
\tt@bsfun@simple\paren{\Gamma(M \restriction \set{\snd{\tt@chan}{\tt@msg}})} =&
\tt@bsfun@simple\paren{\cchan{\tt@msg^{M(\snd{\tt@chan}{\tt@msg})}}{\tt@chan}} \\
=& \update{}{\tt@place{\tt@chan,\tt@msg}{S}}{\mset{\bullet^{M(\snd{\tt@chan}{\tt@msg})}}}\\
=& \update{}{\tt@place{\tt@chan,\tt@msg}{S}}{m_{P}(\snd{\tt@chan}{\tt@msg})} \\
=& \paren{m_{P} \compose \colmap^{-1}} \restriction \set{\tt@place{\tt@chan,\tt@msg}{S}}
\end{align*}
from which we conclude that
\begin{align*}
\paren{m_{P} \compose \colmap^{-1}} \restriction \calG^\Gamma =& 
\Oplus_{\substack{\tt@chan \in \Chan, \\ \tt@msg \in \MMsg}}
\tt@bsfun@simple\paren{\Gamma(M \restriction \set{\snd{\tt@chan}{\tt@msg}})} \\
=& 
\tt@bsfun@simple\paren{\Gamma(M)}
\end{align*}
and $\tt@bsfun@simple(\Gamma') = \tt@bsfun@simple(\Gamma) \oplus \paren{m_{P} \compose \colmap^{-1}} \restriction \calG^\Gamma$.

Further for each $X \in \NonT$ let
$s_X = \update{}{\tt@place{\spn{X}}{S}}{\mset{\bullet}}$; then
it is the case that
$\Oplus_{i=1}^{M(\spn(X))} s_X \in \tt@bsfun@Pi(\Pi(M \restriction \set{\spn{X}}))$
and thus $\Oplus_{X \in \NonT} \Oplus_{i=1}^{M(\spn(X))} s_X \in \tt@bsfun@Pi(\Pi(M))$.
Now $\Oplus_{i=1}^{M(\spn(X))} s_X = \paren{m_{P} \compose \colmap^{-1}} \restriction \set{\tt@place{\spn{X}}{S}}$
and thus $\Oplus_{X \in \NonT} \Oplus_{i=1}^{M(\spn(X))} s_X = \varparen{m_{P} \compose \colmap^{-1}} \restriction \calG^{\nu}$.
Since $\varparen{m_{P} \compose \colmap^{-1}} = \varparen{m_{P} \compose \colmap^{-1}} \restriction \calG^\Gamma \oplus
\varparen{m_{P} \compose \colmap^{-1}} \restriction \calG^\nu$ we have
$$s'_0 = (s'_1 \oplus s_2 \oplus \paren{m_{P} \compose \colmap^{-1} \restriction \calG^\nu}) \oplus (\tt@bsfun@simple(\Gamma') \oplus \tt@marking{\mathit{sim}}).$$

Let us inspect $s'_{\Pi(M)}$. Since $s'_{\Pi(M)} \in \tt@bsfun@Pi(\Pi(M))$ either $s'_{\Pi(M)} = m_{P} \compose \colmap^{-1} \restriction \calG^\nu$
or using a number of weak spawn rules we can transition in $s'_0$ from $m_{P} \compose \colmap^{-1} \restriction \calG^\nu$ to $s'_{\Pi(M)}$.
This implies (using a number of weak spawn rules after the first step) we can make the transition
$s \to[TS=\calN] s'_0 \to[TS=\calN,*] s'$.

Since $s' \in \tt@bsfun(\gamma' \parallel \Pi' \ChanPar \Gamma')$ was arbitrary we can conclude that in fact
$\tt@bsfun(\gamma \parallel \Pi \ChanPar \Gamma) \to[TS={\P[\calN]}] \tt@bsfun(\gamma' \parallel \Pi' \ChanPar \Gamma')$ and thus clearly
$$\tt@bsfun(\gamma \parallel \Pi \ChanPar \Gamma) \to[label=\gamma' \parallel \Pi' \ChanPar \Gamma',LTS={\P[\calN]},*] \tt@bsfun(\gamma' \parallel \Pi' \ChanPar \Gamma')$$
which is what we wanted to prove.
\end{asparaitem}
We can thus deduce that $\tt@bisim$ is a weak simulation.

Let us first prove a lemma that will be the basis of our proof that $\tt@bisim^{-1}$ is a weak simulation.
\begin{lemma*}
Suppose $\vec{s}(1),\cdots\vec{s}(m)$ is a sequence of configurations 
such that
\begin{inparaenum}[(A)]
\item for each $i \in \range{m-1}$
$\vec{s}(i) =s^i_0 \to[rule=\vec{r}^i(1),TS=\calN] s^i_1 \to[rule=\vec{r}^i(2),TS=\calN] \discn \cdots 
\to[rule=\vec{r}^i(m^i-1),TS=\calN] s^i_{m^i-1} \to[rule=\vec{r}^i(m^i),TS=\calN] s^i_{m^i} = \vec{s}(i+1)$ 
where only one rule, $\vec{r}^i(j^i)$ say, is not a weak spawn rule;
\item $\vec{s}(1) \in \tt@bsfun(\Pi \ChanPar \Gamma)$ and $\vec{s}(m) \in \tt@bsfun(\Pi' \ChanPar \Gamma')$; and
\item for all $2 \leq i < m$ there \emph{does not} exist $\Pi_0 \ChanPar \Gamma_0$ such that
$\vec{s}(i) \in \tt@bsfun(\Pi_0 \ChanPar \Gamma_0)$.
\end{inparaenum}
Then $\Pi \ChanPar \Gamma \to[TS=\calG] \Pi' \ChanPar \Gamma'$.
\end{lemma*}
\begin{proof}
It is easy to see that if a configuration $s \in \tt@bsfun(\Pi_0 \ChanPar \Gamma_0)$ 
for some $\Pi_0 \ChanPar \Gamma_0$ and $s \to[rule=r,TS=\calN] s'$ where $r$ is a weak spawn rule 
then $s' \in \tt@bsfun(\Pi_0 \ChanPar \Gamma_0)$. Hence if we look at the transitions
$\vec{s}(1) \to[rule=\vec{r}^1(1),TS=\calN] s^1_1 \to[rule=\vec{r}^1(2),TS=\calN] \cdots 
\to[rule=\vec{r}^1(j^1),TS=\calN] s^1_{j^1}$
we can conclude that for all $0 \leq l < j^1$ we have $s^1_l \in \tt@bsfun(\Pi \ChanPar \Gamma)$.
Thus we can decompose $s^1_{j^1-1}$ as $s^1_{j^1-1} = s^{\Pi} \oplus \tt@bsfun@simple(\Gamma) \oplus \tt@marking{\mathit{sim}}$ where $s^{\Pi} \in \tt@bsfun@Pi(\Pi)$.

Let us do a case analysis on which rule of the \nameref{def:alt-conc-sem} justifies 
the introduction of $\vec{r}^1(j^1)$ into $\calN$'s rules.
\begin{asparaitem}
\item \textbf{Rule} \Cref{altrule:interleave:com}.\newline
There are many candidate rules $\vec{r}^1(j^1)$ at first sight. However, we know that
$s^1_{j^1-1}$ does not have $\tt@marking{\mathit{CCFG}}$ or $\tt@marking{\mathit{CCFG}+}$
as a submarking. Hence we must have
$\vec{r}^1(j^1) = 
\commabr{((\tt@place{A, \ell_{k+1} \, \Xi}{C}, \tt@marking{\mathit{sim}}), (\tt@place{B, \overline{\ell} \, \Xi}{\mathit{CCFG}}, \emptyset, \tt@marking{\mathit{CCFG}+?} \oplus O))}$
for some $\calG$ rule $A \to B\, C$ where $C$ is commutative
and $\overline{\ell} \in \set{+,-}$ and $\mathit{CCFG}+? \in \set{\mathit{CCFG}+,\mathit{CCFG}}$ if $\ell_{k+1} = +$;
$\overline{\ell} = \ell_{k+1}$, and $\mathit{CCFG}+? = \mathit{CCFG}$ otherwise; and
$O = \tt@marking[\mathit{budget}]{|\calN|} \oplus \varupdate{}{\tt@place{C}{\mathit{CCFG}}}{\varmset{\bullet}}$.

Thus with $s^1_{j^1}$ a computation of the CCFG starts and until
for some $i,j$ the rule $\vec{r}^i(j)$ exists the CCFG computation, i.e.~$\vec{r}^i(j) = \commabr{((\tt@place{B, \overline{\ell}' \, \Xi}{\mathit{CCFG}}, \tt@marking{\mathit{CCFG}+?} \oplus \tt@marking[\mathit{budget}]{|\calN|+1}), (\tt@place{B, \overline{\ell}' \, \Xi}{C}, \emptyset, \tt@marking{\mathit{sim}}))}$
where $\overline{\ell}' \in \set{-} \union \set{A^{\text{cov}}_i \mid i \in \range{n} \M(w)(A^{\text{cov}}_i) > 0}$ if $\overline{\ell} = -$
and $\overline{\ell}' = \overline{\ell}$ otherwise.
Along this path between the configurations $s^1_{j^1}$ and $s^i_{j}$ the submarking $\tt@marking{\mathit{sim}}$
cannot appear and hence no weak spawning rule can apply. Thus
we know that $m^1 = j^1$, $j = 1$ and for $1 < i' < i$ it is the case that $m^{i'} = 1$.
We know that $s^{\Pi} \in \tt@bsfun@Pi(\Pi)$ and hence
  $\Pi = A\, M_{k+1} \, X_k \discn\, M_k \cdots X_1 \, M_1 \parallel \Pi_0$
  and $\commabr{\varupdate{}{\tt@place{A, \ell_{k+1} \, \Xi}{C}}{\mset{m}} \in \tt@bsfun@proc(A\, M_{k+1} \, X_k \, M_k \cdots X_1 \, M_1)}$ where $m = \tt@ctoken@fun(M_1,\ldots,M_{k+1})$.
Further it is clear that $m \in \tt@place{A, \ell_{k+1} \, \Xi}{C}$
in $s^1_{j^1-1}$.
Hence the CCFG widget guarantees \cite{Ganty:2012} that
$s^i_{j} = (s^{\Pi} \ominus \varupdate{}{\tt@place{A, \ell_{k+1} \, \Xi}{C}}{\mset{m}}
   \oplus \varupdate{}{\tt@place{B, \overline{\ell}' \, \Xi}{C}}{\mset{m'} }) 
   \oplus \tt@bsfun@simple(\Gamma) 
   \oplus \tt@marking{\mathit{sim}}$ 
   where $m' = \tt@ctoken@fun(M_1,\ldots,M_{k+1} \oplus \M(w))$ for some $w$ such that $C \to[*] w$.

Further the CCFG widget guarantees that if $\overline{\ell}' = +$ then $M_{k+1} \oplus \M(w) \in \TermCache$
and $M_{k+1} \oplus \M(w) \in \MixedCache$ otherwise.
Hence $\varupdate{}{\tt@place{B, \overline{\ell}' \, \Xi}{C}}{\mset{m}} \in \tt@bsfun@proc(B,\discn (M_{k+1} \oplus \M(w)) \, X_k \discn\, M_k \cdots X_1 \, M_1)$ and we can deduce that
  $s^{\Pi} \ominus \varupdate{}{\tt@place{A, \ell_{k+1} \, \Xi}{C}}{\mset{m}}
     \oplus \varupdate{}{\tt@place{B, \overline{\ell}' \, \Xi}{C}}{\mset{m'}} \in \tt@bsfun@Pi(B\, (M_{k+1} \oplus \M(w)) \, X_k \discn \, M_k \cdots X_1 \, M_1 \parallel \Pi_0)$.
It is thus the case that
$s^i_{1} \in \tt@bsfun@Pi(B\, (M_{k+1} \oplus \M(w)) \, X_k \, M_k\discn \cdots X_1 \, M_1 \parallel \Pi_0 \ChanPar \Gamma)$ which implies
  \begin{inparaenum}[(a)]
  \item $m^i = 1$, \item $m = i$ and thus \item $\Pi' \ChanPar \Gamma' = B\, (M_{k+1} \oplus \M(w)) \, X_k \, M_k\discn \cdots X_1 \, M_1 \parallel \Pi_0 \ChanPar \Gamma$.
  \end{inparaenum}
  We can easily check that $\Pi \ChanPar \Gamma \to[TS=\calG] \Pi' \ChanPar \Gamma'$ and
  hence $\Pi \ChanPar \Gamma \to[label=\Pi' \ChanPar \Gamma', LTS=\calG] \Pi' \ChanPar \Gamma'$ which is what we wanted to prove.
\item \emph{Rule} \Cref{altrule:interleave:std}.\newline
We can assume that
$\vec{r}^1(j^1) = 
\commabr{((\tt@place{A, \ell_{k+1} \, \Xi}{C}, \tt@marking{\mathit{sim}}), (\tt@place{\overline{\delta} \, \ell_{k+1} \, \Xi}{C}, \emptyset, \tt@marking{\mathit{sim}}))}$
for some non-terminals and cache characters
 $A \ell_{k+1} \, X_k \discn\, \ell_k \cdots X_1 \, \ell_1$
and $\calG$ rule $A \to \delta$
  where $k \leq K+1$ and
  $\overline{\delta}$ depends on $\delta$ and $k$: $\overline{\delta} = B, \, (+) \, C$ if $\delta = B\,C$, $k < K$; and
$\overline{\delta} = a$ if $\delta = a$ with $a \in \Sigma \union \set{\epsilon}$.

Hence we can deduce there exists a complex token $m$ located at place $\tt@place{A, \ell_{k+1} \, \Xi}{C}$
in $s^1_{j^1-1}$
 and 
 $s^1_{j^1} = (s^{\Pi} \ominus \varupdate{}{\tt@place{A, \ell_{k+1} \, \Xi}{C}}{\mset{m}}
   \oplus \varupdate{}{\tt@place{\overline{\delta}, \ell_{k+1} \, \Xi}{C}}{\mset{m}}) 
   \oplus \tt@bsfun@simple(\Gamma) 
   \oplus \tt@marking{\mathit{sim}}$.

We know $s^{\Pi} \in \tt@bsfun@Pi(\Pi)$ and hence
  $\Pi = A\, M_{k+1} \, X_k \discn\, M_k \cdots X_1\discn \, M_1 \parallel \Pi_0$
  and $\commabr{\varupdate{}{\tt@place{A, \ell_{k+1} \, \Xi}{C}}{\mset{m}} \in \tt@bsfun@proc(A\, M_{k+1} \, X_k \, M_k \cdots \discn X_1 \, M_1)}$.
  Then clearly $\varupdate{}{\tt@place{\overline{\delta}, \ell_{k+1} \, \Xi}{C}}{\mset{m}} \in \tt@bsfun@proc(\delta\, M_{k+1} \, X_k \, M_k \discn \cdots X_1 \, M_1)$ and we can deduce that
  $s^{\Pi} \ominus \varupdate{}{\tt@place{A, \ell_{k+1} \, \Xi}{C}}{\mset{m}}
     \oplus \varupdate{}{\tt@place{\overline{\delta}\, \ell_{k+1} \, \Xi}{C}}{\mset{m}} \in \tt@bsfun@Pi(\delta\, M_{k+1} \, X_k \discn \, M_k \cdots X_1 \, M_1 \parallel \Pi_0)$.

It is thus easy to see that $s^1_{j^1} \in \tt@bsfun@Pi(\delta\, M_{k+1} \, X_k \, M_k\discn \cdots X_1 \, M_1 \parallel \Pi_0 \ChanPar \Gamma)$.
  From our assumptions above this implies
  \begin{inparaenum}[(a)]
  \item $m^1 = j^1$, \item $m = 1$ and thus \item $\Pi' \ChanPar \Gamma' = \delta\, M_{k+1} \, X_k \, M_k\discn \cdots X_1 \, M_1 \parallel \discn \Pi_0 \ChanPar \Gamma$.
  \end{inparaenum}
  It is also easy to see that $\Pi \ChanPar \Gamma \to[TS=\calG] \Pi' \ChanPar \Gamma'$ and
  hence $\Pi \ChanPar \Gamma \to[label=\Pi' \ChanPar \Gamma', LTS=\calG] \Pi' \ChanPar \Gamma'$ which is what we wanted to prove.
\item \textbf{Rules} \Cref{altrule:receive},\Cref{altrule:send}. \newline 
We will prove the case that the introduction of $\vec{r}^1(j^1)$ is justified by rule \Cref{altrule:send}.
The case of \Cref{altrule:receive} are proved similarly.

We can thus assume that
$\vec{r}^1(j^1) = 
\commabr{((\tt@place{\snd{\tt@chan}{\tt@msg}, \ell_{k+1} \, \Xi}{C}, \tt@marking{\mathit{sim}}), 
 (\tt@place{\epsilon, \ell_{k+1} \, \Xi}{C}, \emptyset, \varupdate{}{\tt@place{\tt@chan,\tt@msg}{S}}{\mset{\bullet}} \oplus \tt@marking{\mathit{sim}}))}$ for some non-terminals and cache characters
 $\ell_{k+1} \, X_k \, \ell_k \cdots X_1 \, \ell_1$ where $k \leq K+1$.
 Hence there exists a complex token $m$ located at place $\tt@place{\snd{\tt@chan}{\tt@msg}, \ell_{k+1} \, \Xi}{C}$
 and 
 $s^1_{j^1} = (s^{\Pi} \ominus \varupdate{}{\tt@place{\snd{\tt@chan}{\tt@msg}, \ell_{k+1} \, \Xi}{C}}{\mset{m}}
   \oplus \varupdate{}{\tt@place{\epsilon, \ell_{k+1} \, \Xi}{C}}{\mset{m}}) 
   \oplus (\tt@bsfun@simple(\Gamma) 
   \oplus \varupdate{}{\tt@place{\tt@chan,\tt@msg}{S}}{\mset{\bullet}})
   \oplus \tt@marking{\mathit{sim}}$.

  Since $s^{\Pi} \in \tt@bsfun@Pi(\Pi)$ we can deduce that 
  $\Pi = \snd{\tt@chan}{\tt@msg}\, M_{k+1} \, X_k \discn\, M_k \cdots X_1 \, M_1 \parallel \Pi_0$
  and $\commabr{\varupdate{}{\tt@place{\snd{\tt@chan}{\tt@msg}, \ell_{k+1} \, \Xi}{C}}{\mset{m}} \in 
  \tt@bsfun@proc(\snd{\tt@chan}{\tt@msg}\,\discn M_{k+1} \, X_k \, M_k \cdots X_1 \, M_1)}$.
  Then clearly $\varupdate{}{\tt@place{\epsilon, \ell_{k+1} \, \Xi}{C}}{\mset{m}} \in \tt@bsfun@proc(\epsilon\, M_{k+1} \, X_k \, M_k \cdots X_1 \, M_1)$ and we can deduce that
  $s^{\Pi} \ominus \varupdate{}{\tt@place{\snd{\tt@chan}{\tt@msg}, \ell_{k+1} \, \Xi}{C}}{\mset{m}}
     \oplus \varupdate{}{\tt@place{\epsilon, \ell_{k+1} \, \Xi}{C}}{\mset{m}} \in \tt@bsfun@Pi(\epsilon\, M_{k+1} \, X_k \discn \, M_k \cdots X_1 \, M_1 \parallel \Pi_0)$.
  Further it easy to see that
  $\tt@bsfun@simple(\Gamma) 
     \oplus \varupdate{}{\tt@place{\tt@chan,\tt@msg}{S}}{\mset{\bullet}} =
    \tt@bsfun@simple(\Gamma \oplus \cchan{\tt@msg}{\tt@chan})$.
  We can then conclude that $s^1_{j^1} \in \tt@bsfun@Pi(\epsilon\, M_{k+1} \, X_k \, M_k\discn \cdots X_1 \, M_1 \parallel \Pi_0 \ChanPar \Gamma \oplus \cchan{\tt@msg}{\tt@chan})$.
  From our assumptions above this implies
  \begin{inparaenum}[(a)]
  \item $m^1 = j^1$, \item $m = 1$ and thus \item $\Pi' \ChanPar \Gamma' = \epsilon\, M_{k+1} \, X_k \, M_k\discn \cdots X_1 \, M_1 \parallel \Pi_0 \ChanPar \Gamma \oplus \cchan{\tt@msg}{\tt@chan}$.
  \end{inparaenum}
  It is also easy to see that $\Pi \ChanPar \Gamma \to[TS=\calG] \Pi' \ChanPar \Gamma'$ and
  hence $\Pi \ChanPar \Gamma \to[label=\Pi' \ChanPar \Gamma', LTS=\calG] \Pi' \ChanPar \Gamma'$ which is what we wanted to prove.
\item \textbf{Rule} \Cref{altrule:spawn}.\newline 
Since by definition $\vec{r}^1(j^1)$ is not a weak spawn rule, we can assume that
$\vec{r}^1(j^1) = 
\commabr{
((\tt@place{\spn{X}, \ell_{k+1} \, \Xi}{C}, \tt@marking{\mathit{sim}}), (\tt@place{\epsilon, \ell_{k+1} \, \Xi}{C}, \emptyset, s_X \oplus \tt@marking{\mathit{sim}}))
}$
 for $X \in \NonT$ and some non-terminals and cache characters
 $\ell_{k+1} \, X_k \, \ell_k \cdots X_1 \, \ell_1$ where $k \leq K+1$ and $s_X = \update{}{\tt@place{\spn{X}}{S}}{\mset{\bullet}}$.
 Hence there exists a complex token $m$ located at place $\tt@place{\spn{X}, \ell_{k+1} \, \Xi}{C}$
 and 
 $s^1_{j^1} = (s^{\Pi} \ominus \varupdate{}{\tt@place{\spn{X}, \ell_{k+1} \, \Xi}{C}}{\mset{m}}
   \oplus \varupdate{}{\tt@place{\epsilon, \ell_{k+1} \, \Xi}{C}}{\mset{m}}) 
   \oplus (\tt@bsfun@simple(\Gamma) 
   \oplus \varupdate{}{\tt@place{\spn{X}}{S}}{\mset{\bullet}})
   \oplus \tt@marking{\mathit{sim}}$.

  Since $s^{\Pi} \in \tt@bsfun@Pi(\Pi)$ we can deduce that 
  $\Pi = \spn{X}\, M_{k+1} \, X_k \discn\, M_k \cdots X_1 \, M_1 \parallel \Pi_0$
  and $\commabr{\varupdate{}{\tt@place{\spn{X}, \ell_{k+1} \, \Xi}{C}}{\mset{m}} \in \tt@bsfun@proc(\snd{\tt@chan}{\tt@msg}\,\discn M_{k+1} \, X_k \, M_k \cdots X_1 \, M_1)}$.
  Then clearly $\varupdate{}{\tt@place{\epsilon, \ell_{k+1} \, \Xi}{C}}{\mset{m}} \in \tt@bsfun@proc(\epsilon\, M_{k+1} \, X_k \, M_k \cdots X_1 \, M_1)$ and we can deduce that
  $s' \is s^{\Pi} \ominus \varupdate{}{\tt@place{\spn{X}, \ell_{k+1} \, \Xi}{C}}{\mset{m}}
     \oplus \varupdate{}{\tt@place{\epsilon, \ell_{k+1} \, \Xi}{C}}{\mset{m}} \in \tt@bsfun@Pi(\epsilon\, M_{k+1} \, X_k \discn \, M_k \cdots X_1 \, M_1 \parallel \Pi_0)$.
  Clearly $s_X \in \tt@bsfun@proc(X)$ and thus we can deduce that
  $s' \oplus s_X \in \tt@bsfun@Pi(\epsilon\, M_{k+1} \, X_k \discn \, M_k \cdots X_1 \, M_1 \parallel \Pi_0 \parallel X)$

  We can then conclude that $s^1_{j^1} \in \tt@bsfun@Pi(\epsilon\, M_{k+1} \, X_k \, M_k\discn \cdots X_1 \, M_1 \parallel \Pi_0 \parallel X \ChanPar \Gamma)$.
  From our assumptions above this implies
  \begin{inparaenum}[(a)]
  \item $m^1 = j^1$, \item $m = 1$ and thus \item $\Pi' \ChanPar \Gamma' = \epsilon\, M_{k+1} \, X_k \, M_k\discn \cdots X_1 \, M_1 \parallel \Pi_0 \parallel X\ChanPar \Gamma$.
  \end{inparaenum}
  It is also easy to see that $\Pi \ChanPar \Gamma \to[TS=\calG] \Pi' \ChanPar \Gamma'$ and
  hence $\Pi \ChanPar \Gamma \to[label=\Pi' \ChanPar \Gamma', LTS=\calG] \Pi' \ChanPar \Gamma'$ which is what we wanted to prove.
\item \textbf{Rule} \Cref{altrule:dispatch:term}--\Cref{altrule:dispatch:nonterm}.\newline
We will give a proof of the case that the introduction of $\vec{r}^1(j^1)$ is justified by rule \Cref{altrule:dispatch:nonterm}.
The proof in the case that $\vec{r}^1(j^1)$ introduced to implement rule \Cref{altrule:dispatch:term}
is analogous.

We can thus assume that
$\vec{r}^1(j^1) = 
\commabr{((\tt@place{\epsilon, \ell_{k+1} \, \Xi}{C}, \tt@marking{\mathit{sim}}), (\tt@place{\epsilon, \ell_{k+1} \, \Xi}{C}, P, \tt@marking{\mathit{sim}})}
$ for some non-terminals and cache characters
 $\ell_{k+1} \, X_k \, \ell_k\discn \cdots X_1 \, \ell_1$ where $k \leq K$, $P = \set{\tt@place{k+1, e}{I} \mid e \in \ComSigma}$ and $\ell_{k+1} \in \set{-} \union \PPL$.

 Hence there exists a complex token $m$ located at place $\tt@place{\epsilon, \ell_{k+1} \, \Xi}{C}$
 and 
 $s^1_{j^1} = (s^{\Pi} \ominus \varupdate{}{\tt@place{\epsilon, \ell_{k+1} \, \Xi}{C}}{\mset{m}}
   \oplus \varupdate{}{\tt@place{\epsilon, \ell_{k+1} \, \Xi}{C}}{\mset{m_{\overline{P}}}}) 
   \oplus (\tt@bsfun@simple(\Gamma) 
   \oplus (m_{P} \compose \colmap^{-1}))
   \oplus \tt@marking{\mathit{sim}}$
   where $m_P = m \restriction P$ and $m_{\overline{P}} = m \restriction (\Pinner \setminus P)$.

  Since $s^{\Pi} \in \tt@bsfun@Pi(\Pi)$ we can deduce that 
  $\Pi = \epsilon\, M_{k+1} \, X_k \discn\, M_k \cdots X_1 \, M_1 \parallel \Pi_0$
  and $\commabr{\varupdate{}{\tt@place{\epsilon, \ell_{k+1} \, \Xi}{C}}{\mset{m}} \in \tt@bsfun@proc(\snd{\tt@chan}{\tt@msg}\, M_{k+1} \discn \, X_k \, M_k \cdots X_1 \, M_1)}$.
  We further know that $\commabr{m = \tt@ctoken@fun(M_1,\ldots,M_{k},M_{k+1})}$ and $M_{k+1} \in \MixedCache$. It is easy to see that
  $m_{\overline{P}} = \tt@ctoken@fun(M_1,\ldots,M_{k},M')$ where we write $M' = M_{k+1} \restriction (\ComN \union \PPL)$.
  
  Then clearly $\varupdate{}{\tt@place{\epsilon \, \ell_{k+1} \, \Xi}{C}}{\mset{m_{\overline{P}}}} \in \tt@bsfun@proc(\epsilon, M' \, X_k \, M_k \cdots X_1 \, M_1)$ and we can deduce that
  $s^{\Pi} \ominus \varupdate{}{\tt@place{\epsilon, \ell_{k+1} \, \Xi}{C}}{\mset{m}}
     \oplus \varupdate{}{\tt@place{\epsilon, \ell_{k+1} \, \Xi}{C}}{\mset{m_{\overline{P}}}} \in \tt@bsfun@Pi(\epsilon
     \, M' \, X_k \discn \, M_k \cdots X_1 \, M_1 \parallel \Pi_0)$.

  Inspecting case \textbf{Rules} \Cref{altrule:dispatch:term},\Cref{altrule:dispatch:nonterm} in the proof
   of $\tt@bisim$ is a simulation we can see that
$\paren{m_{P} \compose \colmap^{-1}} \restriction \calG^\Gamma = \tt@bsfun@simple\paren{\Gamma(M_{k+1})}$
and $\paren{m_{P} \compose \colmap^{-1}} \restriction \calG^{\nu} \in \tt@bsfun@Pi(\Pi(M_{k+1}))$
where $\calG^\Gamma = \varset{\tt@place{\tt@chan,\tt@msg}{S} \mid \tt@chan \in \Chan, \tt@msg \in \MMsg}$
and $\calG^\nu = \set{\tt@place{\spn{X}}{S} \mid X \in \NonT}$. Further
$\paren{m_{P} \compose \colmap^{-1}} \restriction \calG^\Gamma \oplus \paren{m_{P} \compose \colmap^{-1}} \restriction \calG^{\nu} = \paren{m_{P} \compose \colmap^{-1}}$.

  Hence we obtain
  $\tt@bsfun@simple(\Gamma) 
     \oplus \paren{m_{P} \compose \colmap^{-1}} \restriction \calG^\Gamma =
    \tt@bsfun@simple(\Gamma \oplus \Gamma(M_{k+1}))$.
  We can then conclude that $s^1_{j^1} \in \tt@bsfun@Pi(\epsilon\, M_{k+1} \, X_k \, M_k\discn \cdots X_1 \, M_1 \parallel \Pi_0 \parallel \Pi(M_{k+1}) \ChanPar \Gamma \oplus \Gamma(M_{k+1}))$.
  From our assumptions above this implies
  \begin{inparaenum}[(a)]
  \item $m^1 = j^1$, \item $m = 1$ and thus \item $\Pi' \ChanPar \Gamma' = \epsilon\, M_{k+1} \, X_k \, M_k\discn \cdots X_1 \, M_1 \parallel \Pi_0 \parallel \Pi(M_{k+1}) \ChanPar \Gamma \discn\oplus \Gamma(M_{k+1})$.
  \end{inparaenum}
  It is also easy to see that $\Pi \ChanPar \Gamma \to[TS=\calG] \Pi' \ChanPar \Gamma'$ and
  hence $\Pi \ChanPar \Gamma \to[label=\Pi' \ChanPar \Gamma', LTS=\calG] \Pi' \ChanPar \Gamma'$ which is what we wanted to prove.
\end{asparaitem}
\end{proof}

Let us now turn to $\tt@bisim^{-1}$.
Suppose $\commabr{\varparen{\Pi \ChanPar \Gamma,\tt@bsfun(\Pi \ChanPar \Gamma)} \in \tt@bisim}$ and 
$\tt@bsfun(\Pi \ChanPar \Gamma) \to[label={\Pi' \ChanPar \Gamma'},LTS={\P[\calN]},*] S$.
Hence
$\tt@bsfun(\Pi \ChanPar \Gamma) \to[label={\Pi' \ChanPar \Gamma'},LTS={\P[\calN]},*] \tt@bsfun(\Pi' \ChanPar \Gamma') \to[label=\epsilon,LTS={\P[\calN]},*] S$.
By definition this implies that there exists non-empty subsets of $\Config$, $\vec{S}(1),\ldots,\vec{S}(m)$ say,
such that  $m \geq 1$,
$\tt@bsfun(\Pi \ChanPar \Gamma) = \vec{S}(1)
\to[label=\epsilon,LTS={\P[\calN]}] \vec{S}(2)
\to[label=\epsilon,LTS={\P[\calN]}] \cdots
\to[label=\epsilon,LTS={\P[\calN]}] \vec{S}(m-1)
\to[label={\Pi' \ChanPar \Gamma'},LTS={\P[\calN]}] \vec{S}(m) = \tt@bsfun(\Pi' \ChanPar \Gamma')$
Thus for $i \in \range{m}$ we may pick configuration $\vec{s}(i) \in \vec{S}(i)$
such that for each $i \in \range{m-1}$
$\vec{s}(i) =s^i_0 \to[rule=\vec{r}^i(1),TS=\calN] s^i_1 \to[rule=\vec{r}^i(2),TS=\calN] \cdots 
\to[rule=\vec{r}^i(m^i-1),TS=\calN] s^i_{m^i-1} \to[rule=\vec{r}^i(m^i),TS=\calN] s^i_{m^i} = \vec{s}(i+1)$ 
where only one rule, $\vec{r}^i(j^i)$ say, is not a weak spawn rule.
We further know that for all $2 \leq i < m$ there \emph{does not} exist $\Pi_0 \ChanPar \Gamma_0$ such that
$\vec{s}(i) \in \tt@bsfun(\Pi_0 \ChanPar \Gamma_0)$.
Since clearly $\vec{s}(1) \in \tt@bsfun(\Pi \ChanPar \Gamma)$ and $\vec{s}(m) \in \tt@bsfun(\Pi' \ChanPar \Gamma')$ the Lemma above applies to give us $\Pi \ChanPar \Gamma \to[label=\Pi' \ChanPar \Gamma', LTS=\calG] \Pi' \ChanPar \Gamma'$.
We can thus conclude that $\tt@bisim^{-1}$ is a weak simulation and hence
$\tt@bisim$ is a weak bisimulation.

Further suppose that we have $s \in \tt@bsfun(\Pi \ChanPar \Gamma)$ for some $\Pi \ChanPar \Gamma$
and $s' \in \tt@bsfun(\Pi' \ChanPar \Gamma')$ and $s \to[TS=\calN,*] s'$ then
we can decompose this path into a sequence of configurations
$\vec{s}(1) = s \to[TS=\calN] \vec{s}(2) \to[TS=\calN] \cdots \to[TS=\calN] \vec{s}(l) = s'$.
Let $i_1,\ldots,i_{l'}$ be the indices, in order, such that  $\vec{s}(i_j) \in \tt@bsfun(\Pi_j \ChanPar \Gamma_j)$ for some $\Pi_1 \ChanPar \Gamma_1, \ldots, \Pi_{l'} \ChanPar \Gamma_{l'}$.
Clearly for all $1 \leq j \leq l'$ we can apply the Lemma above to the path $\vec{s}(i_j) \to[TS=\calN,*] \vec{s}(i_{j+1})$ to obtain 
$\Pi_1 \ChanPar \Gamma_1 \to[label=\Pi_2 \ChanPar \Gamma_2, LTS=\calG] \Pi_2 \ChanPar \Gamma_2 \to[label=\Pi_3 \ChanPar \Gamma_3, LTS=\calG] \cdots \to[label=\Pi_{l'} \ChanPar \Gamma_{l'}, LTS=\calG] \Pi_{l'} \ChanPar \Gamma_{l'}$
which means that since $\tt@bisim$ is a weak bisimulation that
$\tt@bsfun(\Pi \ChanPar \Gamma) \to[TS={\P[\calN]},*] \tt@bsfun(\Pi' \ChanPar \Gamma')$.
Hence we can conclude that if 
$s \in \tt@bsfun(\Pi \ChanPar \Gamma)$ for some $\Pi \ChanPar \Gamma$
and $s' \in \tt@bsfun(\Pi' \ChanPar \Gamma')$ and $s \to[TS=\calN,*] s'$ then
$\tt@bsfun(\Pi \ChanPar \Gamma) \to[TS={\P[\calN]},*] \tt@bsfun(\Pi' \ChanPar \Gamma')$.

We know that $\Pi_0 = A^0_1 \parallel \cdots \parallel A^0_{n'}$ so let
$s_{\mathit{cov}} = \tt@marking{\mathit{Query}} \oplus \varpreupdate{}{\tt@place{A_1}{\mathit{cov}} \mapsto \mset{\bullet}, \ldots, \tt@place{A_n}{\mathit{cov}} \mapsto \mset{\bullet}} \oplus \varupdate[\mid \tt@chan \in \Chan, \tt@msg \in {\Msg[]}]{}{\tt@place{\tt@chan, \tt@msg}{S}}{\varmset{\bullet^{\Gamma^\text{cov}(c)(m)}}}$
and $s_{0} = \tt@marking{\mathit{sim}} \oplus \varupdate[\mid i \in \range{n'}]{}{\tt@place{\spn{A^0_i}}{S}}{\varmset{\bullet}} \oplus \varupdate[\mid \tt@chan \in \Chan, \tt@msg \in {\Msg[]}]{}{\tt@place{\tt@chan, \tt@msg}{S}}{\varmset{\bullet^{\Gamma_0(c)(m)}}}$.
We note that $s_{0} \in \tt@bsfun(\Pi_0 \ChanPar \Gamma_0)$.

We will now show that $(\calN,s_0,s_{\mathit{cov}})$ is a yes-instance of the coverability problem if and only if $\commabr{(\calG,\Pi_0 \ChanPar \Gamma_0,A^{\text{cov}}_1 \parallel \cdots \parallel A^{\text{cov}}_n \ChanPar \Gamma^{\text{cov}})}$ is a yes-instance of simple coverability.

Suppose $(\calN,s_0,s_{\mathit{cov}})$ is a yes-instance of the coverability problem. Then
there exists a $s' \in \Config$ such that $s_0 \to[TS=\calN,*] s'$ and $s_{\mathit{cov}} \leqconfig s'$.
That means that there exists a configuration $s''$ such that
$s_0 \to[TS=\calN,*] s''$, $s'' \to[TS=\calN] \vec{s}(1) \to[TS=\calN] \cdots \to[TS=\calN] \vec{s}(l)$, $\vec{s}(l) = s'$, $s'' \in \tt@bsfun(\Pi \ChanPar \Gamma)$ for some $\Pi \ChanPar \Gamma$, and for each $1 \leq i \leq l$ there does not exists a $\bar{\Pi} \ChanPar \bar{\Gamma}$ such that $\vec{s}(i) \in \tt@bsfun(\bar{\Pi} \ChanPar \bar{\Gamma})$.
We can thus appeal to our reasoning above to obtain that
$\tt@bsfun(\Pi_0 \ChanPar \Gamma_0) \to[TS={\P[\calN]},*] \tt@bsfun(\Pi \ChanPar \Gamma)$.
We can decompose this path to get
$\tt@bsfun(\Pi_0 \ChanPar \Gamma_0) 
\to[label=\Pi_1 \ChanPar \Gamma_1,LTS={\P[\calN]},*] \tt@bsfun(\Pi_1 \ChanPar \Gamma_1)
\to[label=\Pi_2 \ChanPar \Gamma_2,LTS={\P[\calN]},*] \cdots
\to[label=\Pi_{l'} \ChanPar \Gamma_{l'},LTS={\P[\calN]},*] \tt@bsfun(\Pi_{l'} \ChanPar \Gamma_{l'}) = \tt@bsfun(\Pi \ChanPar \Gamma)$.
Since 
$\tt@bisim$ is a weak bisimulation we thus know that
$\Pi_0 \ChanPar \Gamma_0
\to[label=\Pi_1 \ChanPar \Gamma_1,LTS=\calG,*] \Pi_1 \ChanPar \Gamma_1
\to[label=\Pi_2 \ChanPar \Gamma_2,LTS=\calG,*] \cdots
\to[label=\Pi_{l'} \ChanPar \Gamma_{l'},LTS=\calG,*] \Pi_{l'} \ChanPar \Gamma_{l'} = \Pi \ChanPar \Gamma$
and thus $\Pi_0 \ChanPar \Gamma_0 \to[TS=\calG,*] \Pi \ChanPar \Gamma$.

Let us turn to an inspection of $s''$. We know that $s'' \in \tt@bsfun(\Pi \ChanPar \Gamma)$,
for all $1 \leq i \leq l$ there \emph{does not} exist $\bar{\Pi} \ChanPar \bar{\Gamma}$ such that
$\vec{s}(i) \in \tt@bsfun(\bar{\Pi} \ChanPar \bar{\Gamma})$, and
$s_{\mathit{cov}} \leqconfig s' = \vec{s}(l)$.
Since $\tt@marking{\mathit{Query}} \leqconfig s_{\mathit{cov}}$,
we can conclude that $\tt@marking{\mathit{Query}} \leqconfig s'$.
This implies that for some $0 \leq i < l$ we have the transition
$\vec{s}(j) \to[rule=r,TS=\calN] \vec{s}(j+1)$ via the rule $r = (\tt@marking{\mathit{sim}},\tt@marking{\mathit{Query}})$ which switches from simulation mode to query mode and where we write $\vec{s}(0) = s''$.
Inspecting the rules of $\calN$ we note that this mode switch cannot be reversed.
We can deduce that $j = 0$ since using any other rule of $\calN$ for the transition
$s'' \to[rule=r,TS=\calN] \vec{s}(1)$ would either place $\vec{s}(1) \in \tt@bsfun(\bar{\Pi} \ChanPar \bar{\Gamma})$ for some $\bar{\Pi} \ChanPar \bar{\Gamma}$ immediately, or start a CCFG computation disabling $r$ until completion which would produce another $\vec{s}(j) \in \tt@bsfun(\Pi_0 \ChanPar \Gamma_0)$ for some $\Pi_0 \ChanPar \Gamma_0$ before $r$ could be enabled.
Hence we can conclude that $s''(\tt@place{A_i}{\mathit{cov}}) = \emptyset$ for all $1 \leq i \leq n$ since 
$s'' \in \tt@bsfun(\Pi \ChanPar \Gamma)$ and for all $1 \leq j \leq l$
such that $\vec{s}(j) \to[rule=\vec{r}(j),TS=\calN] \vec{s}(j+1)$ the rule 
$\vec{r}(j)$ is a rule of the form 
$((\tt@place{A^{\text{cov}}_i, \ell_{k+1} \, \Xi}{C}, \tt@marking{\mathit{Query}}), (\tt@place{\epsilon, (-) \, \Xi}{C}, \emptyset, \tt@marking{\mathit{Query}} \oplus \varupdate{}{\tt@place{A_i}{\mathit{cov}}}{\mset{\bullet}}))$ or
$((\tt@place{\epsilon, A^{\text{cov}}_i \, \Xi}{C}, \tt@marking{\mathit{Query}}), (\tt@place{\epsilon, (-) \, \Xi}{C}, \emptyset, \tt@marking{\mathit{Query}} \oplus \varupdate{}{\tt@place{A_i}{\mathit{cov}}}{\mset{\bullet}}))$.
Since $s_{\mathit{cov}} \leqconfig s'$ this means that at least there are indices $i_1, \ldots, i_n$ (not necessarily in order)
where for each $1 \leq j \leq n$ rule $\vec{r}(i_j)$ places a $\bullet$-token into place $\tt@place{A_j}{\mathit{cov}}$.
Since the rules $\vec{r}(1),\ldots,\vec{r}(l)$ only remove complex tokens from places of the form
$\tt@place{A, \ell_{k+1} \, \Xi}{C}$ or $\tt@place{\epsilon, A \, \Xi}{C}$ we can conclude that there exists
complex tokens $m_1,...,m_n$ located in places $p_1,...,p_n$ in configuration $s''$
such that for all $1 \leq j \leq n$ the place $p_j = \tt@place{A_j, \ell_{k+1} \, \Xi}{C}$ or $\tt@place{\epsilon, A_j \, \Xi}{C}$. Since $s'' \in \tt@bsfun(\Pi \ChanPar \Gamma)$ this implies there exists
processes $\pi_1,\ldots,\pi_n$ such that $\Pi = \pi_1 \parallel \cdots \parallel \pi_n \parallel \Pi_0$ and
each $\pi_j = A^{\text{cov}}_j\cdot \gamma_j$ or $\epsilon\,M_j\, \cdot \gamma_j$ with $M_j(A^{\text{cov}}_j) > 0$ for $1 \leq j \leq n$. Further it is also easy to see that $\Gamma^{\text{cov}}$ is covered by $\Gamma$.
Hence we can deduce that $\commabr{(\calG,\Pi_0 \ChanPar \Gamma_0,A^{\text{cov}}_1 \parallel \cdots \parallel A^{\text{cov}}_n \ChanPar \Gamma^{\text{cov}})}$ is a yes-instances of simple coverability for the alternative semantics.

For the other direction of the reduction suppose $\commabr{(\calG,\Pi_0 \ChanPar \Gamma_0,A^{\text{cov}}_1 \parallel \cdots \parallel A^{\text{cov}}_n \ChanPar \Gamma^{\text{cov}})}$
is a yes-instances of simple coverability in the alternative semantics. This means that  
$\Pi_0 \ChanPar \Gamma_0
\to[label=\Pi_1 \ChanPar \Gamma_1,LTS=\calG] \Pi_1 \ChanPar \Gamma_1
\to[label=\Pi_2 \ChanPar \Gamma_2,LTS=\calG] \cdots
\to[label=\Pi_{l} \ChanPar \Gamma_{l},LTS=\calG] \Pi_{l} \ChanPar \Gamma_{l} =: \Pi \ChanPar \Gamma$
where $\Pi = \pi_1 \parallel \cdots \parallel \pi_n \parallel \Pi_0$ and
each $\pi_j = A^{\text{cov}}_j\cdot \gamma_j$ or $\epsilon\,M_j\, \cdot \gamma_j$ with $M_j(A^{\text{cov}}_j) > 0$ for $1 \leq j \leq n$.

Since $\tt@bisim$ is a weak bisimulation we obtain
$\tt@bsfun(\Pi_0 \ChanPar \Gamma_0) 
\to[label=\Pi_1 \ChanPar \Gamma_1,LTS={\P[\calN]},*] \tt@bsfun(\Pi_1 \ChanPar \Gamma_1)
\to[label=\Pi_2 \ChanPar \Gamma_2,LTS={\P[\calN]},*] \cdots
\to[label=\Pi_{l-1} \ChanPar \Gamma_{l-1},LTS={\P[\calN]},*] \tt@bsfun(\Pi_{l-1} \ChanPar \Gamma_{l-1}) \discn
\to[label=\Pi \ChanPar \Gamma,LTS={\P[\calN]},*] \tt@bsfun(\Pi \ChanPar \Gamma)$.

Thus we may pick $\vec{s}(1),\ldots,\vec{s}(l)$ such that
for all $1 \leq j \leq l$ we have $\vec{s}(j) \in \tt@bsfun(\Pi_{j} \ChanPar \Gamma_{j})$
and for $j < l$ we have $\vec{s}(j) \to[TS=\calN] \vec{s}(j+1)$.
Hence $\vec{s}(1) \to[TS=\calN,*] \vec{s}(l)$. Further since 
$\tt@bsfun(\Pi_0 \ChanPar \Gamma_0) 
\to[TS={\P[\calN]},*] \tt@bsfun(\Pi_1 \ChanPar \Gamma_1)$ and $s_0 \to[TS=\calN,*] s$ for all $s \in \tt@bsfun(S \ChanPar \emptyset)$ 
we can conclude that $s_0 \to[TS=\calN,*] \vec{s}(l) =: s'$.

Let us now inspect $s'$. Since $s' \in \tt@bsfun(\Pi \ChanPar \Gamma)$ we can deduce that
there exists complex tokens $m_1,\ldots,m_n$ located in places
$p_1,...,p_n$ in configuration $s'$
such that for all $1 \leq j \leq n$ the place $p_j = \tt@place{A^{\text{cov}}_j, \ell_{k+1} \, \Xi}{C}$ or $\tt@place{\epsilon, A^{\text{cov}}_j \, \Xi}{C}$. Hence we may extend the derivation $s_0 \to[TS=\calN,*] s' \to[rule=r,TS=\calN] s'_0 \to[rule=\vec{r}(1),TS=\calN] s'_1 \to[rule=\vec{r}(2),TS=\calN] \cdots \to[rule=\vec{r}(n),TS=\calN] s'_n$
where $r = (\tt@marking{\mathit{sim}},\tt@marking{\mathit{Query}})$ switches from simulation mode to query mode
and the rule 
$\vec{r}(j) =
((\tt@place{A^{\text{cov}}_j, \ell_{k+1} \, \Xi}{C}, \tt@marking{\mathit{Query}}), (\tt@place{\epsilon, (-) \, \Xi}{C}, \emptyset, \tt@marking{\mathit{Query}} \oplus \varupdate{}{\tt@place{A_j}{\mathit{cov}}}{\mset{\bullet}}))$ or
$((\tt@place{\epsilon, A^{\text{cov}}_j \, \Xi}{C}, \tt@marking{\mathit{Query}}), (\tt@place{\epsilon, (-) \, \Xi}{C}, \emptyset, \tt@marking{\mathit{Query}} \oplus \varupdate{}{\tt@place{A_j}{\mathit{cov}}}{\mset{\bullet}}))$
and thus we can deduce that $s'_n$ has one $\bullet$-token located in each place 
$\tt@place{A_1}{\mathit{cov}},\ldots,\tt@place{A_n}{\mathit{cov}}$ and 
$\tt@marking{\mathit{Query}}$ as a submarking, i.e.~$s_{\text{cov}} \leqconfig s'_n$.
Hence we can deduce that $(\calN,s_0,s_{\mathit{cov}})$ is a yes-instance of the coverability problem.

Hence we can conclude that the simple coverability problem for $K$-shaped APCPS in the alternative semantics \Exptime\ reduces to the 
coverability problem of NNCT.
\makeatother
\end{proof}
\renewcommand{\PPL}{\calL}
\begin{customtheorem}[\ref{thm:NNCT-to-APCPS}]
Simple coverability, boundedness and termination for a total-transfer NNCT
\Exptime\ reduces to simple coverability, boundedness and termination respectively for a $\mathit{4}$-shaped APCPS in the alternative semantics.
\end{customtheorem}
\begin{proof}
\makeatletter

\newcommand{\tf@nonterm}[2]{N_{#1}^{\scriptscriptstyle #2}}
\newcommand{\tf@chan}[1]{\mathit{c}_{#1}}
\newcommand{\tf@msg}[1]{\mathit{msg}_{#1}}
\newcommand{\tf@bisim}{\mathit{R}}
\newcommand{\tf@bsfun}{\calF}
\newcommand{\tf@bsfun@chan}{\calF^{\Gamma}}
\newcommand{\tf@bsfun@proc}{\calF^{\Pi}}
\newcommand{\tf@smp@chan}{\widetilde{\Gamma}}

Fix a total-transfer NNCT $\calN = (\Psimple,\Pcomplex,\Pinner,\Rules,\colmap)$ and
an instance of coverability $\paren{\calN,s_0,s_{\text{cov}}}$ with a simple query, i.e.~$s_0(p) = s_{\text{cov}}(p) = \emptyset$ for
all $p \in \Pcomplex$. Let us also fix an enumeration of $\Psimple = \set{\enumelem{p}{1},\ldots,\enumelem{p}{\nsimple}}$, $\Pcomplex = \set{\enumelem{p'}{1},\ldots,\enumelem{p'}{\ncomplex}}$ and
$\Pinner = \set{\enumelem{p^{\I}}{1},\ldots,\enumelem{p^{\I}}{\ninner}}$.
Since $\calN$ is a total-transfer NNCT we know that for all $r \in \TransferRules$ it is the case that $r = ((p,I),(p',\Pinner,O))$.

Let us define an APCPS $\calG=(\Sigma, I, \NonT, \Rules, S)$ that will simulate $\calN$.
The APCPS $\calG$ has a channel $\tf@chan{p}$ for each simple or complex place $p$ plus a special channel $\tf@chan{?}$ and let the set of channels be $\Chan = \set{\tf@chan{p} \mid p \in \Psimple \union \Pcomplex} \union \set{\tf@chan{?}}$. 
For $\calG$'s messages let us first inspect $\calN$'s complex rules:
let $\Xi$ be the set of complex tokens representing ``injected'' coloured tokens in $\calN$'s complex rules
plus the set of complex tokens that are created by simple rules, i.e.~$\Xi = \set{\calc \mid ((p,I),(p',\calc,O)) \in \ComplexRules} \union \varset{\calc \mid p \in \Pcomplex, (I,O) \in \SimpleRules,  O(p)(\calc) \geq 1}$ then $\calG$ will have a message $m_{p,d}$ for each $p \in \Pcomplex$ and $d \in \Xi$
and two special messages $m_{p,\emptyset}, m_{p,\downarrow}$ for each $p \in \Pcomplex$. Further $\calG$ will have a message $\bullet$; and hence we can let the set of messages be 
$\MMsg = \set{\tf@msg{p,d} \mid p \in \Pcomplex, d \in \Xi \union \set{\emptyset,\downarrow}} \union \set{\bullet}$.

The APCPS $\calG$ will have non-terminals $\tf@nonterm{r}{}$, $\tf@nonterm{r}{I}$ for each rule $r \in \Rules$,
non-terminals $\tf@nonterm{r}{\S,O}$, $\tf@nonterm{r}{\C,O}$ for each $r \in \SimpleRules$,
and non-terminals $\tf@nonterm{r}{O}$, $\tf@nonterm{r}{\C}$ for each $r \in \ComplexRules \union \TransferRules$.
Additionally $\calG$ will have non-terminals $\tf@nonterm{}{\spn{p}}$, $\tf@nonterm{}{p}$
for every $p \in \Pcomplex$, a non-terminal $\tf@nonterm{}{\calc}$ for each $\calc \in \Xi$
and two special non-terminals
$\tf@nonterm{}{\spn{?}}$ and $\tf@nonterm{}{\mathit{sim}}$. Hence $\calG$'s set of non-terminals are
\begin{align*}
\NonT =& \set{\tf@nonterm{}{\spn{p}}, \tf@nonterm{}{p} \mid p \in \Pcomplex} 
	     \union \set{\tf@nonterm{}{\calc} \mid \calc \in \Xi} 
	     \union \set{\tf@nonterm{}{\spn{?}}, \tf@nonterm{}{\mathit{sim}}} \\	   
	   & \union \set{\tf@nonterm{r}{}, \tf@nonterm{r}{I} \mid r \in \Rules} 
	     \union \set{\tf@nonterm{r}{\S,O}, \tf@nonterm{r}{\C,O} \mid r \in \SimpleRules} \\
	   & \union \set{\tf@nonterm{r}{O}, \tf@nonterm{r}{\C} \mid r \in \ComplexRules \union \TransferRules} 
\end{align*}
Let us also define a designated non-terminal $A_{\text{cov}}$ to label the coverability query. We can then set $\calG$'s alphabet $\Sigma = \varset{\snd{\tf@chan{}}{\tf@msg{}},\rec{\tf@chan{}}{\tf@msg{}}, \spn{X} \mid \tf@chan{} \in \Chan, \tf@msg{} \in \MMsg, X \in \NonT}$.
Further let us use the standard independence relation for APCPS $I$ over $\Sigma \union \NonT$.

A configuration $s \in \Config$ will be represented by an APCPS configuration in the following way:
\begin{itemize}
\item for each simple place $p \in \Psimple$, the channel $\tf@chan{p}$ will contain $|s(p)|$  $\bullet$-messages --- 
	we can formalise this by a function $\tf@bsfun@chan$ that we define as
	$\tf@bsfun@chan(s) 	= \cchan{\bullet^{|s(\enumelem{p}{1})|}}{\tf@chan{\enumelem{p}{1}}}, \ldots, \cchan{\bullet^{|s(\enumelem{p}{\nsimple})|}}{\tf@chan{\enumelem{p}{\nsimple}}}$;
\item for each complex place $p \in \Pcomplex$, suppose $s(p) = \mset{m_1,\ldots,m_k}$ then for each $m_i$ there will be one process with a process-state $\tf@nonterm{}{p} \cdot \tf@smp@chan(m_i \circ \colmap^{-1}) \cdot \tf@nonterm{}{\spn{?}}$ where $\tf@smp@chan$ is a function that takes a mapping $m_0 : \Psimple \to \N$
and transforms it into a multiset $\M[\snd{\tf@chan{\Psimple}}{\bullet}]$ by
$\tf@smp@chan(m_0)(\snd{\tf@chan{p}}{\bullet}) = |m_0(p)|$ ---
we can formalise this with two functions 
$\tf@bsfun@proc(s,p) =  \Parallel_{i=1}^k \tf@nonterm{}{p} \cdot \tf@smp@chan(m_i \circ \colmap^{-1}) \cdot \tf@nonterm{}{\spn{?}}$ and $\tf@bsfun@proc(s) = \Parallel_{i=1}^{\ncomplex} \tf@bsfun@proc(s,\enumelem{p'}{i})$;
\item in addition we have one administrative process, that implements the execution of $\calN$'s rules and
is in the process-state $\tf@nonterm{}{\mathit{sim}}$ when representing the configuration $s$.
\end{itemize}
Formally, we define a representation function $\tf@bsfun$ by 
$\tf@bsfun(s) 		= \tf@nonterm{}{\mathit{sim}} \parallel \tf@bsfun@proc(s) \ChanPar \tf@bsfun@chan(s)$
and define the relation $\tf@bisim \subseteq \Config \times \APCPSConfig$
by $\tf@bisim 			= \set{(s,\tf@bsfun(s)) \mid s \in \Configuration}$.
We will prove that
$\tf@bisim$ is a weak bisimulation for $(\Config,\to[TS=\calN])$ and $(\commabr{\APCPSConfig,\to[TS=\calG]})$
where we write $\to[TS=\calG]$ for $\toCM$ here and in the following.

We will now define and explain how the administrative process is implementing $\calN$'s rules.

For each rule $r \in \Rules$ the APCPS $\calG$ has the rule
$\tf@nonterm{}{\mathit{sim}} \to \tf@nonterm{r}{} \tf@nonterm{}{\mathit{sim}}$
which guesses one of $\calN$'s rules to execute next.
Depending on whether $r$ is a simple, complex or transfer rule the implementation is different:
\begin{itemize}
\item $r = (I,O) \in \SimpleRules$ \newline
If $r$ is a simple rule, then $\tf@nonterm{r}{}$ rewrites to three non-terminals:
$\tf@nonterm{r}{I}$ which removes $\bullet$-messages as described by $I$,
$\tf@nonterm{r}{O,S}$ which sends $\bullet$-messages as described by the effect of $O$ on $\calN$'s simple places and
$\tf@nonterm{r}{O,C}$ which spawns new processes which will represent the newly created complex tokens as described by $O$ on $\calN$'s complex places. Let us enumerate the set 
$\Xi = \set{\enumelem{\calc}{1},\ldots,\enumelem{\calc}{n_{\Xi}}}$ then we can
formally implement the above by the following rules:
\begin{align*}
\tf@nonterm{r}{} 		\to &\tf@nonterm{r}{I} \cdot \tf@nonterm{r}{\S,O} \cdot \tf@nonterm{r}{\C,O}\\
\tf@nonterm{r}{I} 		\to &\paren{\rec{\tf@chan{\enumelem{p}{1}}}{\bullet}}^{|I(\enumelem{p}{1})|} \cdots 
		   					 \paren{\rec{\tf@chan{\enumelem{p}{\nsimple}}}{\bullet}}^{|I(\enumelem{p}{\nsimple})|}\\
\tf@nonterm{r}{\S,O} 	\to &\paren{\snd{\tf@chan{\enumelem{p}{1}}}{\bullet}}^{|O(\enumelem{p}{1})|} \cdots 
		 					 \paren{\snd{\tf@chan{\enumelem{p}{\nsimple}}}{\bullet}}^{|O(\enumelem{p}{\nsimple})|}\\
\tf@nonterm{r}{\C,O}	\to &\paren{\spn{\tf@nonterm{}{\spn{\enumelem{p'}{1}}}}
									\paren{\snd{\tf@chan{?}}{\tf@msg{\enumelem{p'}{1},\enumelem{\calc}{1}}}} 
									\paren{\rec{\tf@chan{?}}{\bullet}}
								   }^{O(\enumelem{p'}{1})(\enumelem{\calc}{1})} \cdots \\
							&\paren{\spn{\tf@nonterm{}{\spn{\enumelem{p'}{i}}}}
									\paren{\snd{\tf@chan{?}}{\tf@msg{\enumelem{p'}{i},\enumelem{\calc}{j}}}} 
									\paren{\rec{\tf@chan{?}}{\bullet}}
								   }^{O(\enumelem{p'}{i})(\enumelem{\calc}{j})} \cdots \\	   
							&\paren{\spn{\tf@nonterm{}{\spn{\enumelem{p'}{\ncomplex}}}}
									\paren{\snd{\tf@chan{?}}{\tf@msg{\enumelem{p'}{\ncomplex},\enumelem{\calc}{n_{\Xi}}}}} 
									\paren{\rec{\tf@chan{?}}{\bullet}}
								   }^{O(\enumelem{p'}{\ncomplex})(\enumelem{\calc}{n_{\Xi}})}
\shortintertext{and the rules for complex token representing processes we have the following rule for each $\calc \in \Xi$ and $p \in \Pcomplex$:}
\tf@nonterm{}{\spn{p}} \to& \paren{\rec{\tf@chan{?}}{\tf@msg{p,\calc}}} \paren{\snd{\tf@chan{?}}{\bullet}} \tf@nonterm{}{p} \cdot \tf@nonterm{}{\calc} \cdot \tf@nonterm{}{\spn{?}}
\shortintertext{and for each $\calc \in \Xi$:}
\tf@nonterm{}{\calc} \to &\paren{\snd{\tf@chan{\colmap\paren{\enumelem{p^{\I}}{1}}}}{\bullet}}^{|\calc(\enumelem{p^{\I}}{1})|} \cdots 
		   	   \paren{\snd{\tf@chan{\colmap\paren{\enumelem{p^{\I}}{\ninner}}}}{\bullet}}^{|\calc(\enumelem{p^{\I}}{\ninner})|}
\end{align*}
\item $r = ((p,I),(p',\calc,O)) \in \ComplexRules$ \newline
If $r$ is a complex rule, then $\tf@nonterm{r}{}$ rewrites to three non-terminals:
$\tf@nonterm{r}{I}$ which removes $\bullet$-messages as described by $I$,
$\tf@nonterm{r}{O}$ which sends $\bullet$-messages as described by $O$ and
$\tf@nonterm{r}{C}$ which forces one process representing one complex token $m$ in complex place $p$
to change state so that the process afterwards represents $m \oplus \calc$ in complex place $p'$.
Formally we can implement this with the following rules:
\begin{align*}
\tf@nonterm{r}{} 	\to &\tf@nonterm{r}{I} \cdot \tf@nonterm{r}O \cdot \tf@nonterm{r}{\C}\\
\tf@nonterm{r}{I} 	\to &\paren{\rec{\tf@chan{\enumelem{p}{1}}}{\bullet}}^{|I(\enumelem{p}{1})|} \cdots 
		   				 \paren{\rec{\tf@chan{\enumelem{p}{\nsimple}}}{\bullet}}^{|I(\enumelem{p}{\nsimple})|}\\
\tf@nonterm{r}{O} 	\to &\paren{\snd{\tf@chan{\enumelem{p}{1}}}{\bullet}}^{|O(\enumelem{p}{1})|} \cdots 
		   				 \paren{\snd{\tf@chan{\enumelem{p}{\nsimple}}}{\bullet}}^{|O(\enumelem{p}{\nsimple})|}\\
\tf@nonterm{r}{\C} 	\to &(\snd{\tf@chan{p}}{\tf@msg{p',\calc}}) \cdot (\rec{\tf@chan{p}}{\bullet})
\shortintertext{and for each $\calc \in \Xi$ and $p \in \Pcomplex$ we have the following rule for complex token representing processes:}
\tf@nonterm{}{p} \to& (\rec{\tf@chan{p}}{\tf@msg{p',\calc}}) \cdot (\snd{\tf@chan{p}}{\bullet}) \cdot \tf@nonterm{}{p'} \tf@nonterm{}{\calc}
\end{align*}
\item $r = ((p,I),(p',\Pinner,O)) \in \TransferRules$ \newline
If $r$ is a transfer rule, then $\tf@nonterm{r}{}$ rewrites to three non-terminals:
$\tf@nonterm{r}{I}$ which removes $\bullet$-messages as described by $I$,
$\tf@nonterm{r}{O}$ which sends $\bullet$-messages as described by $O$ and
$\tf@nonterm{r}{C}$ which forces one process representing one complex token $m$ in complex place $p$
to change state and
make its summary effective, which transfers $(m \circ \colmap^{-1})$ to the channels,
so that the process afterwards represents the empty complex token $\emptyset$ in complex place $p'$.
Formally we can implement this with the following rules:
\begin{align*}
\tf@nonterm{r}{}	\to &\tf@nonterm{r}{I} \cdot \tf@nonterm{r}{O} \cdot \tf@nonterm{r}{\C}\\
\tf@nonterm{r}{I} 	\to &\paren{\rec{\tf@chan{\enumelem{p}{1}}}{\bullet}}^{|I(\enumelem{p}{1})|} \cdots 
		   				 \paren{\rec{\tf@chan{\enumelem{p}{\nsimple}}}{\bullet}}^{|I(\enumelem{p}{\nsimple})|}\\
\tf@nonterm{r}{O} 	\to &\paren{\snd{\tf@chan{\enumelem{p}{1}}}{\bullet}}^{|O(\enumelem{p}{1})|} \cdots 
		   				 \paren{\snd{\tf@chan{\enumelem{p}{\nsimple}}}{\bullet}}^{|O(\enumelem{p}{\nsimple})|}\\
\tf@nonterm{r}{\C} 	\to &(\snd{\tf@chan{p}}{\tf@msg{p,\downarrow}}) \cdot (\rec{\tf@chan{p}}{\bullet}) 
						 \cdot (\snd{\tf@chan{?}}{\tf@msg{p',\emptyset}}) \cdot (\rec{\tf@chan{?}}{\bullet})
\shortintertext{and for each $p \in \Pcomplex$ we have the following rule for complex token representing processes:}
\tf@nonterm{}{p} \to& (\rec{\tf@chan{p}}{\tf@msg{p,\downarrow}}) \cdot (\snd{\tf@chan{p}}{\bullet})\\
\tf@nonterm{}{\spn{?}} \to& (\rec{\tf@chan{?}}{\tf@msg{p,\emptyset}})\cdot (\snd{\tf@chan{?}}{\bullet}) \cdot \tf@nonterm{}{p} \cdot \tf@nonterm{}{\spn{?}}
\end{align*}
\end{itemize}

Further we add two more rules:
\begin{gather*}
S \to \paren{\snd{\tf@chan{\enumelem{p}{1}}}{\bullet}}^{|s_{0}(\enumelem{p}{1})|} \cdots 
               \paren{\snd{\tf@chan{\enumelem{p}{\nsimple}}}{\bullet}}^{|s_{0}(\enumelem{p}{\nsimple})|} \tf@nonterm{}{\mathit{sim}}\\
\tf@nonterm{}{\mathit{sim}} \to \paren{\rec{\tf@chan{\enumelem{p}{1}}}{\bullet}}^{|s_{\text{cov}}(\enumelem{p}{1})|} \cdots 
                \paren{\rec{\tf@chan{\enumelem{p}{\nsimple}}}{\bullet}}^{|s_{\text{cov}}(\enumelem{p}{\nsimple})|} A_{\text{cov}}.
\end{gather*}
The first rule simply sets up the simulation from $s_0$.
The second rule has the purpose to encode $\calN$'s coverability query with the intention that 
a configuration $A_{\text{cov}} \parallel \Pi \ChanPar \Gamma$ is only reachable if and only if
$s_{\text{cov}}$ is coverable.

First, let us note that clearly $\calG$ can be constructed from $\calN$ in \Exptime.
Also all non-terminals except for the $\tf@nonterm{r}{O}$, $\tf@nonterm{r}{O,S}$ and $\tf@nonterm{}{\calc}$
are non-commutative. It is easy to see that $\calG$ has shaped stacks since the only non-terminal loop increasing
the ``call-stack'' is in the rule
$$\tf@nonterm{}{p} \to (\rec{\tf@chan{p}}{\tf@msg{p',\calc}}) \cdot (\snd{\tf@chan{p}}{\bullet}) \cdot \tf@nonterm{}{p'} \tf@nonterm{}{\calc}$$
and $\tf@nonterm{}{\calc}$ is a commutative non-terminal. It is easy to see that picking $K = 4$ is adequate.

Before we go on to prove $\tf@bisim$ is a weak bisimulation let us first analyse
$\tf@bsfun@proc$ and $\tf@bsfun@chan$.

It is easy to see that for all $s,s' \in \Config$ we have
\begin{gather*}
 \tf@bsfun@proc(s) \parallel \tf@bsfun@proc(s') = \tf@bsfun@proc(s \oplus s'),\\
 \tf@bsfun@chan(s) \oplus \tf@bsfun@chan(s') = \tf@bsfun@chan(s \oplus s'),
\shortintertext{if for all $p \in \Psimple$ we have $s'(p) = \emptyset$ then} 
\tf@bsfun@chan(s) = \tf@bsfun@chan(s \oplus s') = \tf@bsfun@chan(s \ominus s')
\shortintertext{and if for all $p \in \Pcomplex$ we have $s'(p) = \emptyset$ then} 
\tf@bsfun@proc(s) = \tf@bsfun@proc(s \oplus s') = \tf@bsfun@proc(s \ominus s')
\end{gather*} 

Let us label APCPS and NNCT transitions in the following way:
for $s,s' \in \Configuration$ such that $s \to_{\calN} s'$ let us label $\to_{\calN}$ such that
we write $s \to[label=s',LTS=\calN] s'$;
for $\Pi \ChanPar \Gamma,\Pi' \ChanPar \Gamma' \in \APCPSConfig$  such that 
$\Pi \ChanPar \Gamma \to[TS={\calG}] \Pi' \ChanPar \Gamma'$ and $\exists s_0 \in \Configuration$ such that $\tf@bsfun(s_0) = \Pi' \ChanPar \Gamma'$ let us label
$\to[LTS={\calG}]$ such that we write $\Pi \ChanPar \Gamma \to[label=s_0,LTS=\calG] \Pi' \ChanPar \Gamma'$;
if however $\nexists s_0 \in \Configuration$ such that $\tf@bsfun(s_0) = s'$ then let us label
$\to[TS={\calG}]$ such that we write $\Pi \ChanPar \Gamma \to[label=\epsilon,LTS=\calG] \Pi' \ChanPar \Gamma'$.

Let us clarify the definition of a weak bisimulation for a labelled transition system.
\begin{definition*}[(Weak) simulation, bisimulation]
Suppose $(S,\to_S)$ and $(S',\to_{S'})$ are labelled transition systems we say a relation $\mathfrak{R} \subseteq S \times S'$ is a (weak)
\emph{simulation} just if for all $(s,s') \in \mathfrak{R}$, if for some $t \in S$ we have $s \to[label=\alpha,LTS=S] t$ ($s \to[label=\alpha,LTS=S,*] t$) then there exists 
$t' \in S'$ such that $s' \to[label=\alpha,LTS=S'] t'$ ($s' \to[label=\alpha,LTS=S',*] t'$) and $(t,t') \in \mathfrak{R}$. 
We say $\mathfrak{R}$ is a (weak) \emph{bisimulation}
just if both $\mathfrak{R}$ and $\mathfrak{R}^{-1}$ are (weak) simulations.
\end{definition*}

Let us clarify some notation: if $\Gamma \in (\Chan \rightarrow \M[\MMsg])$ such that
for $\Gamma', \Gamma'' \in (\Chan \rightarrow \M[\MMsg])$ we can decompose $\Gamma$ such that $\Gamma = \Gamma' \oplus \Gamma'' $ then $\Gamma' = \Gamma \ominus \Gamma''$.

Let us now prove that $\tf@bisim$ above is a weak simulation. So suppose
we have $(s,\tf@bsfun(s)) \in \tf@bisim$ and $s \to[label=s_0,LTS=\calN,*] s'$ then from the labelling 
we clearly have $s \to[label=s',LTS=\calN] s'$. Hence there exists a $r \in \Rules$ such that
$s \to[rule=r,TS=\calN] s'$. We will perform a case analysis on $r$.
\begin{asparaitem}
\item \emph{Case: } $r \in \SimpleRules$. \newline
Then $r = (I,O)$ and $s' = (s \ominus I) \oplus O$ and also $s\ominus I \in \Config$.
We know that $\tf@bsfun(s) = \tf@nonterm{}{\mathit{sim}} \parallel \tf@bsfun@proc(s) \ChanPar \tf@bsfun@chan(s)$ and 
since $s\ominus I \in \Config$ we have that $|s(p_0)| \geq |I(p_0)|$ for all $p_0 \in \Psimple$.
Hence by definition 
$$\tf@bsfun@chan(s) = \Gamma \oplus \cchan{\bullet^{|I(\enumelem{p}{1})|}}{\tf@chan{\enumelem{p}{1}}} \oplus \cdots \oplus \cchan{\bullet^{|I(\enumelem{p}{\nsimple})|}}{\tf@chan{\enumelem{p}{\nsimple}}}$$
for some $\Gamma \in (\Chan \rightarrow \M[\MMsg])$.
Hence we can see that we can perform the following transitions:
\begin{align*}
\tf@bsfun(s) &\to[label=\epsilon,LTS=\calG] \tf@nonterm{r}{} \tf@nonterm{}{\mathit{sim}} \parallel \tf@bsfun@proc(s) \ChanPar \tf@bsfun@chan(s)\\
			 &\to[label=\epsilon,LTS=\calG] \tf@nonterm{r}{I} \tf@nonterm{r}{O,S} \tf@nonterm{r}{O,C} \tf@nonterm{}{\mathit{sim}} \parallel \tf@bsfun@proc(s) \ChanPar \tf@bsfun@chan(s)\\
			 &\to[label=\epsilon,LTS=\calG] \paren{\rec{\tf@chan{\enumelem{p}{1}}}{\bullet}}^{|I(\enumelem{p}{1})|} \cdots 
		   					 \paren{\rec{\tf@chan{\enumelem{p}{\nsimple}}}{\bullet}}^{|I(\enumelem{p}{\nsimple})|}\\
		   	 &\hspace{2.5em}	 \tf@nonterm{r}{O,S} \tf@nonterm{r}{O,C} \tf@nonterm{}{\mathit{sim}} \parallel \tf@bsfun@proc(s) \ChanPar \tf@bsfun@chan(s)\\
		   	 &\to[label=\epsilon,LTS=\calG] \paren{\rec{\tf@chan{\enumelem{p}{1}}}{\bullet}}^{|I(\enumelem{p}{1})|-1} \cdots 
		   					 \paren{\rec{\tf@chan{\enumelem{p}{\nsimple}}}{\bullet}}^{|I(\enumelem{p}{\nsimple})|}\\
		   	 &\hspace{2.5em}	 \tf@nonterm{r}{O,S} \tf@nonterm{r}{O,C} \tf@nonterm{}{\mathit{sim}} \parallel \tf@bsfun@proc(s) \ChanPar \tf@bsfun@chan(s) \ominus \cchan{\bullet}{\tf@chan{\enumelem{p}{1}}}\\
		   	 &\to[label=\epsilon,LTS=\calG] \cdots\\
		   	 &\to[label=\epsilon,LTS=\calG] \paren{\rec{\tf@chan{\enumelem{p}{2}}}{\bullet}}^{|I(\enumelem{p}{2})|} \cdots 
		   					 \paren{\rec{\tf@chan{\enumelem{p}{\nsimple}}}{\bullet}}^{|I(\enumelem{p}{\nsimple})|}\\
		   	 &\hspace{2.5em}	 \tf@nonterm{r}{O,S} \tf@nonterm{r}{O,C} \tf@nonterm{}{\mathit{sim}} \parallel \tf@bsfun@proc(s) \ChanPar \tf@bsfun@chan(s) \ominus \cchan{\bullet^{|I(\enumelem{p}{1})|}}{\tf@chan{\enumelem{p}{1}}}\\
		   	 &\to[label=\epsilon,LTS=\calG] \cdots\\
		   	 &\to[label=\epsilon,LTS=\calG]	\tf@nonterm{r}{O,S} \tf@nonterm{r}{O,C} \tf@nonterm{}{\mathit{sim}} \parallel \tf@bsfun@proc(s) \ChanPar \tf@bsfun@chan(s) \ominus \cchan{\bullet^{|I(\enumelem{p}{1})|}}{\tf@chan{\enumelem{p}{1}}}\\
		   	 &\hspace{2.5em}	\ominus \cchan{\bullet^{|I(\enumelem{p}{2})|}}{\tf@chan{\enumelem{p}{2}}} \ominus \cdots 
		   	 \ominus \cchan{\bullet^{|I(\enumelem{p}{\nsimple})|}}{\tf@chan{\enumelem{p}{\nsimple}}}\\
		   	 &\hspace{2.5em}=\tf@nonterm{r}{O,S} \tf@nonterm{r}{O,C} \tf@nonterm{}{\mathit{sim}} \parallel \tf@bsfun@proc(s) \ChanPar \tf@bsfun@chan(s \ominus I)
\end{align*}
We can expand $\tf@nonterm{r}{O,S}$ and perform the send actions similarly to the way we expanded
$\tf@nonterm{r}{I}$ and its receive actions.
\begin{align*}
\tf@bsfun(s) &\to[label=\epsilon,LTS=\calG,*]  
				\paren{\snd{\tf@chan{\enumelem{p}{1}}}{\bullet}}^{|O(\enumelem{p}{1})|} \cdots 
		 		\paren{\snd{\tf@chan{\enumelem{p}{\nsimple}}}{\bullet}}^{|O(\enumelem{p}{\nsimple})|}\\
			&\hspace{2.5em}  		   
				\tf@nonterm{r}{O,C} \tf@nonterm{}{\mathit{sim}} \parallel 
				\tf@bsfun@proc(s) \ChanPar \tf@bsfun@chan(s \ominus I)\\
			&\to[label=\epsilon,LTS=\calG,*]   
				\tf@nonterm{r}{O,C} \tf@nonterm{}{\mathit{sim}} \parallel 
				\tf@bsfun@proc(s) \ChanPar \tf@bsfun@chan(s \ominus I) \oplus 
				\cchan{\bullet^{|O(\enumelem{p}{1})|}}{\tf@chan{\enumelem{p}{1}}} \\
			&\hspace{2.5em} 		   
				\oplus \cdots \oplus \cchan{\bullet^{|O(\enumelem{p}{\nsimple})|}}{\tf@chan{\enumelem{p}{\nsimple}}}\\
			&\hspace{2.5em}
				= \tf@nonterm{r}{O,C} \tf@nonterm{}{\mathit{sim}} \parallel 
				  \tf@bsfun@proc(s) \ChanPar \tf@bsfun@chan(s \ominus I \oplus (O \restriction \Psimple))
\end{align*}
We continue with expanding $\tf@nonterm{r}{O,C}$. Similarly to above we can then perform all the spawns and synchronisations.
\begin{align*}
\tf@bsfun(s) &\to[label=\epsilon,LTS=\calG,*] 
				\paren{\spn{\tf@nonterm{}{\spn{\enumelem{p'}{1}}}}
						\paren{\snd{\tf@chan{?}}{\tf@msg{\enumelem{p'}{1},\enumelem{\calc}{1}}}} 
						\paren{\rec{\tf@chan{?}}{\bullet}}
					   }^{O(\enumelem{p'}{1})(\enumelem{\calc}{1})} \cdots \\
			&\hspace{2.5em}  
				\paren{\spn{\tf@nonterm{}{\spn{\enumelem{p'}{i}}}}
						\paren{\snd{\tf@chan{?}}{\tf@msg{\enumelem{p'}{i},\enumelem{\calc}{j}}}} 
						\paren{\rec{\tf@chan{?}}{\bullet}}
					   }^{O(\enumelem{p'}{i})(\enumelem{\calc}{j})} \cdots \\	   
			&\hspace{2.5em}  
				\paren{\spn{\tf@nonterm{}{\spn{\enumelem{p'}{\ncomplex}}}}
						\paren{\snd{\tf@chan{?}}{\tf@msg{\enumelem{p'}{\ncomplex},\enumelem{\calc}{n_{\Xi}}}}} 
						\paren{\rec{\tf@chan{?}}{\bullet}}
					   }^{O(\enumelem{p'}{\ncomplex})(\enumelem{\calc}{n_{\Xi}})}\\
			&\hspace{2.5em}  
				\tf@nonterm{}{\mathit{sim}} \parallel \tf@bsfun@proc(s) 
				\ChanPar \tf@bsfun@chan(s \ominus I \oplus (O \restriction \Psimple))\\
			&\to[label=\epsilon,LTS=\calG]
				\paren{\snd{\tf@chan{?}}{\tf@msg{\enumelem{p'}{1},\enumelem{\calc}{1}}}} 
					\paren{\rec{\tf@chan{?}}{\bullet}}\\
			&\hspace{2.5em}  		
					\paren{\spn{\tf@nonterm{}{\spn{\enumelem{p'}{1}}}}
						\paren{\snd{\tf@chan{?}}{\tf@msg{\enumelem{p'}{1},\enumelem{\calc}{1}}}} 
						\paren{\rec{\tf@chan{?}}{\bullet}}
					   }^{O(\enumelem{p'}{1})(\enumelem{\calc}{1})-1} \cdots \\
			&\hspace{2.5em}  
				\paren{\spn{\tf@nonterm{}{\spn{\enumelem{p'}{i}}}}
						\paren{\snd{\tf@chan{?}}{\tf@msg{\enumelem{p'}{i},\enumelem{\calc}{j}}}} 
						\paren{\rec{\tf@chan{?}}{\bullet}}
					   }^{O(\enumelem{p'}{i})(\enumelem{\calc}{j})} \cdots \\	   
			&\hspace{2.5em}  
				\paren{\spn{\tf@nonterm{}{\spn{\enumelem{p'}{\ncomplex}}}}
						\paren{\snd{\tf@chan{?}}{\tf@msg{\enumelem{p'}{\ncomplex},\enumelem{\calc}{n_{\Xi}}}}} 
						\paren{\rec{\tf@chan{?}}{\bullet}}
					   }^{O(\enumelem{p'}{\ncomplex})(\enumelem{\calc}{n_{\Xi}})} \\
			 &\hspace{2.5em}  
				\tf@nonterm{}{\mathit{sim}} 
				\parallel \tf@nonterm{}{\spn{\enumelem{p'}{1}}}
				\parallel \tf@bsfun@proc(s) 
				\ChanPar \tf@bsfun@chan(s \ominus I \oplus (O \restriction \Psimple))
\end{align*}
Using the rule $\tf@nonterm{}{\spn{\enumelem{p'}{1}}} \to \paren{\rec{\tf@chan{?}}{\tf@msg{p,\enumelem{\calc}{1}}}} \paren{\snd{\tf@chan{?}}{\bullet}} \tf@nonterm{}{\enumelem{p'}{1}} \cdot \tf@nonterm{}{\enumelem{\calc}{1}} \cdot \tf@nonterm{}{\spn{?}}$
we can derive
\begin{align*}
\tf@nonterm{}{\enumelem{\calc}{1}} \to[*] \paren{\snd{\tf@chan{\colmap\paren{\enumelem{p^{\I}}{1}}}}{\bullet}}^{|\enumelem{\calc}{1}(\enumelem{p^{\I}}{1})|} \cdots 
            \paren{\snd{\tf@chan{\colmap\paren{\enumelem{p^{\I}}{\ninner}}}}{\bullet}}&
            \vphantom{\paren{\snd{\tf@chan{\colmap\paren{\enumelem{p^{\I}}{\ninner}}}}{\bullet}}}^{|\enumelem{\calc}{1}(\enumelem{p^{\I}}{\ninner})|} \\
            &=: w
\end{align*}
and the equality $\M(w) = \tf@smp@chan(\enumelem{\calc}{1} \circ \colmap^{-1})$ holds. Thus
\begin{align*}				
\tf@bsfun(s) &\to[label=\epsilon,LTS=\calG,*]
				\paren{\snd{\tf@chan{?}}{\tf@msg{\enumelem{p'}{1},\enumelem{\calc}{1}}}} 
					\paren{\rec{\tf@chan{?}}{\bullet}}\\
			&\hspace{1.5em}  		
					\paren{\spn{\tf@nonterm{}{\spn{\enumelem{p'}{1}}}}
						\paren{\snd{\tf@chan{?}}{\tf@msg{\enumelem{p'}{1},\enumelem{\calc}{1}}}} 
						\paren{\rec{\tf@chan{?}}{\bullet}}
					   }^{O(\enumelem{p'}{1})(\enumelem{\calc}{1})-1} \cdots \\
			&\hspace{1.5em}  
				\paren{\spn{\tf@nonterm{}{\spn{\enumelem{p'}{i}}}}
						\paren{\snd{\tf@chan{?}}{\tf@msg{\enumelem{p'}{i},\enumelem{\calc}{j}}}} 
						\paren{\rec{\tf@chan{?}}{\bullet}}
					   }^{O(\enumelem{p'}{i})(\enumelem{\calc}{j})} \cdots \\	   
			&\hspace{1.5em}  
				\paren{\spn{\tf@nonterm{}{\spn{\enumelem{p'}{\ncomplex}}}}
						\paren{\snd{\tf@chan{?}}{\tf@msg{\enumelem{p'}{\ncomplex},\enumelem{\calc}{n_{\Xi}}}}} 
						\paren{\rec{\tf@chan{?}}{\bullet}}
					   }^{O(\enumelem{p'}{\ncomplex})(\enumelem{\calc}{n_{\Xi}})} \\
			 &\hspace{1.5em}  
				\tf@nonterm{}{\mathit{sim}} 
				\parallel \paren{\rec{\tf@chan{?}}{\tf@msg{p,\calc}}} \paren{\snd{\tf@chan{?}}{\bullet}} \tf@nonterm{}{\enumelem{p'}{1}} \cdot \tf@smp@chan(\enumelem{\calc}{1} \circ \colmap^{-1}) \cdot \tf@nonterm{}{\spn{?}}\\
			&\hspace{1.5em}  	
				\parallel \tf@bsfun@proc(s) 
				\ChanPar \tf@bsfun@chan(s \ominus I \oplus (O \restriction \Psimple))\\
			&\to[label=\epsilon,LTS=\calG,*]				
					\paren{\spn{\tf@nonterm{}{\spn{\enumelem{p'}{1}}}}
						\paren{\snd{\tf@chan{?}}{\tf@msg{\enumelem{p'}{1},\enumelem{\calc}{1}}}} 
						\paren{\rec{\tf@chan{?}}{\bullet}}
					   }^{O(\enumelem{p'}{1})(\enumelem{\calc}{1})-1} \cdots \\
			&\hspace{1.5em}  
				\paren{\spn{\tf@nonterm{}{\spn{\enumelem{p'}{i}}}}
						\paren{\snd{\tf@chan{?}}{\tf@msg{\enumelem{p'}{i},\enumelem{\calc}{j}}}} 
						\paren{\rec{\tf@chan{?}}{\bullet}}
					   }^{O(\enumelem{p'}{i})(\enumelem{\calc}{j})} \cdots \\	   
			&\hspace{1.5em}  
				\paren{\spn{\tf@nonterm{}{\spn{\enumelem{p'}{\ncomplex}}}}
						\paren{\snd{\tf@chan{?}}{\tf@msg{\enumelem{p'}{\ncomplex},\enumelem{\calc}{n_{\Xi}}}}} 
						\paren{\rec{\tf@chan{?}}{\bullet}}
					   }^{O(\enumelem{p'}{\ncomplex})(\enumelem{\calc}{n_{\Xi}})} \\
			 &\hspace{1.5em}  
				\tf@nonterm{}{\mathit{sim}} 
				\parallel \tf@nonterm{}{\enumelem{p'}{1}} \cdot \tf@smp@chan(\enumelem{\calc}{1} \circ \colmap^{-1}) \cdot \tf@nonterm{}{\spn{?}}\\
			&\hspace{1.5em}  	
				\parallel \tf@bsfun@proc(s) 
				\ChanPar \tf@bsfun@chan(s \ominus I \oplus (O \restriction \Psimple))\\
\shortintertext{we can continue like this:}		
\tf@bsfun(s) &\to[label=\epsilon,LTS=\calG,*]				
					\paren{\spn{\tf@nonterm{}{\spn{\enumelem{p'}{1}}}}
						\paren{\snd{\tf@chan{?}}{\tf@msg{\enumelem{p'}{1},\enumelem{\calc}{2}}}} 
						\paren{\rec{\tf@chan{?}}{\bullet}}
					   }^{O(\enumelem{p'}{1})(\enumelem{\calc}{2})} \cdots \\
			&\hspace{1.5em}  
				\paren{\spn{\tf@nonterm{}{\spn{\enumelem{p'}{i}}}}
						\paren{\snd{\tf@chan{?}}{\tf@msg{\enumelem{p'}{i},\enumelem{\calc}{j}}}} 
						\paren{\rec{\tf@chan{?}}{\bullet}}
					   }^{O(\enumelem{p'}{i})(\enumelem{\calc}{j})} \cdots \\	   
			&\hspace{1.5em}  
				\paren{\spn{\tf@nonterm{}{\spn{\enumelem{p'}{\ncomplex}}}}
						\paren{\snd{\tf@chan{?}}{\tf@msg{\enumelem{p'}{\ncomplex},\enumelem{\calc}{n_{\Xi}}}}} 
						\paren{\rec{\tf@chan{?}}{\bullet}}
					   }^{O(\enumelem{p'}{\ncomplex})(\enumelem{\calc}{n_{\Xi}})} \\
			 &\hspace{1.5em}  
				\tf@nonterm{}{\mathit{sim}} 
				\parallel \Parallel_{i=1}^{O(\enumelem{p'}{1}{\enumelem{\calc}{1}})} \tf@nonterm{}{\enumelem{p'}{1}} \cdot \tf@smp@chan(\enumelem{\calc}{1} \circ \colmap^{-1}) \cdot \tf@nonterm{}{\spn{?}}\\
			&\hspace{1.5em}  	
				\parallel \tf@bsfun@proc(s) 
				\ChanPar \tf@bsfun@chan(s \ominus I \oplus (O \restriction \Psimple))\\		
\tf@bsfun(s) &\to[label=\epsilon,LTS=\calG,*]				
					\paren{\spn{\tf@nonterm{}{\spn{\enumelem{p'}{2}}}}
						\paren{\snd{\tf@chan{?}}{\tf@msg{\enumelem{p'}{2},\enumelem{\calc}{1}}}} 
						\paren{\rec{\tf@chan{?}}{\bullet}}
					   }^{O(\enumelem{p'}{2})(\enumelem{\calc}{1})} \cdots \\
			&\hspace{1.5em}  
				\paren{\spn{\tf@nonterm{}{\spn{\enumelem{p'}{i}}}}
						\paren{\snd{\tf@chan{?}}{\tf@msg{\enumelem{p'}{i},\enumelem{\calc}{j}}}} 
						\paren{\rec{\tf@chan{?}}{\bullet}}
					   }^{O(\enumelem{p'}{i})(\enumelem{\calc}{j})} \cdots \\	   
			&\hspace{1.5em}  
				\paren{\spn{\tf@nonterm{}{\spn{\enumelem{p'}{\ncomplex}}}}
						\paren{\snd{\tf@chan{?}}{\tf@msg{\enumelem{p'}{\ncomplex},\enumelem{\calc}{n_{\Xi}}}}} 
						\paren{\rec{\tf@chan{?}}{\bullet}}
					   }^{O(\enumelem{p'}{\ncomplex})(\enumelem{\calc}{n_{\Xi}})} \\
			 &\hspace{1.5em}  
				\tf@nonterm{}{\mathit{sim}} 
				\parallel \Parallel_{i=1}^{O(\enumelem{p'}{1}{\enumelem{\calc}{1}})} \tf@nonterm{}{\enumelem{p'}{1}} \cdot \tf@smp@chan(\enumelem{\calc}{1} \circ \colmap^{-1}) \cdot \tf@nonterm{}{\spn{?}}\\
			&\hspace{1.5em}  	
				\parallel \cdots\\
			&\hspace{1.5em}
				\parallel \Parallel_{i=1}^{O(\enumelem{p'}{1}{\enumelem{\calc}{n_{\Xi}}})} \tf@nonterm{}{\enumelem{p'}{1}} \cdot \tf@smp@chan(\enumelem{\calc}{n_{\Xi}} \circ \colmap^{-1}) \cdot \tf@nonterm{}{\spn{?}}\\
			&\hspace{1.5em}	
				\parallel \tf@bsfun@proc(s) 
				\ChanPar \tf@bsfun@chan(s \ominus I \oplus (O \restriction \Psimple))\\						
			&\hspace{1.5em} =
					\paren{\spn{\tf@nonterm{}{\spn{\enumelem{p'}{2}}}}
						\paren{\snd{\tf@chan{?}}{\tf@msg{\enumelem{p'}{2},\enumelem{\calc}{1}}}} 
						\paren{\rec{\tf@chan{?}}{\bullet}}
					   }^{O(\enumelem{p'}{2})(\enumelem{\calc}{1})} \cdots \\
			&\hspace{1.5em}  
				\paren{\spn{\tf@nonterm{}{\spn{\enumelem{p'}{i}}}}
						\paren{\snd{\tf@chan{?}}{\tf@msg{\enumelem{p'}{i},\enumelem{\calc}{j}}}} 
						\paren{\rec{\tf@chan{?}}{\bullet}}
					   }^{O(\enumelem{p'}{i})(\enumelem{\calc}{j})} \cdots \\	   
			&\hspace{1.5em}  
				\paren{\spn{\tf@nonterm{}{\spn{\enumelem{p'}{\ncomplex}}}}
						\paren{\snd{\tf@chan{?}}{\tf@msg{\enumelem{p'}{\ncomplex},\enumelem{\calc}{n_{\Xi}}}}} 
						\paren{\rec{\tf@chan{?}}{\bullet}}
					   }^{O(\enumelem{p'}{\ncomplex})(\enumelem{\calc}{n_{\Xi}})} \\
			 &\hspace{1.5em}  
				\tf@nonterm{}{\mathit{sim}} 
				\parallel \tf@bsfun@proc(O,\enumelem{p'}{1})
				\parallel \tf@bsfun@proc(s) 
				\ChanPar \tf@bsfun@chan(s \ominus I \oplus (O \restriction \Psimple))\\	
			&\to[label=\alpha,LTS=\calG,*]
				\tf@nonterm{}{\mathit{sim}} \parallel \tf@bsfun@proc(s) 
				\parallel \paren{\Parallel_{i=1}^{\ncomplex} \tf@bsfun@proc(O,\enumelem{p'}{i})}\\
			&\hspace{2.5em}  	
				\ChanPar \tf@bsfun@chan(s \ominus I \oplus (O \restriction \Psimple))\\
			&\hspace{1.5em} = 
				\tf@nonterm{}{\mathit{sim}} \parallel \tf@bsfun@proc(s) 
				\parallel \tf@bsfun@proc(O \restriction \Pcomplex)\\
			&\hspace{2.5em}  
				\ChanPar \tf@bsfun@chan(s \ominus I \oplus (O \restriction \Psimple))\\
\end{align*}
It just remains to analyse the last configuration:
\begin{gather*}
\begin{aligned}
\tf@bsfun@proc(s) \parallel \tf@bsfun@proc(O \restriction \Pcomplex) &= \tf@bsfun@proc(s) 
				\parallel \tf@bsfun@proc(O) = \tf@bsfun@proc(s \oplus O) \\
				&= \tf@bsfun@proc(s \ominus I \oplus O) = \tf@bsfun@proc(s')
\end{aligned}\\
\shortintertext{ and }
\tf@bsfun@chan(s \ominus I \oplus (O \restriction \Psimple)) = \tf@bsfun@chan(s \ominus I \oplus O) = \tf@bsfun@chan(s')  
\end{gather*}  
from which we can deduce $\tf@bsfun(s) \to[label=\alpha,LTS=\calG,*] \tf@bsfun(s')$, $(s',\tf@bsfun(s')) \in \tf@bisim$ and clearly $\alpha = s'$
which is what we wanted to prove.
\item \emph{Case: } $r \in \ComplexRules$. \newline
	Then $r = ((p,I),(p',\oplus,O))$ and $s' = (s \ominus I \ominus \update{}{p}{\mset{m}}) \oplus O \oplus \update{}{p'}{\mset{m \oplus \calc}}$ for some $m \in s(p)$.
Again $\tf@bsfun(s) = \tf@nonterm{}{\mathit{sim}} \parallel \tf@bsfun@proc(s) \ChanPar \tf@bsfun@chan(s)$ and 
similarly to the case above
\begin{align*}
\tf@bsfun(s) &\to[label=\epsilon,LTS=\calG,*] 
				  \tf@nonterm{r}{C} \tf@nonterm{}{\mathit{sim}} \parallel 
				  \tf@bsfun@proc(s) \ChanPar \tf@bsfun@chan(s \ominus I \oplus O)\\
			 &\to[label=\epsilon,LTS=\calG]	  
			 	  (\snd{\tf@chan{p}}{\tf@msg{p',\calc}})  \cdot (\rec{\tf@chan{p}}{\bullet}) \tf@nonterm{}{\mathit{sim}} \parallel 
				  \tf@bsfun@proc(s) \ChanPar \tf@bsfun@chan(s \ominus I \oplus O)\\
			 &\to[label=\epsilon,LTS=\calG]	  
			 	  (\rec{\tf@chan{p}}{\bullet}) \tf@nonterm{}{\mathit{sim}} \parallel 
				  \tf@bsfun@proc(s) \ChanPar \tf@bsfun@chan(s \ominus I \oplus O) \oplus
				  \cchan{\tf@msg{p',\calc}}{\tf@chan{p}}\\	
			 &\hspace{1.5em} = 
			 	  (\rec{\tf@chan{p}}{\bullet}) \tf@nonterm{}{\mathit{sim}} 
			 	  \parallel \tf@bsfun@proc(\update{}{p}{\mset{m}}) 
			 	  \parallel \tf@bsfun@proc(s \ominus \update{}{p}{\mset{m}}) \\
			 &\hspace{2.5em}
			 	  \ChanPar \tf@bsfun@chan(s \ominus I \oplus O) \oplus
				  \cchan{\tf@msg{p',\calc}}{\tf@chan{p}}\\	
			 &\hspace{1.5em} = 
			 	  (\rec{\tf@chan{p}}{\bullet}) \tf@nonterm{}{\mathit{sim}} 
			 	  \parallel \tf@nonterm{}{p} \tf@smp@chan\paren{m \circ \colmap^{-1}}\tf@nonterm{}{\spn{?}}\\
			 &\hspace{2.5em}	  
			 	  \parallel \tf@bsfun@proc(s \ominus \update{}{p}{\mset{m}}) \\
			 &\hspace{2.5em}
			 	  \ChanPar \tf@bsfun@chan(s \ominus I \oplus O) \oplus
				  \cchan{\tf@msg{p',\calc}}{\tf@chan{p}}
\end{align*}
It is immediate that we can derive
$$\tf@nonterm{}{\calc} \to[*] \paren{\snd{\tf@chan{\colmap\paren{\enumelem{p^{\I}}{1}}}}{\bullet}}^{|\calc(\enumelem{p^{\I}}{1})|} \cdots 
		   	   \paren{\snd{\tf@chan{\colmap\paren{\enumelem{p^{\I}}{\ninner}}}}{\bullet}}^{|\calc(\enumelem{p^{\I}}{\ninner})|} =: w$$
and that the equality $\M(w) = \tf@smp@chan(\calc \circ \colmap^{-1})$ holds. Thus
\begin{align*}	
\tf@bsfun(s)&\to[label=\epsilon,LTS=\calG,*]	  
				(\rec{\tf@chan{p}}{\bullet}) \tf@nonterm{}{\mathit{sim}} \\
			&\hspace{2.5em}	  
			 	  \parallel (\rec{\tf@chan{p}}{\tf@msg{p',\calc}}) \cdot (\snd{\tf@chan{p}}{\bullet}) \cdot \tf@nonterm{}{p'} \tf@smp@chan\paren{m \circ \colmap^{-1}} \oplus \tf@smp@chan\paren{\calc \circ \colmap^{-1}} \tf@nonterm{}{\spn{?}}\\
			 &\hspace{2.5em}	  
			 	  \parallel \tf@bsfun@proc(s \ominus \update{}{p}{\mset{m}}) \\
			 &\hspace{2.5em}
			 	  \ChanPar \tf@bsfun@chan(s \ominus I \oplus O) \oplus
				  \cchan{\tf@msg{p',\calc}}{\tf@chan{p}}\\	
			 &\hspace{1.5em} =
				(\rec{\tf@chan{p}}{\bullet}) \tf@nonterm{}{\mathit{sim}} \\
			&\hspace{2.5em}	  
			 	  \parallel (\rec{\tf@chan{p}}{\tf@msg{p',\calc}}) \cdot (\snd{\tf@chan{p}}{\bullet}) \cdot \tf@nonterm{}{p'} \tf@smp@chan\paren{(m \oplus \calc) \circ \colmap^{-1}} \tf@nonterm{}{\spn{?}}\\
			 &\hspace{2.5em}	  
			 	  \parallel \tf@bsfun@proc(s \ominus \update{}{p}{\mset{m}}) \\
			 &\hspace{2.5em}
			 	  \ChanPar \tf@bsfun@chan(s \ominus I \oplus O) \oplus
				  \cchan{\tf@msg{p',\calc}}{\tf@chan{p}}\\	
			&\to[label=\alpha,LTS=\calG,*]	  
				\tf@nonterm{}{\mathit{sim}} 
			 	\parallel \tf@nonterm{}{p'} \paren{(m \oplus \calc) \circ \colmap^{-1} \circ \tf@smp@chan} \tf@nonterm{}{\spn{?}}\\
			 &\hspace{2.5em}	  
			 	  \parallel \tf@bsfun@proc(s \ominus \update{}{p}{\mset{m}}) \\
			 &\hspace{2.5em}
			 	  \ChanPar \tf@bsfun@chan(s \ominus I \oplus O) \\
			 &\hspace{1.5em} =	
			 	\tf@nonterm{}{\mathit{sim}} 
			 	  \parallel \tf@bsfun@proc(s \ominus \update{}{p}{\mset{m}} \oplus \update{}{p'}{\mset{m \oplus \calc}}) \\
			 &\hspace{2.5em}
			 	  \ChanPar \tf@bsfun@chan(s \ominus I \oplus O) 
\end{align*}
Lastly, we clearly have
\begin{gather*}
\begin{aligned}
&\tf@bsfun@proc(s \ominus \update{}{p}{\mset{m}} \oplus \update{}{p'}{\mset{m \oplus \calc}}) \\
&\;\;\;\; = \tf@bsfun@proc(s \ominus \update{}{p}{\mset{m}} \ominus I \oplus \update{}{p'}{\mset{m \oplus \calc}} \oplus O) =
\tf@bsfun@proc(s')
\end{aligned}\\
\begin{aligned}
&\tf@bsfun@chan(s \ominus I \oplus O) \\
&\;\;\;\; = \tf@bsfun@chan(s \ominus \update{}{p}{\mset{m}} \ominus I \oplus \update{}{p'}{\mset{m \oplus \calc}} \oplus O) = \tf@bsfun@chan(s')
\end{aligned}
\end{gather*}
from which we can deduce $\tf@bsfun(s) \to[label=\alpha,LTS=\calG,*] \tf@bsfun(s')$, $(s',\tf@bsfun(s')) \in \tf@bisim$ and clearly $\alpha = s'$
which is what we wanted to prove.
\item $r \in \TransferRules$. \newline
	Then $r = ((p,I),(p',\Pinner,O))$ since $r$ is total and 
	$s' = (s \ominus I \ominus \update{}{p}{\mset{m}}) 
		  \oplus O 
		  \oplus \update{}{p'}{\mset{\emptyset}} 
		  \oplus (m \circ \colmap^{-1})$ for some $m \in s(p)$.
Analogously to the cases above:
\begin{align*}
\tf@bsfun(s) &\to[label=\epsilon,LTS=\calG,*] 
				  \tf@nonterm{r}{C} \tf@nonterm{}{\mathit{sim}} \parallel 
				  \tf@bsfun@proc(s) \ChanPar \tf@bsfun@chan(s \ominus I \oplus O)\\
			 &\to[label=\epsilon,LTS=\calG] 
			 	  (\snd{\tf@chan{p}}{\tf@msg{p,\downarrow}}) \cdot (\rec{\tf@chan{p}}{\bullet}) 
				  \cdot (\snd{\tf@chan{?}}{\tf@msg{p',\emptyset}}) \cdot (\rec{\tf@chan{?}}{\bullet})
				  \tf@nonterm{}{\mathit{sim}} \\
			&\hspace{2.5em}
				  \parallel 
				  \tf@bsfun@proc(s) \ChanPar \tf@bsfun@chan(s \ominus I \oplus O)\\
			&\to[label=\epsilon,LTS=\calG,*]
				  (\snd{\tf@chan{p}}{\tf@msg{p,\downarrow}}) \cdot (\rec{\tf@chan{p}}{\bullet}) 
				  \cdot (\snd{\tf@chan{?}}{\tf@msg{p',\emptyset}}) \cdot (\rec{\tf@chan{?}}{\bullet})
				  \tf@nonterm{}{\mathit{sim}} \\
			&\hspace{2.5em}
				  \parallel \tf@nonterm{}{p} \tf@smp@chan\paren{m \circ \colmap^{-1}}\tf@nonterm{}{\spn{?}}\\
			&\hspace{2.5em}	  
			 	  \parallel \tf@bsfun@proc(s \ominus \update{}{p}{\mset{m}}) \\
			&\hspace{2.5em}
			 	  \ChanPar \tf@bsfun@chan(s \ominus I \oplus O)\\
			&\to[label=\epsilon,LTS=\calG]
				  (\snd{\tf@chan{p}}{\tf@msg{p,\downarrow}}) \cdot (\rec{\tf@chan{p}}{\bullet}) 
				  \cdot (\snd{\tf@chan{?}}{\tf@msg{p',\emptyset}}) \cdot (\rec{\tf@chan{?}}{\bullet})
				  \tf@nonterm{}{\mathit{sim}} \\
			&\hspace{2.5em}
				  \parallel (\rec{\tf@chan{p}}{\tf@msg{p',\downarrow}}) \cdot (\snd{\tf@chan{p}}{\bullet})
				  \tf@smp@chan\paren{m \circ \colmap^{-1}}\tf@nonterm{}{\spn{?}}\\
			&\hspace{2.5em}	  
			 	  \parallel \tf@bsfun@proc(s \ominus \update{}{p}{\mset{m}}) \\
			&\hspace{2.5em}
			 	  \ChanPar \tf@bsfun@chan(s \ominus I \oplus O)\\
			&\to[label=\epsilon,LTS=\calG,*]
				  (\snd{\tf@chan{?}}{\tf@msg{p',\emptyset}}) \cdot (\rec{\tf@chan{?}}{\bullet})
				  \tf@nonterm{}{\mathit{sim}} \\
			&\hspace{2.5em}
				  \parallel \tf@smp@chan\paren{m \circ \colmap^{-1}}\tf@nonterm{}{\spn{?}}\\
			&\hspace{2.5em}	  
			 	  \parallel \tf@bsfun@proc(s \ominus \update{}{p}{\mset{m}}) 
			 	  \ChanPar \tf@bsfun@chan(s \ominus I \oplus O)\\
			&\to[label=\epsilon,LTS=\calG]
				  (\snd{\tf@chan{?}}{\tf@msg{p',\emptyset}}) \cdot (\rec{\tf@chan{?}}{\bullet})
				  \tf@nonterm{}{\mathit{sim}} \\
			&\hspace{2.5em}
				  \parallel \tf@nonterm{}{\spn{?}}
			 	  \parallel \tf@bsfun@proc(s \ominus \update{}{p}{\mset{m}}) \\
			&\hspace{2.5em}
			 	  \ChanPar \tf@bsfun@chan(s \ominus I \oplus O) \oplus \Gamma\paren{\tf@smp@chan(m \circ \colmap^{-1})}
\end{align*}
We can see that 
$$\tf@bsfun@chan(s \ominus I \oplus O) \oplus \Gamma\paren{\tf@smp@chan(m \circ \colmap^{-1})} = 
\tf@bsfun@chan(s \ominus I \oplus O \oplus m \circ \colmap^{-1})$$
and so
\begin{align*}
\tf@bsfun(s)&\to[label=\epsilon,LTS=\calG,*]
				  (\snd{\tf@chan{?}}{\tf@msg{p',\emptyset}}) \cdot (\rec{\tf@chan{?}}{\bullet})
				  \tf@nonterm{}{\mathit{sim}} \\
			&\hspace{2.5em}
				  \parallel (\rec{\tf@chan{?}}{\tf@msg{p,\emptyset}})\cdot (\snd{\tf@chan{?}}{\bullet}) \cdot \tf@nonterm{}{p'} \cdot \tf@nonterm{}{\spn{?}} \\
			&\hspace{2.5em}	  
			 	  \parallel \tf@bsfun@proc(s \ominus \update{}{p}{\mset{m}}) \\
			&\hspace{2.5em}
			 	  \ChanPar \tf@bsfun@chan(s \ominus I \oplus O \oplus (m \circ \colmap^{-1}))\\
			&\to[label=\alpha,LTS=\calG,*]
				  \tf@nonterm{}{\mathit{sim}}
				  \parallel \tf@nonterm{}{p'} \cdot \tf@nonterm{}{\spn{?}} 
			 	  \parallel \tf@bsfun@proc(s \ominus \update{}{p}{\mset{m}}) \\
			&\hspace{2.5em}
			 	  \ChanPar \tf@bsfun@chan(s \ominus I \oplus O \oplus (m \circ \colmap^{-1})) \\
			&\hspace{1.5em} = \tf@nonterm{}{\mathit{sim}}
			 	  \parallel \tf@bsfun@proc(s \ominus \update{}{p}{\mset{m}} \oplus \update{}{p'}{\mset{\emptyset}}) \\
			&\hspace{2.5em}
			 	  \ChanPar \tf@bsfun@chan(s \ominus I \oplus O \oplus (m \circ \colmap^{-1}))
\end{align*}
Analysing the last configuration:
\begin{gather*}
\begin{aligned}
&\tf@bsfun@proc(s \ominus \update{}{p}{\mset{m}} \oplus \update{}{p'}{\mset{\emptyset}}) \\
&\;\;\;\; = \tf@bsfun@proc(s \ominus \update{}{p}{\mset{m}} \ominus I \oplus \update{}{p'}{\mset{\emptyset}} \oplus O \oplus (m \circ \colmap^{-1})) \\
&\;\;\;\; = \tf@bsfun@proc(s')
\end{aligned}\\
\begin{aligned}
&\tf@bsfun@chan(s \ominus I \oplus O \oplus (m \circ \colmap^{-1})) \\
&\;\;\;\; = \tf@bsfun@chan(s \ominus \update{}{p}{\mset{m}} \ominus I \oplus \update{}{p'}{\mset{\emptyset}} \oplus O \oplus (m \circ \colmap^{-1})) \\
&\;\;\;\; = \tf@bsfun@chan(s')
\end{aligned}
\end{gather*}
from which we can deduce $\tf@bsfun(s) \to[label=\alpha,LTS=\calG,*] \tf@bsfun(s')$, $(s',\tf@bsfun(s')) \in \tf@bisim$ and clearly $\alpha = s'$
which is what we wanted to prove.
\end{asparaitem}
Hence we can conclude that $\tf@bisim$ is a weak simulation.

Let us now investigate $\tf@bisim^{-1}$.
Suppose
we have $(s,\tf@bsfun(s)) \in \tf@bisim$ and $F(s) \to[label=s',LTS=\calG,*] t$ then from the labelling we know we have
$\tf@bsfun(s) \to[label=s',LTS=\calG,*] \tf@bsfun(s') \to[label=\epsilon,LTS=\calG,*] t$.
Inspecting the rules of $\calG$ we can see that the only paths (modulo different interleavings)
to make the transitions $\tf@bsfun(s) \to[label=s',LTS=\calG,*] \tf@bsfun(s')$ are the three described in the proof that 
$\tf@bisim$ is a simulation and depends on which rule $r$ of $\calN$ is simulated.

Inspecting the transition paths we can see that if $r \in \SimpleRules$
then $s' = s \ominus I \oplus O$ and since $\tf@nonterm{r}{I}$ can be fully executed we can deduce that 
$|s(p)| \geq |I(p)|$ for all $p \in \Psimple$ and hence $s \ominus I \in \Config$.
Thus we can conclude that $s \to[rule=r,TS=\calN] s'$ and hence $s \to[label=s',LTS=\calN] s'$.

If $r \in \ComplexRules$ then we can see that $(s \ominus I \ominus \update{}{p}{\mset{m}}) \oplus O \oplus \update{}{p'}{\mset{m \oplus \calc}}$ for some $m \in s(p)$ and as above since $\tf@nonterm{r}{I}$ can be fully executed we can deduce that $s \ominus I \in \Config$ which implies $s \to[rule=r,TS=\calN] s'$ and hence $s \to[label=s',LTS=\calN] s'$.

For the third case if $r \in \ComplexRules$ inspecting the transition path above we can see that
$s' = s \ominus \update{}{p}{\mset{m}} \ominus I \oplus \update{}{p'}{\mset{\emptyset}} \oplus O \oplus (m \circ \colmap^{-1})$ for some $m \in s(p)$ as above since $\tf@nonterm{r}{I}$ can be fully executed we can deduce that $s \ominus I \in \Config$. Hence $s \to[rule=r,TS=\calN] s'$ and hence $s \to[label=s',LTS=\calN] s'$.

Thus we can conclude that $\tf@bisim^{-1}$ is a weak simulation and thus $\tf@bisim$ is a weak bisimulation.

In order to finish the proof we will now prove that a check of coverability of $s_{\text{cov}}$ for $\calN$
is equivalent to a program-point coverability check of $l_{\text{cov}}$ for $\calG$.

First suppose $(\calN, s_0, s_{\text{cov}})$ is a yes-instance of NNCT coverability with a simple query.
Hence there exists a path $s_0 \to[TS=\calN] s_1 \to[TS=\calN] \cdots \to[TS=\calN] s_n$ and $s_{\text{cov}} \leqconfig s_n$. Attaching the labels to the transition system as we have described above this path turns into
$s_0 \to[label=s_1,LTS=\calN] s_1 \to[label=s_2,LTS=\calN] \cdots \to[label=s_n,LTS=\calN] s_n$.
Since $\tf@bisim$ is a weak bisimulation we know that we can find a path
$S \to[label=s_0,LTS=\calG,*] \tf@bsfun(s_0) \to[label=s_1,LTS=\calG,*] \tf@bsfun(s_1) \to[label=s_2,LTS=\calG,*] \cdots \to[label=s_n,LTS=\calG,*] \tf@bsfun(s_n)$.
We know that for all $p \in \Psimple$ it is the case that $|s_{\text{cov}}(p)| \leq |s_n(p)|$ hence 
$s_n \ominus (s_{\text{cov}} \restriction \Psimple) \in \APCPSConfig$.
Thus
\begin{align*}
\tf@bsfun(s_n) & =  \tf@nonterm{}{\mathit{sim}} \parallel \tf@bsfun@proc(s_n) \ChanPar \tf@bsfun@chan(s_n)\\
			   &\to[label=\epsilon,LTS=\calG]
			   		\paren{\rec{\tf@chan{\enumelem{p}{1}}}{\bullet}}^{|s_{\text{cov}}(\enumelem{p}{1})|} \cdots 
		   				 \paren{\rec{\tf@chan{\enumelem{p}{\nsimple}}}{\bullet}}^{|s_{\text{cov}}(\enumelem{p}{\nsimple})|} l_{\text{cov}}\\
		   	   &\hspace{2.5em}
		   			\parallel \tf@bsfun@proc(s_n) \ChanPar \tf@bsfun@chan(s_n)\\
		   		&\to[label=\epsilon,LTS=\calG,*]
			   		 l_{\text{cov}}
		   			\parallel \tf@bsfun@proc(s_n) \ChanPar \tf@bsfun@chan(s_n)
		   				\ominus \cchan{\bullet^{|s_{\text{cov}}(\enumelem{p}{1})|}}{\tf@chan{\enumelem{p}{1}}} \\
		   		&\hspace{2.5em}
		   				\ominus \cchan{\bullet^{|s_{\text{cov}}(\enumelem{p}{2})|}}{\tf@chan{\enumelem{p}{2}}}
		   				\ominus
		   				\cdots
		   				\ominus \cchan{\bullet^{|s_{\text{cov}}(\enumelem{p}{\nsimple})|}}{\tf@chan{\enumelem{p}{\nsimple}}}\\
		   	   & = A_{\text{cov}} \parallel \tf@bsfun@proc(s_n) \ChanPar \tf@bsfun@chan(s_n \ominus (s_{\text{cov}} \restriction \Psimple))\\
\end{align*}
Hence $(\calG,S \ChanPar \emptyset,A_{\text{cov}} \ChanPar \emptyset)$ is a yes-instance of alternative simple coverability and
thus $(\calG,S \ChanPar \emptyset,A_{\text{cov}} \ChanPar \emptyset)$ is a yes-instance of standard simple coverability.

Suppose $(\calG,S \ChanPar \emptyset,A_{\text{cov}} \ChanPar \emptyset)$ is a yes-instance of standard simple coverability.
Hence we know that $(\calG,S \ChanPar \emptyset,A_{\text{cov}} \ChanPar \emptyset)$ is a yes-instance of alternative simple coverability.
Hence $S \to[TS=\calG,*] A_{\text{cov}}\alpha \parallel \Pi \ChanPar \Gamma$
for some $\alpha$, $\Pi$ and $\Gamma$.
Let us view this path in the labelled transition system we can then obtain:
$S \to[label=s_0,LTS=\calG,*] \tf@bsfun(s_0) \to[label=s_1,LTS=\calG,*] \tf@bsfun(s_1) \to[label=s_2,LTS=\calG,*] \cdots \to[label=s_n,LTS=\calG,*] \tf@bsfun(s_n) \to[label=\epsilon,LTS=\calG,*] A_{\text{cov}}\alpha \parallel \Pi \ChanPar \Gamma$ where it is clear that the first simulation state set up by $S$ can only by $s_0$.

We know that $\tf@bsfun(s_n) = \tf@nonterm{}{\mathit{sim}} \parallel \tf@bsfun@proc(s_n) \ChanPar \tf@bsfun@chan(s_n)$.
Since on the sequential level it must be the case that
$\tf@nonterm{}{\mathit{sim}} \to[*] \beta A_{\text{cov}}\alpha$, but inspecting the rules of $\calG$ we can see that only
$\tf@nonterm{}{\mathit{sim}} \to[*] \paren{\rec{\tf@chan{\enumelem{p}{1}}}{\bullet}}^{|s_{\text{cov}}(\enumelem{p}{1})|} \cdots 
 		   				 \paren{\rec{\tf@chan{\enumelem{p}{\nsimple}}}{\bullet}}^{|s_{\text{cov}}(\enumelem{p}{\nsimple})|} A_{\text{cov}}$ we can conclude that $\alpha = \epsilon$.
Further, along these transition no administrative messages to channels $\tf@chan{?}$, $\tf@chan{p}$ for $p \in \Pcomplex$
are sent and thus any process in $\tf@bsfun@proc(s_n)$ is blocked on its initial receive action and thus cannot receive or send any messages to a $\tf@chan{p}$ where $p \in \Psimple$. Hence it must be the case
that 
\begin{align*}
\Gamma = \tf@bsfun@chan(s_n) \ominus \cchan{\bullet^{|s_{\text{cov}}(\enumelem{p}{1})|}}{\tf@chan{\enumelem{p}{1}}}
		   				\ominus
		   				\cdots
		   				\ominus \cchan{\bullet^{|s_{\text{cov}}(\enumelem{p}{\nsimple})|}}{\tf@chan{\enumelem{p}{\nsimple}}}\\
		= \tf@bsfun@chan(s_n \ominus s_{\text{cov}})		   				
\end{align*}
since $s_{\text{cov}}(p) = \emptyset$ for all $p \in \Pcomplex$.
We can thus conclude that $s_{\text{cov}} \leqconfig s_n$.
Since $\tf@bisim$ is a weak bisimulation we know that there is a path
$s_0 \to[label=s_1,LTS=\calN,*] s_1 \to[label=s_2,LTS=\calN] \cdots \to[label=s_n,LTS=\calN,*] s_n$
and hence
$s_0 \to[TS=\calN,*] s_n$ and $s_{\text{cov}} \leqconfig s_n$ thus 
$(\calN, s_0, s_{\text{cov}})$ is a yes-instance of coverability with a simple query.

For boundedness, suppose the set $\varset{\Pi \ChanPar \Gamma : S \ChanPar \emptyset \to[LTS=\calG,*] \Pi \ChanPar \Gamma}$ is finite. This implies that the set
$\varset{\tf@bsfun(s) : \tf@bsfun(s_0) \to[LTS=\calG,*] \tf@bsfun(s)}$ is finite and since $\tf@bisim$ is a weak bisimulation this means that $\varset{s : s_0 \ChanPar \ChanPar \to[LTS=\calN,*] s}$ is finite.
Conversely, suppose $\varset{s : s_0 \ChanPar \ChanPar \to[LTS=\calN,*] s}$ is finite then clearly
$\varset{\tf@bsfun(s) : \tf@bsfun(s_0) \to[LTS=\calG,*] \tf@bsfun(s)}$ is finite. Since in $\calG$ there are only finitely many more steps along weak bisimulation steps we can also conclude that 
$\varset{\Pi \ChanPar \Gamma : S \ChanPar \emptyset \to[LTS=\calG,*] \Pi \ChanPar \Gamma}$ is finite.
Hence $\calG$ is bounded from from $S \ChanPar \emptyset$ if, and only if, 
$\calN$ is bounded from $s_0$.

For termination, suppose there is an infinite path from $S \ChanPar \emptyset$ in $\calG$, then since there are only finitely many $\epsilon$ steps along weak bisimulation steps and $\tf@bisim$ is a weak bisimulation we can conclude that 
$\calN$ has an infinite path from $s_0$. 
Conversely, if $\calN$ has an infinite path from $s_0$ then there is an infinite path from $S \ChanPar \emptyset$ in $\calG$ since $\tf@bisim$ is a weak bisimulation and every step .
$\calN$ is bounded from $s_0$.
Hence $\calG$ is terminating from from $S \ChanPar \emptyset$ if, and only if, 
$\calN$ is terminating from $s_0$.
\makeatother
\end{proof}        
\newpage
\subsection{Proofs Upper Bound}

\begin{customlemma}[\ref{lemma:pivot_strengthening:tools:i}]
Suppose $\vec{s}$ is a covering path for $s_{\text{cov}}$ in $\calN_{i+1}$ and $|\vec{s}(1)(\enumelem{p}{i+1})| \geq R' \cdot \length{\calN_i}{s_{\text{cov}}}$. Then $\exists \vec{s}'$ a covering path
	   for $s_{\text{cov}}$ in $\calN_{i+1}$ such that $\vec{s}'(1) = \vec{s}(1)$ and $|\vec{s}'| \leq  \length{\calN_i}{s_{\text{cov}}}$.
\end{customlemma}%
\begin{proof}
Let $m_\infty = \update{}{\bullet}{\infty}$, let $s = \vec{s}(1)$ and define 
$\vec{s}_\infty(i) = \update{\vec{s}(j)}{\enumelem{p}{i+1}}{m_\infty}$ for all $1 \leq j \leq |\vec{s}|$. 
Since $\vec{s}$ is covering path for $s_{\text{cov}}$ in $\calN_{i+1}$, 
we can easily see that $\vec{s}_\infty$ is a covering path for $s_{\text{cov}}$ in $\calN_i$.
 Hence there must exists a covering path $\vec{s}'_\infty$ for $s_{\text{cov}}$ in $\calN_i$
 such that $|\vec{s}'_\infty| \leq \length{\calN_i}{s_{\text{cov}}}$.
	We can ``replay'' the path $\vec{s}'_\infty$ from $s$ instead of 
	$\update{s}{\enumelem{p}{i+1}}{m_\infty}$ to obtain a configuration sequence $\vec{s}_0$ from $s$ so that 
	$\vec{s}_0(j)(k) = \vec{s}'_\infty(j)(k)$ for all $1 \leq j \leq |\vec{s}'_\infty|$ and $k \neq i+1$. An easy induction shows that for all $1 \leq j \leq |\vec{s}'_\infty|$ we have
	$|\vec{s}_0(j)(i+1)| \geq R' \times \length{\calN_i}{s_{\text{cov}}} - j \times R$ from which we can conclude
	\begin{align*}
	|\vec{s}_0(j)(i+1)| &\geq R' - R + (R'-R) \times (\length{\calN_i}{s_{\text{cov}}}-1) \\
						&\geq R' - R \geq s_{\text{cov}}(\enumelem{p}{i+1}).
	\end{align*}
	since $\length{\calN_i}{s_{\text{cov}}} \geq 1$ for all $i$.
	Hence we can conclude that
	$\vec{s}_0$ is a path in $\calN_{i+1}$, further $\vec{s}_0$ is a path from $s$ covering $s_{\text{cov}}$ in $\calN_{i+1}$ and by construction $|\vec{s}_0| \leq \length{\calN_i}{s_{\text{cov}}}$.
	which is what we wanted to prove.
\end{proof}

\begin{customlemma}[\ref{lem:rackoff_case_split}]
For all covering paths $\vec{s}$ for $s_{\text{cov}}$ one of two cases applies:
\begin{inparaenum}[(C-1)]
\item $\norm[\calN_{i+1}]{\vec{s}}^* < R' \cdot \length{\calN_i}{s_{\text{cov}}}$; or
\item $\vec{s} = \vec{s}_1 \cdot s_p \cdot \vec{s}_2$ such that
$\norm[\calN_{i+1}]{\vec{s}_1}^* < R' \cdot \length{\calN_i}{s_{\text{cov}}}$ and
$\norm[\calN_{i+1}]{s_p} \geq R' \cdot \length{\calN_i}{s_{\text{cov}}}$. 
\end{inparaenum}
\end{customlemma}
\begin{proof}
Let $\vec{s}$ be a covering path $\vec{s}$ for $s_{\text{cov}}$.
Let us do a case analysis:
\begin{enumerate}[(i)]
	\item \emph{Case 1: $\norm[\calN_{i+1}]{\vec{s}}^* < R' \cdot \length{\calN_i}{s_{\text{cov}}}$.}\label{case:bounded_path}
	\newline
	Case C-1 holds.
	\item \emph{Case 2: $\norm[\calN_{i+1}]{\vec{s}}^* \nleq R' \cdot \length{\calN_i}{s_{\text{cov}}}$.}
	\newline
	Hence we can split $\vec{s}$ in two paths $\vec{s} =\vec{s}_1\vec{s}_0$ 
	such that $\vec{s}_0$ is the longest prefix of $\vec{s}$ such that $\norm[\calN_{i+1}]{\vec{s}_1}^* < R' \cdot \length{\calN_i}{s_{\text{cov}}}$ and $\norm[\calN_{i+1}]{\vec{s}_0(1)} \nleq R' \cdot \length{\calN_i}{s_{\text{cov}}}$, i.e.~$\norm[\calN_{i+1}]{\vec{s}_0(1)} \geq R' \cdot \length{\calN_i}{s_{\text{cov}}}$.
	Letting $s_p = \vec{s}_0(1)$ and $\vec{s}_2 = \vec{s}_0(2) \cdots \vec{s}_0(|\vec{s}_0|)$ means C-2 applies.
\end{enumerate}
\end{proof}


\begin{customlemma}[\ref{lemma:pivot_strengthening:tools:ii}]
Suppose $\vec{s}$ is a covering path for $s_{\text{cov}}$ in $\calN_{i+1}$.
If $L \in \N$ such that $|\vec{s}| \leq L$ and $\norm[\calN_{i+1}]{\vec{s}(1)} < L \cdot R'$ then there exists a covering path $\vec{s}''$ for $s_{\text{cov}}$ in $\calN_{i+1}$ such that $\vec{s}''(1) \leqconfig \vec{s}(1)$, $\norm[\calN_{i+1};C]{\vec{s}''(1)} < L \cdot R'$, and $|\vec{s}''| \leq L$.
\end{customlemma}
\begin{proof}
We will construct a path removing all superfluous complex tokens first.
Suppose we label every complex token in $\vec{s}(1)$ as ``not-moved'' and along $\vec{s}$ if a complex token is moved we change the label to ``moved''. Since $|\vec{s}| \leq L$ and each transition can move at most one complex token, and remove $R$ complex empty token
it must be the case that for each $p \in \Pcomplex$ the size of the multiset 
$M_{\text{remove}} = |\mset{m\text{ ``not-moved'' } \mid m \in \vec{s}(|\vec{s}|)(p)}| \geq |\varmset{m\text{ ``not-moved'' } \mid m \in \vec{s}(1)(p)}| - R \times L$.
Hence the complex tokens in $M_{\text{remove}}$ play no r\^ole in $\vec{s}$, appear in $\vec{s}(j)(p)$ for all $1 \leq j \leq |\vec{s}|$ and could thus safely be remove without affecting the validity of the path.
However, in $\vec{s}(|\vec{s}|)(p)$ there are complex tokens, at most $\norm[\calN_{\nsimple};C]{s_{\text{cov}}}$ many, that are required to justify 
$s_{\text{cov}} \leqconfig \vec{s}(|\vec{s}|)$ and hence we cannot remove these tokens without losing the cover of $s_{\text{cov}}$. We can be on the safe side an assume the former tokens are all ``not-moved'' tokens and hence
we may only remove $|\vec{s}(1)(p)| - (R \times L + \norm[\calN_{\nsimple};C]{s_{\text{cov}}})$ tokens. We can do this for all $p \in \Pcomplex$ and hence we get a new path $\vec{s}''_0$ which is
a covering path for $s_{\text{cov}}$ in $\calN_{i+1}$, $|\vec{s}''_0| \leq L$ and
$|\vec{s}''_0(1)(p')| < R \times L + \norm[\calN_{\nsimple};C]{s_{\text{cov}}}$ for all $p' \in \Pcomplex$ and since $\vec{s}''_0(1)$ is obtained from 
$\vec{s}(1)$ by removing tokens we can see that $\vec{s}''_0(1) \leqconfig \vec{s}(1)$ and hence also
$\norm[\calN_{i+1}]{\vec{s}''_0(1)} < (L+1) \cdot R'$.

Let us now show that we can also reduce the number of coloured tokens. 
We will label coloured tokens in $\vec{s}''_0(1)$ with three statuses \emph{untouched}, \emph{ejected} and \emph{consumed}. Initially this record of $\vec{s}''_0(1)$'s coloured tokens marks all of them as \emph{untouched}.
Along each transition of $\vec{s}''_0$ we update this record as follows. If the transition 
$\vec{s}''_0(j) \to[TS=\calN_{i+1}] \vec{s}''_0(j+1)$ is a transfer transition then we look at all the coloured tokens that are ejected in this transition to simple places and label the coloured tokens whose origin lies in $\vec{s}''_0(1)$ as \emph{ejected}. If in the transition $\vec{s}''_0(j) \to[TS=\calN_{i+1}] \vec{s}''_0(j+1)$ a simple token is removed and its origin is as a coloured token in $\vec{s}''_0(1)$ then we mark that coloured token in $\vec{s}''_0(1)$ as \emph{consumed}.
Since $|\vec{s}| \leq L$ and each transition can remove at most $R$ simple tokens
it must be the case that at most $R \times L$ coloured tokens were marked as \emph{consumed} by the above process.

However, for each $p' \in \Pcomplex$ and $p^{\I} \in \Pinner$ there are coloured tokens in 
$\vec{s}''_0(|\vec{s}''_0|)(p')(p^{\I})$, at most $\norm[\calN_{\nsimple};C]{s_{\text{cov}}}$ many, 
that are required to justify 
$s_{\text{cov}} \leqconfig \vec{s}''_0(|\vec{s}''_0|)$ and hence we cannot remove these tokens without losing the cover of $s_{\text{cov}}$. As above it is safe to label these coloured tokens \emph{consumed} in $\vec{s}''_0(1)$.
Hence for each $p' \in \Pcomplex$ and $p^{\I} \in \Pinner$ there are at most $R \times L + \norm[\calN_{\nsimple};C]{s_{\text{cov}}}$ that are labelled \emph{consumed} in $\vec{s}''_0(1)$.

It should be clear that we can remove all coloured tokens labelled \emph{untouched} along $\vec{s}''_0$
without affecting the validity of the path. We can also remove coloured tokens that are labelled \emph{ejected}
along $\vec{s}''_0$ since the transfer transition may still fire, however, fewer tokens will be transferred.
Since these coloured tokens are not labelled as consumed it is clear that their absence as simple tokens cannot disable
a rule that is fired along $\vec{s}''_0$. For a rule $r$ that removes an empty complex token $m$ we can see that if $m$ is empty along $\vec{s}''_0$ then removing coloured tokens from $\vec{s}''_0(1)$ could only possibly lead to less tokens inside $m$, hence it is impossible to disable $r$.
Thus we can obtain a new covering path $\vec{s}''$ for $s_{\text{cov}}$ in $\calN_{i+1}$, such that
$|\vec{s}''| \leq L$ and
$|\vec{s}''(1)(p')| < R \times L + \norm[\calN_{\nsimple};C]{s_{\text{cov}}}$ for all $p' \in \Pcomplex$ and 
$|\vec{s}''(1)(p')(p^{\I})| < R \times L + \norm[\calN_{\nsimple};C]{s_{\text{cov}}}$ for all $p' \in \Pcomplex$, for all $p^{\I} \in \Pinner$
since $\vec{s}''(1)$ is obtained from 
$\vec{s}(1)$ by removing tokens we can see that $\vec{s}''(1) \leqconfig \vec{s}(1)$ and hence also
$\norm[\calN_{i+1}]{\vec{s}''_0(1)} < L \cdot R'$.

Since 
\begin{align*}
R \times L + \norm[\calN_{\nsimple};C]{s_{\text{cov}}} &\leq R + R \times (L-1) + \norm[\calN_{\nsimple};C]{s_{\text{cov}}} + 1 \\
&= R \times (L-1) + R' \\
&\leq R' \times (L-1) + R'\\
&= R' \times L 
\end{align*}
we can see that the above implies that
$\norm[\calN_{i+1};C]{\vec{s}''_0(1)} < L \cdot R'$ which is what we wanted to prove.
\end{proof}

\begin{customcorollary}[\ref{cor:rackoff_case_split}]
For all covering paths $\vec{s}$ for $s_{\text{cov}}$ in $\calN_{i+1}$ one of two cases applies:
\begin{asparaenum}[(C$'_1$)]
\item(app:cor:rackoff_case_split:c-1) $\norm[\calN_{i+1}]{\vec{s}}^* < R' \cdot \length{\calN_i}{s_{\text{cov}}}$; or
\item(app:cor:rackoff_case_split:c-2) there exist paths $\vec{s}_1$ and $s_{p'} \cdot \vec{s}_2$ such that 
$\vec{s}_1$ is a covering path for $s_{p'}$, $\vec{s}_1(1) = \vec{s}(1)$ and
$\norm[\calN_{i+1}]{\vec{s}_1}^* < R' \cdot \length{\calN_i}{s_{\text{cov}}}$;
$s_{p'} \cdot \vec{s}_2$ is a covering path for $s_{\text{cov}}$ and
$|\vec{s}'_2| \leq \length{\calN_i}{s_{\text{cov}}}$; and
$\norm[\calN_{i+1};C]{s_{p'}} \leq R' \cdot \length{\calN_i}{s_{\text{cov}}}$. 
\end{asparaenum}
\end{customcorollary}
\begin{proof}
Let $\vec{s}$ be a covering path $\vec{s}$ for $s_{\text{cov}}$ in $\calN_{i+1}$.
According to Lemma~\ref{lem:rackoff_case_split} we can consider the following two cases:
\begin{enumerate}[(i)]
\item(lem:rackoff_case_split_c-1) $\norm[\calN_{i+1}]{\vec{s}}^* < R' \cdot \length{\calN_i}{s_{\text{cov}}}$. \newline
Clearly case \ref{cor:rackoff_case_split:c-1} applies.
\item(lem:rackoff_case_split_c-2) $\vec{s} = \vec{s}_1 \cdot s_p \cdot \vec{s}_2$ such that
$\norm[\calN_{i+1}]{\vec{s}_1}^* < R' \cdot \length{\calN_i}{s_{\text{cov}}}$ and
$\norm[\calN_{i+1}]{s_p} \geq R' \cdot \length{\calN_i}{s_{\text{cov}}}$. \newline 
Since $\norm[\calN_{i+1}]{s_p} \geq R' \cdot \length{\calN_i}{s_{\text{cov}}}$ we can assume w.l.o.g. that 
$|s_p(\enumelem{p}{i+1})| \geq R' \cdot \length{\calN_i}{s_{\text{cov}}}$ (we can always change the enumeration $\Psimple$'s elements). Lemma~\ref{lemma:pivot_strengthening:tools:i} applies to $s_p \cdot \vec{s}_2$ and gives us a covering path $\vec{s}'_2$ for $s_{\text{cov}}$ from $s_p$ in $\calN_{i+1}$ such that $|\vec{s}'_2| \leq \length{\calN_i}{s_{\text{cov}}}$.

Let $s_0 = \vec{s}_1(|\vec{s}_1|)$. Clearly $s_0 \cdot \vec{s}'_2$ is a covering path for $s_{\text{cov}}$ in $\calN_{i+1}$,
with length $|s_0 \cdot \vec{s}'_2| \leq \length{\calN_i}{s_{\text{cov}}} + 1$, and
$\norm[\calN_{i+1}]{s_0} < R' \cdot \length{\calN_i}{s_{\text{cov}}}$ since $\norm[\calN_{i+1}]{\vec{s}_1}^* < R' \cdot \length{\calN_i}{s_{\text{cov}}}$. 

We are thus able to apply Lemma~\ref{lemma:pivot_strengthening:tools:ii} to 
 $s_0 \cdot \vec{s}'_2$ which yields a new covering path $s_{p'} \cdot \vec{s}''_2$ for $s_{\text{cov}}$ in $\calN_{i+1}$
 such that $s_{p'} \leqconfig s_0$, $\norm[\calN_{i+1};C]{s_{p'}} < R' \cdot (\length{\calN_i}{s_{\text{cov}}}+1)$, and $|s_{p'} \cdot \vec{s}''_2| \leq \length{\calN_i}{s_{\text{cov}}} + 1$.

We can then conclude that $\vec{s}_1$ is a covering path in $\calN_{i+1}$ for $s_{p'}$,
$\vec{s}_1(1) = \vec{s}(1)$, and $\norm[\calN_{i+1}]{\vec{s}_1}^* < R' \cdot \length{\calN_i}{s_{\text{cov}}}$.
Further $s_{p'} \cdot \vec{s}''_2$ is a covering path for $s_{\text{cov}}$ in $\calN_{i+1}$,
$|\vec{s}''_2| \leq \length{\calN_i}{s_{\text{cov}}}$, and
$\norm[\calN_{i+1};C]{s_{p'}} \leq R' \cdot \length{\calN_i}{s_{\text{cov}}}$.
\end{enumerate}
\end{proof}

\begin{customproposition}[\ref{prop:rackoff_recurrence}]
\begin{asparaenum}[(i)]
\item $\length{\calN_{0}}{s_{\text{cov}}} \leq \diameter[\calN_{0}]{S_{(0,R')}} $
\item $\length{\calN_{i+1}}{s_{\text{cov}}} \leq \diameter[\calN_{i+1}]{S_{(i+1,B_i)}} + 
\length{\calN_i}{s_{\text{cov}}}$.
\end{asparaenum}
\end{customproposition}
\begin{proof}
In order to prove (i) 
let $\vec{s}$ be a covering path for $s_{\text{cov}}$ in $\calN_0$. Since all simple places are ignored it is trivial 
to see that $\norm[\calN_0]{\vec{s}}^* < 1 \leq R'$. Further 
$\norm[\calN_0;C]{s_{\text{cov}}} \leq \norm[\calN_{\nsimple};C]{s_{\text{cov}}} \leq R'$. 
Hence we have $\vec{s}(1), s_{\text{cov}} \in S_{(0,R')}$ which implies we can find a path $\vec{s}'$ from $\vec{s}(1)$ covering $s_{\text{cov}}$ in $\calN_0$ with 
$|\vec{s}'| \leq \diameter[\calN_0]{S_{(0,R')}}$. 
Hence $\dist{\vec{s}(1)}{s_{\text{cov}}}{\calN_0} \leq \diameter[\calN_0]{S_{(0,R')}}$. Since $\vec{s}$ was arbitrary we can conclude that
$\length{\calN_{0}}{s_{\text{cov}}} \leq \diameter[\calN_{0}]{S_{(0,R')}}$.

For claim (ii)
let $s$ be an element of $\Configuration^\infty$ such that $\norm[\calN_{i+1}]{s} < \infty$. We want to show that
$$\dist{s}{s_\text{cov}}{\calN_{i+1}} \leq \diameter[\calN_{i+1}]{S_{(i+1,B_i)}} + \length{\calN_i}{s_{\text{cov}}}.$$
Let us do a case analysis:
\begin{enumerate}[(i)]
	\item \emph{Case 1: There is no covering path for $s_{\text{cov}}$ from $s$ in $\calN_{i+1}$.}
	\newline
	Then $\dist{s}{s_{\text{cov}}}{\calN_{i+1}} = 0$ by definition and the inequality holds
	trivially.
	\item \emph{Case 2: There is a covering path $\vec{s}$ for $s_{\text{cov}}$ from $s$ in $\calN_{i+1}$.}\label{case:bounded_path}
	\newline
	Corollary~\ref{cor:rackoff_case_split} allows us to consider two cases:
	\begin{asparaenum}[label={\emph{Case (\Alph*):}}]
	\item \emph{$\norm[\calN_{i+1}]{\vec{s}'}^* < B_i$}.\newline
		In this case it is easy to see that
		$\norm[\calN_i+1;C]{s_{\text{cov}}} \leq \norm[\calN_{\nsimple};C]{s_{\text{cov}}} \discn \leq R' \leq B_i$. 
		Hence we can see $s, s_{\text{cov}} \in S_{(i+1,B_i)}$ which implies we can find a path 
		$\vec{s}'$ from $s$ covering $s_{\text{cov}}$ in $\calN_{i+1}$ with 
		$|\vec{s}'| \leq \diameter[\calN_{i+1}]{S_{(i+1,B_i)}}$. 
		Hence $\dist{s}{s_{\text{cov}}}{\calN_{i+1}} \leq \diameter[\calN_{i+1}]{S_{(i+1,B_i)}}$.
	\item \emph{There exist paths $\vec{s}_1$ and $s_{p'} \cdot \vec{s}_2$ such that 
				$\vec{s}_1$ is a covering path for $s_{p'}$, $\vec{s}_1(1) = s$ and
				$\norm[\calN_{i+1}]{\vec{s}_1}^* < B_i$;
				$s_{p'} \cdot \vec{s}_2$ is a covering path for $s_{\text{cov}}$ and
				$|\vec{s}'_2| \leq \length{\calN_i}{s_{\text{cov}}}$; and
				$\norm[\calN_{i+1};C]{s_{p'}} \leq B_i$}.\newline
		We can see $s, s_{p'} \in S_{(i+1,B_i)}$ which implies we can find a path 
		$\vec{s}'$ from $s$ covering $s_{p'}$ in $\calN_{i+1}$ with 
		$|\vec{s}'| \leq \diameter[\calN_{i+1}]{S_{(i+1,B_i)}}$. 
		Since $\calN_{i+1}$ is a WSTS and $s_{p'} \leqconfig \vec{s}'(|\vec{s}'|)$, 
		we can replay $s_{p'} \cdot \vec{s}_2$ from $\vec{s}'(|\vec{s}'|)$ yielding a path
		$\vec{s}'(|\vec{s}'|) \cdot \vec{s}''$ in $\calN_{i+1}$ such that $|\vec{s}''| = |\vec{s}_2|$ and
		$\vec{s}''$ covers $s_{\text{cov}}$.
		Thus $\vec{s}'\vec{s}''$ is a covering path in $\calN_{i+1}$ for $s_{\text{cov}}$ with
		$|\vec{s}'\vec{s}''| \leq \diameter[\calN_{i+1}]{S_{(i+1,B_i)}}+\length{\calN_i}{s_{\text{cov}}}$.
		Hence we can conclude that
		$\dist{s}{s_{\text{cov}}}{\calN_{i+1}} \leq \diameter[\calN_{i+1}]{S_{(i+1,B_i)}}+\length{\calN_i}{s_{\text{cov}}}$.
\end{asparaenum}
\end{enumerate}
Since the inequality holds in all cases and $s$ was arbitrary we can conclude that
$$\length{\calN_{i+1}}{s_{\text{cov}}} \leq \diameter[\calN_{i+1}]{S_{(i+1,B_i)}}+\length{\calN_i}{s_{\text{cov}}}.$$
\end{proof}
\newpage
\subsubsection{Proof of Theorem~\ref{thm:existence_counter_abstraction_petri_net}}
Our analysis of $S_{(i,B)}$ is driven by two observations: 
\begin{inparaenum}[(O$'_1$)]
\item(obs:1) $S_{(i,B)}$ is the set of starting and covering configurations of paths in $\calN_i$ along which the $i$ not ignored places contain less than $B$ tokens;
\item along such paths a complex token cannot carry more than $B$ tokens of a particular colour if these coloured tokens are ejected on the path; and if they are not ejected then they play no r\^ole in enabling the path. 
\end{inparaenum}
In order to exploit these observations we derive a transition system from $\calN_i$ that generates exactly the paths of bounded norm. 
Suppose $(S,\to[TS=\calS],\leq_{\calS})$ is a PSTS and $\norm{-}$ is a norm for $S$
then we say $\calS = (S,\to[TS=\calS],\leq_{\calS},\discn\norm{-})$ is a \emph{normed PSTS}.
Given a normed PSTS $\calS = (S,\to[TS=\calS],\leq_{\calS},\norm{-})$ and a \emph{norm-radius} $\rho \in \N^{\infty}$ we define a normed sub-PSTS $\BTS{\calS}{}{\rho}$ --- the \emph{ball} of norm-radius $\rho$ in $\calS$ ---
by $\BTS{\calS}{}{\rho} = \paren{\BTS{S}{}{\rho},\to[TS={\BTS{\calS}{}{\rho}}],\leq_\calS,\norm{-}}$
where $\BTS{S}{}{\rho} \is \varset{s \in S \mid \norm{s} < B }$, and
${\to[TS={\BTS{\calS}{}{\rho}}]} := {{\to[TS={\calS}]} \intersect (\BTS{S}{}{\rho})^2}$.
Note that by construction all $\BTS{\calS}{}{\rho}$-paths $\vec{s}$ satisfy $\norm{\vec{s}}^{*} < \rho$.

Attaching the norm $\norm[\calN_i]{-}$ to $\calN_i$ we then obtain a normed PSTS $(\Config^{\infty},\to[TS=\calN_i],\leqconfig,\norm[\calN_i]{-})$ to which we will also refer to as $\calN_i$.
The transition system ${\BTS{\calN_i}{}{B}}$ produces the paths in $P_{i,B}$ which generate the set $S_{(i,B)}$.

Let $B \in \N$ and $i \leq \nsimple$, we will now define the Petri net $\cabspetrinet$.
It is our intention to show there is a tight (length preserving) correspondence between ${\BTS{\calN_i}{}{B}}$-paths and paths in $\cabspetrinet$.
Set the dimension 
$\cabspnetdim := i + \paren{B+1}\paren{\paren{\ncomplex + 1}\paren{B+1}^{\ninner-1}-1}$.

\noindent We define the counter abstraction function $\cabsfun[i,B] : \Configuration^{\infty} \to \N^{\cabspnetdim}$ as
$\cabsfun[i,B](s)(j) := |s(\enumelem{p}{j})|$ for $1 \leq j \leq i$ and
$$\cabsfun[i,B]\paren{s}\paren{{i+\theta}} := \sum_{m^{\I} \in \cconfunI(j_0,\ldots,j_{\ninner - 1})} s\paren{\enumelem{p'}{j_{\ninner}}}(m^{\I})$$ 
where {$\theta = 
\sum_{k = 0}^{\ninner}j_k \cdot \paren{B+1}^k$},
for all $1 \leq  j_{\ninner} \leq \ncomplex$ and $0 \leq j_k \leq B$, $k = 0,\ldots, \ninner-1$,
and we write
\[\cconfunI(j_0,\ldots,j_{\ninner - 1}) = \set{m \in \Minner \mid \min(m(\changed[jk]{\enumelem{p^{\I}}{k}}),B) = j_{k-1}}.\]
Let us establish a \changed[jk]{``linearity''} property of $\cabsfun[i,B]$ to simplify proofs.
\begin{lemma}\label{lemma:counter_abstraction_homomorphism}
Let $i \leq \nsimple$, $B \in \N$. 
Suppose $s,s' \in \Configuration^{\infty}$ such that both $\norm[\calN_i]{s},\norm[\calN_i]{s'} < \infty$, then
$\cabsfun[i,B]\paren{s \oplus s'} = \cabsfun[i,B](s) + \cabsfun[i,B](s')$
and if $s \ominus s' \in \Configuration^{\infty}$ then
$\cabsfun[i,B]\paren{s \ominus s'} = \cabsfun[i,B](s) - \cabsfun[i,B](s')$.
\end{lemma} 
\begin{proof}
It is clear that for all $1 \leq j \leq i$
\begin{align*}
\cabsfun[i,B]\paren{s \oplus s'}(j) 
&= |\paren{s \oplus s'}(\enumelem{p}{j})| 
= \paren{s \oplus s'}(\enumelem{p}{j})(\bullet) \\
&= s(\enumelem{p}{j})(\bullet) + s'(\enumelem{p}{j})(\bullet) \\
&= |s(\enumelem{p}{j})| + |s'(\enumelem{p}{j})| \\
&= \cabsfun[i,B](s)(j) + \cabsfun[i,B](s')(j)
\shortintertext{and since $s \ominus s' \in \Configuration^{\infty}$ we know $|s(\enumelem{p}{j})| - |s'(\enumelem{p}{j})| \geq 0$ and hence }
\cabsfun[i,B]\paren{s \ominus s'}(j) 
&= |\paren{s \ominus s'}(\enumelem{p}{j})| 
= \paren{s \ominus s'}(\enumelem{p}{j})(\bullet) \\
&= s(\enumelem{p}{j})(\bullet) - s'(\enumelem{p}{j})(\bullet) \\
&= |s(\enumelem{p}{j})| - |s'(\enumelem{p}{j})|\\ 
&= \cabsfun[i,B](s)(j) - \cabsfun[i,B](s')(j)
\end{align*}
Further let $1 \leq j_{\ninner} \leq \ncomplex$ and $(j_0,\ldots,j_{\ninner-1}) \in \N^{\leq B}$. Then
\begin{align*}
&\cabsfun[i,B]\paren{s \oplus s'}\paren{i+\sum_{k = 0}^{\ninner}j_k \cdot \paren{B+1}^k} \\
&\;\;\;\; = \sum_{m^{\I} \in \cconfunI(j_0,\ldots,j_{\ninner - 1})} \paren{s \oplus s'}(\enumelem{p'}{j_{\ninner}})(m^{\I})\\
&\;\;\;\; = \sum_{m^{\I} \in \cconfunI(j_0,\ldots,j_{\ninner - 1})} \paren{s(\enumelem{p'}{j_{\ninner}})(m^{\I}) + s'(\enumelem{p'}{j_{\ninner}})(m^{\I})}\\
&\;\;\;\; = \paren{\sum_{m^{\I} \in \cconfunI(j_0,\ldots,j_{\ninner - 1})} s(\enumelem{p'}{j_{\ninner}})(m^{\I})} \\
&\;\;\;\;\;\; + \paren{\sum_{m^{\I} \in \cconfunI(j_0,\ldots,j_{\ninner - 1})} s'(\enumelem{p'}{j_{\ninner}})(m^{\I})}\\
&\;\;\;\; = \cabsfun[i,B]\paren{s}\paren{i+\sum_{k = 0}^{\ninner}j_k \cdot \paren{B+1}^k}  \\
&\;\;\;\;\;\; + \cabsfun[i,B]\paren{s'}\paren{i+\sum_{k = 0}^{\ninner}j_k \cdot \paren{B+1}^k}
\shortintertext{ and since $s \ominus s' \in \Configuration^{\infty}$ we know $s(\enumelem{p'}{j_{\ninner}})(m^{\I}) - s'(\enumelem{p'}{j_{\ninner}})(m^{\I}) \geq 0$ and hence }
&\cabsfun[i,B]\paren{s \ominus s'}\paren{i+\sum_{k = 0}^{\ninner}j_k \cdot \paren{B+1}^k} \\
&\;\;\;\; = \sum_{m^{\I} \in \cconfunI(j_0,\ldots,j_{\ninner - 1})} \paren{s \ominus s'}(\enumelem{p'}{j_{\ninner}})(m^{\I})\\
&\;\;\;\; = \sum_{m^{\I} \in \cconfunI(j_0,\ldots,j_{\ninner - 1})} \paren{s(\enumelem{p'}{j_{\ninner}})(m^{\I}) - s'(\enumelem{p'}{j_{\ninner}})(m^{\I})}\\
&\;\;\;\; = \paren{\sum_{m^{\I} \in \cconfunI(j_0,\ldots,j_{\ninner - 1})} s(\enumelem{p'}{j_{\ninner}})(m^{\I})} \\
&\;\;\;\;\;\; - \paren{\sum_{m^{\I} \in \cconfunI(j_0,\ldots,j_{\ninner - 1})} s'(\enumelem{p'}{j_{\ninner}})(m^{\I})}\\
&\;\;\;\; = \cabsfun[i,B]\paren{s}\paren{i+\sum_{k = 0}^{\ninner}j_k \cdot \paren{B+1}^k}  \\
&\;\;\;\;\;\; - \cabsfun[i,B]\paren{s'}\paren{i+\sum_{k = 0}^{\ninner}j_k \cdot \paren{B+1}^k}
\end{align*}

Hence we know that for all $1 \leq j \leq \cabspnetdim$
$\cabsfun[i,B]\paren{s \oplus s'}(j) = \cabsfun[i,B](s)(j) + \cabsfun[i,B](s')(j)$ and thus
$\cabsfun[i,B]\paren{s \oplus s'} = \cabsfun[i,B](s) + \cabsfun[i,B](s')$,
and if $s \ominus s' \in \Configuration^{\infty}$ then for $1 \leq j \leq \cabspnetdim$
$\cabsfun[i,B]\paren{s \ominus s'}(j) = \cabsfun[i,B](s)(j) - \cabsfun[i,B](s')(j)$
and thus
$$\cabsfun[i,B]\paren{s \ominus s'} = \cabsfun[i,B](s) - \cabsfun[i,B](s').$$
\end{proof}%
\noindent
We define the counter abstraction Petri net
$\commabr{\cabspetrinet = \varparen{\cabspnetdim,\cabspnetrules}}$.
The set of rules $\cabspnetrules$ is derived from $\Rules$ by case analysis:
if $r = (I,O) \in \SimpleRules$ then add $\ind{r} := \varparen{\cabsfun[i,B](O \ominus I),\cabsfun[i,B](I)}$ to $\cabspnetrules$.

\noindent For $r = ((p,I),(p',\calc,O)) \in \ComplexRules$, add for all $m \in \Minner$ such that $\norm[\I]{m} \leq B$
the rule $\ind[m]{r} := (\ind{r}, d)$ to $\cabspnetrules$ where 
\begin{align*}
\ind{r} & = \cabsfun[i,B]\paren{O\oplus\update{}{p'}{\mset{m \oplus \calc}} \ominus \paren{I\oplus\update{}{p}{\mset{m}}}}\\
d & = \cabsfun[i,B]\paren{I\oplus\update{}{p}{\mset{m}}}.
\end{align*}
For $r = ((p,I),(p',P, O)) \in \TransferRules$, define a collection of rules parametrised for all 
$m \in \Minner$ such that $\norm[\I]{m} \leq B$ and $\max_{p \in P}(|m(p)|) < B$ as
$\ind[m]{r} := (\ind{r},d)$ and add them to $\cabspnetrules$ where $m_{P} = m\restriction P$ and $m_{\overline{P}} = m \restriction (\Pinner \setminus P)$:
\begin{align*}
\ind{r} & = \cabsfun[i,B]\varparen{O\oplus\paren{m_{P} \circ \colmap^{-1}}\oplus\update{}{p'}{\mset{m_{\overline{P}}}} \ominus \varparen{I\oplus\update{}{p}{\mset{m}}}}\\
d & = \cabsfun[i,B]\paren{I\oplus\update{}{p}{\mset{m}}}.
\end{align*}
\noindent Coverability instances for $\cabspetrinet$ and ${\BTS{\calN_i}{}{B}}$ are unfortunately not isomorphic  since the order on the former is ignorant of complex tokens. To remedy this we add rules to $\cabspetrinet$ which yields the Petri net $\cabspnetord = \varparen{\cabspnetdim,\cabspnetrules \union \cabspnetordrules}$.
Let us write $s_{p,m} \is \update{}{p}{\mset{m}}$ for $p \in \Pcomplex$, $m \in \Minner$
and define 
$$\cabspnetordrules = \set{\paren{\cabsfun[i,B](s_{p,m'})-\cabsfun[i,B](s_{p,m})),\cabsfun[i,B](s_{p,m'}} \mid \Xi}$$
where $\Xi = m,m' \in \Minner_{\ninner,B+1}, m' >_{\Minner} m, p \in \Pcomplex$.


\begin{remark*}
Even though the Petri nets $\cabspetrinet$ and $\cabspnetord$ have a vast number of rules, the maximal entry 
in any rule is tightly controlled: 
$R \geq \max\varset{\ind{r}(i) \mid \ind{r} \in \cabspnetrules \union \cabspnetordrules}$. 
We aim to exploit a result by \citeauthor{Bonnet:12} that shows that covering radii in SIAN are sensitive to the maximal entry in any rule rather than their number.
\end{remark*}%
\noindent
We will now set out to show the connection between $\cabspetrinet$ and $\calN_i$.
First let us add a norm to $\cabspetrinet$: define $\norm[\cabspetrinet]{v} \!= \max_{j \in \range{i}}(|v(j)|)$.
In the following we prove two lemmas to help us prove a lockstep property (\Cref{lemma:bisimulation_lemma_lockstep}) 
between $\cabspetrinet$ and $\calN_i$.

\begin{lemma}\label{lemma:counter_abstraction_Mabs}
Let $1 \leq i \leq \nsimple$, $B \in \N$. 
Suppose $p = \enumelem{p'}{j_{\ninner}} \in \Pcomplex$, 
and $j_0,\ldots,j_{\ninner-1} \in \N^{\leq B}$ such that $m,m' \in \cconfunI(j_0,\ldots,j_{\ninner-1})\discn \subseteq \Minner$
and $m'' \in \Minner$ then
$$\cabsfun[i,B]\paren{\update{}{p}{\mset{m \oplus m''}}} = \cabsfun[i,B]\paren{\update{}{p}{\mset{m' \oplus m''}}}$$
\end{lemma}
\begin{proof}
Let $j'_0,\ldots,j'_{\ninner-1}$ such that $m \oplus m'' \in \cconfunI(j'_0,\ldots,j'_{\ninner-1})$ and write $j'_{\ninner} = j_{\ninner}$.
Let us consider $m_0 \oplus \calc$. Let $1 \leq k \leq \ninner$ and 
note that since $m,m' \in \cconfunI(j_0,\ldots,j_{\ninner-1})$ we have 
$\min(m(\enumelem{p^{\I}}{k}),B) = j_{k-1} = \min(m(\enumelem{p^{\I}}{k}),B)$; 
and thus
\begin{align*}
j_{k-1} &= \min((m \oplus m'')(\enumelem{p^{\I}}{k}),B) \\
	&= \min(m(\enumelem{p^{\I}}{k}) + m''(\enumelem{p^{\I}}{k}),B) \\
	&= \min(\min(m(\enumelem{p^{\I}}{k}),B) + m''(\enumelem{p^{\I}}{k}),B)\\
	&= \min(\min(m'(\enumelem{p^{\I}}{k}),B) + m''(\enumelem{p^{\I}}{k}),B) \\
	&= \min(m'(\enumelem{p^{\I}}{k}) + m''(\enumelem{p^{\I}}{k}),B)
\end{align*}
from which we can conclude that $m' \oplus m'' \in \cconfunI(j'_0,\ldots,j'_{\ninner-1})$.
Further 

\begin{align*}
&\cabsfun[i,B]\paren{\update{}{p}{\mset{m \oplus m''}}}\paren{i+\sum_{k = 0}^{\ninner}j'_k \cdot \paren{B+1}^k} \\
&\;\;\;\;= \sum_{m^{\I}\in \cconfunI(j'_0,\ldots,j'_{\ninner-1})} \paren{\update{}{p}{\mset{m \oplus m''}}}(p)(m^{\I})\\
&\;\;\;\;= \paren{\update{}{p}{\mset{m \oplus m''}}}(p)(m \oplus m'') \\
&\;\;\;\;= 1
= \paren{\update{}{p}{\mset{m' \oplus m''}}}(p)(m \oplus m'') \\
&\;\;\;\;= \sum_{m^{\I}\in \cconfunI(j'_0,\ldots,j'_{\ninner-1})} \paren{\update{}{p}{\mset{m' \oplus m''}}}(p)(m^{\I})\\
&\;\;\;\;= \cabsfun[i,B]\paren{\update{}{p}{\mset{m' \oplus m''}}}\paren{i+\sum_{k = 0}^{\ninner}j'_k \cdot \paren{B+1}^k} 
\end{align*}
For any $\enumelem{p'}{j''_{\ninner}} \in \Pcomplex$, 
and $j''_0,\ldots,j''_{\ninner-1} \in \N^{\leq B}$ such that 
$i+\sum_{k = 0}^{\ninner}j'_k \cdot \paren{B+1}^k \neq i+\sum_{k = 0}^{\ninner}j''_k \cdot \paren{B+1}^k$
it is clear that $m \oplus m'', m' \oplus m'' \notin \cconfunI(j''_0,\ldots,j''_{\ninner-1})$ and so we have
\begin{align*}
&\cabsfun[i,B]\paren{\update{}{p}{\mset{m \oplus m''}}}\paren{i+\sum_{k = 0}^{\ninner}j''_k \cdot \paren{B+1}^k} \\
&\;\;\;\;= \sum_{m^{\I}\in \cconfunI(j''_0,\ldots,j''_{\ninner-1})} \paren{\update{}{p}{\mset{m \oplus m''}}}(\enumelem{p'}{j''_{\ninner}})(m^{\I})\\
&\;\;\;\;= \sum_{m^{\I}\in \cconfunI(j''_0,\ldots,j''_{\ninner-1})} 0 = 0 \\
&\;\;\;\;= \sum_{m^{\I}\in \cconfunI(j''_0,\ldots,j''_{\ninner-1})} \paren{\update{}{p}{\mset{m' \oplus m''}}}(\enumelem{p'}{j''_{\ninner}})(m^{\I})\\
&\;\;\;\;= \cabsfun[i,B]\paren{\update{}{p}{\mset{m' \oplus m''}}}\paren{i+\sum_{k = 0}^{\ninner}j''_k \cdot \paren{B+1}^k} 
\end{align*}
and also for all $1 \leq j \leq \nsimple$
\begin{align*}
\cabsfun[i,B]\paren{\update{}{p}{\mset{m \oplus m''}}}\paren{j} &=
|\paren{\update{}{p}{\mset{m \oplus m''}}}\paren{\enumelem{p}{j}}| \\
&= 0 \\
&= |\paren{\update{}{p}{\mset{m' \oplus m''}}}\paren{\enumelem{p}{j}}| \\
&= \cabsfun[i,B]\paren{\update{}{p}{\mset{m' \oplus m''}}}\paren{j}
\end{align*}
Hence we can conclude that
$$\cabsfun[i,B]\paren{\update{}{p}{\mset{m \oplus m''}}} = \cabsfun[i,B]\paren{\update{}{p}{\mset{m' \oplus m''}}}$$
\end{proof}

\begin{lemma}\label{lemma:counter_abstraction_Meject}
Let $1 \leq i \leq \nsimple$, $B \in \N$. 
Suppose $p = \enumelem{p'}{j_{\ninner}} \in \Pcomplex$, 
and $j_0,\ldots,j_{\ninner-1} \in \N^{\leq B}$ such that $m,m' \in \cconfunI(j_0,\ldots,j_{\ninner-1})\discn \subseteq \Minner$
and $P \subseteq \Pinner$ then
$$\cabsfun[i,B]\paren{\update{}{p}{\mset{m \restriction \overline{P}}}} = \cabsfun[i,B]\paren{\update{}{p}{\mset{m' \restriction \overline{P}}}}$$
where we write $\overline{P} = \Pinner \setminus P$.
\end{lemma}
\begin{proof}
Let $j'_0,\ldots,j'_{\ninner-1}$ such that $m \restriction \overline{P} \in \cconfunI(j'_0,\ldots,j'_{\ninner-1})$ and write $j'_{\ninner} = j_{\ninner}$.
Let us consider $m' \restriction \overline{P}$. Let $1 \leq k \leq \ninner$ and 
note that since $m,m' \in \cconfunI(j_0,\ldots,j_{\ninner-1})$ we have 
$\min(m(\enumelem{p^{\I}}{k}),B) = j_{k-1} = \min(m(\enumelem{p^{\I}}{k}),B)$; 
and thus if $\enumelem{p^{\I}}{k} \notin P$
\begin{align*}
j_{k-1} &= \min(m \restriction \overline{P}(\enumelem{p^{\I}}{k}),B) 
	= \min(m(\enumelem{p^{\I}}{k}),B) \\
	&= \min(m'(\enumelem{p^{\I}}{k}),B) \\
	&\min(m' \restriction \overline{P}(\enumelem{p^{\I}}{k}),B) 
\end{align*}
and for $\enumelem{p^{\I}}{k} \in P$
\begin{align*}
j_{k-1} &= \min(m \restriction \overline{P}(\enumelem{p^{\I}}{k}),B) 
	= \min(0,B) = 0 = \\
	&\min(m' \restriction \overline{P}(\enumelem{p^{\I}}{k}),B) 
\end{align*}
from which we can conclude that $m \restriction \overline{P},m' \restriction \overline{P} \in 
\commabr{\cconfunI(j'_0,\ldots,j'_{\ninner-1})}$.
Lemma~\ref{lemma:counter_abstraction_Mabs} then gives
$$\cabsfun[i,B]\paren{\update{}{p}{\mset{m \restriction \overline{P}}}} = \cabsfun[i,B]\paren{\update{}{p}{\mset{m' \restriction \overline{P}}}}$$
\end{proof}

\begin{lemma}\label{lemma:bisimulation_lemma_lockstep}
Suppose $s \in \Configuration^\infty$, $i \leq \nsimple$, $B \in \N$ such that 
$|s(\enumelem{p}{j})| = \infty$ for all $i < j \leq \nsimple$. 
\begin{asparaenum}[(i)]
\item If $s \to[TS=\BTS{\calN_i}{}{B}] s'$ then we have $\cabsfun[i,B](s) \to[TS=\cabspetrinet] \cabsfun[i,B](s')$.
\item If $\cabsfun[i,B](s) \to[TS=\cabspetrinet] \cabsfun[i,B](s) + r$ then there exists an $s' \in \Configuration^\infty$ such that $\cabsfun[i,B](s') = \cabsfun[i,B](s) + r$ and $s \to[TS=\calN_i] s'$.
\end{asparaenum}
\end{lemma}
\begin{proof}
We will first prove (i). Suppose $s \to[rule=r,TS={\BTS{\calN_i}{}{B}}] s'$ for some $r \in \Rules$ and let us write
\begin{inparaenum}[(a)]
\item $r = (I,O)$ if $r \in \SimpleRules$;
\item $r = ((p,I),(p',\calc,O))$ if $r \in \ComplexRules$; and
\item $r = ((p,I),(p',P,O))$ if $r \in \TransferRules$.
\end{inparaenum}

Since $s \to[rule=r,TS={\calN_i}] s'$ we know for all $\enumelem{p}{j} \in \Psimple$ we have
$|s(\enumelem{p}{j})| \geq |I(\enumelem{p}{j})|$ and for all $p' \in \Pcomplex$
$s(p') \supseteq I(p')$. Hence $s \ominus I \in \Configuration^{\infty}$
and $\cabsfun(s)(i) = |s(\enumelem{p}{j})| \geq |I(\enumelem{p}{j})| = \cabsfun(I)(j)$.
\begin{itemize}
\item \emph{Case: } $s \to[rule=r,TS={\calN_i}] s'$, $r = (I,O) \in \SimpleRules$. \newline
In this case we can see that for all $\enumelem{p'}{j} \in \Pcomplex$ 
\begin{align*}
&\cabsfun[i,B](s)\paren{i+j\paren{B+1}^{\ninner}} = s(\enumelem{p'}{j})(\vec{0}) \\
&\;\;\;\; \geq I(\enumelem{p'}{j})(\vec{0}) = \cabsfun[i,B](I)\paren{i+j\paren{B+1}^{\ninner}}.
\end{align*}
We see that $\cabsfun[i,B](r)$ is enabled at $\cabsfun[i,B](s)$ and hence
there is a transition $\cabsfun[i,B](s) \to_{\cabspetrinet} \cabsfun[i,B](s) + \ind{r}$.
We know that $s' = \paren{s \ominus I} \oplus O$ so since $s \ominus I \in \Configuration^{\infty}$
Lemma~\ref{lemma:counter_abstraction_homomorphism} gives us 
$\cabsfun[i,B](s') = \paren{\cabsfun[i,B](s) - \cabsfun[i,B](I)} + \cabsfun[i,B](O)$
and clearly $\ind{r} = (-\cabsfun[i,B](I)) + \cabsfun[i,B](O)$ so
$\cabsfun[i,B](s') = \cabsfun[i,B](s) + \ind{r}$ and thus
$\cabsfun[i,B](s) \to_{\cabspetrinet} \cabsfun[i,B](s')$ which is what we wanted to prove.
\item \emph{Case: } $s \to[rule=r,TS={\calN_i}] s'$, $r = ((p,I),(p',\calc,O)) \in \ComplexRules$. \newline
In this case we know that for some $m \in s(p)$ we have $s \ominus \update{}{p}{\mset{m}} \ominus I \in \Configuration^{\infty}$ and
$$s' = \paren{s \ominus \update{}{p}{\mset{m}} \ominus I} \oplus \update{}{p'}{\mset{m \oplus \calc}} \oplus O$$
Suppose $p = \enumelem{p'}{j_{\ninner}}$, 
and $j_0,\ldots,j_{\ninner-1} \in \N^{\leq B}$ such that $m \in \cconfunI(j_0,\ldots,j_{\ninner-1})$ then we know that
\begin{align*}
&\cabsfun[i,B](s)\paren{i+\sum_{k=0}^{\ninner}j_k\paren{B+1}^k} \\
&\;\;\;\; = \sum_{m^{\I}\in \cconfunI(j_0,\ldots,j_{\ninner})} s(p)(m^{\I})| \geq s(p)(m) \geq 1.
\end{align*}
We can thus see that for some $m_0 \in \cconfunI(j_0,\ldots,j_{\ninner-1})$ 
such that $\norm[I]{m_0} \leq B$ the rule
   $\ind[m_0]{r}$ is enabled at $\cabsfun[i,B](s)$ and thus $\cabsfun[i,B](s) \to[rule={\ind[m_0]{r}},TS=\cabspetrinet] \cabsfun[i,B](s)+\ind[m_0]{r}$.
Since $s \ominus \update{}{p}{\mset{m}} \ominus I \in \Configuration^{\infty}$ Lemma~\ref{lemma:counter_abstraction_homomorphism} gives us
\begin{align*}
\cabsfun[i,B](s') =&\, \cabsfun[i,B](s) - \cabsfun[i,B]\paren{\update{}{p}{\mset{m}}} \\
& - \cabsfun[i,B]\paren{I} + \cabsfun[i,B]\paren{\update{}{p'}{\mset{m \oplus \calc}}} + \cabsfun[i,B](O)
\end{align*}
Since $m,m_0 \in \cconfunI(j_0,\ldots,j_{\ninner-1})$ Lemma~\ref{lemma:counter_abstraction_Mabs} applies to give both
\begin{gather*}
\cabsfun[i,B]\paren{\update{}{p}{\mset{m}}} = \cabsfun[i,B]\paren{\update{}{p}{\mset{m_0}}}\\
\cabsfun[i,B]\paren{\update{}{p'}{\mset{m \oplus \calc}}} = \cabsfun[i,B]\paren{\update{}{p'}{\mset{m_0 \oplus \calc}}}
\end{gather*}
and hence
\begin{align*}
\cabsfun[i,B](s') = &\, \cabsfun[i,B](s) - \cabsfun[i,B]\paren{\update{}{p}{\mset{m_0}}} \\
& - \cabsfun[i,B]\paren{I} + \cabsfun[i,B]\paren{\update{}{p'}{\mset{m_0 \oplus \calc}}} + \cabsfun[i,B](O)\\
= &\, \cabsfun[i,B](s) + \ind[m_0]{r}
\end{align*}
and thus $\cabsfun[i,B](s) \to[rule={\ind[m_0]{r}},TS=\cabspetrinet] \cabsfun[i,B](s')$ which is what we wanted to prove.
\item \emph{Case: } $s \to[rule=r,TS={\calN_i}] s'$, $r = ((p,I),(p',P,O)) \in \TransferRules$. \newline
In this case we know that for some $m \in s(p)$ we have $s \ominus \update{}{p}{\mset{m}} \ominus I \in \Configuration^{\infty}$ and
\begin{align*}
s' &= \paren{s \ominus \update{}{p}{\mset{m}} \ominus I}  \oplus \update{}{p'}{\mset{m_{\overline{P}}}} 
\oplus O  \oplus \paren{m_{P} \circ \colmap^{-1}} 
\end{align*}
where $m_{P} = m \restriction P$ and $m_{\overline{P}} = m \restriction \varset{\Pinner \setminus P}$.
Since $\norm[\calN_i]{s'} < B$ we can conclude that for all $p \in \colmap(P)$ we can bound
$|\paren{m_{P} \circ \colmap^{-1}}(p)| < B$ which implies that 
$$\max_{p \in P}(|m(p)|) < B.$$
Suppose $p = \enumelem{p'}{j_{\ninner}}$, 
and $j_0,\ldots,j_{\ninner-1} \in \N^{\leq B}$ such that $m \in \cconfunI(j_0,\ldots,j_{\ninner-1})$ then we know that
\begin{align*}
\cabsfun[i,B](s)&\paren{i+\sum_{k=0}^{\ninner}j_k\paren{B+1}^k} 
= \sum_{m^{\I}\in \cconfunI(j_0,\ldots,j_{\ninner})} s(p)(m^{\I}) \\
&\geq s(p)(m) \geq 1.
\end{align*}
Further for all $k$ such that $\enumelem{p^{\I}}{k} \in P$ it is the case that $j_k < B$.
Hence for some $m_0 \in \cconfunI(j_0,\ldots,j_{\ninner-1})$ 
such that $\norm[\I]{m_0} \leq B$ and $\max_{p \in P} (|m_0(p)|) < B$ the rule
$\ind[m_0]{r}$ is enabled at $\cabsfun[i,B](s)$ and so $\cabsfun[i,B](s) \to[rule={\ind[m_0]{r}},TS=\cabspetrinet] \cabsfun[i,B](s) + \ind[m_0]{r}$.
Since $s \ominus \update{}{p}{\mset{m}} \ominus I \in \Configuration^{\infty}$
Lemma~\ref{lemma:counter_abstraction_homomorphism} yields
\begin{align*}
\cabsfun[i,B](s') &= \cabsfun[i,B](s) - \cabsfun[i,B]\paren{\update{}{p}{\mset{m}}} - \cabsfun[i,B](I) \\
				  &\;\;\;\;+\cabsfun[i,B]\paren{\update{}{p'}{\mset{m_{\overline{P}}}}} \\
				  &\;\;\;\;+\cabsfun[i,B](O) + \cabsfun[i,B](\paren{m_{P} \circ \colmap^{-1}}).
\end{align*}
Since $m,m_0 \in \cconfunI(j_0,\ldots,j_{\ninner-1})$ we can apply Lemma~\ref{lemma:counter_abstraction_Mabs} and Lemma~\ref{lemma:counter_abstraction_Meject} to get
\begin{gather*}
\cabsfun[i,B]\paren{\update{}{p}{\mset{m}}} = \cabsfun[i,B]\paren{\update{}{p}{\mset{m_0}}}\\
\cabsfun[i,B]\paren{\update{}{p'}{\mset{m_{\overline{P}}}}} = \cabsfun[i,B]\paren{\update{}{p'}{\mset{{m_0}_{\overline{P}}}}}
\end{gather*}
where ${m_0}_{\overline{P}} = m_0 \restriction \varparen{\Pinner \setminus P}$.
Since $\max_{p \in P}(|m(p)|) < B$ and $\max_{p \in P}(|m_0(p)|) < B$ we can see that
\begin{align*}
m(\enumelem{p}{k}) 	&= \min(m(\enumelem{p}{k}),B) = j_{k-1} \\
					&= \min(m_0(\enumelem{p}{k}),B) = m_0(\enumelem{p}{k})
\end{align*}
for all $1 \leq k \leq \Pinner$ such that $\enumelem{p^{\I}}{k} \in P$. 
Hence we can see that
$$m_{P} \circ \colmap^{-1} = {m_0}_{P} \circ \colmap^{-1}$$
where ${m_0}_{P} = {m_0} \restriction P$ which implies
\begin{align*}
\cabsfun[i,B](s') &= \cabsfun[i,B](s) - \cabsfun[i,B]\paren{\update{}{p}{\mset{m_0}}} - \cabsfun[i,B](I) \\
				  &\;\;\;\;+\cabsfun[i,B]\paren{\update{}{p'}{\mset{{m_0}_{\overline{P}}}}} \\
				  &\;\;\;\;+\cabsfun[i,B](O) + \paren{{m_0}_{P} \circ \colmap^{-1}} \\
				  &= \cabsfun[i,B](s) + \ind[m_0]{r}
\end{align*}
and hence
$\cabsfun[i,B](s) \to[rule={\ind[m_0]{r}},TS=\cabspetrinet] \cabsfun[i,B](s')$ which is what we wanted to prove.
\end{itemize}
Hence in all cases $\cabsfun[i,B](s) \to_{\cabspetrinet} \cabsfun[i,B](s')$ thus (i) holds.

For (ii) suppose $\cabsfun[i,B](s) \to_{\cabspetrinet} \cabsfun[i,B](s)+r$ for some 
$r \in \cabspnetrules$. We know that either 
\begin{inparaenum}[(a)]
\item for some $r_0=(I,O) \in \SimpleRules$ we have $r = \ind{r_0}$; or
\item for some $m \in \Minner$ such that $\norm[\I]{m} \leq B$, $r = \ind[m]{r_0}$ for some 
$r_0=((p,I),(p',\calc,O)) \in \ComplexRules$; or
\item for some $m \in \Minner$ such that $\norm[\I]{m} \leq B$ and for all $\max_{p \in P}(|m(p)|) < B$, $r = \ind[m]{r_0}$ for some $r_0=((p,I),(p',P,O))\discn \in \TransferRules$.
\end{inparaenum}

Before we do a case analysis on $r_0$ let us have a look at the simple places.
First of all $1 \leq j \leq i$,  $\enumelem{p}{j} \in \Psimple$ we know that $\cabsfun[i,B](s)(j) + r(j) \geq 0$, 
i.e.~$|s(\enumelem{p}{j})| \geq |I(\enumelem{p}{j})|$ and we know that
$|s(\enumelem{p}{j})| = \infty$ for $i < j \leq \nsimple$.
Hence clearly $s \ominus (I \restriction \Psimple )\in \Configuration^{\infty}$.

We proceed the proof by a case analysis on $r$.
\begin{itemize}
\item \emph{Case: } $r = \ind{r_0}$ for $r_0 \in \SimpleRules$. \newline
In this case we can see that for all $\enumelem{p'}{j} \in \Pcomplex$ 
\begin{align*}
&s(\enumelem{p'}{j})(\vec{0}) = \cabsfun[i,B](s)\paren{i+j\paren{B+1}^{\ninner}}  \\
&\;\;\;\; \geq \cabsfun[i,B](I)\paren{i+j\paren{B+1}^{\ninner}} = I(\enumelem{p'}{j})(\vec{0}).
\end{align*}
Hence $s \ominus I \in \Configuration^{\infty}$.
Further we know that 
$$\cabsfun[i,B](s)+r = \cabsfun[i,B](s)-\cabsfun[i,B](I)+\cabsfun[i,B](O)$$
and since by above $s \ominus I \in \Configuration^{\infty}$ we get
$$\cabsfun[i,B](s)+r = \cabsfun[i,B](s \ominus I \oplus O)$$
if we define $s' = s \ominus I \oplus O$ then clearly $s \to[TS={\calN_i}] s'$ which is what we wanted to prove.

\item \emph{Case: } $r = \ind[m]{r_0}$ for $r_0 \in \ComplexRules$. \newline
Suppose $p = \enumelem{p'}{j_{\ninner}}$ and $j_0,\ldots,j_{\ninner-1} \in \N^{\leq B}$ such that $m \in \cconfunI(j_0,\ldots,j_{\ninner-1})$
then by definition $\ind[m]{r_0}\varparen{i+\sum_{k=0}^{\ninner}j_k\varparen{B+1}^k} = -1$, 
since $\cabsfun[i,B](s) \to_{\cabspetrinet} \cabsfun[i,B](s) + r$ 
we must have that 
\begin{align*}
&\cabsfun(s)\paren{\theta} + \ind[m]{r_0}\paren{\theta} \geq 0
\end{align*}
where $\theta = i+\sum_{k=0}^{\ninner}j_k\paren{B+1}^k$.
Hence we can deduce that
$$\sum_{m^{\I}\in \cconfunI(j_0,\ldots,j_{\ninner-1})} |s(\enumelem{p'}{j_{\ninner}})(m^{\I})| \geq 1$$
Hence there exists a $m_0 \in \cconfunI(j_0,\ldots,j_{\ninner-1})$ such that $|s(\enumelem{p'}{j_{\ninner}})(m_0)| \geq 1$.

It is easy to see that $s \ominus I \ominus \update{}{p}{\mset{m_0}} \in \Configuration^\infty$
$$s \to[TS={\calN_i}] \paren{s \ominus I \ominus \update{}{p}{\mset{m_0}}}\oplus O \oplus \update{}{p'}{\mset{m_0 \oplus \calc}} =: s'$$
We can then use Lemma~\ref{lemma:counter_abstraction_homomorphism} to see that
\begin{align*}
\cabsfun[i,B](s')   =&\, \cabsfun[i,B](s) - \cabsfun[i,B](I)
					-\cabsfun[i,B]\paren{\update{}{p}{\mset{m_0}}}\\
					&+\cabsfun[i,B]\paren{O}
					+\cabsfun[i,B]\paren{\update{}{p'}{\mset{m_0 \oplus \calc}}}
\end{align*}
Since $m,m_0 \in \cconfunI(j_0,\ldots,j_{\ninner-1})$ and Lemma~\ref{lemma:counter_abstraction_Mabs} yields
\begin{gather*}
\cabsfun[i,B]\paren{\update{}{p}{\mset{m_0}}} = \cabsfun[i,B]\paren{\update{}{p}{\mset{m}}}\\
\cabsfun[i,B]\paren{\update{}{p'}{\mset{m_0 \oplus \calc}}} = \cabsfun[i,B]\paren{\update{}{p'}{\mset{m \oplus \calc}}}
\end{gather*}
Hence we can see that 
\begin{align*}
\cabsfun[i,B](s')   =&\, \cabsfun[i,B](s) - \cabsfun[i,B](I)
					-\cabsfun[i,B]\paren{\update{}{p}{\mset{m}}}\\
					&+\cabsfun[i,B]\paren{O}
					+\cabsfun[i,B]\paren{\update{}{p'}{\mset{m \oplus \calc}}}\\
					=&\, \cabsfun[i,B](s) + r
\end{align*}
which is what we wanted to prove.

\item \emph{Case: } $r = \ind[m]{r_0}$ for $r_0 \in \TransferRules$. \newline
Suppose $p = \enumelem{p'}{j_{\ninner}}$ and $j_0,\ldots,j_{\ninner-1} \in \N^{\leq B}$ such that $m \in \cconfunI(j_0,\ldots,j_{\ninner-1})$
then by definition $\ind[m]{r_0}\varparen{i+\sum_{k=0}^{\ninner}j_k\varparen{B+1}^k} = -1$, 
since $\cabsfun[i,B](s) \to_{\cabspetrinet} \cabsfun[i,B](s) + r$ 
we must have that 
$$\cabsfun(s)\paren{\theta} + \ind[m]{r_0}\paren{\theta} \geq 0.$$
where $\theta = i+\sum_{k=0}^{\ninner}j_k\paren{B+1}^k$.
Hence we can deduce that
$$\sum_{m^{\I}\in \cconfunI(j_0,\ldots,j_{\ninner-1})} |s(\enumelem{p'}{j_{\ninner}})(m^{\I})| \geq 1$$
Hence there exists a $m_0 \in \cconfunI(j_0,\ldots,j_{\ninner-1})$ such that $|s(\enumelem{p'}{j_{\ninner}})(m_0)| \geq 1$.

Hence it is easy to see that $s \ominus I \ominus \update{}{p}{\mset{m_0}} \in \Configuration^\infty$ and
$$s \to[TS={\calN_i}] s'$$
where
\begin{align*}
s' &= s \ominus \update{}{p}{\mset{m_0}} \ominus I
	\oplus \update{}{p'}{\mset{{m_0}_{\overline{P}}}} \oplus O \oplus \paren{{m_0}_{P} \circ \colmap^{-1}} 
\end{align*}
where ${m_0}_{P} = m_0 \restriction P$ and ${m_0}_{\overline{P}} \restriction \varparen{\Pinner \setminus P}$.
Since $s \ominus \update{}{p}{\mset{m_0}} \ominus I \in \Configuration^{\infty}$
Lemma~\ref{lemma:counter_abstraction_homomorphism} yields
\begin{align*}
\cabsfun[i,B](s') &= \cabsfun[i,B](s) - \cabsfun[i,B]\paren{\update{}{p}{\mset{m_0}}} - \cabsfun[i,B](I) \\
				  &\;\;\;\;+\cabsfun[i,B]\paren{\update{}{p'}{\mset{\update[\mid p_0 \in P]{m_0}{p_0}{0}}}} \\
				  &\;\;\;\;+\cabsfun[i,B](O) + \cabsfun[i,B](\paren{(m_0\restriction P) \circ \colmap^{-1}+\vec{0}}).
\end{align*}
Since $m,m_0 \in \cconfunI(j_0,\ldots,j_{\ninner-1})$ we can apply Lemma~\ref{lemma:counter_abstraction_Mabs} and Lemma~\ref{lemma:counter_abstraction_Meject} to get
\begin{gather*}
\cabsfun[i,B]\paren{\update{}{p}{\mset{m}}} = \cabsfun[i,B]\paren{\update{}{p}{\mset{m_0}}}\\
\cabsfun[i,B]\paren{\update{}{p'}{\mset{m_{\overline{P}}}}} = \cabsfun[i,B]\paren{\update{}{p'}{\mset{{m_0}_{\overline{P}}}}}
\end{gather*}
where $m_{\overline{P}} = m \restriction (\Pinner \setminus P)$.
Further we notice by assumption $\max{p \in P} |m(p)| < B$  
and thus for all $1 \leq k \leq \Pinner$ such that $\enumelem{p^{\I}}{k} \in P$ we have
\begin{align*}
m(\enumelem{p}{k}) &= \min(m(\enumelem{p}{k}),B) = j_{k-1} \\
				   &= \min(m_0(\enumelem{p}{k}),B) = m_0(\enumelem{p}{k})
\end{align*}
and hence we can see that
$$\paren{m_P \circ \colmap^{-1}} = \paren{{m_0}_{P} \circ \colmap^{-1} + \vec{0}}$$
where $m_P = m\restriction P$ which implies
\begin{align*}
\cabsfun[i,B](s') &= \cabsfun[i,B](s) - \cabsfun[i,B]\paren{\update{}{p}{\mset{m}}} - \cabsfun[i,B](I) \\
				  &\;\;\;\;+\cabsfun[i,B]\paren{\update{}{p'}{\mset{m_{\overline{P}}}}} \\
				  &\;\;\;\;+\cabsfun[i,B](O) + \cabsfun[i,B](m_P \circ \colmap^{-1})\\
				  &= \cabsfun[i,B](s') + r
\end{align*}
which is what we wanted to prove.
\end{itemize}
Thus in all cases also (ii) holds and hence we can conclude the proof.
\end{proof}

\begin{definition}[Simulation, bisimulation]
Suppose $(S,\to[label=u,LTS=S])$ and $(S',\to[label=u,LTS={S'}])$ are labelled transition systems we say a relation $\calB \subseteq S \times S'$ is a 
\emph{(weak) simulation} if for all $(s,s') \in \calB$, if for some $t \in S$ we have $s \to[label=u,LTS=S] t$ ($s \to[label=u,LTS=S,*] t$) then there exists 
$t' \in S'$ such that $s' \to[label=u,LTS=S'] t'$ ($s' \to[label=u,LTS=S',*] t'$) and $(t,t') \in \calB$. 
We say $\calB$ is a \emph{bisimulation} relation just if both $\calB$ and $\calB^{-1}$ are simulation relations.
\end{definition}\noindent
%
Let us temporarily label the transition systems relations of $\calN_i$, $\cabspetrinet$ in the following way:
$s \to[label=s',LTS=\calN_i] s'$ if $s \to[TS=\calN_i] s'$;
$v \to[label=u,LTS=\cabspetrinet] v'$ if $v \to[TS=\cabspetrinet] v'$ where  $u = s'$ if $\cabsfun[i,B](s') = v'$ and $u = \epsilon$ otherwise. 

\begin{proposition}\label{prop:bisimulation_counter_abstraction}
The relation 
$\commabr{\{\paren{s,\cabsfun[i,B](s)} \mid s \in \BTS{\Configuration^\infty}{}{B}, |s(\enumelem{p}{j})| = \infty, i < j\}}$ is a bisimulation between the labelled transition systems $\BTS{\calN_i}{}{B}$ and $\BTS{\cabspetrinet}{}{B}$. 
\end{proposition} 
\begin{proof}
Let $\calB = \commabr{\{\paren{s,\cabsfun[i,B](s)} \mid s \in \BTS{\Configuration^\infty}{}{B}, |s(\enumelem{p}{j})| = \infty, i < j\}}$ and suppose $(s,\ind{s}) \in \calB$. 
By definition we know that $\ind{s} = \cabsfun[i,B](s)$.
Suppose $s \to[label=s',LTS=\BTS{\calN_i}{}{B}] s'$ is a transition then Lemma~\ref{lemma:bisimulation_lemma_lockstep} yields that
$\cabsfun[i,B](s) \to[label=s',LTS=\cabspetrinet] \cabsfun[i,B](s')$.
Clearly $(s,\cabsfun(s')) \in \calB$ since $s' \in \BTS{\Configuration^\infty}{}{B}$ which also implies that for all
$1 \leq j \leq i$ we have $|\cabsfun(s')(j)| = |s'(\enumelem{p}{j})| < B$ and thus 
$\norm[\cabspetrinet]{\cabsfun(s')} < B$.
Hence $\cabsfun[i,B](s) \to[label=s',LTS=\BTS{\cabspetrinet}{}{B}]{\cabspetrinet} \cabsfun[i,B](s')$ is transition and so we can conclude that $\calB$ is a simulation relation.

Let us now consider $\calB^{-1}$. Suppose $(s,\ind{s}) \in \calB$, again we know that $\ind{s} = \cabsfun[i,B](s)$. Suppose $\cabsfun[i,B](s) \to[label=u,LTS=\BTS{\cabspetrinet}{}{B}] \cabsfun[i,B](s) + r$ is a transition for some $r \in R_{i,B}$. Lemma~\ref{lemma:bisimulation_lemma_lockstep} then yields that there exists 
$s' \in \Configuration^\infty$ such that $\cabsfun[i,B](s') = \cabsfun[i,B](s) + r$ and $s \to[label=s',LTS={\calN_i}] s'$. Hence we can conclude $ u = s'$.
Further since $\norm[\cabspetrinet]{\cabsfun[i,B](s) + r} < B$ we know that $\norm[\cabspetrinet]{\cabsfun[i,B](s')} < B$ which implies that that for all $1 \leq i \leq j$ we have $|s'(\enumelem{p}{j})| = \cabsfun(s')(i) < B$ and thus $\norm{\calN_i}{s'} < B$.
Hence $(s,\cabsfun(s')) \in \calB$ and so we know that $\calB^{-1}$ is a simulation.

We can thus conclude that $\calB$ is a bisimulation.
\end{proof}
To establish the relationship between $\cabspnetord$, $\cabspetrinet$ and $\calN_i$ we relabel their transition relations 
temporarily with a label indicating which rule set was used: $\vsimrule$ for a rule from $\cabspnetrules$ and $\epsilon$ otherwise:\newline
$v \to[label=\vsimrule,LTS=\cabspnetord] v'$ if $v \to[rule=r,TS=\cabspetrinet] v'$, and
$v \to[label=\epsilon,LTS=\cabspnetord] v'$ if $v \to[rule=r,TS=\cabspnetord] v'$ and $r \in \cabspnetordrules$;
and for $\cabspetrinet$ and $\calN_i$ we label all transitions with $\vsimrule$.

\begin{proposition}\label{prop:weaksimulation_counter_abstraction}
The relation $\{\paren{\cabsfun[i,B](s),\cabsfun[i,B](s)} \mid s \in \Configuration^\infty\}$ is a simulation relation for $\cabspetrinet$ and $\cabspnetord$. 
The relation $\{\paren{\cabsfun[i,B](s),s'} \mid s \leqconfig s', s,s' \in \Configuration^\infty,|s(\enumelem{p}{j})| = \infty, i < j\}$ is a weak simulation for $\cabspnetord$ and $\calN_{i}$. 
\end{proposition} 
\begin{proof}
That $\{\paren{\cabsfun[i,B](s),\cabsfun[i,B](s)} \mid s \in \Configuration^\infty\}$ is a simulation relation for $\cabspetrinet$ and $\cabspnetord$ is obvious, since if $v \to[rule=\vsimrule,TS=\cabspetrinet] v'$ then $v \to[rule=\vsimrule,TS=\cabspnetord] v'$.

For the other direction let $\calW = \{\paren{\cabsfun[i,B](s),s'} \mid s \leqconfig s', s,s' \in \Configuration^\infty,|s(\enumelem{p}{j})| = \infty, i < j\}$.
Suppose $(\cabsfun[i,B](s), \discn \cabsfun[i,B](s')) \in \calW$ then we know that $s \leqconfig s'$.
Suppose that $\cabsfun[i,B](s) \to[label=\vsimrule,LTS=\cabspnetord,*] v$ then we may split this transition up
$\cabsfun[i,B](s) \to[label=\epsilon,LTS=\cabspnetord] v_1 \to[label=\epsilon,LTS=\cabspnetord] \cdots \to[label=\epsilon,LTS=\cabspnetord] v_{n-1} \to[label=\vsimrule,LTS=\cabspnetord] v_n \to[label=\epsilon,LTS=\cabspnetord] \cdots \to[label=\epsilon,LTS=\cabspnetord] v$.

Let us prove a Lemma:
\begin{lemma*}
If $\cabsfun[i,B](s_0) \to[rule=r,TS={\cabspetrinet}] \cabsfun[i,B](s_0) + r$ with $r \in \cabspnetordrules$
then there exists $s_1 \in \Config^\infty$ such that $\cabsfun[i,B](s_1) = \cabsfun[i,B](s_0) + r$ and $s_1 \leqconfig s_0$.
\end{lemma*}
\begin{proof}
Since $r \in \cabspnetordrules$ we know $r = \cabsfun[i,B](M_{p,m'})-\cabsfun[i,B](M_{p,m})$ for some $m,m' \in \Minner$, 
such that both $\norm[\I]{m},\norm[\I]{m'} \leq B$ and $m >_{\Minner} m'$.
Suppose $p = \enumelem{p'}{j_{\ninner}}$ and $j_0,\ldots,j_{\ninner-1} \in \N^{\leq B}$ such that $m \in \cconfunI(j_0,\ldots,j_{\ninner-1})$
then by definition 
$$r\paren{i+\sum_{k=0}^{\ninner}j_k\paren{B+1}^k} = -1,$$ 
since $\cabsfun[i,B](s_0) \to_{\cabspnetord} \cabsfun[i,B](s_0) + r$. 
we must have that 
\begin{align*}
\cabsfun(s_0)&\paren{\theta} + r\paren{\theta} \geq 0
\end{align*}
where $\theta= i+\sum_{k=0}^{\ninner}j_k\paren{B+1}^k$.
Hence we can deduce that
$$\sum_{m^{\I}\in \cconfunI(j_0,\ldots,j_{\ninner-1})} |s_0(\enumelem{p'}{j_{\ninner}})(m^{\I})| \geq 1$$
Thus there exists a complex token $m_0 \in \cconfunI(j_0,\ldots,j_{\ninner-1})$ such that $|s_0(\enumelem{p'}{j_{\ninner}})(m_0)| \geq 1$.

Since $\norm[\I]{m} \leq B$ we know that 
\begin{align*}
m(\enumelem{p^{\I}}{k}) &= \min(m(\enumelem{p^{\I}}{k}),B) = j_{k-1} \\
	&= \min(m_0(\enumelem{p^{\I}}{k}),B) \leq m_0(\enumelem{p^{\I}}{k})
\end{align*}
for all $1 \leq k \leq \ninner$
and hence $m' <_{\Minner} m \leq_{\Minner} m_0$.

Clearly $s_0 \ominus s_{p,m_0} \in \Configuration^{\infty}$ and hence Lemma~\ref{lemma:counter_abstraction_homomorphism} gives
\begin{align*}
\cabsfun[i,B](s_0 \oplus s_{p,m'} \ominus s_{p,m_0}) &= \cabsfun[i,B](s_0) + \cabsfun[i,B](s_{p,m'})\\
&\;\; -\cabsfun[i,B](s_{p,m_0})
\end{align*}
since $m,m_0 \in \cconfunI(j_0,\ldots,j_{\ninner-1})$ Lemma~\ref{lemma:counter_abstraction_Mabs} gives
\begin{align*}
\cabsfun[i,B](s_0 \oplus s_{p,m'} \ominus s_{p,m_0}) &= \cabsfun[i,B](s_0) + \cabsfun[i,B](s_{p,m'})\\
&\;\; -\cabsfun[i,B](s_{p,m})\\
&= \cabsfun[i,B](s_0) + r 
\end{align*}
It is easy to see that $s_1 \is s_0 \oplus s_{p,m'} \ominus s_{p,m_0} \leqconfig s$ 
which concludes the proof of the Lemma.
\end{proof}

The Lemma allows us to conclude that there are $s_j \in \Config^\infty$ such that $v_j = \cabsfun[i,B](s_j)$ 
for all $j \in \range{n-1}$, $|s_j(\enumelem{p}{j'})| = \infty$ for all $i < j'$ and $s_j \leqconfig s_{j-1} \leqconfig s$ for $1 < j < n$.
Inspecting the transition $\cabsfun[i,B](s_{n-1}) \to[label=\vsimrule,LTS=\cabspnetord] v_n$
then due to the labelling we know that $\cabsfun[i,B](s_{n-1}) \to[TS=\cabspetrinet] v_n$.
Lemma~\ref{lemma:bisimulation_lemma_lockstep} then yields that there is a $s'' \in \Configuration^\infty$
such that $s_{n-1} \to[TS={\calN_i}] s''$ and $v_n = \cabsfun[i,B](s'')$.
Since $s \leqconfig s'$ and $\calN_i$ is a WSTS we obtain $\exists s''' \in \Configuration^\infty$
such that $s' \to[label=\vsimrule,LTS=\calN_i] s'''$ and $s'' \leqconfig s'''$.
Applying the above lemma further we can see that $v = \cabsfun[i,B](s'''')$ for some $s''' \leqconfig s''$
which implies $s'''' \leqconfig s'''$ and hence that
$(\cabsfun[i,B](s'''),s''') \in \calW$ which is what we wanted to prove.

Thus $\calW$ is a weak simulation.
\end{proof}
\begin{corollary}\label{cor:bounded_path_implies_bounded_counter_abstraction_path}
Let $1 \leq i \leq \nsimple$, $B \in \N$.
Suppose $s,s' \in \Configuration^\infty$ such that
there exists a covering path from $s$ for $s'$ in $\BTS{\calN_i}{}{B}$ then
there exists a covering path from $\cabsfun[i,B](s)$ for $\cabsfun[i,B](s')$ in $\BTS{\cabspnetord}{}{B}$.
\end{corollary}
\begin{proof}
Let $\vec{s}$ be a covering path for $s'$ from $s$ in $\BTS{\calN_i}{}{B}$.
Then Proposition~\ref{prop:bisimulation_counter_abstraction} yields a bisimulation for $\BTS{\calN_i}{}{B}$ and $\BTS{\cabspetrinet}{}{B}$. Hence
$\cabsfun[i,B](\vec{s}) = \cabsfun[i,B](\vec{s}(1)) \cdots \cabsfun[i,B](\vec{s}(|\vec{s}|))$ is a path in
$\BTS{\cabspetrinet}{i}{B}$ from $\cabsfun[i,B](s)$ to $\cabsfun[i,B](\vec{s}(|\vec{s}|))$.

It is immediate that $\cabsfun[i,B](\vec{s})$ is also a path in $\BTS{\cabspnetord}{}{B}$. 

Since $\vec{s}$ is a covering path for $s'$ in $\BTS{\calN}{i}{B}$ we have $s' \leqconfig \vec{s}(|\vec{s}|)$
and hence for all $1 \leq j \leq i$ we have $|\vec{s}(|\vec{s}|)(\enumelem{p}{j})| \geq |s'(\enumelem{p}{j})|$,
further for all $1 \leq j \leq \ncomplex$ we have 
$\vec{s}(\vec{s})(\enumelem{p'}{j}) \leq_{\M[\Minner]} s'(\enumelem{p'}{j})$.
Hence we know that for all $1 \leq j \leq \ncomplex$
\begin{align*}
s'(\enumelem{p'}{j}) = \mset{m^{j}_1,\ldots,m^{j}_{k_{j}}} &&\text{and}&&
\vec{s}(\vec{s})(\enumelem{p'}{j}) = \mset{{m'}^{j}_1,\ldots,{m'}^{j}_{k'_{j}}}
\end{align*}
such that for all $1 \leq j \leq \ncomplex$ and $1 \leq j' \leq k_j$ we have $m^j_{j'} \leq_{\Minner} {m'}^j_{j'}$ (where we wlog assume the injection $h(j') = j'$).

We will now extend $\cabsfun[i,B](\vec{s})$ to a new path $\cabsfun[i,B](\vec{s})\vec{s}'$ so that 
$\vec{s}'(|\vec{s}'|)$ covers $\cabsfun[i,B](s')$ in $\cabspetrinet$. This will be achieved by using rules in $\cabspnetordrules$
that will allow us to swap the counter abstraction representation of the complex tokens ${m'}^j_{j'}$ for $m^j_{j'}$.

Let us write $t_0 = \vec{s}(\vec{s})$ using the rules in $\cabspnetordrules$
$$\paren{\cabsfun[i,B](M_{\enumelem{p}{1},m^1_{j'}})-\cabsfun[i,B](M_{\enumelem{p}{1},{m'}^1_{j'}})),\cabsfun[i,B](M_{\enumelem{p}{1},{m'}^1_{j'}}))}$$ 
for $1 \leq j \leq k_1$
we can see using Lemma~\ref{lemma:counter_abstraction_homomorphism} that
\begin{align*}
 \cabsfun[i,B]\paren{t_0} &\to_{\cabspnetord} 
 	\cabsfun[i,B]\varparen{\varupdate{t_0}{\enumelem{p}{1}}{\paren{s(\enumelem{p}{1}) \ominus \varmset{{m'}^1_{1}} \oplus \varmset{{m}^1_{1}}}} }\\
 	&\to_{\cabspnetord} \cdots \\
 	&\begin{aligned}
 	\to_{\cabspnetord} \cabsfun[i,B]\varparen{\varupdate{t_0}{\enumelem{p}{1}}{\varparen{s(\enumelem{p}{1}) 
 		&\ominus \varmset{{m'}^1_{1},\ldots,{m'}^1_{k_1}} \\
 	 	&\oplus \varmset{{m}^1_{j'},\ldots,m^1_{k_1}}}}}
 	\end{aligned}
 \end{align*} 
Clearly 
\begin{align*}
&\update{t_0}{\enumelem{p}{1}}{\paren{s(\enumelem{p}{1}) \ominus \mset{{m'}^1_{1},\ldots,{m'}^1_{k_1}} \oplus \mset{{m}^1_{j'},\ldots,m^1_{k_1}}}} \\
&\;\;\;\;\;\;\;= \update{t_0}{\enumelem{p}{1}}{\mset{{m}^1_{1},\ldots,m^1_{k_1},{m'}^1_{k_1+1},\ldots,m{'}^1_{k'_1}}} =: t_1
\end{align*}
and we note that $s' \leq_{\BTS{\calN}{i}{B}} t_1$.
We can simply carry on in this fashion and obtain
\begin{align*}
 \cabsfun[i,B]\paren{t_1} &\to_{\cabspnetord}^* \cabsfun[i,B]\paren{t_2} 
 \to_{\cabspnetord}^* \cdots \to_{\cabspnetord}^* \cabsfun[i,B]\paren{t_{\ncomplex}}
 \end{align*} 
and for all $1 \leq j \leq \ncomplex$ we have
\begin{align*}
t_j = \update{t_{j-1}}{\enumelem{p}{j}}{\mset{{m}^{j}_{1},\ldots,m^j_{k_j},{m'}^j_{k_j+1},\ldots,m{'}^j_{k'_j}}}
\end{align*}

By construction we know that $\cabsfun[i,B](\vec{s}(|\vec{s}|)) \to_{\cabspnetord}^* \cabsfun[i,B]\paren{t_{\ncomplex}}$ and so there is a path $\vec{s}'$ from $\cabsfun[i,B](\vec{s}(|\vec{s}|))$ to $\cabsfun[i,B]\paren{t_{\ncomplex}}$ and we further have for all $p \in \Pcomplex$, $m \in \Minner$, $|s'(p)(m)| \leq |t_{\ncomplex}(p)(m)|$,
for all $p \in \Psimple$, $|s'(p)| \leq |t_{\ncomplex}(p)|$.

Thus it is routine to check that $\cabsfun[i,B](s') \leq_{\N^{\cabspnetdim}} \cabsfun[i,B]\paren{t_{\ncomplex}}$
and thus $\cabsfun[i,B](\vec{s})\vec{s}'$ is a covering path for $\cabsfun[i,B](s')$ from $\cabsfun[i,B](s)$ 
and we note that clearly $\norm[\cabspetrinet]{\cabsfun[i,B](\vec{s})\vec{s}'}^* < B$.

Hence we can conclude that $\dist{s}{s'}{\BTS{\calN}{i}{B}} > 0 \implies \dist{\cabsfun[i,B](s)}{\cabsfun[i,B](s')}{\BTS{\cabspnetord}{i}{B}} > 0$.
\end{proof} 
As a consequence we can reason about covering distance on $\cabspnetord$ rather than $\calN_i$ for configuration pairs in $S_{i,B}$.

\begin{corollary}\label{cor:counter_abstraction_yields_quasi_path}
Let $i \leq \nsimple$, $B \in \N$.
For all $(s,s') \in S_{i,B}$ we have
$\dist{s}{s'}{\calN_i} \leq \dist{\cabsfun[i,B](s)}{\cabsfun[i,B](s')}{\cabspnetord}$.
\end{corollary}
\begin{proof}
Since $(s,s') \in S_{i,B}$ we know that there exists a covering path $\vec{s}$ in $\calN_i$ from $s$ for $s'$
with norm $\norm[\calN_i]{\vec{s}}^* < B$, i.e.~$\vec{s}$ is a path in $\BTS{\calN_i}{}{B}$.

Corollary~\ref{cor:bounded_path_implies_bounded_counter_abstraction_path} then yields that there exists
a covering path $\vec{s}_0$ from $\cabsfun[i,B](s)$ for $\cabsfun[i,B](s')$ in $\BTS{\cabspnetord}{}{B}$

In particular
$\vec{s}_0$ is a path in $\cabspnetord$ from $\cabsfun[i,B](s)$ that covers $\cabsfun[i,B](s')$.
Hence we can find a path $\vec{s}_1$ in $\cabspnetord$ from $\cabsfun[i,B](s)$ that covers $\cabsfun[i,B](s')$
such that $|\vec{s}_1| = \dist{\cabsfun[i,B](s)}{\cabsfun[i,B](s')}{\cabspnetord}$.

Proposition~\ref{prop:weaksimulation_counter_abstraction} then gives us that $\calW$ is a weak simulation for $\cabspnetord$ and $\calN_i$,
and since clearly
$(\cabsfun[i,B](s),s) \in \calW$ an easy induction on the length of $\vec{s}_1$ gives us a path $\vec{s}'$ in $\calN_i$
such that $|\vec{s}'| \leq |\vec{s}_1| = \dist{\cabsfun[i,B](s)}{\cabsfun[i,B](s')}{\cabspnetord}$ and 
$(\cabsfun[i,B](s''),\vec{s}'(|\vec{s}'|)) \in \calW$
where $\vec{s}_1(|\vec{s}|_1) = \cabsfun[i,B](s'')$
which by definition means $s'' \leqconfig \vec{s}(|\vec{s}|)$

Since $\cabsfun[i,B](s'') \geq_{\N^{\cabspnetdim}} \cabsfun[i,B](s')$ we know that
for all $1 \leq j \leq \nsimple$, $|s''(\enumelem{p}{j})| \geq |s'(\enumelem{p}{j})|$.
In order to compare the complex places let us fix $p \in \Pcomplex$
and enumerate all elements of $\paren{\N^{\leq B}}^{\ninner}$ by $\vec{j}_1,\ldots,\vec{j}_{N}$
where $N = |\paren{\N^{\leq B}}^{\ninner}|$.
Since $\cabsfun[i,B](s'') \geq_{\N^{\cabspnetdim}} \cabsfun[i,B](s')$
we know that for all $1 \leq k \leq N$
$$\sum_{m^{\vec{j}_k} \in \Xi} s'(p)(m^{\vec{j}_k}) \leq 
\sum_{m^{\vec{j}_k} \in \Xi} s''(p)(m^{\vec{j}_k}) $$
where we abbreviate $\Xi = \cconfunI(\vec{j}_k(1),\ldots,\vec{j}_k(\ninner))$.
Hence we can conclude
that 
\begin{gather*}
s'(p) = \mset{m^{\vec{j}_1}_1,\ldots,m^{\vec{j}_1}_{n_1},\ldots,m^{\vec{j}_N}_1,\ldots,m^{\vec{j}_N}_{n_N}}\\
s''(p) = \mset{{m'}^{\vec{j}_1}_1,\ldots,{m'}^{\vec{j}_1}_{n'_1},\ldots,m^{\vec{j}_N}_1,\ldots,m^{\vec{j}_N}_{n'_N}}
\end{gather*}
where for all $1 \leq k \leq N$ and $1 \leq k' \leq n_k$, $m^{\vec{j}_k}_{k'},{m'}^{\vec{j}_k}_{k'} \in \cconfunI(\vec{j}_k(1),\ldots,\vec{j}_k(\ninner))$
and
\begin{align*}
n_k = \sum_{m^{\vec{j}_k} \in \Xi} s'(p)(m^{\vec{j}_k}),  &&
n'_k = \sum_{m^{\vec{j}_k} \in \Xi} s''(p)(m^{\vec{j}_k}) 
\end{align*}
and hence $n_k \leq n'_k$.
Let us pair up $m^{\vec{j}_k}_{k'}$ and ${m'}^{\vec{j}_k}_{k'}$ for all $1 \leq k \leq N$ and $1 \leq k' \leq n_k$. 
Since $\norm[\calN_i;C]{s'} \leq B$ we know that $\norm[\I]{m^{\vec{j}_k}_{k'}} \leq B$ and hence since
$m^{\vec{j}_k}_{k'},{m'}^{\vec{j}_k}_{k'} \in \cconfunI(\vec{j}_k(1),\ldots,\vec{j}_k(\ninner))$
we can deduce that $m^{\vec{j}_k}_{k'} \leqMinner {m'}^{\vec{j}_k}_{k'}$.
Hence we can deduce that
$s'(p) \leq_{\M[\Minner]} s''(p)$ for all $p \in \Pcomplex$ and thus $s' \leqconfig s''$.
Since $s'' \leqconfig \vec{s}'(|\vec{s}'|)$ we can conclude that $s' \leqconfig \vec{s}'(|\vec{s}'|)$.

Hence $\vec{s}'$ is a path in $\calN_i$ from $s$ that covers $s'$ and $|\vec{s}'| \leq \dist{\cabsfun[i,B](s)}{\cabsfun[i,B](s')}{\cabspnetord}$ from which we can conclude
that $\dist{s}{s'}{\calN_i} \leq \dist{\cabsfun[i,B](s)}{\cabsfun[i,B](s')}{\cabspnetord}$.
\end{proof}

\begin{customtheorem}[\ref{thm:existence_counter_abstraction_petri_net}]
For all $i \leq \nsimple$, $B \in \N$ there exists a Petri net $\cabspnetGen = (\cabspnetdimGen, \cabspnetrulesGen)$ and a function $\cabsfun[i,B]$ such that
\begin{enumerate}[({A}1)]
\item $\cabspnetdimGen \leq i + \ncomplex \times \paren{B+1}^{\ninner}$, 
\item $R \geq \max\varset{r(i) \mid r \in \cabspnetrulesGen}$, and 
\item for all $s, s' \in S_{i,B}$:  $R' \geq \max_{j \in \range{\cabspnetdimGen}}(\cabsfun[i,B](s)(j))$,
\item $\dist{s}{s'}{\calN_i} \leq \dist{\cabsfun[i,B](s)}{\cabsfun[i,B](s')}{\cabspnetGen}$.
\end{enumerate}
\end{customtheorem}
\begin{proof}
Taking $\cabspnetGen = \cabspnetord$ it is easy to see that \ref{thm:existence_counter_abstraction_petri_net:i}--\ref{thm:existence_counter_abstraction_petri_net:iii} hold. \ref{thm:existence_counter_abstraction_petri_net:i} follows from \Cref{cor:counter_abstraction_yields_quasi_path}.
\end{proof}

\subsubsection{Proofs of Section~\ref{sec:covdiameter}}

\begin{corollary}\label{cor:pre_length_bound_quasi_path}
Let $1 \leq i \leq \nsimple$, $B \in \N$. Then
$$\diameter[\calN_{i}]{S_{(i,B)}} \leq \max\set{\length{\cabspnetord}{\cabsfun[i,B](s')} \mid \norm[\calN_i;C]{s'} \leq B}$$
\end{corollary}
\begin{proof}
Let $(s,s') \in S_{(i,B)}$ then Corollary~\ref{cor:counter_abstraction_yields_quasi_path} gives that 
\begin{align*}
\dist{s}{s'}{\calN_i} &\leq \dist{\cabsfun[i,B](s)}{\cabsfun[i,B](s')}{\cabspnetord}\\
					  &\leq \length{\cabspnetord}{\cabsfun[i,B](s')}
\end{align*}
because $\cabsfun[i,B](s) \in \N^{\cabspnetdim}$.
Taking $\max$ over $S_{(i,B)}$:
\begin{align*}
\diameter[\calN_{i}]{S_{(i,B)}} &=
\max\set{\dist{s}{s'}{\calN_i} \mid (s,s') \in S_{(i,B)}} \\
&\leq \max\set{\length{\cabspnetord}{\cabsfun[i,B](s') \mid (s,s') \in S_{i,B}}}\\
&\leq \max\set{\length{\cabspnetord}{\cabsfun[i,B](s')} : \norm[\calN_i;C]{s'} < B}
\end{align*}
which concludes the proof.
\end{proof}

\begin{customcorollary}[\ref{cor:length_bound_quasi_path}]
Let $i \leq \nsimple$, $B \in \N$. Then
\begin{align*}
\diameter[\calN_{i}]{S_{(i,B)}} &\leq \max\set{\length{\cabspnetGen}{\cabsfun[i,B](s')} : \norm[\calN_i;C]{s'} \leq B}\\
&\leq (6\max\set{R,B,1}\max\set{R',B})^{\paren{\cabspnetdim+1}!}. 
\end{align*}
\end{customcorollary}
\begin{proof}
Let $s$ be such that $\norm[\calN_i;C]{s} \leq B$, then
Lemma~\ref{lemma:bound_lemma_bonnet} \cite[Lemma 12]{Bonnet:12} applies to give us 
$\length{\cabspnetord}{\cabsfun[i,B](s)} \leq (6R_{\cabspnetord}R'_{\cabspnetord})^{(\cabspnetdim+1)!}$.

Further it is easy to see that $R_{\cabspnetord} \leq \max\set{R,B,1}$ and $R'_{\cabspnetord} \leq \max\set{R',B}$ and hence
we obtain
$\length{\cabspnetord}{\cabsfun[i,B](s)} \leq (6\max\set{R,B,1}\max\set{R',B})^{(\cabspnetdim+1)!}$.
Invoking Corollary~\ref{cor:pre_length_bound_quasi_path} yields:
\begin{align*}
\diameter[\calN_{i}]{S_{(i,B)}} 
	&\leq \max\set{\length{\cabspnetord}{\cabsfun[i,B](s')} \mid \norm[\calN_i;C]{s'} \leq B}\\
	&\leq \max_{\norm[\calN_i;C]{s'} \leq B}\varparen{(6\max\varset{R,B,1}\max\varset{R',B})^{(\cabspnetdim+1)!}}\\
	&\leq (6\max\set{R,B,1}\max\set{R',B})^{(\cabspnetdim+1)!}
\end{align*}
where the last inequality is justified since the argument of $\max$ does not depend on $s'$ only on $B$.
\end{proof}

\begin{customtheorem}[\ref{thm:covering_radius:bound}]
Let us write $\slog$, super-logarithm, for the inverse of $\tetration[2]{(-)}$, i.e.,~$n=\tetration[2]{\slog(n)}$. Then for all $1 \leq i \leq \nsimple$:
\begin{asparaenum}[(i)]
\item $\length{\calN_{0}}{s_{\text{cov}}} \leq \tetration[2]{(2\slog(48\ncomplex \nsimple \ninner R'))}$
\item $\length{\calN_{i+1}}{s_{\text{cov}}} \leq 2^{2^{\varparen{\varparen{\max\varset{\length{\calN_i}{s_{\text{cov}}},2}}^{48 \ninner \nsimple {\ncomplex}{R'}}}}}$ and
\item $\length{\calN}{s_{\text{cov}}} \leq \tetration[2]{(2\nsimple +2 \slog(48(\nsimple+1) \ninner \nsimple {\ncomplex}{R'}))}$. 
\end{asparaenum}
\end{customtheorem}
\begin{proof}
For claim (i) applying Proposition~\ref{prop:rackoff_recurrence} yields that
$$\length{\calN_{0}}{s_{\text{cov}}} \leq \diameter[\calN_{0}]{S_{(0,R')}}$$
Applying Corollary~\ref{cor:length_bound_quasi_path} then gives us:
\begin{align*}
\length{\calN_{0}}{s_{\text{cov}}} \leq&\, (6\max\set{R,R',1}R')^{(\cabspnetdim[i,R']+1)!} \\
&+ \length{\calN_i}{s_{\text{cov}}}
\end{align*}
Clearly $R' = \max\set{R,R',1}$. Hence
\begin{align*}
\length{\calN_{0}}{s_{\text{cov}}} \leq&\, (6{R'}^2)^{(\cabspnetdim[0,R']+1)!} \\
								    \leq&\, 2^{\log_2(6{R'}^2)(\cabspnetdim[0,R']+1)!}
\end{align*}
Further
\begin{align*}
\log_2(6&R'^2){(\cabspnetdim[0,R']+1)!}+1 \\
	&\leq (4 + 2\log_2(R'))(\cabspnetdim[0,R']+1)! \\
	&\leq (\cabspnetdim[0,R']+6+R')!
\shortintertext{since $4 + 2\log_2(R') \leq \cabspnetdim[0,R']+6+R'$, and hence}
\log_2(6&R'^2){(\cabspnetdim[0,R']+1)!}+1 \\
&\leq 2^{\sum^{\cabspnetdim[0,R']+6+R'}_{k=1} \log_2(k)}\\
&\leq 2^{(\cabspnetdim[0,R']+6+R')^2}\\
&= 2^{\varparen{\varparen{R'+1}\varparen{\varparen{\ncomplex + 1}\varparen{R'+1}^{\ninner-1}-1}+6+R'}^2}\\
&= 2^{\varparen{\varparen{\varparen{\ncomplex + 1}\varparen{R'+1}^{\ninner}}+5}^2}
\end{align*}
And 
\begin{align*}
2\log_2&\varparen{\varparen{\varparen{\ncomplex + 1}\varparen{R'+1}^{\ninner}}+5} \\
&\leq
2\varparen{\log_2\varparen{\ncomplex + 1}+ \ninner\log_2\varparen{R'+1}+\log_2(5)}\\
&\leq 2\ncomplex + 2\ninner R' + 6
\leq 48\ncomplex \nsimple \ninner R'
\end{align*}
Putting it all together gives
$$\length{\calN_{0}}{s_{\text{cov}}} \leq \tetration[2]{(2\slog(48\ncomplex \nsimple \ninner R'))}$$

For the second claim we know from
Proposition~\ref{prop:rackoff_recurrence} that
$$\length{\calN_{i+1}}{s_{\text{cov}}} \leq \diameter[\calN_{i+1}]{S_{(i+1,B_i)}} + \length{\calN_i}{s_{\text{cov}}},$$
where $B_i = R' \times \length{\calN_i}{s_{\text{cov}}}$.
Corollary~\ref{cor:length_bound_quasi_path} then tells us that
\begin{align*}
\length{\calN_{i+1}}{s_{\text{cov}}} \leq&\, (6\max\set{R,B,1}\max\set{R',B})^{(\cabspnetdim+1)!} \\
&+ \length{\calN_i}{s_{\text{cov}}}
\end{align*}
Then since we know $B_i + 1 \geq \max\set{R,R',1}$ the above implies
\begin{align*}
\length{\calN_{i+1}}{s_{\text{cov}}} \leq&\, (6(B_i + 1)^2)^{(\cabspnetdim[i+1,B_i]+1)!} + \length{\calN_i}{s_{\text{cov}}}\\
\leq&\, 2^{\log_2(6(B_i + 1)^2){(\cabspnetdim[i+1,B_i]+1)!}} + 2^{\log_2\paren{\length{\calN_i}{s_{\text{cov}}}}}\\
\leq&\, 2^{\log_2(6(B_i + 1)^2){(\cabspnetdim[i+1,B_i]+1)!}+1}
\end{align*}
where the last line is justified since the inequality $\log_2(6(B_i + 1)^2) \geq \discn\log_2\paren{\length{\calN_i}{s_{\text{cov}}}}$ holds.
Further 
\begin{align*}
\log_2(6(B_i &+ 1)^2){(\cabspnetdim[i+1,B_i]+1)!}+1 \\
	&\leq (4 + 2\log_2(B_i + 1))(\cabspnetdim[i+1,B_i]+1)! \\
	&\leq (\cabspnetdim[i+1,B_i]+6+B_i)!
\shortintertext{since $4 + 2\log_2(B_i + 1) \leq \cabspnetdim[i+1,B_i]+6+B_i$, and hence}
\log_2(6(B_i &+ 1)^2){(\cabspnetdim[i+1,B_i]+1)!}+1 \\
&\leq 2^{\sum^{\cabspnetdim[i+1,B_i]+6+B_i}_{k=1} \log_2(k)}\\
&\leq 2^{(\cabspnetdim[i+1,B_i]+6+B_i)^2}
\end{align*}
Expanding $\cabspnetdim[i+1,B_i]$ then gives
\begin{align*}
\cabspnetdim[i+1,B_i]+6+B_i =&\, i + \paren{B_i+1}\paren{\paren{\ncomplex + 1}\paren{B_i+1}^{\ninner-1}-1}\\
						&+6+B_i + 1 \\
					  	=&\, i + \paren{\ncomplex + 1}\paren{B_i+1}^{\ninner}+6\\
					  	\leq&\, 2^{\log_2(\nsimple)} +  2^{\log_2(\ncomplex + 1)}\paren{B_i+2}^{\ninner}+2^3\\
					  	\leq&\, 2^{3+\log_2(\nsimple)+\log_2(\ncomplex + 1)}\paren{B_i+2}^{\ninner}\\
					  	\leq&\, \paren{B_i+2}^{\ninner+3+\log_2(\nsimple)+\log_2(\ncomplex + 1)}
\end{align*}
Expanding $B_i$
\begin{align*}
B_i+2 &\leq R' \times \length{\calN_i}{s_{\text{cov}}} + 2 \\
	  &\leq (R' + 1) \max\set{\length{\calN_i}{s_{\text{cov}}},2}\\
	  &\leq 2^{\log_2(R' + 1)} \max\set{\length{\calN_i}{s_{\text{cov}}},2}\\
	  &\leq \paren{\max\set{\length{\calN_i}{s_{\text{cov}}},2}}^{\log_2(R' + 1)+1}
\end{align*}
Putting it all together
\begin{align*}
\length{\calN_{i+1}}{s_{\text{cov}}} &\leq 2^{2^{\paren{\paren{\paren{\max\set{\length{\calN_i}{s_{\text{cov}}},2}}^{\log_2(R' + 1)+1}}^{\ninner+3+\log_2(\nsimple)+\log_2(\ncomplex + 1)}}^2}}.
\end{align*}

And we can see
\begin{align*}
&2\paren{\log_2(R' + 1)+1}\paren{\ninner+3+\log_2(\nsimple)+\log_2(\ncomplex + 1)} \\
&\;\;\;\;\leq 2\paren{R'+1}\paren{3+\ninner+\nsimple+\ncomplex}\\
&\;\;\;\;\leq 48 \ninner \nsimple \ncomplex R'
\end{align*}
since $\ninner, \nsimple, \ncomplex, R' \geq 1$ and hence we can conclude
\begin{align*}
\length{\calN_{i+1}}{s_{\text{cov}}} 
&\leq 2^{2^{\paren{\paren{\max\set{\length{\calN_i}{s_{\text{cov}}},2}}^{48 \ninner \nsimple {\ncomplex}{R'}}}}}.
\end{align*}

We can prove claim (iii) by a simple induction.
The base case is given by (i).
The inductive step is:
\begin{align*}
\length{\calN_{i+1}}{s_{\text{cov}}} &\leq 2^{2^{\paren{\max\set{\length{\BTS{\calN_i}{s_{\text{cov}}},2}}^{48 \ninner \nsimple \ncomplex R'}}}}\\
&\leq 2^{2^{\paren{\tetration[2]{(2\cdot i +2 \slog(48(i+1) \ninner \nsimple \ncomplex {R'}))}}^{48 \ninner \nsimple {\ncomplex}{R'}}}}\\ 
&\leq \tetration[2]{(2\cdot (i+1) +2 \slog(48(i+2) \ninner \nsimple {\ncomplex}{R'}))}
\end{align*}
which concludes the proof.
\end{proof}

\begin{customcorollary}[\ref{cor:cov_nnct_tower}]
Coverability for NNCTs is decidable and in \Tower.
\end{customcorollary}
\begin{proof}
Let $\calN$ be an NNCT with $\nsimple$ places, $\ncomplex$ complex places, $\ninner$ colours and rules $\calR$ and a coverability query $\mathpzc{Q}$ giving rise to a bound $R'$ and
$$B(\nsimple,\ncomplex,\ninner,R') = \tetration[2]{(2\nsimple +2 \slog(48(\nsimple+1) \discn \ninner \nsimple {\ncomplex}{R'}))}.$$
Then the theorem above tells
us that along a covering path a simple place cannot contain more than $R \cdot B(\nsimple,\ncomplex,\ninner,R')$ tokens,
a complex tokens cannot have more that $R \cdot B(\nsimple,\ncomplex,\ninner,R')$ tokens of a particular colour and
there are at most $R \cdot B(\nsimple,\ncomplex,\ninner,R')$ complex tokens. Hence along a covering path simple place 
can be represented using $\log_2(R \cdot B(\nsimple,\ncomplex,\ninner,R'))$ bits, the state of complex token can be represented using $\ninner \log_2(R \cdot B(\nsimple,\ncomplex,\ninner,R'))$ bits and hence a complex place may be represented by $\ncomplex \ninner R \cdot B(\nsimple,\ncomplex,\ninner,R') \log_2(R \cdot B(\nsimple,\ncomplex,\ninner,R'))$ bits. Hence a non-deterministic Turing machine requiring at most $O(B(\nsimple,\ncomplex,\ninner,R'))$ space can decide the coverability problem. Using Savitch's theorem we know there is a deterministic turing machine deciding coverability in $O(B(\nsimple,\ncomplex,\ninner,R')^2)$ space and using an exponential to obtain a time bounded turing machine we find that the coverability problem can be decided in time $O(2^{B(\nsimple,\ncomplex,\ninner,R')^2})$ and is thus clearly in \Tower.

\end{proof}        
\newpage
\subsection{Proof of Theorem \ref{thm:nnct:coverability:tower_hard}}\label{sec:lowerbound:full}

\begin{customtheorem}[\ref{thm:nnct:coverability:tower_hard}]
Coverability for a simple query for total-transfer NNCT is \Tower-hard.
\end{customtheorem}
\begin{proof}
We can deduce \Tower-hardness by showing that given a deterministic \changed[jk]{bounded} two-counter machine,
 \changed[jk]{$\calM$,
 of size $n$ with counters that are $(\tetration[2]{n})$-bounded }
\jk{I think the size of a Minsky machine is defined to be cardinality of the set of instructions or equivalently the number of control states.}
 we can construct an NNCT $\calN_{\calM}$ in polynomial-time that weakly bisimulates $\calM$ in such a way that we can reduce the halting problem for $\calM$ to coverability for a simple query for $\calN_{\calM}$. The machine $\calM$ can use the following operations: $\inc{x}$, $\dec{x}$, $\reset{x}$, $\iszero{x}$, $\ismax{x}$ for each counter $x$.

Each simulation state of $\calN_{\calM}$ will represent a valuation $v$ of $6n+2$ \emph{active} and \emph{inactive} counters, and $n$ arrays. In addition to the counters $x$, $y$ of $\calM$ the NNCT $\calN_{\calM}$ will simulate the auxiliary counters $s_i$, $p_i$, $p'_i$,  $c_i$, $c'_i$ and an auxiliary array $a_i$ for each $i \leq n$. 
Each active counter $d \in \set{s_i, p_i, p'_i, c_i, c'_i} \union \set{x,y \mid i = n}$ is $(\tetration[2]{i})$-bounded, each inactive counter has an undefined value. 
For each $i$ the array has length exactly $(\tetration[2]{i})+1$ and carries values $a_i(j) \in \set{0,1,2}$ for $0 \leq j \leq \tetration[2]{i}$. The NNCT $\calN_{\calM}$ will have two simple places $\encp{d}$ and $\encp[\overline]{d}$ for each counter $d \in \set{s_i, p_i, c_i, c'_i} \union \set{x,y \mid i = n}$, three complex places $\encp{a_i}$ and $\encp[\overline]{a_i}$, $\encp{\mathit{aux}_i}$ and two colours $\encp{j_i}$ and $\encp[\overline]{j_i}$ for each array $a_i$, and a complex ``sink'' place $\encp{\mathit{disc}}$; $\colmap$
maps $\encp{j_i}$ to $\encp{p'_i}$ and $\encp[\overline]{j_i}$ to $\encp[\overline]{p'_i}$,
in addition to a (polynomial in $n$) number of simple places encoding the control of $\calM$ and the ``internal'' control of $\calN_{\calM}$. \changed[jk]{Further $\calN_{\calM}$'s transfer rules will all be total, hence $\calN_{\calM}$ will be a total-transfer NNCT.}
A valuation $v$ is represented by a configuration $s$ as follows:
\begin{itemize}
\item For each $i$ and $d \in \set{s_i, p_i, p'_i, c_i, c'_i} \union \set{x,y \mid i = n}$ if $d$ is active then there are exactly $v(d)$ $\bullet$-tokens in $\encp{d}$ and $\tetration[2]{i} - v(d)$ $\bullet$-tokens in $\encp[\overline]{d}$.
\item For each $i$ and $d \in \set{s_i, p_i, p'_i, c_i, c'_i} \union \set{x,y \mid i = n}$ if $d$ is inactive then both places $\encp{d}$, $\encp[\overline]{d}$ are empty.
\item For each $i$ the array $a_i$ is represented as follows: for each $0 \leq k \leq \tetration{i}$
let $m_{(i,k)}$ be a complex token such that $m_{(i,k)}$ contains \changed[jk]{exclusively tokens of colours $\encp{j_i}$ and $\encp[\overline]{j_i}$: $k$ tokens of colour $\encp{j_i}$ and $\tetration{i} - k$ tokens of colour  $\encp[\overline]{j_i}$};
then there are exactly $v(a_i(k))$ complex tokens $m_{(i,k)}$ in $\encp{a_i}$ and
$2-v(a_i(k))$ complex tokens $m_{(i,k)}$ in $\encp[\overline]{a_i}$.
\item For each $i$ the place $\encp{\mathit{aux}_i}$ is empty.
\end{itemize}

The question whether $\calM$ halts, i.e.~whether $\calM$ reaches a halting control state from its initial state, can then be answered by performing a coverability query on $\calN_{\calM}$ where the query marking involves only the simple places of $\calN_{\calM}$ encoding $\calM$'s finite control. Assuming that only $\calM$'s halting control states\jk{abnormal termination?} have no successors, $\calM$'s halting problem also reduces to 
the termination problem for $\calN_{\calM}$. And we can augment $\calN_{\calM}$ with an additional simple place that is incremented with every transition so that the halting problem for $\calM$ reduces to deciding boundedness of $\calN_{\calM}$. \jk{It may make sense to introduce complex token removal as otherwise the number of complex tokens always increases.}

Let $d \in \set{s_i, p_i, p'_i, c_i, c'_i, x,y \mid 1 \leq i \leq n}$ we will implement operation
$\inc{d}$ as follows
\begin{quote}
\begin{tabular}{lp{0.5eM}p{0.5eM}l}
	\mathline[{&}]{3}{\text{\em Add a $\bullet$-token to $\encp{d}$}}\\
	\mathline[{&}]{3}{\text{\em Remove a $\bullet$-token from $\encp[\overline]{d}$}}
\end{tabular}
\end{quote}
Suppose $d$ is $\tetration[2]{i}$-bounded. Clearly if configuration $s$ represents valuation $v$, $d$ is active and $v(d) < \tetration[2]{i}$ then 
there are $v(d)$ $\bullet$-tokens in $\encp{d}$ and $\tetration[2]{i} - v(d)$ $\bullet$-tokens in $\encp[\overline]{d}$.
Hence there is at least one $\bullet$-token in $\encp[\overline]{d}$.
Performing the operation $\inc{d}$ yields a new configuration $s'$ where there are $v(d)+1$ $\bullet$-tokens in $\encp{d}$ and $\tetration[2]{i} - (v(d)+1)$ $\bullet$-tokens in $\encp[\overline]{d}$ and all other (non-control) places are unchanged. The configuration $s'$ then represents the valuation $\update{v}{d}{v(d)+1}$. Further if 
$v(d) = \tetration[2]{i}$ we note $\encp[\overline]{d}$ is empty and thus the simulation of $\inc{d}$ blocks at the attempt to remove a token from $\encp[\overline]{d}$.

We note that we can implement $\dec{d}$ by
\begin{quote}
\begin{tabular}{lp{0.5eM}p{0.5eM}l}
	\mathline[{&}]{3}{\text{\em Remove a $\bullet$-token from $\encp{d}$}}\\
	\mathline[{&}]{3}{\text{\em Add a $\bullet$-token to $\encp[\overline]{d}$}}
\end{tabular}
\end{quote}
Similarly to the reasoning on $\inc{d}$ we can see that if configuration $s$ represents valuation $v$, $d$ is active and $v(d) > 0$ we can successfully simulate $\dec{d}$ and obtain a configuration $s'$ that represents the valuation $\update{v}{d}{v(d)-1}$. However if $v(d) = 0$ simulating $\dec{d}$ will get stuck.

In addition to the operations $\calM$ supports, we will implement further instructions to simplify and improve readability.
The NNCT $\calN_{\calM}$ will simulate $\activate{d}$, $\deactivate{d}$ for $d \in \set{s_i, p_i, p'_i, c_i, c'_i, x,y \mid 1 \leq i \leq n}$, and operations
$\reset{d}$, $\iszero{d}$, $\ismax{d}$ for $d \in \set{ p_i, p'_i, c_i, c'_i, x,y \mid 1 \leq i \leq n}$.
Further we implement \changed[jk]{the operation $\ismaxandreset{s_i}$} and the counter-specialised operations:
\begin{quote} 
 $\isequal{p_i}{p'_i}$,  $\inc{a_i(p_i)}$, $\dec{a_i(p_i)}$, $\reset{a_i(p_i)}$, \newline $\iszero{a_i(p_i)}$, $\ismax{a_i(p_i)}$, $\activate{a_i}$
\end{quote} 
for each array $a_i$ and counter $p_i$. The latter counter operations will be used to implement $\ismaxandreset{s_i}$. All above operations are only guaranteed to succeed if the counter in question is active at the start of the operation. 

Counters $s_1$, $p_1$, $p'_1$, $c_1$, $c'_1$ are $2$-bounded so implementing operations on them is trivial.
For $i < n$, operations on $a_i$ are simulated using $s_{i}$, $p_{i}$, $p'_{i}$, $c_{i}$, $c'_{i}$ and
operations on $s_{i+1}$, $p_{i+1}$, $p'_{i+1}$, $c_{i+1}$, $c'_{i+1}$ are simulated using operations on 
$p_{i}$, $p'_{i}$, $c_{i}$, $c'_{i}$ and $a_i$. 

\lo{I would propose that we aim to devote about a page to the lower bound proof. We can relegate the simulation of at least half of the 8 operations in the following to the appendix. I have no doubt that instead of 2.5 pages of lower-bound proof (with as many as 8 cases), POPL reviewers would much prefer to see a partial (with the rest in the appendix) but ``representative'' proof plus additional analyses (such as the connections between the three models we have discussed) and examples.}
\jk{I agree. Some implementations are more obvious than others. I placed everything here so it would be easier for you to read the entire argument. As it stands do you think the argument is sufficient for a proof? Which implementations do you think can be relegated to the appendix? 
In my opinion the following are obvious:
\begin{itemize}
\item $\iszero{a_i(p_i)}, \ismax{a_i(p_i)}$,
\item $\reset{a_i(p_i)}$,
\item $\reset{d}$,
\item $\isequal{p_i,p'_i}$,
\item the details of the initialisation and
\item possibly $\inc{d}$, $\dec{d}$?
\end{itemize}
I think this would already take us to about 1.5 pages. It may be that the remainder of the argument can be trimmed.
}

\begin{asparaenum}[(i)]
\item The following shows how to implement $\inc{a_i(p_i)}$ and can only succeed if $p_i, p'_i$ are active.

\noindent
\begin{tabular}{lp{1eM}p{1eM}l}
	\tline[{&}]{3}{Move a complex-token from $\encp[\overline]{a_i}$ to $\encp{\mathit{aux}_i}$;}{7cm}\\
	\mathline[{&}]{3}{\deactivate{p'_i};}\\
	\tline[{&}]{3}{Eject the contents of a complex-token $m$ in $\encp{\mathit{aux}_i}$ and its remains into $\encp{\mathit{disc}}$;}{7cm}\\
	\mathline[{&}]{3}{\isequal{p_i}{p'_i};\; \reset{p'_i};}\\
	\mathline[{&}]{3}{\text{\em Create an empty complex-token in $\encp{\mathit{aux}_i}$;}}\\
	\mathline[{&}]{3}{\textbf{while }\paren{p_i \neq p'_i} \textbf{ do}}\\
		\tline[{&&}]{2}{Inject a $\encp{j_i}$-coloured $\bullet$-token into a complex token in $\encp{\mathit{aux}_i}$;}{6cm}\\
		\mathline[{&&}]{2}{\inc{p'_i};}\\
	\mathline[{&}]{3}{\textbf{while }\paren{\neg(\ismax{p'_i})} \textbf{ do}}\\
		\tline[{&&}]{2}{Inject a $\encp[\overline]{j_i}$-coloured $\bullet$-token into a complex token in $\encp{\mathit{aux}_i}$;}{6cm}\\
		\mathline[{&&}]{2}{\inc{p'_i};}\\
	\mathline[{&}]{3}{\text{\em Move a complex-token from $\encp{\mathit{aux}_i}$ to $\encp{a_i}$;}}\\
	\mathline[{&}]{3}{\reset{p'_i};}\\
\end{tabular}

Suppose $\inc{a_i(p_i)}$ is executed in a configuration $s$ that represents valuation $v$ and $p_i, p'_i$ are active.
If $v(a_i(v(p_i))) < 2$ then we know there exists a complex token $m_{(i,v(p_i))}$ in $\encp[\overline]{a_i}$ and all complex tokens in $\encp[\overline]{a_i}$ are of the form $m_{(i,k)}$ for some $0 \leq k \leq \tetration[2]{i}$.
The move of one $m_{(i,k)}$ complex token from $\encp[\overline]{a_i}$ to $\encp{\mathit{aux}_i}$ results in
$\encp{\mathit{aux}_i}$ containing just $m_{(i,k)}$, since by assumption $\encp{\mathit{aux}_i}$ was empty before.
After deactivating $p'_i$ we know that both $\encp{p'_i}$  and $\encp[\overline]{p'_i}$ are empty. Ejecting
the contents of $m_{(i,k)}$ removes $m_{(i,k)}$ from $\encp{\mathit{aux}_i}$, inserts $k$ $\bullet$-tokens into $\encp{p'_i}$ and $(\tetration[2]{i}) - k$ $\bullet$-tokens into $\encp[\overline]{p'_i}$ and 
places the remaining (now empty) complex token into $\encp{\mathit{disc}}$. We can see that, disregarding
$a_i$, the configuration we have reached represents a partial valuation $v'$ that sets $v'(p'_i) = k$ and $v'(p_i) = v(p_i)$. After executing $\isequal{p_i}{p'_i}$ the simulation only succeeds if $k = v(p_i)$.
Hence $\encp[\overline]{a_i}$ now contains $2-(v(a_i(v(p_i)))+1)$ complex tokens $m_{(i,v(p_i))}$ and the same number of other complex tokens $m_{(i,k)}$.
Clearly the two following while loops carefully inject $\encp{j_i}$-coloured $\bullet$-tokens and 
$\encp[\overline]{j_i}$-coloured $\bullet$-tokens into the newly created
token at $\encp{\mathit{aux}_i}$ to yield a new $m_{(i,v(p_i))}$ located in $\encp{\mathit{aux}_i}$
which we move to $\encp{a_i}$. 
Thus $\encp{a_i}$ now contains $v(a_i(v(p_i)))+1$ complex tokens $m_{(i,v(p_i))}$ and the same number of other complex tokens $m_{(i,k)}$ as before. The configuration $s'$ we have reached thus represents a valuation $v''$ such that
$v''(a_i(j)) = v(a_i(j))$ for all $0 \leq j \leq \tetration[2]{i}$ and $j \neq v(p_i)$ and $v''(a_i(v(p_i))) = v(a_i(v(p_i)))+1$.
Otherwise if $v(a_i(v(p_i))) = 2$ then the simulation either blocks while attempting to move a complex token
$m_{(i,k)}$ from $\encp[\overline]{a_i}$ to $\encp{\mathit{aux}_i}$ or on the execution of $\isequal{p_i}{p'_i}$
since it is impossible to obtain $k = v(p_i)$.

By swapping $\encp{a_i}$ and $\encp[\overline]{a_i}$ in the above we obtain an implementation of $\dec{a_i(p_i)}$.

\item The following shows how to implement $\activate{a_i}$ and can only succeed if $p_i$ and $p'_i$ are active.

\noindent
\begin{tabular}{lp{1eM}p{1eM}p{1eM}l}
	\mathline[{&}]{4}{\textbf{for } p_i \is 0 \textbf{ to } \tetration[2]{i} \textbf{ do }}\\
	\mathline[{&&}]{3}{\textbf{for } z \is 0 \textbf{ to } 1 \textbf{ do }}\\
	\mathline[{&&&}]{2}{\text{\em Create an empty complex-token in $\encp{\mathit{aux}_i}$;}}\\
	\mathline[{&&&}]{2}{\textbf{while }\paren{p_i \neq p'_i} \textbf{ do}}\\
		\tline[{&&&&}]{1}{Inject a $\encp{j_i}$-coloured $\bullet$-token into a complex token in $\encp{\mathit{aux}_i}$;}{5.25cm}\\
		\mathline[{&&&&}]{1}{\inc{p'_i};}\\
	\mathline[{&&&}]{2}{\textbf{while }\paren{\neg(\ismax{p'_i})} \textbf{ do}}\\
		\tline[{&&&&}]{1}{Inject a $\encp[\overline]{j_i}$-coloured $\bullet$-token into a complex token in $\encp{\mathit{aux}_i}$;}{5.25cm}\\
		\mathline[{&&&&}]{1}{\inc{p'_i};}\\
	\mathline[{&&&}]{2}{\text{\em Move a complex-token from $\encp{\mathit{aux}_i}$ to $\encp[\overline]{a_i}$;}}\\
	\mathline[{&&&}]{2}{\reset{p'_i};}\\
\end{tabular}

We first note that the interior for-loop simply repeats its body twice and can thus be considered syntactic sugar.
Suppose $\activate{a_i}$ is invoked in a configuration $s$ such that in $s$ the places $\encp{a_i}$ and $\encp[\overline]{a_i}$ are empty. We can then follow our analysis in the second part of $\inc{a_i(p_i)}$'s implementation
to see that the interior for-loop places two complex tokens $m_{i,k}$ into $\encp[\overline]{a_i}$ for all $k \in 0,\ldots, \tetration[2]{i}$. Hence when the outer for-loop terminates we have reached a configuration representing a valuation $v$ such that $v(a_i(k)) = 0$ for all $0 \leq k \leq \tetration[2]{i}$.

\item The following shows how to implement $\iszero{a_i(p_i)}$.

\begin{tabular}{lp{1eM}p{1eM}l}
	\mathline[{&}]{3}{\inc{a_i(p_i)};\; \inc{a_i(p_i)};}\\
	\mathline[{&}]{3}{\dec{a_i(p_i)};\; \dec{a_i(p_i)};}
\end{tabular}

Suppose $\iszero{a_i(p_i)}$ is executed in a configuration $s$ that represents valuation $v$ and $p_i, p'_i$ are active. If $v(a_i(v(p_i))) = 0$ we can clearly execute $\inc{a_i(p_i)}$ twice and undo the former by executing $\dec{a_i(p_i)}$ twice.
We thus end up at a configuration representing $v$. If however $v(a_i(v(p_i))) > 0$ then 
either $v(a_i(v(p_i))) = 2$ initially or after the first invocation of $\inc{a_i(p_i)}$ and
our analysis above assures us
that the subsequent invocation of $\inc{a_i(p_i)}$ blocks. Hence $\iszero{a_i(p_i)}$ only succeeds if $v(a_i(v(p_i))) = 0$.
We obtain an implementation of $\ismax{a_i(p_i)}$ by swapping $\inc{a_i(p_i)}$ for $\dec{a_i(p_i)}$ and vice versa.
Analogous reasoning to above then yields that $\iszero{a_i(p_i)}$ only succeeds if $v(a_i(v(p_i))) = 2$.

\item The following shows how to implement $\reset{a_i(p_i)}$.

\noindent
\begin{tabular}{lp{1eM}p{1eM}l}
	\mathline[{&}]{3}{\textbf{while }\paren{a_i(p_i) \neq 0} \textbf{ do}}\\	
		\mathline[{&&}]{2}{\dec{a_i(p_i)}}
\end{tabular}

It is trivial to see that the while loop only terminates once a configuration is reached representing a valuation $v$
that sets $v(a_i(v(p_i))) = 0$. In every iteration $v(a_i(v(p_i)))$ is decremented by $1$, so termination is ensured.

\item The following shows how to implement $\ismaxandreset{s_{i+1}}$.

\noindent 
\begin{tabular}{lp{0.5eM}p{0.5eM}l}
	\mathline[{&}]{3}{\textbf{for } p_i \is 0 \textbf{ to } \tetration[2]{i} \textbf{ do }}\\
		\mathline[{&&}]{2}{\reset{a_i(p_i)}}\\
	\mathline[{&}]{3}{\textbf{while }(\iszero{a_i(\tetration[2]{i})}) \textbf{ do}}\\	
		\mathline[{&&}]{2}{\dec{s_{i+1}};\;\reset{p_i};\;\inc{a_i(p_i)};}\\
		\mathline[{&&}]{2}{\textbf{while } \ismax{a_i(p_i)} \textbf{ do }}\\
			\mathline[{&&&}]{1}{\dec{a_i(p_i)};\;\inc{p_i};\;\inc{a_i(p_i)};}
\end{tabular}

We know that after the for-loop $a_i$ is the array such that $a_i(x) = 0$ for all $0 \leq x \leq \tetration[2]{i}$.
The array $a_i$ is meant to be binary representation of a number between $0$ and $\tetration[2]{(i+1)}$. This number is initially $0$ and we can see that the outer while loop performs a long addition of $1$ for each iteration. If $v(a_i(v(p_i))) = 2$ then $v(p_i)$ is an index representing a carry bit in the long addition computation. Clearly for each
number represented by $a_i$ we have perform $\dec{s_{i+1}}$. Hence if initially $v(s_{i+1}) = \tetration[2]{(i+1)}$
then after performing $\ismaxandreset{s_{i+1}}$ it is the case that the resulting valuation $v'(s_{i+1}) = 0$, and
if $v(s_{i+1}) < \tetration[2]{(i+1)}$ then after $v(s_{i+1})$ iterations the resulting valuation $v'$ would set $v'(s_{i+1}) = 0$ and $a_i$ would represent the number $v(s_{i+1})$. Since $v(s_{i+1}) < \tetration[2]{(i+1)}$ this implies that
$a_i(\tetration[2]{i}) = 0$ and hence body of the outer while loop is executed again leading to an invocation of
$\dec{s_{i+1}}$ which will block. Hence $\ismaxandreset{s_{i+1}}$ will block when executed in a configuration representing $v$ such that $v(s_{i+1}) < \tetration[2]{(i+1)}$.

We modify the implementation of $\ismaxandreset{s_{i+1}}$ by replacing
$\dec{s_{i+1}}$ by {\em Add a $\bullet$-token to $\encp[\overline]{d}$} we obtain an implementation of 
$\activate{d}$.
If $d$ is inactive then $\encp{d}$ and $\encp[\overline]{d}$ are empty. By an analogous argument to above we can see that invoking $\activate{d}$ will add $\tetration[2]{(i+1)}$ $\bullet$-token to $\encp[\overline]{d}$ yielding a configuration where $d$ is active and is valued $0$.
We can similarly obtain an implementation for $\deactivate{d}$.
First modify $\ismaxandreset{s_{i+1}}$ by replacing $\dec{s_{i+1}}$ by
{\em Remove a $\bullet$-token from $\encp[\overline]{d}$} which gives us an implementation of an intermediate operation
$\deactivate[']{d}$. If $d$ is active and valued $0$ then $\deactivate[']{d}$ clearly succeeds and removes all tokens from $\encp[\overline]{d}$. Hence we can implement $\deactivate{d}$ by $\reset{d};\; \deactivate[']{d};$ which  succeeds if $d$ is active.

\item The following shows how to implement $\iszero{d}$ for $d \in \set{p_{i+1}, p'_{i+1}, c_{i+1}, c'_{i+1}} \union \set{x,y \mid i+1 = n}$
if $s_i$ is $0$. 

\noindent
\begin{tabular}{lp{1eM}p{1eM}l}
	\mathline[{&}]{3}{\textbf{while }(*) \textbf{ do}}\\	
		\mathline[{&&}]{2}{\inc{d}; \inc{s_{i+1}};}\\
	\mathline[{&}]{3}{\ismaxandreset{s_{i+1}}}\\	
	\mathline[{&}]{3}{\textbf{while }(*) \textbf{ do}}\\	
		\mathline[{&&}]{2}{\dec{d}; \inc{s_{i+1}};}\\
	\mathline[{&}]{3}{\ismaxandreset{s_{i+1}}}
\end{tabular}

Suppose $\iszero{d}$ is started in a configuration $s_0$ that represents valuation $v_0$ such that $v_0(s_i) = 0$.
After the first while-loop non-deterministically terminates, just before invoking $\ismaxandreset{s_{i+1}}$ we reach a configuration that represents a valuation $v$ such that $v(s_{i+1}) = k$ and $v(d) = v_0(d) + k$. Since
neither $\inc{d}$ nor $\inc{s_{i+1}}$ blocked we know that $0 \leq k \leq \tetration[2]{(i+1)}-v_0(d)$. In fact for all
$0 \leq k \leq \tetration[2]{(i+1)}-v_0(d)$ it is possible to reach such a configuration where $v(s_{i+1}) = k$ and $v(d) = v_0(d) + k$.
Invoking $\ismaxandreset{s_{i+1}}$ only succeeds if $k = \tetration[2]{(i+1)}$. 
Thus all other configurations where $k < \tetration[2]{(i+1)}$ $\ismaxandreset{s_{i+1}}$ blocks. Hence if
$v_0(d) > 0$ it will be the case that $k < \tetration[2]{(i+1)}$ and thus $\ismaxandreset{s_{i+1}}$ blocks.
If $v_0(d) = 0$ then there is one configuration that we can reach which represents a valuation $v_1$
such that $v_1(s_{i+1}) = \tetration[2]{(i+1)}$, $v_1(d) = v_0(d) + \tetration[2]{(i+1)} = \tetration[2]{(i+1)}$
and we can successfully invoke $\ismaxandreset{s_{i+1}}$ to yield a new configuration representing a valuation $v_2$
$v_2(s_{i+1}) = 0$, $v_2(d) = \tetration[2]{(i+1)}$.
The second while loop then again non-deterministically decrements $d$ and increments $s_{i+1}$ in the same fashion as above. Again similar reasoning to above yields that there is only one configuration that we can reach that represents a valuation $v_3$ which values
$v_3(s_{i+1}) = \tetration[2]{(i+1)}$ and consequently $v_3(d) = 0$ and which guarantees $\ismaxandreset{s_{i+1}}$ succeeds to yield a configuration that represents a valuation $v_4$ such that
$v_4(s_{i+1}) = v_4(d) = 0$. Hence $\iszero{d}$ succeeds only if started in a configuration that values $d$ as $0$.

We can modify the above to obtain an implementation $\ismax{d}$. We simply swap the first for the second while loop
and analogous reasoning to $\iszero{d}$ then yields that $\ismax{d}$ succeeds only when $d$ is valued at 
$\tetration[2]{(i+1)}$ and $s_{i+1}$ at $0$.

\item The following shows how to implement $\isequal{p_i}{p'_i}$ for $d \in \set{p_i, p'_i, c_i, c'_i} \union \set{x,y \mid i = n}$
if both $p_i$ and $p'_i$ are active and $c_i$ is zero.

\noindent
\begin{tabular}{lp{1eM}p{1eM}l}
	\mathline[{&}]{3}{\textbf{while }\paren{p_i \neq 0 \textbf{ or } p'_i \neq 0} \textbf{ do}}\\	
		\mathline[{&&}]{2}{\dec{p_i};\;\dec{p'_i};\;\inc{c_i};}\\
	\mathline[{&}]{3}{\textbf{while }\paren{c_i \neq 0} \textbf{ do}}\\	
		\mathline[{&&}]{2}{\inc{p_i};\;\inc{p'_i};\;\dec{c_i};}
\end{tabular}

If we invoke $\isequal{p_i}{p'_i}$ in a configuration that represents a valuation $v$ such that $v(p_i) = k$, 
$v(p'_i) = k'$ and $v(c_i) = 0$ then for each iteration of the first while loop
we obtain a configuration representing a valuation $v'$ such that 
$v'(p_i) = k - v'(c_i)$, 
$v'(p'_i) = k' - v'(c_i)$. The inequality $v'(c_i) \leq \min(k,k')$ must hold at the end of each iteration as otherwise either $\dec{p_i}$ or $\dec{p'_i}$ would have blocked. Since the loop-guard is false only if
both $v'(p_i) = 0$ and $v'(p'_i) =  0$ we can see that the while loop only terminates successfully
once $v'(c_i) = k = k'$ i.e.~if initially $v(p_i) = v(p'_i)$. Otherwise the loop guard would never be false and
thus at some point the invocation of $\dec{p_i}$ or $\dec{p'_i}$ blocks.
Thus after the successful breaking out of the while loop $v'(c_i) = v(p_i) = v(p'_i)$ and it is easy to see
that the second while loop simply copies the contents of $c_i$ back to $p_i$ and $p'_i$.
Hence invoking $\isequal{p_i}{p'_i}$ only succeeds if started in a configuration where
$v(p_i) = v(p'_i)$ and $v(c_i) = 0$.

\item The following shows how to implement $\reset{d}$ for $d \in \set{p_i, p'_i, c_i, c'_i} \union \set{x,y \mid i = n}$.

\noindent
\begin{tabular}{lp{1eM}p{1eM}l}
	\mathline[{&}]{3}{\textbf{while }\paren{d \neq 0} \textbf{ do}}\\	
		\mathline[{&&}]{2}{\dec{d}}
\end{tabular}

As in the implementation of  $\reset{a_i(p_i)}$ the while-loop simply decrements $d$ until $d$ has value 0.
Hence $\reset{d}$ works as expected.
\end{asparaenum}
In order to set up a configuration that simulates the initial configuration of $\calM$ we set up $\calN_{\calM}$ where all non-control places are empty and $\calN_{\calM}$'s finite control initiates the execution of

\noindent\begin{tabular}{lp{1eM}p{1eM}l}
	\mathline[{&}]{3}{\activate{s_1};\; \activate{p_1};\; \activate{p'_1};\; }\\
	\mathline[{&}]{3}{\activate{c_1};\; \activate{c'_1};\; \activate{a_1};}\\	
	\mathline[{&}]{3}{\cdots}\\	
	\mathline[{&}]{3}{\activate{s_{n-1}};\; \activate{p_{n-1}};\; \activate{p'_{n-1}};\; }\\
	\mathline[{&}]{3}{\activate{c_{n-1}};\; \activate{c'_{n-1}};\; \activate{a_{n-1}};}\\	
	\mathline[{&}]{3}{\activate{s_{n}};\; \activate{p_{n}};\; \activate{p'_{n}};\; \activate{c_{n}};\;}\\
	\mathline[{&}]{3}{\activate{c'_{n}};\; \activate{x};\; \activate{y};\;}\\	
\end{tabular}
After executing these operations clearly we reach a configuration $s$ in which all counters are active and which represents a valuation $v$ such that $v(d) = 0$ for all counters $d$ and $v(a_i(k)) = 0$ for all $0 \leq k \leq \tetration[2]{i}$.

\todo{
	\begin{itemize}
		\item explain reduction relation ?
	\end{itemize}
}

\end{proof}          
\fi
\end{document}